\documentclass [a4paper,11pt,twoside]{book}
  \usepackage{amsthm}
  \usepackage{amssymb}
  \usepackage{amsmath}
\usepackage[utf8]{inputenc} 
\usepackage[T2A,T1]{fontenc}
\usepackage[russian,polish,english]{babel} 
\usepackage{framed}
\usepackage[Glenn]{fncychap}
\ChTitleVar{\large}
\ChNameVar{\large\rm}
\ChNumVar{\bf\LARGE}
\makeatletter
 \renewcommand{\DOCH}{%
    \vskip -8em 
    \settoheight{\myhi}{\CTV\FmTi{Test}}
    \setlength{\py}{\baselineskip}
    \addtolength{\py}{\RW}
    \addtolength{\py}{\myhi}
    \setlength{\pyy}{\py}
    \addtolength{\pyy}{-1\RW}
    \raggedright
    \CNV\FmN{\@chapapp}\space\CNoV\thechapter
    \hskip 3pt\mghrulefill{\RW}\rule[-1\pyy]{2\RW}{\py}\par\nobreak
    }
 \renewcommand{\DOTIS}[1]{%
    \vskip -3em 
    \setlength{\py}{25pt}
    \setlength{\pyy}{\py}
    \setlength{\backskip}{\py}
    \addtolength{\backskip}{2pt}
    \addtolength{\pyy}{\RW}
    \setlength{\myhi}{\baselineskip}
    \addtolength{\myhi}{\pyy}
    \mghrulefill{\RW}\rule[-1\py]{2\RW}{\pyy}\par\nobreak
    \vskip -1\backskip
    \rule{2\RW}{\myhi}\mghrulefill{\RW}\hskip 3pt %
    \raggedleft\CTV\FmTi{#1}\par\nobreak
    \vskip 32\p@
    }  
\def\@part[#1]#2{
    \ifnum \c@secnumdepth >\m@ne
      \refstepcounter{part}
      \addcontentsline{toc}{part}{Part \thepart\hspace{0.25em} -- \hspace{0.25em}#1}
    \else
      \addcontentsline{toc}{part}{Part #1}
    \fi
    {\parindent \z@
     \interlinepenalty \@M
     \normalfont
     \ifnum \c@secnumdepth >\m@ne
       \huge\hfil\textsc{\partname~\textbf{\thepart}}\hfil   
       \par\nobreak\vspace{0.5\baselineskip}
     \fi
     \Huge \centering \textsc{\textbf{#2}}
     \markboth{}{}\par}
    \nobreak
    \vskip 3ex
    \@afterheading}
\makeatother
  \usepackage{enumerate}
  \usepackage{amscd}
  \usepackage{accents}
    \usepackage{mathtools}
  \usepackage{eucal}
  \usepackage{mathrsfs}
  \usepackage{upgreek}
   \usepackage{dsfont}
\usepackage{yfonts}
  \usepackage{bbm}
\usepackage{geometry}
\usepackage{array}
\usepackage{mathtools}
\usepackage{booktabs}
\usepackage{pgf}
\usepackage{tikz} 
\usetikzlibrary{matrix,arrows,shapes}
\usepackage{makeidx}
\usepackage[a4paper,hyperindex=true]{hyperref}
\hypersetup
{
bookmarks=true,
bookmarksopen=true,
pdftitle="BV formalism",
pdfauthor="Katarzyna Rejzner", 
pdfsubject="BV formalism", 
pdfmenubar=true, 
pdfhighlight=/O, 
colorlinks=true, 
pdfpagemode=None, 
pdfpagelayout=SinglePage, 
pdffitwindow=true, 
linkcolor=see, 
citecolor=red!75!darkgray, 
urlcolor=see 
}
\makeindex
\usepackage{fancyhdr}
\fancypagestyle{plain}{%
\fancyhf{}

}
 
\makeatletter
\def\cleardoublepage{\clearpage\if@twoside \ifodd\c@page\else
\hbox{}
\thispagestyle{empty}
\newpage
\if@twocolumn\hbox{}\newpage\fi\fi\fi}
\makeatother

\makeatletter
\newcommand\figcaption{\def\@captype{figure}\caption}
\newcommand\tabcaption{\def\@captype{table}\caption}
\makeatother
\newcommand\xqed[1]{%
  \leavevmode\unskip\penalty9999 \hbox{}\nobreak\hfill
  \quad\hbox{#1}}
\newcommand\eex{\xqed{$\blacktriangle$}}
\newcommand{\Tens}{\mathfrak{Tens}}
\newcommand{\CE}{\mathfrak{CE}}  
\newcommand{\E}{\mathfrak{E}}
\newcommand{\sm}{\mathfrak{s}}
\newcommand{\sr}{\mathfrak{s}_{\textrm{r}}}
\newcommand{\smp}{\mathfrak{s}_\p}
\newcommand{\V}{\mathfrak{V}}

\newcommand{\Ba}{\mathfrak{Ba}}
\newcommand{\fA}{\mathfrak{A}}
\newcommand{\fa}{\mathfrak{a}}
\newcommand{\D}{\mathfrak{D}}
\newcommand{\Cc}{\mathfrak{C}}
\newcommand{\F}{\mathfrak{F}}
\newcommand{\fG}{\mathfrak{G}}
\newcommand{\W}{\mathfrak{W}}

\newcommand{\KT}{\mathfrak{KT}}

\newcommand{\BV}{\mathfrak{BV}}
\newcommand{\Nm}{\mathfrak{Nm}}

\newcommand{\frakgc}{\mathfrak{g}_c}
\newcommand{\frakg}{\mathfrak{g}}
\newcommand{\X}{\mathfrak{X}}

\newcommand{\euA}{\mathscr{A}}  
\newcommand{\euF}{\mathscr{F}}  
\newcommand{\A}{\mathcal{A}}  
\newcommand{\Gcal}{\mathcal{G}}  

\newcommand{\Ccal}{\mathcal{C}}
\newcommand{\Dcal}{\mathcal{D}}
\newcommand{\Ecal}{\mathcal{E}} 
\newcommand{\Fcal}{\mathcal{F}}

\newcommand{\Mcal}{\mathcal{M}}
\newcommand{\Ocal}{\mathcal{O}}
\newcommand{\Scal}{\mathcal{S}}
\newcommand{\Rcal}{\mathcal{R}}
\newcommand{\Tcal}{\mathcal{T}}
\newcommand{\Vcal}{\mathcal{V}}
\newcommand{\Wcal}{\mathcal{W}}
\newcommand{\Xcal}{\mathcal{X}}

\newcommand{\Ci}{\mathcal{C}^\infty} 
\newcommand{\Ca}{\mathrm{\mathbf{C}}}
\newcommand{\Da}{\mathrm{\mathbf{D}}}
\newcommand{\obj}{\mathrm{Obj}}
\newcommand{\Hom}{\mathrm{Hom}}
\newcommand{\Nat}{\mathrm{Nat}}

\newcommand{\Loc}{\mathrm{\mathbf{Loc}}}       
\newcommand{\SLoc}{\mathrm{\mathbf{SLoc}}}   
\newcommand{\Obs}{\mathrm{\mathbf{Obs}}}       
\newcommand{\Sts}{\mathrm{\mathbf{Sts}}}          
\newcommand{\Top}{\mathrm{\mathbf{Top}}}          
\newcommand{\Vect}{\mathrm{\mathbf{Vec}}}       
\newcommand{\dgA}{\mathrm{\mathbf{dgA}}}      
\newcommand{\TA}{\mathrm{\mathbf{TAlg}}}      
\newcommand{\PgAlg}{\mathrm{\mathbf{PgAlg}}}
\newcommand{\FM}{\mathrm{\mathbf{Fm}}}        
\newcommand{\PB}{\mathrm{\mathbf{Pb}}}
\newcommand{\WF}{\mathrm{WF}}         
\newcommand{\id}{\mathrm{id}}               
\newcommand{\supp}{\mathrm{supp}}      
\newcommand{\dvol}{\!\!\mathrm{dvol}_{\sst{M}}} 
\newcommand{\Diff}{\mathrm{Diff}}        
\newcommand{\im}{\mathrm{Im}}             
\newcommand{\ke}{\mathrm{Ker}}            
\newcommand{\Der}{\mathrm{Der}}          
\newcommand{\tr}{\mathrm{tr}}                 
\newcommand{\loc}{\mathrm{loc}}
\newcommand{\inv}{\mathrm{inv}}
\newcommand{\reg}{\mathrm{reg}}
\newcommand{\ren}{\mathrm{r}}
\newcommand{\alte}{\mathrm{alt}}
\newcommand{\pg}{\mathrm{pg}}
\newcommand{\p}{\mathrm{ph}}
\newcommand{\af}{\mathrm{af}}
\newcommand{\ta}{\mathrm{ta}}
\newcommand{\gh}{\mathrm{gh}}

\newcommand{\mc}{\mathrm{mc}}
\newcommand{\ml}{\mathrm{ml}}
\newcommand{\ex}{\mathrm{ext}}

\newcommand{\inte}{\mathrm{int}}

\newcommand{\NN}{\mathbb{N}}          
\newcommand{\RR}{\mathbb{R}}           
\newcommand{\CC}{\mathbb{C}}           
\newcommand{\M}{\mathbb{M}} 	     
\newcommand{\al}{\alpha}
\newcommand{\bet}{\beta}

\newcommand{\Ga}{\Gamma}
\newcommand{\de}{\delta}
\newcommand{\De}{\Delta}

\newcommand{\la}{\lambda}
\newcommand{\La}{\Lambda}
\newcommand{\si}{\sigma}
\newcommand{\ph}{\varphi}
\newcommand{\T}{\cdot_{{}^\Tcal}}
\newcommand{\TL}{\cdot_{{}^{\Tcal_\Lambda}}}
\newcommand{\TR}{\cdot_{{}^{\TTR}}}
\newcommand{\delT}{\delta^{\sst{\TT}}_{S}}
\newcommand{\delTR}{\delta^{\sst{\TTR}}_S}
\newcommand{\deL}{\delta^{{\Lambda}}_S}
\newcommand{\TT}{\Tcal}
\newcommand{\TTR}{\Tcal_\ren}
\newcommand{\TRH}{\cdot_{{}^{\TTH}}}
\newcommand{\TTH}{\Tcal_{\!\ren\!,H}}
\newcommand{\TTHb}{\overline{\Tcal}_{\!\ren\!,H}}
\newcommand{\TTL}{\Tcal_{\!\sst{\Lambda}}}

\newcommand{\DDp}{\Gamma'_{\Delta_{D}}}
\newcommand{\DD}{\Gamma_{\Delta_{D}}}
\newcommand{\DC}{\Gamma_{\Delta}}

\newcommand{\HL}{\Gamma_{\Lambda}}
\newcommand{\paqft}{{p\textsc{aqft}}}
\newcommand{\tvs}{{\textsc{tvs}}}

\newcommand{\lcvs}{{\textsc{lcvs}}}
\newcommand{\eom}{{\textsc{eom}}}
\newcommand{\qme}{{\textsc{qme}}}
\newcommand{\qft}{{\textsc{qft}}}
\newcommand{\cme}{{\textsc{cme}}}
\newcommand{\mwi}{{\textsc{mwi}}}
\newcommand{\sst}[1]{\scriptscriptstyle{#1}}  
\newcommand{\hinv}{*^{\!\sst{-\!1}}}             
\newcommand{\minus}{\sst{-1}}   
\newcommand{\1}{\mathds{1}}                         
\newcommand{\pa}{\partial}                              
\newcommand{\be}{\begin{equation}}
\newcommand{\ee}{\end{equation}}

\newcommand{\adP}{\mathrm{ad}P}

\newcommand{\ad}{\mathrm{ad}}
\newcommand{\Lap}{\bigtriangleup}

\newcommand{\form}{\stackrel{\mathrm{formal}}{=}}
\newcommand{\os}{\stackrel{\mathrm{o.s.}}{=}}
\newcommand{\gt}[2]{\tilde{\textfrak{g}}^{#1#2}}

\newcommand{\otoprule}{\midrule[\heavyrulewidth]} 
\newcommand{\Ftens}{\underset{{\sst\F(M)}}{\otimes}}
\newcommand{\Ft}{\underset{{\sst\F}}{\otimes}}
\definecolor{see}{RGB}{67,75,179}
\definecolor{darksee}{RGB}{42,44,148}
\definecolor{honey}{RGB}{232,180,129}
\definecolor{lighthoney}{RGB}{255,254,220}
\definecolor{citecol}{rgb}{0.5,0,0} 
 \theoremstyle{plain}
  \newtheorem{defn}{Definition}[section]
  \newtheorem{thm}[defn]{Theorem}
  \newtheorem{prop}[defn]{Proposition}
  \newtheorem{cor}[defn]{Corollary}

  \theoremstyle{plain}

  \theoremstyle{definition}
  \newtheorem{rem}[defn]{Remark}
  \newtheorem{exa}[defn]{Example}
 \theoremstyle{definition}
  \newtheorem{ass}{\underline{\textit{Assumption}}}
\author{Katarzyna Rejzner}
\title{Batalin-Vilkovisky formalism\\ in locally covariant field theory}
\begin{document}
%
\pagestyle{empty}
\begin{center}
\textsc{
\Huge Batalin-Vilkovisky formalism\\
in locally covariant\\
field theory\\
\linespread{2}}
\vspace{27ex}
\Large Dissertation\\
zur Erlangung des Doktorgrades\\
des Department Physik\\
der Universit\"at Hamburg\\
\linespread{1.2}
\vspace{26ex}
\Large vorgelegt von\\ 
\Large \textbf{Katarzyna Anna Rejzner}\\
aus Krak\'ow\\
\linespread{2}
\vspace{8ex}
\linespread{2}
\Large \textsc{Hamburg 2011}
\end{center}
\newpage

\clearpage
\pagestyle{empty}
$\qquad$\\
\vspace{14cm}
$\qquad$\\
{\addtolength{\tabcolsep}{15pt}
\begin{tabular}{ll}
Gutachter der Dissertation:& Prof. Dr. K. Fredenhagen\\
				      &  Prof. Dr. J. Teschner\\\addlinespace[1ex]
Gutachter der Disputation:&Prof. Dr. K. Fredenhagen\\
				      &Prof. Dr. Jochen Bartels\\\addlinespace[1ex]
Datum der Disputation:     &Mittwoch, 26. Oktober 2011\\\addlinespace[1ex]
Vorsitzender der Pr\"ufungsausschusses:&Prof. Dr. G\"unter Sigl\\\addlinespace[1ex]
Vorsitzender der Promotionsausschusses:&Prof. Dr. G\"unter Sigl\\\addlinespace[1ex]
Dekan der Fakult\"at f\"ur Mathematik,&Prof. Dr Heinrich Graener\\
Informatik und Naturwissenschaften:&\\\addlinespace[1ex]
\end{tabular}}
\newpage

\clearpage
%
\thispagestyle{empty}
\begin{center}
\Large{\textsc{Zusammenfassung}}
\end{center}
In der vorliegenden Arbeit wird eine vollst\"andige Formulierung des Batalin-Vilkovisky (BV) Formalismus
im Rahmen der lokal kovarianten Feldtheorie vorgeschlagen. Im ersten Teil der Arbeit wird die klassische Theorie untersucht, wobei  der Schwerpunkt auf die zugrundeliegenden unendlich dimensionalen Strukturen gelegt wird. Es wird gezeigt, dass die Anwendung der unendlich dimensionalen Geometrie eine konzeptionell elegante Formulierung der Theorie erm\"oglicht. Die Konstruktion des BV-Komplexes ist völlig kovariant, und eine abstrakte Verallgemeinerung auf der Ebene Funktoren und nat\"urlichen Transformationen wird vorgegeben. Dies erm\"oglicht die Anwendung des BV-Komplexes in der klassischen Gravitationstheorie.  Anschließend wird eine homologische Interpretation der diffeomorphismusinvarianten physikalischen Gr{\"o}{\ss}en vorgeschlagen.

Im zweiten Teil der Arbeit wird die Quantentheorie untersucht. Ein Rahmen f\"ur die BV-Quantisierung, der vom Pfadintegralformalismus v\"ollig unabh\"angig ist und nur  auf der perturbativen algebraischen Quantenfeldtheorie basiert, wird formuliert. Um solche Formulierung zu erm\"oglichen, wird zuerst bewiesen, dass das renormierte zeitgeordnete Produkt als eine bin{\"a}re Operation auf einem geeigneten Definitionsbereich aufgefasst werden kann. Mittels dieses Resultats wird die Assoziativit\"at dieses Produkts gezeigt und dadurch lassen sich die renormierte BV Strukturen konsistent definieren. Insbesondere werden die Quantenmastergleichung und der Quanten-BV-Operator definiert. Dabei wird die Master-Ward-Identit\"at, eine wichtige Struktur der kausalen St\"orungstheorie, benutzt. 
\begin{center}
\Large{\textsc{Abstract}}
\end{center}
The present work contains a complete formulation of the Batalin-Vilkovisky (BV) formalism in the framework of locally covariant field theory. In the first part of the thesis the classical theory is investigated with a particular focus on the infinite dimensional character of the underlying structures. It is shown that the use of infinite dimensional differential geometry allows for a conceptually clear and elegant formulation. The construction of the BV complex is performed in a fully covariant way and we also generalize the BV framework to a more abstract level, using functors and natural transformations. In this setting we construct the BV complex for classical gravity. This allows us to give a homological interpretation to the notion of diffeomorphism invariant physical quantities in general relativity.

The second part of the thesis concerns the quantum theory. We provide a framework for the BV quantization that doesn't rely on the path integral formalism, but is completely formulated within perturbative algebraic quantum field theory. To make such a formulation possible we first prove that the renormalized time-ordered product can be understood as a binary operation on a suitable domain. Using this result we prove the associativity of this product and provide a consistent framework for the renormalized BV structures. In particular the renormalized quantum master equation and the renormalized quantum BV operator are defined. To give a precise meaning to theses objects we make a use of the master Ward identity, which is an important structure in causal perturbation theory.
\cleardoublepage
\pagestyle{empty}
$\qquad$\\
\vspace{2cm}
$\qquad$\\
\begin{flushright}
\begin{minipage}{8cm}
{\large
\textit{We shall not cease from exploration\\
And the end of all our exploring\\
Will be to arrive where we started\\
And know the place for the first time.\\}
$\ $\\
T.S. Eliot, \textit{Four Quartets}}
 \end{minipage}
 \end{flushright}
\cleardoublepage
\pagestyle{fancy}
\setcounter{page}{1}
\tableofcontents
\listoffigures
\chapter{Introduction}
The quest to find simple and beautiful principles underlying the laws of nature is driving the progress of human thought since thousands of years. By applying the scientific method we search for fundamental mathematical theories that describe the world around us. We believe that every step we make in exploring the Universe takes us closer to the underlying truth. We are travelers in this exciting journey in search of the truth and at each turning there are new surprises awaiting for us. While looking back to the road we have come so far, one sees that there are some principles that were chosen as guidelines and they mark our way ever since. 
One of them is the \textit{locality principle}. The importance of this principle as a guideline for the rigorous study of QFT was first put forward by Rudolf Haag over 50 years ago in his seminal paper \cite{Haag0}. It basically says that the laws of physics are local. In classical theory it is reflected by the fact that everything is described by a system of partial differential equations. In the quantum case it means that the whole physics of the system is encoded in a net of local observables, constructed by associating the corresponding observable algebras to regions of spacetime. This point of view is a cornerstone of the local quantum physics approach. To stress how important this principle is, let me cite at this point the book of Rudolf Haag \textit{Local quantum physics} \cite{Haag}.
\begin{quotation}
\textit{The German term ``Nahwirkungsprinzip'' is more impressive than the somewhat colourless word ``locality''. Certainly the idea behind these words proposed by Faraday around 1830, initiated the most significant conceptual advance in physics after Newton's Principia. It guided Maxwell in his formulation of the laws of electrodynamics, was sharpened by Einstein in the theory of special relativity and again it was the strict adherence to this idea which led Einstein ultimately to his theory of gravitation, the general theory of relativity}
\begin{flushright}
R. Haag
\end{flushright}
\end{quotation}
This is a very deep observation and it shows how universal the principle of locality is. It underlines also our everyday intuition. The relation with the special relativity theory is also clear, since the locality principle entails that causally separated processes can be measured simultaneously without any restrictions. This presence of a causal structure is crucial for the axiomatic formulation of quantum field theory. 

There are also other principles that guide our search for mathematical models correctly describing the reality. Another principle, commonly applied in physical reasoning is the \textit{covariance principle}. It was proposed by Einstein as the underlying principle of his theory of relativity. Its first incarnation, the principle of Poincar\'e-covariance, states that there are no preferred Lorentzian coordinates for the description of physical processes. In other words, there is no absolute time and space, but we can still speak of events as something localized in given spacetime points. The debate among physicists about the notion of absolute space and time is actually very old. Traditionally we mark its beginning with the famous papers of the Leibniz-Clarke correspondence 1715-1716. Although this was an exchange of letters between Gottfried Wilhelm Leibniz and Samuel Clarke, it is clear that the latter was actually presenting the position of sir Isaac Newton. In this debate Leibniz is putting forth the relational point of view on the notion of space:
\begin{quote}
\textit{As for my own opinion, I have said more than once, that I hold space to be something merely relative, as time is; that I hold it to be an order of coexistences, as time is an order of successions. }
\end{quote}
On the other hand Clark was defending Newton's absolute notion of space and time:
\begin{quote}
\textit{The reality of space is not a supposition, but is proved by the foregoing arguments, to which no answer has been given. Nor is any answer given to that other argument, that space and time are quantities, which situation and order are not.}
\end{quote} 
The history showed that the last sentence in this argument was actually said by Nature itself, since Newton's idea of space and time was proved incorrect upon the experimental verification of special relativity. Yet this was not the end of the story. Our idea of space and time had to be revised even more deeply in the light of general relativity.
In this theory even the concept of spacetime events looses meaning. 
The principle of \textit{general covariance} says that the laws of physics should be formulated in a coordinate independent way. In general relativity the geometry itself is a dynamical object. Its behavior is governed by Einstein's Equations. The theory of general relativity is not only very successful in describing the physical phenomena, but also has a beautiful mathematical structure. This interplay of physical intuition and mathematical reasoning became a hallmark in the development of modern physical theories. This is however not the end of the journey. As mentioned above, the diffeomorphism invariance of the theory forces us to abandon the concept of ``points'' as physical entities. It is rather the relations between distinguished events that should be given a physical meaning. The fact that the spacetime itself is dynamical seems at a first glance to be in conflict with the framework of QFT, where the theory is formulated with respect to a fixed causal structure. To understand how locality and covariance can fit into a consistent framework, a new paradigm is needed.  The notion of a \textit{locally covariant  quantum field} was first proposed by K. Fredenhagen at \cite{F00} and was developed in a collaboration with R. Brunetti, S. Hollands,  R. Verch  and R. Wald \cite{HW,Ver,BFV}. 
This was also an important step for the conceptual understanding of quantum field theory on curved backgrounds. The principle of \textit{general local covariance} is a fundament of  a new axiomatic framework for {\qft} on curved spacetimes proposed by R. Brunetti, K. Fredenhagen and R. Verch \cite{BFV}. The idea is to define a quantum field theory at once in all the spacetimes in a coherent way. The physical information is then encoded in a way in which algebras of observables are associated to Lorentzian manifolds (spacetime). From the mathematical point of view this amounts to the construction of a certain functor. It can also be applied to classical theories, where we associate certain Poisson algebras to spacetimes. The principle of local covariance was up to now applied in many interesting examples, including scalar \cite{BFV}, Dirac \cite{Ko-Dirac} and electromagnetic fields \cite{DS,DL}. Some aspects of the category theory side of the framework were further investigated in works of Fewster and Verch \cite{Few,Few2,Spass}.

Let us now make a short stop in our journey and look around what we have already found. We recognized two leading principles of theoretical physics, namely \textit{locality} and  the \textit{covariance}. Next we discussed how these two principles can be combined into a consistent framework by the \textit{principle of general local covariance}. There is however one more important aspect common to many modern physical theories that we have to fit into the picture. It is the principle of \textit{gauge invariance}\index{gauge!invariance}. In many aspects it is similar to diffeomorphism invariance and creates also some problems in quantum field theory. 
This principle turned out to be very powerful and universal and led to the formulation of the Standard Model of particle physics. It is also very attractive from a mathematical point of view, since classical gauge theory can be described with the use of simple geometrical structures. This is in agreement with the program of ``geometrization of physics''. Despite the great success of gauge theories, many open questions still remain. For example, the problem of rigorous nonperturbative quantization of Yang-Mills theory, confinement, asymptotic freedom, etc. There are also some technical issues related to the fact that in certain constructions of classical and quantum field theories an appropriate gauge fixing has to be done in intermediate steps. From the point of view of general local covariance the gauge fixing procedure also has to be performed in a covariant way.  An obvious candidate for a consistent framework that makes it possible is the BRST\index{BRST}  (Becchi, Rouet, Sora, Tyutin) method. It  was originally introduced in \cite{BRST1,BRST2} and later on it was put in a more general setting, called BV\index{BV!formalism} (Batalin, Vilkovisky) formalism \cite{Batalin:1977pb,Batalin:1981jr,Batalin:1983wj,Batalin:1983jr}. The present work aims at a systematic treatment of this formalism in the framework of locally covariant field theory. We treat both the classical and the quantum case. The most important new insight of the present thesis is the treatment of the renormalized quantum BV operator and the quantum master equation within the framework of perturbative algebraic quantum field theory. Moreover we propose an extension of the BV formalism to a more abstract level of natural transformations, in agreement with the principle of general local covariance. This formulation makes it also possible to apply the BV construction in the context of general relativity. 

The thesis is divided into two parts. In the first one, we present the basic structures of the BV formalism in classical theory, since it allows us to avoid technical complications related to the renormalization procedure. In the second part, we treat quantum field theory. The first chapter introduces mathematical tools that will be needed for our formulation. Beside some basic notions of category theory (section \ref{kat}) and distribution theory (section \ref{distr}) we will also need differential calculus on infinite dimensional manifolds (section \ref{idc}). This is a quite natural framework, since field theory is intrinsically defined as a theory with infinitely many degrees of freedom. Moreover the symmetry groups, important for the formulation of gauge invariance are infinite dimensional. As an example we can consider the diffeomorphism group of a finite dimensional manifold or the  group of local gauge transformations (a gauge parameter is associated to each spacetime point).

In the second chapter we give an overview of the functional approach to classical field theory and put it into the framework of general local covariance. In this  approach \cite{DF02} one constructs a Poisson algebra, by defining the Peierls bracket\index{Peierls bracket} \cite{Pei,Mar} on the space of smooth functionals on the
 configuration space. Since we work in an off-shell setting, this space contains all the possible field configurations of a certain type. For example in case of the scalar field the configuration space is just $\Ci(M)$, the space of smooth functions on the manifold $M$. The dynamics is introduced later on. This configuration space of a classical field theory is an infinite dimensional locally convex vector space.  Fields that satisfy equations of motion constitute a subspace of the configuration space. The formulation for the scalar fields was done already in \cite{DF,DF04} whereas the generalization for the anticommuting fields was systematically described in \cite{Rej}. We give a short overview of these results in sections \ref{scal} and \ref{fer} accordingly. 
 
The construction of the Peierls bracket relies on the fact that the equations of motion of a given theory form a normally hyperbolic system. This is unfortunately not the case when one has a system subject to a local symmetry. Standard examples are Yang-Mills theories and gravity. In this case, to obtain a normally hyperbolic system, one has to ``fix the gauge''. As we already mentioned, this can be achieved in the framework of the BV formalism. We start chapter \ref{BVform} with some historical remarks on the development of the BRST and the BV method (section \ref{history}). Next we propose a new formulation of this framework that is based on infinite dimensional differential geometry and fits well with the principle of general local covariance. In sections  \ref{YM} and \ref{grav} we discuss the concrete examples, namely Yang-Mills theory and classical gravity. In the last example a further extension of the formalism is needed to encompass the notion of diffeomorphism invariant physical quantities. This issue is discussed in detail in section \ref{difftrans} and in section \ref{BVnat} the corresponding BV complex is constructed.

The second part of the thesis concerns quantum field theory. We start chapter \ref{pAQFT} with some general ideas on how to incorporate the Batalin-Vilkovisky formalism into the framework of perturbative algebraic quantum field theory ({\paqft}). We describe this framework in detail in section \ref{renorm}. Using these tools we finally provide the definition of the renormalized quantum master equation and the quantum BV operator in section \ref{renormBV}. Our construction differs from other approaches to mathematically rigorous formulations of the BV formalism, since we don't employ any explicit regularization scheme. Instead we work all the time with objects that are well defined and no divergences appear in the intermediate steps. Moreover our approach uses the notion of causality, therefore we stay all the time in spacetime with the physical signature. We also don't require the compactness of the underlying spacetime, since we work on manifolds that are globally hyperbolic. All these requirements are physically motivated, but they couldn't be consistently employed in versions of the BV quantization present in the mathematical literature \cite{Costello, CostBV,Froe,Stasheff2}. This was one of the motivations to take a fresh look of the BV formalism and try to understand its structure from the point of view of {\paqft}. The reason to do this is not only the conceptual understanding of already existing methods. The main motivation to put the BV quantization into the framework of locally covariant field theory is the perspective to apply it in quantum gravity. Indeed, it was proposed in \cite{F} (see also \cite{BF1}) that employing the causal perturbation theory one can define the perturbative quantum gravity as an effective theory and the background independence is achieved by using the principle of local covariance. This is a very promising perspective and the present thesis makes a first step towards fulfilling this program. Some of the results of the first part of this thesis were already published in \cite{FR,FR2,Rej} and these of the second part will be included in the upcoming paper \cite{FR3}. 
%
\part{Classical field theory}  %
It is remarkable that nearly 200 years after Faraday and Maxwell the structure of classical field theory can still 
provide us with surprises. Although it doesn't carry the intriguing and slightly magical flavor of quantum field theory, it has the beauty and rich structure of its own. The problem of finding a coherent mathematical
structure for classical field theories has been addressed in various ways. In particular we want to mention two approaches: the multisymplectic geometry \cite{Got,Kij,CCI} and a formalism based on jet bundles and variational bicomplex \cite{Vin1,Vin2,And,Sard}. Both approaches aim basically at mathematically precise formulation of the variational calculus. There is however yet another possibility. One can take the infinite-dimensionality of field theory ``seriously'' and formulate it in the language of calculus on infinite dimensional spaces (see section \ref{clcvs}). This approach is motivated by the recent developments in perturbative algebraic quantum field theory \cite{BDF,DF,DF02,DF04}  and have opened a new perspective also for the classical field theory. Above all it provides us with a deeper conceptual understanding of the problem and allows to formulate the theory in a concise mathematical language. We follow this approach in the present work (section \ref{classfunc}), using some results of \cite{FR,BFR}.

Even more subtle than the classical theory of bosonic fields is the conceptual basis for the classical theory where fermions are present. Various attempts were made to tackle this problem. On the mathematical side there is again the variational bicomplex approach \cite{Sard,Sard2,SardGia} and the supermanifold \cite{FranPol} or graded manifold \cite{CarFig97,MonVal02}  formalism. The supermanifold approach was used in \cite{Bruz-Cian-84,Bruz-Cian-86,Bruz85,Jad}. The geometrical foundations of supermechanics on graded manifolds were formulated in \cite{IboMar93,Mon92,MonVal02}, including the notion of graded Lagrangian, tangent supermanifold, space of velocities and Hamiltonian mechanics of a graded system. Instead of following these approaches we want again to take a different perspective and look at the problem form the point of view presented in \cite{BDF,DF,DF02,DF04}. Since fermions arise primary in the quantum field theory and the classical equivalent has to be seen as a kind of a limit, it seems natural to use a framework which in a simple way can be related to the quantum case. This is easily realized in the functional approach to classical field theory and moreover the anticommuting variables are in this formalism on the equal footing with the commuting ones. This will be especially important in the context of BV formalism, where both kinds of fields appear. The complete treatment of classical field theory of fermions was presented in \cite{Rej}. Here we recall only the most important results. This will be done in section \ref{fer}.

We start our discussion of classical field theory with a chapter introducing mathematical structures that we shall need later for the formulation of the classical field in the locally covariant framework. In particular, in the discussion of gauge theories and the BV complex\index{BV!complex} we shall apply methods of infinite dimensional differential geometry. It is remarkable in physics, that to solve a problem one often has to use techniques from many different fields of mathematics. This is perhaps what makes the research so challenging and exciting, but on the other hand a single scientist is not able to explore in detail all the subtleties of the methods he is using. This is how the research in theoretical physics differs from pure mathematics. Bearing this in mind we don't attempt here to give a complete introduction into the mathematical methods we are using, since it would go far beyond the scope of this thesis. Instead we want to provide the reader with a basic vocabulary and give an overview of the fields of mathematics, that turned out to be useful in our framework. Before taking a jump into a vast ocean of definitions and theorems we want first to justify why those structures would be needed. 

The underlying idea of the BV formalism\index{BV!formalism} is actually very simple. To present it, we start with a toy model.
 Imagine that your physical system is characterized by a configuration space $E$, which has a structure of a finite dimensional manifold. The dynamics is implemented by a set of equations on the configuration space and  solutions of these equations form a subspace $E_S$ of $E$. It can now happen that your system has some symmetries, i.e. there exist one-parameter groups of transformations of $E_S$. Such transformations map solutions into new solutions. They correspond to vector fields on the submanifold $E_S$ (see figure \ref{geom}). 
  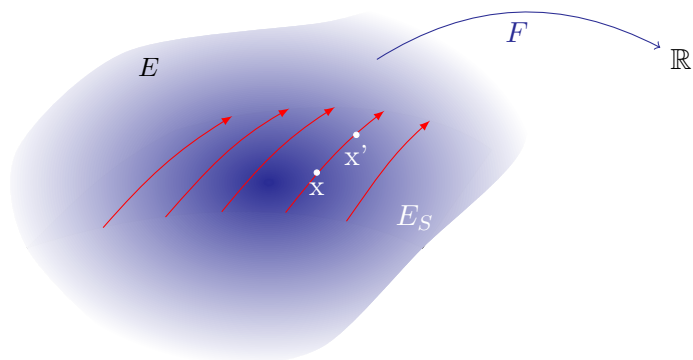
\begin{figure}[!htb]
 \begin{center}
 \begin{tikzpicture}
\filldraw[fill=see, scale=1.3]  (-2,0) .. controls (-0.5,0.5) and (1,0.5) .. (2,0) -- (2.7,1) .. controls (2.1,1.45) and (0.9,1.6) .. (-0.5,1.2) .. controls (-1.18,0.85) and (-1.5,0.5) .. (-2,0);
\shade[shading=radial,  inner color=darksee,outer color=white,opacity=0.75,scale=1.3] plot[smooth cycle] coordinates{(-1,-1) (-2,0)(-2.1,1)(-1,2)(1,2.3)(2,2.5)(3,2)(3,1)(2,0)(1,-1)(0,-1.2)};
\path[-latex] (-1.6,0.27) edge  [color=red,out=49,in=212]  (0.1,1.75);
\path[-latex] (-0.78,0.41) edge  [color=red,out=49,in=210]  (0.85,1.85);
\path[-latex] (-0.04,0.48) edge  [color=red,out=51,in=213]  (1.45,1.87);
\path[-latex] (0.8,0.47) edge  [color=red,out=52,in=218]  (2.1,1.82);
\path[-latex] (1.6,0.355) edge  [color=red,out=55,in=222]  (2.7,1.69);
\draw(-1,2.4) node {\small{$E$}};
\draw(6,2.5) node [name=R] {$\RR$};
\draw(2.5,0.4) node [color=white] {$E_S$};
\filldraw[color=white] (1.21,1) circle (1pt) node [anchor=north,color=white]{x};
\filldraw[color=white] (1.73,1.5) circle (1pt) node [anchor=north,color=white]{x'};
\path[->] (2,2.5) edge  [color=darksee, bend left]  node[right,below] {$F$} (R);
\end{tikzpicture}
\end{center}
\caption[Geometrical setting]{The geometrical setting of a classical theory with symmetries. Symmetry orbits are depicted as red arrows. Two solutions $x$ and $x'$ are laying on the same orbit. By performing an experiment we measure the observable $F$ and obtain a result which is a real number. If $F$ is invariant, then in particular $F(x)=F(x')$.\label{geom}}
\end{figure} 
In this toy model you can think of observables as functionals $F\in\Ci(E_S)$ on the solution space. In this interpretation performing an experiment means assigning to a given configuration $x$ of the system a certain number $F(x)$. But if your system has symmetries you should not be able to distinguish between solutions that lie on the same orbit. Therefore your physical observables are those functionals that are constant along those orbits: $\Ci_\inv(E_S)$. The aim of BV formalism is to use some simple geometrical consideration to describe the space $\Ci_\inv(E_S)$ as a certain homology.

But wait a minute! Isn't the configuration space of the field theory really infinite dimensional? Well, yes it is.
Actually there are two ways out of this problem. First is to use the jet space formalism, which is slightly heavy on the notation, but avoids delving into the infinite dimensional aspects of the problem. Another way out is to take the geometrical picture presented above quite literarily and translate it simply to the language of infinite dimensional manifolds. We chose the second option because in the end it provides a relatively simple language to describe those problems in a very general setting. However, to use the machinery of infinite dimensional geometry, one has to learn some vocabulary. 
\chapter{Mathematical preliminaries}\label{math}
\vspace{-5ex}
\begin{flushright}
 \begin{minipage}{10cm}
\textit{Ich glaube, da{\ss} es, im strengsten Verstand, f\"ur den Menschen nur eine einzige Wissenschaft gibt, und diese ist reine Mathematik. Hierzu bed\"urfen wir nichts weiter als unseren Geist.}
\vspace{-1ex}
\begin{center}
$***$
\end{center}
\textit{Die Mathematiker sind eine Art Franzosen: Redet man zu ihnen, so \"ubersetzen sie es in ihre Sprache, und dann ist es also bald ganz etwas anderes.}
\begin{flushright}
G. Ch. Lichtenberg
\end{flushright}
 \end{minipage}
\end{flushright}
\vspace{1ex}
\noindent\rule[2pt]{\textwidth}{1pt}
\vspace{3ex}\\
This is the program of our short tour into the land of mathematics. It all starts in the realm of functional analysis with some general facts concerning locally convex vector spaces. The reader who is well acquainted with those issues may just skip this part. After crossing the jungle of definitions and properties, where strange structures live, we will enter more safe territory of nuclear vector spaces. There the tensor products are well behaving and every practitioner of calculus can feel comfortable there. From there we go by a very frequently used path straight into the distributions' spaces. Since that land is quite well known to all mathematical physicists, we shall not spend too much time there, only shortly stopping to discuss the generalization to vector-valued distributions. Finally, we leave the realm of functional analysis and pay a visit to the category theory, careful to avoid too much ``abstract nonsense''. After this short trip into the land of mathematics we can come back to physics in chapter 2, bringing along new techniques and theorems.
\clearpage
\section{Functional analysis}\label{topo}
\subsection{Locally convex vector spaces}\label{locconvvs}
Let the journey begin! We start this mathematical introduction with some basic definitions concerning the topological vector spaces. This section is based on \cite{Rud,Jar,Sch0,Koe}. 
\begin{defn}
A \textbf{topological space}\index{topological vector space} is a set $X$ in which a collection $\tau$\index{topology} of 
subsets (called open sets) has been specified, with the following properties:
\begin{itemize}
\item $X\in\tau$
\item $\emptyset\in \tau$
\item the intersection of any two open sets is open: $U\cap V\in\tau$ for $U,V\in\tau$
\item the union of every collection of open sets\index{set!open} is open: \\
$\bigcup\limits_{\alpha\in A}U_{\alpha}\in \tau$ for $U_{\alpha}\in \tau\ \forall \alpha\in A$. 
\end{itemize}
\end{defn}
Well, that was quite general. A topology doesn't give us too much structure, but let's see what we can do with it. A first thing to do is to look at mappings between the spaces. A topology already tells us something about the regularity of those mappings, since it contains already a notion of convergence.
\begin{defn}
A function $f : X \rightarrow Y$, where $X$ and $Y$ are topological spaces, is \textbf{continuous} if and only if for every open set $V \subseteq Y$, the inverse image:
\begin{equation}
    f^{-1}(V) = \{x \in X \; | \; f(x) \in V \}
\end{equation}
is open.
\end{defn}
In our applications the topology will not be enough to capture all the structure we need. In the physics context it is common that we want to add certain quantities and scale them. This leads in a natural way to a vector space structure. Now we want this structure to be compatible also with the topology.
\begin{defn}
A \textbf{Topological vector space} (\textsc{\textsc{tvs}}) is a pair $(X,\tau)$, where $\tau$ is a topology on a vector space $X$ such that:
\begin{itemize}
\item every point of $X$ is a closed set
\item the vector space operations are continuous with respect to $\tau$. 
\end{itemize}
\end{defn}
The subsets of topological vector spaces can have certain special properties. Here we list the most important ones:
\begin{defn}
Let $E$ be a vector space over a field $\mathbb{K}=\mathbb{C}$ or $\mathbb{R}$ and $A, B\subseteq E$:
\begin{enumerate}
\item $A$ is called \textbf{circled}\index{set!circled} if  $\, \forall\lambda \in \mathbb{K}, |\lambda| \leq 1 : \lambda A \subseteq A$.
\item $A$ is called \textbf{balanced}\index{set!balanced} if $\forall\lambda \in \mathbb{K}, |\lambda| = 1 : \lambda A \subseteq A$.
\item $A$ is said to \textbf{absorb} $B$ if there exists a $\lambda > 0$ with $[0, \lambda]\cdot B \subseteq A$. 
\item $A$ is called \textbf{absorbent}\index{set!absorbent} if $\forall x \in E$ : $A$ absorbs $\{x\}$. 
\item $A$ is called \textbf{convex}\index{set!convex} if $\mathbb{R}\ni\lambda_1 , \lambda_2 
\geq 0$, $\lambda_1 + \lambda_2 = 1$, $x_1 , x_2 \in A$ implies: $\lambda_1 x_1 + \lambda_2 x_2 \in A$. 
\item $A$ is called \textbf{absolutely convex}\index{set!absolutely convex} if $\lambda_1 , \ldots, \lambda_n \in\mathbb{K}$. $\sum\limits^n_ {i=1}|\lambda_i|\leq 1$, $x_1 , \ldots, x_n\in A$ implies $\sum\limits^n_ {i=1}\lambda_i x_i\in A$.
\item $A$ is called \textbf{bounded}\index{set!bounded} if for every neighborhood $U$ of zero, there exists a scalar $\lambda$ so that
$A\subseteq \lambda U$. In other words a set is called bounded if it is absorbed by every zero neighborhood.
\end{enumerate} 
\end{defn}
An important tool to characterize a topological vector space is a \textit{base}. It is defined as follows:
\begin{defn}
A \textbf{local base} of a topological vector space $X$ is thus a collection $\mathscr B$, of neighborhoods of 0 such that every neighborhood of 0 contains a member of $\mathscr B$. The open sets 
of $X$ are then precisely those that are unions of translates of members of $\mathscr B$.
\end{defn}
There are some important types of topological vector spaces, that have many nice properties and are therefore commonly used in mathematics and physics. In the following definitions, $X$ always 
denotes a topological vector space, with topology $\tau$.
\begin{enumerate}
\item $X$ is a \textbf{\textit{locally convex vector space}}\index{topological vector space!locally convex} (\textsc{\textsc{lcvs}}) if there is a local base  $\mathscr B$ whose members are convex. 
\item $X$ is locally bounded if 0 has a bounded neighborhood. 
\item $X$ is locally compact if 0 has a neighborhood whose closure is compact. 
\item $X$ is metrizable\index{topological vector space!metrizable} if $\tau$ is compatible with some metric d. 
\item $X$ is a Fr\'echet space\index{topological vector space!Fr\'echet} if $X$ is a complete locally convex space with a metrizable topology
\item $X$ is normable if a norm exists on $X$ such that the metric induced by the norm is 
compatible with $\tau$. 
\end{enumerate}
In our framework we shall use always topological vector spaces that are locally convex. On the practical grounds they can be also characterized in terms of \textit{seminorms}\index{seminorm} . This is how they are usually defined in the context of physics.
\begin{defn}
A \textbf{seminorm} on a vector space $X$ is a real-valued function $p$ on $X$ such that: 
\begin{enumerate}
\item $p(x + y) < p(x) + p(y)$ for all $x,y \in X$.
\item $p(\lambda x)=|\lambda|p(x)$ for all $x \in X$ and all scalars $\lambda\in\mathbb{K}$. 
\end{enumerate}
\end{defn}
We see that a seminorm already provides us with a lot of information, but there is one property that is still missing. We would like to have some notion of distance in our space, to compare different elements and a single seminorm is not enough to do it. Indeed, it can happen that a seminorm evaluated on an element is $0$ even though the element itself is non zero. If we want to avoid it, we arrive as
the notion of a \textit{norm}.\index{norm} 
\begin{defn}
A seminorm $p$ is a \textbf{norm} if it satisfies: $p(x)\neq0$ if $x\neq 0$.
\end{defn}
But what if our {\tvs} doesn't admit a norm? If we still want to get some notion of a distance, we can compare two elements not with one seminorm, but use a whole family.
\begin{defn}
A family $\mathscr{P}$ of seminorms on $X$ is said to be \textbf{separating}\index{separating!family of seminorms} if to each $x \neq0$ corresponds at least one $p\in\mathscr{P}$ with $p(x)\neq 0$. 
\end{defn}
We can see that a separating family of seminorms already allows us to distinguish two elements. From the following theorem it becomes clear why locally convex vector spaces are so commonly used in physics.
\begin{thm}
With each separating family of seminorms on $X$ we can associate a locally convex topology $\tau$ on $X$ and vice versa: every locally convex topology is generated by some family of separating seminorms. 
\end{thm}
\begin{proof}
See \cite{Rud}.
\end{proof}
Now we know, why are seminorms so nice. That's usually the best you can get, if you cannot equip your space with a norm. Besides they also provide a notion of distance. This is quite important, since if you want to think of the configuration space of your physical model you would like to be able to tell if a certain solution is ``close'' to a given one. It turns out that if our family of seminorms is countable this notion of distance is actually very precisely defined.
\begin{thm}\label{metrizable}
A locally convex vector space $(X,\tau)$ is metrizable\index{topological vector space!metrizable} iff $\tau$ can be defined by $\mathscr{P} = \{p_n : n \in \mathbb{N}\}$ a countable separating family of seminorms on $X$. One can equip $X$ with a metric which is compatible with $\tau$ and which provides a family of convex balls.
\end{thm}
\begin{proof}
See \cite{Koe,Rud}.
\end{proof}
A \textsc{lcvs} from theorem \ref{metrizable} can be equipped with the metric:
\begin{equation}
d(x, y) := \sum\limits_{\in \mathbb{N}} 2^{-n} \frac{p_n (x-y)}{1 + p_n(x -y)}\label{metric}
\end{equation}
This metric is compatible with $\tau$ but in general this metric doesn't provide convex balls (see the discussion in \cite{Rud} after theorem 1.24 and exercise 18). Nevertheless it is good to know that you have a metric that can actually be written down in a closed form. You can also invent as many modifications of this definition as you want. If $X$ is complete with respect to the metric from theorem \ref{metrizable} it is a Fr\'echet space. Usually a Fr\'echet space topology is defined explicitly by providing a countable separating family of seminorms. Those spaces are already well behaving but still not optimal for the calculus. Everything becomes much easier if we go one step further and equip our topological space with a norm. Among normed topological spaces there is a class that is especially favored by all functional analysts.
\begin{defn}
A \textbf{Banach space}\index{topological vector space!Banach} is a normed \textsc{tvs} which is complete with respect to the norm. 
\end{defn}
\begin{thm}
A topological vector space $X$ is normable if and only if its origin has 
a convex bounded neighborhood. 
\end{thm}
\begin{proof}
See \cite{Rud}.
\end{proof}
Many of the fundamental theorems of functional analysis don't work outside Banach spaces, or they need some more assumptions. Nevertheless notions of calculus exist on general locally convex vector spaces and a lot can be proven also in this setting. We give an overview of the crucial results in section \ref{idc}. It turns out that at some point we would have to leave the save realm of Banach and move the calculus to a more general setting. But before we do it, we shall stop for a while and admire the dualities.

A very important notion in the theory of \textsc{lcvs} is the duality. A dual space of $E$, denoted by $E'$ is the space of continuous linear mappings $L(E,\RR)$.  In general it can be equipped with different inequivalent topologies:
\begin{enumerate}
\item \textbf{Weak topology}:\index{topology!weak} the so called pointwise convergence topology. A functional $x'\in E'$ converges to $0$ in this topology if all its values $\left<x',x\right>$ converge to $0$, for all $x\in E$.
\item \textbf{Strong topology}:\index{topology!strong} topology of uniform convergence on bounded sets of  $E$.
\item \textbf{Compact convergence topology}\index{topology!compact convergence}: topology of uniform convergence on compact sets of  $E$.
\end{enumerate}
One notices that the weak topology contains the least information about how good the convergence is. One can only see what is happening at a given point. It is of course better to control convergence in the whole bounded set. Therefore, strong and compact convergence topology are more reliable, but of course it is usually more difficult to prove  them. At this point one more thing requires a comment. Until now we didn't really give a definition of a \textit{uniform convergence}. Actually, it's a very general notion and can be applied even in a structure more general than a {\lcvs}, namely in a \textit{uniform space}. Since we don't need such a level of generality we only give a definition of a uniform convergence here for maps between topological vector spaces. Before we do it there is one more important notion we have to introduce, namely a \textit{net}\index{net}. It is a generalization of a sequence and will be useful in the context of algebraic formulation of field thory.
\begin{defn}
Let $I$ be a directed set (i.e. a nonempty set with a reflexive and transitive binary relation $\leq$ with the additional property that every pair of elements has an upper bound) and $V\subset X$ a general set in a {\tvs} $X$. If for each $\alpha\in I$ an element $x_\al\in V$ is given, then $(x_\al)_{\al\in Z}$ is a net in $V$. 
\end{defn}
The notion of convergence of nets is defined in the following way:
\begin{defn}
A net $(x_\al)_{\al\in Z}$  is said to be convergent\index{convergence!of a net} to $x\in X$ if for each neighborhood $U$ of $x$, there exists $\beta(U)\in I$, such that $x_\gamma\in U$ for all $\gamma\geq \beta(U)$.
\end{defn}
Now we can come back to the discussion of uniform convergence. Consider $X$ to be a subspace of the space of all mappings between {\lcvs} $E$ and $F$. We say that a net $(f_\al)_{\al\in I}$ in $X$ converges uniformly\index{convergence!uniform} to an element $f\in X$ on a set $B\subset E$ if  for every $U$ open in $F$ there exists $\beta(U)$ such that $f_\al(x)-f(x)\in U$ \textit{for all} $x\in B$ and $\alpha\geq \beta(U)$. Now we can give a precise definition of the topology of uniform convergence on compact (bounded) sets. We say that a net $(f_\al)_{\al\in I}$ in $X$ converges to $f\in X$ in this topology if it converges uniformly on all the compact (bounded) sets. It is clear that this notion contains more information on the ``quality'' of the convergence than the pointwise topology.
\subsection{Nuclear locally convex vector spaces}\label{nuclear}
Working on the level of infinite dimensional vector spaces often causes problems, since there are many counterintuitive properties and one has to be careful with using notions known from the finite dimensional case. A particular difficulty arises with the definition of tensor products. On the level of infinite dimensional \textsc{lcvs} the projective and injective tensor products do not coincide. There is however a class of \textsc{lcvs} for which this is the case. Those are called \textit{nuclear spaces}. We give now an abstract definition of nuclear \textsc{lcvs} and list some properties that are usefull in proving that a given space is nuclear. Many of the spaces used in this work have the nuclearity property. 

First we introduce a notion of a \textit{\textbf{nuclear operator}}. In principle an operator is nuclear if it can be approximated by operators of finite rank (finite dimensional image). A nuclear operator has certain nice properties analogous to finite-dimensional operators. For example a nuclear operator mapping a space with a basis into itself has a finite trace, which is given by the sum of the series formed from the diagonal elements of the matrix of this operator relative to an arbitrary basis. Below we give a precise definition of a nuclear operator on a locally convex vector space. Let $E$ and $F$ be {\lcvs} over the field $\RR$ or $\CC$, let  $E'$ and  $F'$ be their strong duals. 
\begin{defn}
A linear operator $A:E\rightarrow F$ is called \textit{\textbf{nuclear}}\index{nuclear!operator} if it can be represented in the form
\[
x\mapsto Ax=\sum\limits_{i=1}^\infty\lambda_i\left<x,x_i'\right>y_i\,,
\]
where $\{\lambda_i\}$ is a summable numerical sequence, $\{x_i'\}$ is an equicontinuous\footnote{A set $A$ of continuous functions between two topological spaces $E$ and $F$ is equicontinuous at the points $x_0 \in E$ and $y_0\in F$  if for any open set $\Ocal$ around  $y$, there are neighborhoods $U$ of $x_0$ and $V$ of $y_0$ such that for every $f \in A$, if the intersection of $f(U)$ and $V$ is nonempty, then $f(U) \subseteq \Ocal$. One says that $A$ is equicontinuous if it is equicontinuous for all points $x_0\in E$, $y_0\in F$. The notion of equicontinuity becomes more intuitive, if we choose $E$ and $F$ to be metric spaces. The family $A$ is equicontinuous at a point $x$ if for every $\epsilon>0$, there exists a  $\delta > 0$ such that $d(f(x_0), f(x)) < \epsilon$ for all $f \in A$ and all $x$ such that $d(x_0, x) <\delta$. In other words we require all member of the familiy $A$ to be continuous and to have equal variation over a given neighbourhood.}\index{equicontinuity} sequence in $E'$, $\{y_i\}$ is a sequence of elements from a certain complete bounded convex circled set in $F$ and $\left<x,x_i'\right>$ denotes the value of the linear functional $x_i'$ at a vector $x$.
\end{defn}
Now we can define what is a  \textit{\textbf{nuclear space}}.
\begin{defn}
A nuclear space\index{nuclear!space} is a locally convex vector space for which all continuous linear mappings into an arbitrary Banach space are nuclear operators.
\end{defn}
Nuclear spaces are commonly used in analysis, since the projective and injective tensor products (see section \ref{tensor}) are equivalent for such spaces and also an analogue of Schwartz' kernel theorem is valid \cite{Gro}. They also have quite good permanence properties that allow us to prove that a given space is nuclear without direct calculations. Here we give a list of those properties based on \cite{Pie,Jar}.
\begin{thm}\label{nuc} The following spaces are nuclear:
\begin{enumerate}[i)]
\item a linear subspace of a nuclear \textsc{lcvs},\label{nsubs}
\item a quotient of a nuclear \textsc{lcvs} by a closed linear subspace,
\item a cartesian product and a projective limit of an arbitrary family of nuclear \textsc{lcvs},\label{nprod}
\item a countable direct sum and a countable inductive limit of  nuclear \textsc{lcvs},
\item a projective tensor product of nuclear spaces.
\end{enumerate}
\end{thm}
From properties \ref{nsubs}) and \ref{nprod}) one can obtain a following useful corollary (see \cite{Pie}, proposition 5.2.3):
\begin{cor}[ \cite{Pie}, 5.2.3]\label{ninit}
Let $E$ be a linear space which is mapped by certain linear mappings $T_i$ into locally convex vector spaces $E_i$, $i\in I$ in such a way that for each element $x\neq 0$ there exists an index $i_0\in I$ with $T_{i_0}x\neq 0$. Space $E$ can identified with a subspace of $\prod\limits_I E_i$ and equipped with the initial topology with respect to this family of mappings. If all spaces $E_i$ are nuclear, then this topology on $E$ is also nuclear.
\end{cor}
Many of the spaces commonly used in analysis are nuclear. In particular the spaces of smooth function $\Ecal$, compactly supported functions  $\Dcal$ and Schwartz functions $\Scal$ as well as their strong and weak duals are nuclear.
\subsection{Tensor products}\label{tensor}
Nuclear spaces introduced in the previous section are particularly useful in the context of tensor products on \textsc{lcvs}. We will explain in section \ref{axioms} how the tensor structure is related to the causality. More on that issue can be found in \cite{FR2}. It is well known that the definition of topological tensor product is not unique for general locally convex vector spaces. Most of the results on this subject can be found in the thesis of A. Grothendieck \cite{Gro}. In principle there are two natural notions that can be applied in this case: projective and injective tensor product. We recall here both definitions. More details can be found in \cite{Jar, Bou}.
\begin{defn}
Let $E$ and $F$ be locally convex vector spaces and let $\otimes: E\times F\rightarrow E\otimes F$ be the canonical map into the corresponding tensor product. The finest topology on $E\otimes F$ which makes $\otimes$ continuous is called \textit{\textbf{the projective tensor topology}}\index{tensor product topology!projective} or the $\pi$-topology. The space $E\otimes F$ equipped with this topology is denoted by $E\otimes_\pi F$
\end{defn}
It can be shown that the topology $\pi$ is locally convex. Another possible topology on $E\otimes F$ is the so called  \textit{\textbf{injective tensor topology}}\index{tensor product topology!injective}. Its definition is a little bit more involved. In some sense it is the weakest well behaving topology one can put on $E\otimes F$. The idea is to define it via the topology on the space of continuous linear mappings $L(E'_\gamma, F)$. We equipped $E'$ with the finest locally convex topology $\gamma$ that coincides with the weak one on equicontinuous sets. One can identify $E\otimes F$ with a subspace of $L(E'_\gamma, F)$. Next we equip $L(E'_\gamma, F)$ with a topology of uniform convergence on equicontinuous compact sets in $E'$.  We denote the resulting topological space by $E\varepsilon F$. It is called \textit{the $\varepsilon$-product} of $E$ and $F$. The corresponding topology induced on $E\otimes F$ is called the $\varepsilon$-topology and $E\otimes F$ equipped with it is the injective tensor product $E\otimes_\varepsilon F$. This topology is better behaving if we want to consider for example vector valued distributions (see section \ref{vvalued}) and was used (in a slightly modified version) by L. Schwartz in \cite{Sch1,Sch2}. Inequivalent notions of tensor products on \textsc{lcvs} can possibly create a problem, but there is a large class of spaces, where they coincide. The crucial result, proved by A. Grothendieck \cite{Gro} says that:
\begin{thm}
$E$ is a nuclear locally convex vector space if and only if for each arbitrary \textsc{lcvs} $F$ the projective and injective tensor products coincide, i.e.
\[
E\otimes_\varepsilon F=E\otimes_\pi F\,.
\]
\end{thm}
\section{Distributions}\label{distr}
In the previous section we introduced many abstract notions, now it's time for something more practical. Most commonly used locally convex vector spaces in physics are spaces of functions and distributions. Therefore we want to concentrate on them in this section, where we give a short summary of the definitions and theorems from the theory of distributions that we are going to use later on. The purpose of this short revision is twofold. Firstly we want to show how the abstract notions from topology work in specific 
examples. Secondly we want to fix the notation and present the concepts we need in a language consistent with the rest of the thesis. Since the theory of distributions is quite a standard tool in mathematical physics we can now feel more relaxed in our journey through the land of mathematics and enjoy nice views, seen now from a slightly different perspective. First we discuss the distributions on $\RR^n$, next we remark on generalization to distributions on manifolds and vector valued distributions. 
\subsection{Distributions on $\Omega\subset\RR^n$}\label{distr1}
We start with defining locally convex topologies on spaces of smooth functions on $\RR^n$. 
Definitions and theorems in this section are taken from \cite{Rud,Hoer,Trev,Sch0}. The part on wave front sets is based on the chapter 4 of \cite{BaeF}.

Let $\Omega\subset\RR^n$ be an open subset and $\Ecal(\Omega)\doteq\Ci(\Omega)$ the space of smooth functions on it. We equip this space with a Fr\'echet topology generated by the family of seminorms:
\be\label{topE}
p_{K,m}(\ph)=\sup_{x\in K\atop |\alpha|\leq m}|\partial^\alpha\ph(x)|\,,
\ee
where $\alpha\in\NN^N$ is a multiindex and $K\subset \Omega$ is a compact set. This is just the topology of uniform convergence on compact sets mentioned in section \ref{locconvvs} of all the derivatives. 

The space of smooth compactly supported functions $\Dcal(\Omega)\doteq\Ci_c(\Omega)$ can be equipped with a locally convex topology in a similar way. The fundamental system of seminorms is given by \cite{Sch0}:
\be\label{topD}
p_{\{m\},\{\epsilon\},a}(\ph)=\sup_\nu\big(\sup_{|x|\geq\nu,\atop |p|\leq m_\nu} \big|D^p\ph^a(x)\big|/\epsilon_\nu\big)\,,
\ee
where $\{m\}$ is an increasing sequence of positive numbers going to $+\infty$ and $\{\epsilon\}$ is a decreasing one tending to $0$. This topology is no longer Fr\'echet.

The space of \textbf{\textit{distributions}}\index{distribution} is defined to be the dual  $\Dcal'(\Omega)$ of $\Dcal(\Omega)$ with respect to the topology given by (\ref{topD}). Equivalently, given a linear map $L$ on $\Dcal(\Omega)$ we can decide if it is a distribution by checking one of the equivalent conditions given in the theorem below \cite{Trev,Rud,Hoer}.
\begin{thm}
A linear map $u$ on $\Ecal(\Omega)$ is a distribution if it satisfies the following equivalent conditions:
\begin{enumerate}
\item To every compact subset $K$ of $\Omega$ there exists an integer $m$ and a constant $C>0$ such that for all $\ph\in\Dcal$ with support contained in $K$ it holds:
\[
|u(\ph)|\leq C\max_{p\leq k}\sup_{x\in\Omega}|\pa^p\ph(x)|\,.
\]
We call $||u ||_{\Ccal^k(\Omega)}\doteq\max_{p\leq k}\sup_{x\in\Omega}|\pa^p\ph(x)|$ the $\Ccal^k$-norm and if  the same integer $k$ can be used in all $K$ for a given distribution $u$, then we say that $u$ is of order $k$\index{order of a distribution}.
\item If a sequence of test functions $\{\ph_k\}$, as well as all their derivatives converge uniformly to 0 and if all the test functions $\ph_k$ have their supports contained in a compact subset $K\subset\Omega$ independent of the index $k$, then $u(\ph_k)\rightarrow 0$.
\end{enumerate}
\end{thm}
An important property of a distribution is its support\index{distribution!support}.  If $U' \subset U$ is an open subset then $\Dcal(U')$ is a closed subspace of $\Dcal(U)$ and there is a natural restriction map $\Dcal'(U) \rightarrow \Dcal'(U')$. We denote the restriction of a distribution $u$ to an open subset $U'$ by $u|_{U'}$.
\begin{defn}
The support $\supp u$ of a distribution $u \in \Dcal'(\Omega)$ is the smallest closed set $\Ocal$ such that $u|_{\Omega\setminus \Ocal} = 0$. In other words:
\[
\supp u\doteq \{x\in\Omega|\, \forall U\,\textrm{open neigh. of }x,\, U\subset\Omega\ \exists \ph\in\Dcal(\Omega), \supp\ph\subset U,\,\mathrm{s.t. }<\!u,\ph\!>\neq 0\}\,.
\]
\end{defn}
Distributions with compact support\index{distribution!with compact support} can be characterized by means of a following theorem:
\begin{thm}
The set of distributions in $\Omega$ with compact support is identical with the dual $\Ecal'(\Omega)$ of $\Ecal(\Omega)$ with respect to the topology given by (\ref{topE}).
\end{thm}

Now we want to consider topologies on $\Ecal'(\Omega)$ and $\Dcal'(\Omega)$. Most natural choices are the strong or the weak topology. Spaces $\Ecal(\Omega)$, $\Dcal(\Omega)$, as well as their strong duals are reflexive (coinciding with the dual of the dual space) nuclear spaces. Moreover it is shown in \cite{Sch0} that $\Ecal'(\Omega)$ is embedded in $\Dcal'(\Omega)$ also as a topological vector space, if we equip both with their strong topologies.

There are many examples of distributions. Clearly any locally integrable function $f \in L_\loc^1(\Omega)$ defines a distribution in $\Dcal'(\Omega)$ by:
\[
\ph\mapsto \int_\Omega f(x)\ph(x) dx
\]
There is also a characterization of a wider class of distributions in terms of measures. This will be important later on in the context of local functionals.
\begin{thm}
Let $u\in\Dcal'(\Omega)$. If $\supp u$ is a compact subset of $\Omega$, then $u$ has a finite order $N<\infty$. In case when $\supp u=\{x\}$ consists of a single point, then there are constants $c_\al$ such that:
\[
u=\sum\limits_{|\alpha|\leq N}c_\al D^\al\delta_x\,,
\]
where $\delta_x$ is the evaluation functional, i.e. $\delta_x(\ph)=\ph(x)$.
\end{thm}
\begin{proof}See \cite{Rud}, theorems 6.24, 6.25.
\end{proof}
This can be further generalized to a theorem that allows us to write the compactly supported distributions as certain measures:
\begin{thm}
Let $u\in\Dcal'(\Omega)$ be a distribution with compact support $K$ and of rank $N$. Let $K\subset V\subset \Omega$, where $V$ is an open set. Then there exists finitely many functions $f_\beta$ in $\Omega$ (one for each multiindex $\beta_i\leq N+2$ for $i=1,...,n$) with supports in $V$ such that:
\[
u(\ph)=\sum_\beta (-1)^{|\beta|}\int_\Omega f_\beta(x)(D^\beta\ph)(x)dx\qquad \ph\in\Dcal(\Omega)\,.
\]
\end{thm}
\begin{proof}
See \cite{Rud}, theorem 6.27.
\end{proof}
The theorem above justifies somewhat the notation used commonly in physics, where the evaluation of a distribution  on a test function is written as an integral. However this characterization must be taken with certain care, since it is in general not possible to write a distribution $u$ as a sum of measures with supports contained in the support of $u$. 

Now we discuss the singularity structure of distributions. This is mainly based on \cite{Hoer} and chapter 4 of \cite{BaeF}.
\begin{defn}
The singular support\index{singular support} $\mathrm{sing\, supp}\, u$ of $u \in \Dcal'(\Omega)$ is the smallest closed subset $\Ocal$ such that $u|_{\Omega\setminus \Ocal} \in \Ecal(\Omega\setminus \Ocal)$.
\end{defn}
We recall an important theorem giving the criterium for a compactly distribution to have an empty singular support:
\begin{thm}
A distribution $u \in \Ecal'(\Omega)$ is smooth if and only if for every $N$ there is a constant $C_N$ such that:
\[
|\hat{u}(\xi )| \leq C_N (1 + |\xi |)^{-N}\,,
\]
where $\hat{u}$ denotes the Fourier transform of $u$.
\end{thm}
If a distribution has a nonempty singular support we can give a further characterization of its singularity structure by specifying the direction in which it is singular. This is exactly the purpose of the definition of a wave front set.
\begin{defn}
For a distribution $u \in \Dcal'(\Omega)$ the wavefront set\index{wavefront set} $\WF(u)$ is the complement in $\Omega \times \RR^n\setminus\{0\}$ of the set of points $(x,\xi) \in \Omega \times \RR^n\setminus\{0\}$ such that there exist
\begin{itemize}
\item a function $f \in\Dcal(\Omega)$ with $f(x)=1$,
\item an open conic neighborhood $C$ of $\xi$, with
\[
\sup_{\xi\in C}(1+|\xi|)^N|\widehat{f \cdot }(\xi)|<\infty\qquad\forall N \in \NN_0\,.
\]
\end{itemize}
\end{defn}
We shall come back to the discussion of wave front sets in section \ref{scal} when we recall the construction of a Poisson structure of the classical theory of the scalar field.
\subsection{Vector-valued distributions on manifolds}\label{distr2}
In this section we mainly follow the introduction given in \cite{Baer}. Let $M$ be a manifold equipped with a smooth volume density $\ \dvol$. In particular we can use the volume form induced by a Lorentzian metric. We consider a real or complex vector bundle $B\xrightarrow{\pi}M$ with fiber $V$. Let $\Gamma(B)$ be the space of its smooth sections. We also use the notation $\Gamma(B)\equiv\Gamma(M,V)$ to stress the fact that locally this is isomorphic to the space of smooth functions with values of $V$. All spaces of field configurations have more or less this structure.
The space of compactly supported sections will be denoted by $\Gamma_c(B)\equiv\Gamma_c(M,V)$.

We equip $B$ and $T^*M$ with connections, both denoted by $\nabla$. 
They induce connections on the tensor bundles $\underbrace{T^*M\otimes\cdots\otimes  T^*M}_{n-1}\otimes B$, again denoted by $\nabla$. By a tensor product of bundles we understand a vector bundle over $M$ whose fiber is the $n$-fold tensor product of corresponding fibers.
 For a continuously differentiable section $\ph\in \mathcal{C}^1(B)$ the covariant derivative is a continuous section in $T^*M\otimes B$, $\nabla\ph \in \mathcal{C}^0(T^*M\otimes B)$. More generally, for $\ph\in \mathcal{C}^k(B)$ we get $\nabla\ph \in \mathcal{C}^0(\underbrace{T^*M\otimes\cdots\otimes  T^*M}_{k}\otimes B)$. We choose a metric on $T^*M$ and on $B$. This induces metrics on all bundles $T^*M\otimes\cdots\otimes  T^*M\otimes B$ and for a subset $U \subset M$ and $\ph \in \mathcal{C}^k(B)$ we define the $\Ccal^k$-norm by
 \[
||\ph||_{\Ccal^k(U)} \doteq \max_{j=0,...,k}\sup_{x\in U} |\nabla^j\ph(x)|\,.
\]
For compact $U$ all choices of metrics and connections yield equivalent norms. Using the $\Ccal^k$-norm defined above one can introduce locally convex topologies on spaces $\Gamma_c(B)$ and $\Gamma(B)$ and define spaces of distributions as their duals $\Gamma_c'(B)$ and $\Gamma'(B)$\index{distributions!vector-valued}. The theory of distributions on $\RR^n$ can be easily generalized to the distributions on a manifold. 

Since on Lorentzian manifolds we have a distinguished volume element we identify smooth functions with distributions using the following prescription:
\begin{exa}
Every locally integrable section $f \in L^1_\loc(M,V)$ can be interpreted as a
 distribution by setting for any $\ph\in \Gamma_c(M,V^*)$:
\[
f(\ph)\doteq\int\limits_M\dvol \ph(f)\,.
\]
\end{exa}
All the definitions and properties mentioned in the previous section extend easily to manifolds, since locally they are isomorphic to $\RR^n$.
\subsection{Distributions with values in a graded algebra}\label{vvalued}
In this section we describe in details the theory of vector valued distributions. We focus on the case when the vector space in question is a graded infinite dimensional algebra $\A$. We denote the graded product of $\A$ by $\wedge$. L. Schwartz  \cite{Sch1} defines vector valued distributions in a following way:
\begin{defn}\label{vector}
Let $X$ be a \textsc{lcvs} with a topology defined by a separable family of seminorms $\{p_\alpha\}_{\alpha\in I}$. We say that $T$ is a distribution on $\RR^n$ with values in $X$ if it is a continuous linear mapping from $\Dcal$ to $X$, where $\Dcal$ denotes the space of compactly supported functions on $\RR^n$.
\end{defn}
Under some technical assumptions we can identify the space of distributions with values in $\A$ with the appropriately completed tensor product $\Dcal'\widehat{\otimes}\A$. In the context of this thesis the choice of a suitable topology for this completion will always be quite natural. We will discuss it later on with specific examples. The notion of vector-valued distribution enables us to formulate the classical field theory involving anticommuting fields in a mathematically elegant way. This will be discussed in section \ref{fer}. One can generalize all well known operations like convolution, Fourier transform and pullback to such objects \cite{Hoer,Sch1,Sch2}.
\begin{defn}
Let  $T=t\otimes f$ and $\phi=\varphi\otimes g$, where $f,g\in \A$, $t\in\Dcal'$ and $\varphi\in\Dcal$. We have an antisymmetric bilinear product on $\A$ defined as: $m_a(T,S)\doteq T\wedge S$. We define the convolution of $T$ and $\phi$ by setting:
\begin{equation}
(T*\phi)(x)\doteq t(\phi(x-.))\otimes m_a(f,g)\,.
\end{equation}
The extension by the sequential continuity to $\Dcal'\hat{\otimes}\A$ defines a convolution of a vector-valued distribution with a vector-valued function.
\end{defn}
\begin{defn}
Let  $T=t\otimes f$ and $S=s\otimes g$, where $f,g\in \A$, $t\in\Ecal'(\RR^2)$ and $s\in\Dcal'$. We define the convolution of $T$ and $S$ by setting:
\begin{equation}
T*S\doteq \int t(.,y)s(y)dy\otimes m_a(f,g)\,,
\end{equation}
This expression is well defined by \cite[4.2.2]{Hoer} and can be extended by continuity to arbitrary $S\in\Dcal'\hat{\otimes}\A$,  $T\in\Ecal'\hat{\otimes}\A$.
\end{defn}
\begin{defn}
In a similar spirit we define the evaluation of $T=t\otimes f$ on $\phi=\varphi\otimes g$,  by:
\begin{equation}
\left<T,\phi\right>\doteq \left<t,\varphi\right>\otimes m_a(f,g)\,,
\end{equation}
where $f,g\in \A$, $t\in\Dcal'$ and $\varphi\in\Dcal$. Also this can be extended by continuity to  $\Dcal'\hat{\otimes}\A$.
\end{defn}
Let  $\mathscr{S}$ denote the space of rapidly decreasing functions, i.e. such that: $\sup_x|x^\beta\partial^\alpha\phi(x)|<~\infty$ for all multi-indices $\alpha,\beta$.
\begin{defn}
Let $T\in\mathscr{S}'\hat{\otimes}\A$. We define $\hat{T}\in\mathscr{S}'\hat{\otimes}\A$, the Fourier transform of $T$ as:
\begin{equation}
\hat{T}(\phi)=T(\hat\phi)\qquad\phi\in\mathscr{S}\,.
\end{equation}
\end{defn}
Also the notion of the wave front set\index{wavefront set} \cite{Hoer} can be extended to distributions with values in a \textsc{lcvs}. The case of Banach spaces was already treated in detail in \cite{Ko}.
\begin{defn}
Let $\{p_\alpha\}_{\alpha\in A}$ be the family of seminorms generating the locally convex topology on $\A$. Let  $T\in\mathscr{S}'\hat{\otimes}\A$. A point $(x,\xi_0)\in T^*\RR^{n}\setminus 0$ is not
in $\textrm{WF}(T)$, if and only if $p_\alpha(\widehat{\phi u}(\xi))$
is fast decreasing as $|\xi|\rightarrow\infty$ for all $\xi$ in an open
conical neighbourhood of $\xi_0$, for some $\phi\in \Dcal$
with $\phi(x)\neq 0$, $\forall \alpha\in A$.
\end{defn}
With the notion of the wave front set we can define a ``pointwise product'' of two distributions $T,S\in\Dcal'\hat{\otimes}\A$ by a straightforward extension of \cite[8.2.10]{Hoer}:
\begin{prop}
Let $T,S\in\Dcal'\hat{\otimes}\A$, $U\in M$ (open). The product $T\cdot S$ can be defined as the pullback of $m_a\circ(T\otimes S)$ by the diagonal map $\delta:U\rightarrow U\times U$ unless $(x,\xi)\in\textrm{WF}(T)$ and $(x,-\xi)\in\textrm{WF}(S)$ for some $(x,\xi)$.
\end{prop}
Obviously we have:  $T\cdot S=(-1)^{|T||S|}S\cdot T$, whenever these expressions are well defined. In the paper we also use a more suggestive notation: $T\cdot S\doteq\left<T,S\right>$.
\section{Infinite dimensional calculus}\label{idc}
After a short trip into the realm of distributions now it's time for a really exciting adventure! We set for a journey into the area of the Land of Mathematics, that was discovered not so long ago and still hides some dangerous surprises. \textsc{hic svnt leones}\ldots Nevertheless it's worth to take a risk, because it is undoubtedly beautiful and fascinating. We enter now the realm of the infinite dimensional calculus\index{infinite dimensional!calculus}. As a guide in this strange land we shall use the lecture notes of K.-H. Neeb \cite{Neeb} and the book of P. Michor \cite{Michor}. We start with some historical remarks. Infinite dimensional differential calculus came into general attention quite recently, mainly due to works of Hamilton \cite{Ham} and Milnor \cite{Mil}. Nevertheless the idea itself seems to be much older. Perhaps the need for such a generalization became apparent already to Bernoulli and Euler at the beginnings of variational calculus. During the 20-th century the motivation for a calculus in spaces more general than Banach spaces became stronger partially due to possible applications in physics. There were many different approaches and definitions but they all met similar difficulties. One of them is the fact that the composition of continuous linear mappings ceases to be a jointly continuous operation for any suitable 
topology on spaces of linear mappings.
\begin{exa}[after \cite{Michor}]\label{evex}
Consider the evaluation $ev : E \times E^* \rightarrow \mathbb{R}$, where $E$ is a locally convex space and 
$E^*$ is its dual of continuous linear functionals equipped with any locally convex 
topology. Let us assume that the evaluation is jointly continuous. Then there are 
neighborhoods $U \subseteq E$ and $V \subseteq E^*$ of zero such that $ev(U \times V ) \subseteq [-1, 1]$. But then 
$U$ is contained in the polar of $V$ , so it is bounded in $E$, and so $E$ admits a bounded 
neighborhood and is thus normable.
\end{exa}
We recal that the polar of a set is defined as follows:
\begin{defn}
Given a dual pair $(X,Y)$ the polar set or polar of a subset $A$ of $X$ is a set $A^\circ$ in $Y$ defined as:
\begin{equation}
    A^\circ := \{y \in Y : \sup\{|\langle x,y \rangle |: x \in A \} \le 1\}
\end{equation}
 \end{defn}
 This simple counterexample shows that as soon as we leave Banach spaces, we get into trouble even with very harmless operations like an evaluation. Having this in mind we will now cautiously move forward into the dangerous realm of infinite dimensional calculus.
 
The problem of defining a derivative on a locally convex space has roots in the variational calculus. The calculus of variations started evolving into a rather formal procedure, used extensively in physics. At the same time Weierstrass in his lectures gave more reliable foundations to the theory, which 
was made public by Kneser (A. Kneser, \textit{Lehrbuch der Variationsrechnung}, Vieweg, Braunschweig, 1900). Further development went into the direction of the theory of partial differential equations. 

The most commonly used definitions of a derivative are the 
\textit{Fr\'echet derivative} and the \textit{G\^ateaux derivative}, but there are many more. In \cite{Michor} the authors recall (after Averbukh, Smolyanov \cite{Av}) that in the literature one finds 25 inequivalent definitions of the first derivative (in \textsc{tvs}) in 
a single point. This shows that finite order differentiability beyond Banach spaces is really a nontrivial issue. For continuously differentiable mappings the 
many possible notions reduce to 9 inequivalent ones (fewer for Fr\'echet spaces). 
And if we look for infinitely often differentiable mappings, then we end up with 
6 inequivalent notions (only 3 for Fr\'echet spaces).

The next step into the direction of infinite dimensional geometry is the definition of a manifold.
First idea to generalize this notion to infinite dimensions was a manifold modeled on a Banach space (Banach manifold). Later it turned out that Banach manifolds are not suitable for some of the important 
questions of global analysis. A counterexample, interesting also from a physical point of view is due to 
\cite{Omo}, see also \cite{Omo2}: If a Banach Lie group acts 
effectively on a finite dimensional compact smooth manifold it must be finite dimensional itself. Since in the context of gauge theories we want to consider infinite dimensional Lie groups,
we would like to have a notion of a manifold modeled on a more general space: Fr\'echet or just locally convex one.

Another issue related to the smooth calculus is the so called \textit{\textbf{cartesian closedness}}. One would like to have a property:
 \begin{equation}
 \Ci(E \times F, G) \cong \Ci(E, \Ci(F, G))\qquad\textrm{(NOT TRUE!)},\label{cartcl}
\end{equation}
which is called the \textit{\textbf{cartesian closedness}}. This is a property fulfilled by many well behaving categories, but the category of smooth manifolds doesn't have it. This was a motivation for developing generalizations of smooth manifolds, so called \textit{smooth spaces}. Best known approaches are:
\begin{itemize}
\item Chen spaces (Chen, 1977)
\item diffeological spaces (Souriau 1980)
\end{itemize}
The categories of smooth spaces defined in this way have all the nice properties, but their objects are quite difficult to handle. This was a motivation for development of different approaches, that aim at cartesian closedness but on the same time provide a calculus which is relatively ``easy to use''. In \cite{FK} (see also \cite{Michor,Froe}) a smooth calculus was proposed which has a property (\ref{cartcl}) holding without any restrictions for the so called \textbf{\textit{convenient vector spaces}}. The key idea is to define a different, finer topology on the product $E \times F$. The approach of  \cite{FK} is based on bornological instead of topological concepts. 

However if one wants to define a smooth manifold basing on a concept of charts, then the cartesian closedness is very limited even in the convenient setting (see the discussion in \cite{Michor}, chapter IV). A way out is to base a definition of a manifold on the concept of the family of smooth mappings (see: \cite{Seip,Kri}). In other words, one specifies explicitly which mappings are ``smooth''.

In the present work we take a rather ``pragmatic'' point of view and choose a setting that resembles closely the finite dimensional case and is easily applicable to problems at hand. Besides, the physical examples we consider involve relatively well behaving spaces, namely Fr\'echet or nuclear \textsc{lcvs}. In the next section we give a short introduction to calculus on locally convex vector spaces  based on the lecture notes \cite{Neeb}.
\subsection{Calculus on locally convex vector spaces}\label{clcvs}
We start with the notion of derivative. Actually it resembles just the directional derivative we know from the finite dimensional calculus.
Let $X$ and $Y$ be topological vector spaces, $U \subseteq X$ an open set and $f:U \rightarrow Y$ a map. The derivative of $f$ at $x$ in the direction\index{derivative!on a locally convex vector space} of $h$ is defined as
\be\label{de}
df(x)(h) \doteq \lim_{t\rightarrow 0}\frac{1}{t}\left(f(x + th) - f(x)\right)
\ee
whenever the limit exists. The function $f$ is called differentiable\index{infinite dimensional!calculus} at $x$ if $df(x)(h)$ exists for all $h \in X$. It is called continuously differentiable if it is differentiable at all points of $U$ and
$df:U\times X\rightarrow Y, (x,h)\mapsto df(x)(h)$
is a continuous map. It is called a $\Ccal^1$-map if it is continuous and continuously differentiable. Higher derivatives are defined for $\Ccal^n$-maps by 
\be
d^n f (x)(h_1 , \ldots , h_n ) \doteq \lim_{t\rightarrow 0}\frac{1}{t}\big(d^{n-1} f (x + th_n )(h_1 , \ldots, h_{n-1} ) -
 d^{n-1}f (x)(h_1 , \ldots, h_{n-1}) \big)
 \ee
The derivative defined by (\ref{de}) has many nice properties. It is shown for example in \cite{Neeb,Ham}, that the following proposition holds:
\begin{prop}
Let $X$ and $Y$ be locally convex spaces, $U \subseteq X$ an open subset, and $f:U \rightarrow Y$ a continuously differentiable function. Then:
\begin{enumerate}
\item	For any $x \in U$ , the map $df(x): X \rightarrow Y$ is real linear and continuous.
\item (Fundamental Theorem of Calculus). If $x + [0, 1]h \subseteq U$ , then 
\[
f (x + h) = f (x) +\int\limits_0^1 df (x + th)(h) dt\,.
\]
\item $f$ is continuous.
\item If $f$ is $\Ccal^n$, $n \geq 2$, then the functions $(h_1,...,h_n) \mapsto d^nf(x)(h_1,...,h_n)$, $x \in U$, are
symmetric $n$-linear maps.
\item If $x + [0, 1]h \subseteq U$,then we have the Taylor Formula:
\begin{multline*}
f (x + h) = f (x) + df (x)(h) + \ldots+\frac{1}{(n-1)!}d^{n-1} f (x)(h,\ldots, h)+\\
+\frac{1}{(n-1)!}\int\limits_0^1(1-t)^{n-1}d^nf(x+th)(h,...,h)dt\,.
\end{multline*}
\end{enumerate}
\end{prop}
Now, following \cite{Neeb} we shall introduce a notion of an infinite dimensional manifold\index{infinite dimensional!manifold}.
Let $M$ be a Hausdorff topological space and $E$ a \textsc{lcvs}. An $E$-chart of an open subset $U \subseteq M$ is a homeomorphism $\varphi:U \rightarrow \varphi(U) \subseteq E$ onto an open subset $\varphi(U)$ of $E$. We denote such a chart as a pair $(\varphi,U)$. Two charts $(\varphi,U)$ and $(\psi,V)$ are said to be smoothly compatible if the map
$\psi \circ \varphi^{-1}\Big|_{\varphi(U\cap V)}: \varphi(U\cap V ) \rightarrow \psi(U \cap V )$
is smooth.	\\
An $E$-atlas	of	$M$	is	a family	$(\varphi_i, U_i)_{i\in I}$	of	pairwise compatible $E$-charts of $M$ for which $\bigcup_i U_i=M$. Many of the objects used in differential geometry can be defined also in the infinite dimensional case. We start with the notion of a tangent space. 
\begin{defn}
Let $x$ be an element of a locally convex vector space $E$. A kinematic tangent vector with foot point $x$ is a pair $(x, Q)$ with $Q\in E$. $T_x E \cong E$ is the space of all kinematic tangent vectors with foot point $x$. It consists of all derivatives $c'(0)$ at $0$ of smooth curves $c : \RR \rightarrow E$ with $c(0) = x$. The kinematic tangent space of a locally convex vector space $E$ will be denoted by $TE$ and the space of vector fields by $\Gamma(TE)$.
\end{defn}
We use the term \textit{kinematic} since in the most general case this definition doesn't coincide with the definition of vector fields as derivations. Fortunately for the spaces considered in the context of BV formalism this doesn't pose a problem.
\begin {defn}
Let $M$ be a smooth manifold with the atlas $(\varphi_i, U_i)_{i\in I}$, where $\varphi_i:U_i\rightarrow E_i$. We consider the following equivalence relation on the disjoint union
\[
\bigcup\limits_{i\in I}U_i\times E_i \times \{i\}\,,
\]
\[
(x,v,i) \sim (y,w,j) \Leftrightarrow x = y\ \textrm{and}\ d(\ph_{ij})(\ph_j(x))w = v\,,
\]
where $\ph_{ij}$ are the transition functions. One denotes the quotient set by $TM$, the \textbf{kinematic tangent bundle} of $M$.  A kinematic vector field\index{infinite dimensional!vector fields} $X$ on M is just a smooth section of the kinematic tangent bundle. 
\end {defn}
On finite dimensional manifolds we can define vector fields equivalently as bounded derivations of the sheaf of smooth functions. In infinite dimensional geometry those notions don't coincide in general. The vector fields defined as derivations are called \textit{\textbf{operational vector fields}}.
\begin{defn}[\cite{Michor}, 32.1]
By an operational vector field $X$ on $M$ we mean a bounded derivation of the sheaf $\Ci(. ,\RR)$, i.e. for the open $U \subset M$ we are given bounded derivations $X_U : \Ci(U, \RR) \rightarrow \Ci(U, \RR)$ commuting with the restriction mappings.
\end{defn}
Kinematic vector fields are contained in the space of operational vector fields but the opposite is not always true. In the present work by vector fields we always mean the kinematic vector fields. We already see that although many of the finite dimensional concepts can be generalized to the infinite dimensional case, one has to be extremely cautious. For example the existence of local flows is not guaranteed if we go beyond Banach spaces, since the implicit function theorem doesn't hold in the simple form.  A detailed discussion of these issues can be found in section 32 of \cite{Michor}.

Now we want to define differential forms on an infinite dimensional manifold. This turns out to be a problem, since there is no natural notion of a cotangent space. To see the source of the problem note that while we can define for each manifold $M$ modeled on a \textsc{lcvs} $E$ a cotangent bundle  by $T^*M = \bigcup\limits_{m\in M} (T_mM)'$ and endow it with a vector bundle structure over $M$, we cannot endow it with a smooth manifold structure. To be able to do it we would need a locally convex topology on the dual space $E'$ such that for each local diffeomorphism $f:U \rightarrow E$, $U$ open in $E$, the map $U \times E' \rightarrow E'$, $(x,\lambda) \mapsto \lambda \circ df(x)$ is smooth. This is fulfilled for the norm topology on a Banach space but becomes a problem if we consider manifolds modeled on more general \textsc{lcvs} (compare with the example \ref{evex}).

Nevertheless one can still introduce differential forms by a direct definition without reference to the cotangent bundle. The existing alternative definitions of differential forms are unfortunatelly not equivalent. A detailed discussion of this problem is given in \cite{Michor} (VII.33). It turns out that only one of these notions is stable under Lie derivatives, exterior derivative, and the pullback. 
\begin {defn}[\cite{Michor},VII.33.22]\label{33.22}
Let $M$ be a smooth infinite dimensional manifold. We will define
the space of differential forms\index{infinite dimensional!differential forms} on $M$ as:
\be\label{forms}
\Omega^k(M) \doteq \Ci(M \leftarrow L^k_\alte(TM,M \times \RR))\,.
\ee
Similarly, we denote by $\Omega^k(M;V) \doteq \Ci(M \leftarrow L^k_\alte(TM,M \times V))$
the space of differential forms with values in a locally convex vector space $V$.
\end {defn}
\subsection{Infinite dimensional Lie groups}
In physics one would like to treat certain spaces of functions as infinite dimensional Lie groups\index{infinite dimensional!Lie group}. To put it in an appropriate mathematical setting one needs first a notion of an infinite dimensional manifold. A definition proposed in \cite{Neeb} makes it possible to provide certain infinite dimensional spaces with the structure of a manifold modeled on a locally convex vector space. One can apply all tools of locally convex analysis  in this case. Unfortunately this definition doesn't cover all the interesting cases. In particular it fails for the spaces of mappings between noncompact manifolds.
\begin{exa}[after \cite{Neeb}] If $M$ is a non-compact finite-dimensional manifold, then one cannot expect 
the topological groups $\Ci(M, K)$ to be Lie groups. A typical example arises for $M = \mathbb{N}$ (a 
0 - dimensional manifold) and $K = \mathbb{T} := \mathbb{R}/\mathbb{Z}$. Then $\Ci(M, K) \cong \mathbb{T}^{\mathbb{N}}$ is a topological group for which no  1-neighborhood is contractible.
\end{exa}
This means that one cannot consider $\Diff(M)$ to be an infinite dimensional Lie group for a noncompact manifold $M$. One can overcome this problem by restricting to compactly supported diffeomorphisms \cite{Neeb04,Gloe,Gloe06}. We discuss those issues on concrete examples in chapter \ref{BVform}.
\section{Categories and functors}\label{kat}
In the locally covariant framework \cite{BFV} one defines a classical or a quantum field theory as a functor between certain categories. This is motivated by the fact, that we need to define the theory on all the spacetimes of the given class in a coherent way. Before we recapitulate the framework of locally covariant field theory in section \ref{classfunc} we recall some notions from the category theory that will be used later on.
\begin{defn}
A category\index{category} $\Ca$ consists of
\begin{itemize}
\item  a class $\obj(\Ca)$ of objects,
\item a class $\Hom(\Ca)$ of morphisms between the objects. Each morphism $f$ has a unique  {source object} $a$ and  {target object} $b$ where $a,b \in \obj(\Ca)$.\\ Notation: if $f: a \rightarrow b$  then we write $f\in \Hom(a, b)$ 
\item for $a,b,c\in \obj(\Ca)$, a binary operation $\Hom(a, b) \times \Hom(b, c) \rightarrow \Hom(a, c)$ called  {composition of morphisms}. It is denoted by $f\circ g$.
\end{itemize}
such that the following axioms hold:
\begin{itemize}
\item  \textbf{associativity}: if $f : a \rightarrow b$, $g : b \rightarrow c$ and $h : c \rightarrow d$ then $h \circ (g \circ f) = (h \circ g) \circ f$
\item  \textbf{identity}: for every object $c$, there exists a morphism $\mathrm{id}_c : c \rightarrow c$, such that for every $\Hom(a,b)\ni f$we have: $\mathrm{id}_b \circ f = f \circ \mathrm{id}_a$=f.
\end{itemize}
\end{defn}
An important category in our formalism is the category of spacetimes\index{category!of spacetimes}, with isometric causal embeddings as morphisms. More precisely we have:
\begin{itemize}
	\item[$\Loc$]
		\begin{itemize}
 			\item[$\ $] $\obj(\Loc)$: all four-dimensional, globally
			hyperbolic oriented and time-oriented spacetimes $(M,g)$.
			\item[$\ $] {\bf Morphisms}: Isometric embeddings that fulfill:
		\begin{itemize}
			\item[$\bullet$] Given $(M_1,g_1),(M_2,g_2)\in\obj(\Loc)$, for any 
			causal curve $\gamma : [a,b]\to M_2$, if $\gamma(a),\gamma(b)\in\chi(M_1)$ then for all $t	\in ]a,b[$ we have: $\gamma(t)\in\chi(M_1)$ (see figure \ref{ccurve}). This property is called: causality preserving.
			\item[$\bullet$] They preserve the orientation and
			time-orientation of the embedded spacetime.
		\end{itemize} 
		\end{itemize}
		\end{itemize}
\begin{figure}[tb]
\begin{center}
 \begin{tikzpicture}[scale=0.85]
\shadedraw[shading=axis,shading angle=90, left color=see!80!white,right color=white,opacity=0.8] plot[smooth cycle] coordinates{(-1,-1) (-2,0)(-2.1,1)(-1,2)(1,2.3)(2,2.5)(3,2)(3,1)(2,0)(1,-1)(0,-1.2)};
\draw(-1.7,0.5) node {$M_2$};
\draw(0.2,1.3) .. controls (0,1) and (1,0.8) .. (0.2,-0.5);
\filldraw[red] (0.2,1.3) circle (1pt)
		     (0.2,-0.5) circle (1pt);
\draw node[left,red]  at (0.2,1.3) {\footnotesize$\gamma(b)$};
\draw node[left,red] at (0.2,-0.5){\footnotesize$\gamma(a)$};
\draw (0.2,0.5) node[name=s,minimum size=2cm,draw,diamond] {}
node[anchor=south west] {};
\draw (0,-3) node[shading=axis,shading angle=90, left color=see!80!white,right color=white,name=t,minimum size=2cm,draw,diamond,opacity=0.8] {}
node {$M_1$};
\path[->] (t.north east) edge  [color=red, bend right]  node[right] {$\psi$} (s.south east);
\end{tikzpicture}
\end{center}
\vspace{-3ex}
\caption{Causality preserving embedding $\psi$.\label{ccurve}}
\end{figure}
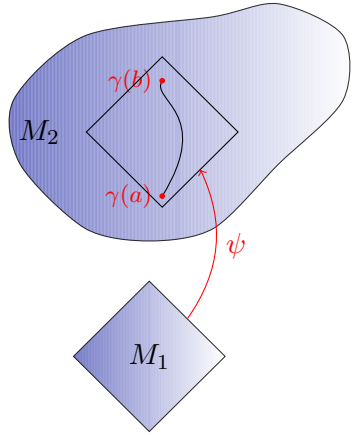
The other important categories are: the category of locally convex topological vector spaces and the category of Poisson algebras.
\begin{itemize}
\item[$\Vect$]
\begin{itemize}
 \item[$\ $] $\obj(\Vect)$: (small) topological vector spaces\index{category!of topological vector spaces}
 \item[$\ $] {\bf Morphisms}: homomorphisms of  topological vector spaces
  \end{itemize}
\item[$\Obs$]
\begin{itemize}
 \item[$\ $] $\obj(\Obs)$: Poisson (graded) algebras\index{category!of Poisson algebras} 
 \item[$\ $] {\bf Morphisms}: Poisson (graded) morphisms
  \end{itemize}
\item[$\Top$]
\begin{itemize}
 \item[$\ $] $\obj(\Top)$: topological spaces \index{category!of topological spaces}
 \item[$\ $] {\bf Morphisms}: continuous maps
  \end{itemize}
\end{itemize} 
Now we define the next important notion from the category theory, namely a \textit{functor}\index{functor!covariant}.
\begin{defn}
Let $\Ca$ and $\Da$ be categories. A \textit{covariant functor} ${\F}$ from $\Ca$ to $\Da$ is a mapping that:
\begin{itemize}
\item associates to each object $c \in Obj(\Ca)$ an object ${\F}(c) \in Obj(\Da)$,
 \item associates to each morphism $\Hom(\Ca)\ni f:a\rightarrow b \in$, a morphism $\Hom(\Da)\ni {\F}(f):{\F}(a) \rightarrow {\F}(b) $
\end{itemize}
 \be\label{covfunctor}
\begin{CD}
a @>f>> b\\
@V{{\F}}VV     @VV{{\F}}V\\
{\F}(a)@>{\F}(f)>> {\F}(b)
\end{CD}
\ee
such that the following two conditions hold:
\begin{itemize}
    \item ${\F}(\mathrm{id}_{c}) = \mathrm{id}_{{\F}(c)}$ for every object $c \in \Ca$.
    \item ${\F}(g \circ f) = {\F}(g) \circ {\F}(f)$ for all morphisms $f:a \rightarrow b$ and $g:b\rightarrow c.\,\!$
\end{itemize}
\end{defn}
If a similar definition holds with arrows reversed, we say that a functor is contravariant\index{functor!contravariant}.
Just as the functors provide a notion of mappings between categories, one can go a level of abstraction higher and think how the functors should transform into each other. This leads to the notion of a \textit{natural transformation}\index{natural transformation}.
\begin{defn}
Let ${\F}$ and $\fG$ be functors between categories $\Ca$ and $\Da$, then a natural transformation $\eta$ from ${\F}$ to $\fG$ associates to every object $a \in \Ca$ a morphism $\Hom(\Da)\ni\eta_a : {\F}(a)\rightarrow \fG(a)$, such that for every morphism $\Ca\ni f : a \rightarrow b$ we have:
\[
    \eta_b \circ {\F}(f) = \fG(f) \circ \eta_a\,.
\]
This equation can be expressed by:
 \[
\begin{CD}
{\F}(a) @>{\F}(f)>> {\F}(b)\\
@V{\eta_a}VV     @VV{\eta_b}V\\
\fG(a)@>\fG(f)>> \fG(b)
\end{CD}
\]
\end{defn}
\chapter{Classical field theory in the locally covariant framework}\label{classfunc}
\vspace{-5ex}
\begin{flushright}
 \begin{minipage}{10cm}
\textit{Ordnung und Sichtung sind der Anfang der Beherrschung, und der eigentlich furchtbare Feind ist der unbekannte.}
\begin{flushright} 
T. Mann \textit{Der Zauberberg}
\end{flushright}
 \end{minipage}
\end{flushright}
\vspace{3ex}
\noindent\rule[2pt]{\textwidth}{1pt}
\vspace{1ex}\\
From the Land of Mathematics our journey leads us back to physics. It's nice to be home again and it's about time to use  in practice the tools and techniques we learnt. We start with the classical theory since it is technically less complicated than {\qft} and, remarkably, some structures used in the quantum world are present already in the classical setting. Moreover they can be better understood by investigating their classical counterparts. The main goal of this chapter is to provide a brief introduction into the locally covariant formulation of classical field theory. Starting classically we have time to get a little  bit acquainted with the concepts, that will be used later on in the quantum case. 
\section{Kinematical structure}
In the locally covariant framework \cite{BFV} with a physical system we associate first the configuration space\index{configuration space} of all the fields of the theory. The main principle of the covariant approach is to make this assignment for all the spacetimes in a coherent way. This can be formulated with the use of the category theory language. We define a contravariant functor $\E$ from the category of spacetimes $\Loc$ to the category of locally convex vector spaces $\Vect$. It assigns to each spacetime $M$ the configuration space $\E(M)$ of fields defined on it and the isometric embeddings $\chi:M\rightarrow N$ are mapped into pullbacks $\chi^*:\E(N)\rightarrow \E(M)$. For example in case of the scalar field we have $\E(M)=\Ci(M)$. We illustrate this on figure \ref{conf}.
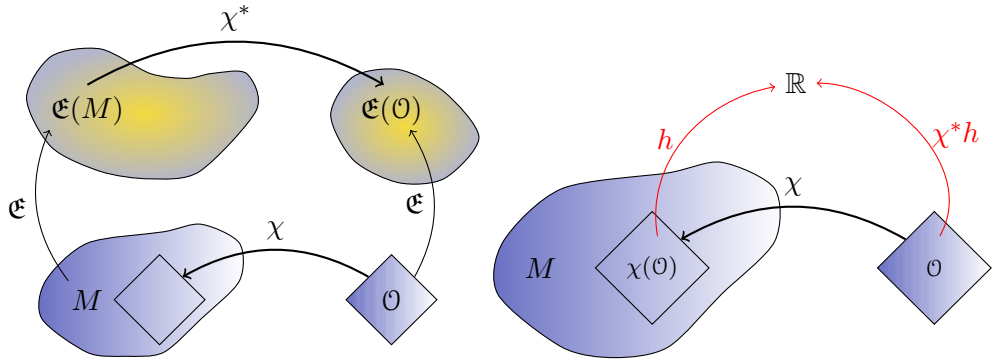
\begin{figure}[!b]
\begin{center}
\begin{tikzpicture}[scale=0.5]
\shadedraw[shading=radial, inner color=honey!50!yellow,outer color=see!40!lighthoney,scale=0.65] plot[smooth cycle] coordinates{(-1,6)(-3,7)(-4,9)(-3,10.7)(-1,11)(1,10)(2.5,9.8)(3.5,10)(4.9,10)(5.3,8.5)(3,6)(1,5.7)};
\draw(-1,5.5) node {$\E(M)$};
\shadedraw[shading=axis,shading angle=90, left color=see!80!white,right color=white] plot[smooth cycle] coordinates{(-1,-.9) (-2.2,0)(-2.1,1)(-1,2)(1,2.3)(2,2.5)(3,2)(3,1)(2,0)(1,-.9)(0,-1)};
\draw(-1,0.5) node {$M$};
\path[->] (-1.5,1) edge  [bend left]  node[left] {$\E$} (-2,5);
\draw(0.9,0.5) node[name=s,minimum size=1.2cm,draw,diamond] {};
\draw(7,0.5) node[shading=axis,shading angle=90, left color=see!80!white,right color=white,name=t,minimum size=1.2cm,draw,diamond] {}
node {$\Ocal$};
\path[->] (t.north west) edge  [bend right,thick]  node[above] {$\chi$} (s.north east);
\shadedraw[shading=radial, inner color=honey!50!yellow,outer color=see!40!lighthoney,scale=0.6] plot[smooth cycle] coordinates{(12,6)(10,7)(9,9)(10,10.7)(12,11)(14,10)(15.5,8)(14,6.2)};
\path[->] (t.north east) edge  [bend right]  node[left] {$\E$} (7.5,5);
\draw(7,5.5) node {$\E(\Ocal)$};
\path[->] (-1,6.2) edge  [bend left,thick]  node[above] {$\chi^*$} (6.8,6.2);
\end{tikzpicture}
\begin{tikzpicture}[scale=0.7]
\shadedraw[shading=axis,shading angle=90, left color=see!80!white,right color=white] plot[smooth cycle] coordinates{(-1,-1) (-2.2,0)(-2.1,1)(-1,2)(1,2.3)(2,2.5)(3,2)(3,1)(2,0)(1,-1)(0,-1.2)};
\draw(-1.4,0.5) node{$M$};
\draw(3.4,4) node [name=real] {$\RR$};
\draw(.7,0.5) node[name=s,minimum size=1.5cm,draw,diamond, inner sep=0pt] {\footnotesize$\chi(\Ocal)$};
\draw (6,0.5) node[shading=axis,shading angle=90, left color=see!80!white,right color=white,name=t,minimum size=1.5cm,draw,diamond] {}
node { \footnotesize$\Ocal$};
\path[->] (t.north west) edge  [bend right,thick]  node[above] {$\chi$} (s.north east);
\path[->] (6.1,1.1) edge  [out=60, in=0,color=red]  node[right] {$\chi^*h$} (real);
\path[->] (0.8,1.1) edge  [out=100, in=190,color=red]  node[left] {$h$} (real);
\end{tikzpicture}
\end{center}
\caption{Construction of the functor $\E$.\label{conf}}
\end{figure}
In all the physical cases one can identify $\E(M)$ with the space of sections of some vector bundle $B$ over $M$ with a fiber $V$, i.e. $\E(M)=\Gamma(B)\equiv\Gamma(M,V)$. As discussed in section \ref{topo}, the relevant topology is the topology of uniform convergence, together with all the derivatives, on compact subsets of $M$. 

Another functor between categories $\Loc$ and $\Vect$ is the functor which associates to a manifold the space $\E_c(M)$ of compactly supported configurations. This functor is covariant, because now the isometric embeddings can be mapped into the push forwards $\E_c\chi=\chi_*$, namely for $\ph\in\E_c(M)$ we define
\[
\chi_*\ph(x)=\left\{
\begin{array}{ccc}
\ph(\chi^{-1}(x)) & , & x\in\chi(M),\\
0                    & , & \text{else}.
\end{array} \right.
\]
In the same way we can construct a functor that associates to a manifold $M$ the space of test functions $\D(M)\doteq\Ci_c(M)$. The morphisms of $\Loc$ are again mapped to pushforwards (see picture \ref{testf}). Now we introduce the structure that is crucial in the functional approach to classical field theory. We want to identify the space which would contain the observables of our theory. It seems natural to define them simply as functionals $F:\E(M)\to\RR$. They are required to be smooth in the sense of the calculus on locally convex vector spaces (section \ref{idc}). Moreover one can see that if $F\in\Ci(\E(M))$ is a smooth functional, then
for each $n\in\NN$ and each $\ph\in\E(M)$, the functional derivative $F^{(n)}(\ph)$ is a symmetric distribution with compact support on $M^n$. The spacetime support\index{spacetime support!of a functional} of a functional $F$ is defined as the set of points $x\in M$ such that $F$ depends on the field configuration in any neighbourhood of $x$.
\begin{align}\label{support}
\supp\, F\doteq\{ & x\in M|\forall \text{ neighbourhoods }U\text{ of }x\ \exists \ph,\psi\in\E(M), \supp\,\psi\subset U 
\\ & \text{ such that }F(\ph+\psi)\not= F(\ph)\}\ .\nonumber
\end{align}
Equivalently we can understand the support of a functional $F$ as the closure of the sum of all the supports of its derivatives: 
\be\label{support2}
\supp\, F=\overline{\bigcup\limits_{{\ph\in\E(M)}\atop n\in\NN}}\supp(F^{(n)}(\ph))\,.
\ee
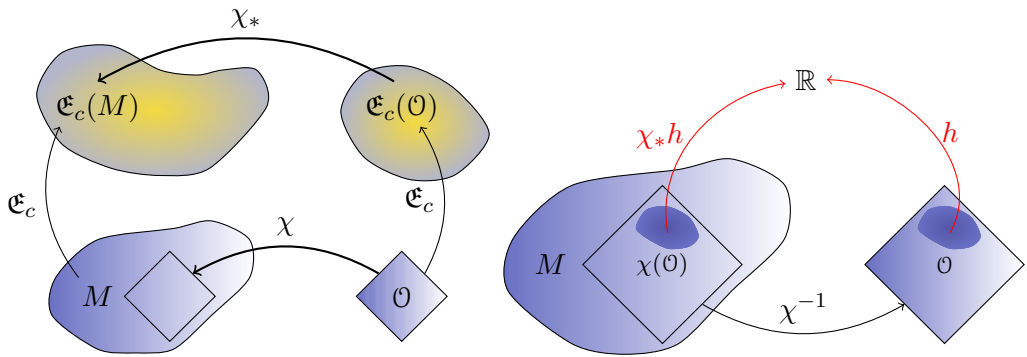
\begin{figure}[!tb]
\begin{center}
\begin{tikzpicture}[scale=0.5]
\shadedraw[shading=radial, inner color=honey!50!yellow,outer color=see!40!lighthoney,scale=0.65] plot[smooth cycle] coordinates{(-1,6)(-3,7)(-4,9)(-3,10.7)(-1,11)(1,10)(2.5,9.8)(3.5,10)(4.9,10)(5.3,8.5)(3,6)(1,5.7)};
\draw(-1,5.5) node {$\E_c(M)$};
\shadedraw[shading=axis,shading angle=90, left color=see!80!white,right color=white] plot[smooth cycle] coordinates{(-1,-.9) (-2.2,0)(-2.1,1)(-1,2)(1,2.3)(2,2.5)(3,2)(3,1)(2,0)(1,-.9)(0,-1)};
\draw(-1,0.5) node {$M$};
\path[->] (-1.5,1) edge  [bend left]  node[left] {$\E_c$} (-2,5);
\draw(0.9,0.5) node[name=s,minimum size=1.2cm,draw,diamond] {};
\draw(7,0.5) node[shading=axis,shading angle=90, left color=see!80!white,right color=white,name=t,minimum size=1.2cm,draw,diamond] {}
node {$\Ocal$};
\path[->] (t.north west) edge  [bend right,thick]  node[above] {$\chi$} (s.north east);
\shadedraw[shading=radial, inner color=honey!50!yellow,outer color=see!40!lighthoney,scale=0.6] plot[smooth cycle] coordinates{(12,6)(10,7)(9,9)(10,10.7)(12,11)(14,10)(15.5,8)(14,6.2)};
\path[->] (t.north east) edge  [bend right]  node[left] {$\E_c$} (7.5,5);
\draw(7,5.5) node {$\E_c(\Ocal)$};
\path[<-] (-1,6.2) edge  [bend left,thick]  node[above] {$\chi_*$} (6.8,6.2);
\end{tikzpicture}
\begin{tikzpicture}[scale=0.7]
\shadedraw[shading=axis,shading angle=90, left color=see!80!white,right color=white] plot[smooth cycle] coordinates{(-1,-1) (-2.2,0)(-2.1,1)(-1,2)(1,2.3)(2,2.5)(3,2)(3,1)(2,0)(1,-1)(0,-1.2)};
\draw(-1.4,0.5) node{$M$};
\draw(3.4,4) node [name=real] {$\RR$};
\draw(.7,0.5) node[name=s,minimum size=2.1cm,draw,diamond, inner sep=0pt] {\footnotesize$\chi(\Ocal)$};
\draw (6,0.5) node[shading=axis,shading angle=90, left color=see!80!white,right color=white,name=t,minimum size=2.1cm,draw,diamond] {}
node { \footnotesize$\Ocal$};
\shade[shading=radial, inner color=darksee!80!white,outer color=see!80!white,scale=0.2] plot[smooth cycle] coordinates{(30.5,4)(28.5,5)(27.5,6)(28.5,7.7)(30.5,8)(32.5,7)(33.4,5.2)(32.5,4.2)};
\shade[shading=radial, inner color=darksee!80!white,outer color=see!80!white,scale=0.2] plot[smooth cycle] coordinates{(4,4)(2,5)(1,6)(2,7.7)(4,8)(6,7)(6.9,5.2)(6,4.2)};
\path[<-] (t.south west) edge  [bend left]  node[above] {$\chi^{-1}$} (s.south east);
\path[->] (6.1,1.1) edge  [out=60, in=0,color=red]  node[right] {$h$} (real);
\path[->] (0.8,1.1) edge  [out=100, in=190,color=red]  node[left] {$\chi_*h$} (real);
\end{tikzpicture}
\end{center}
\caption{Construction of the functor $\E_c$.\label{testf}}
\end{figure}
We identify now a class of functionals, that are of particular interest in classical field theory. Those are so called \textit{local functionals}\index{local!functional}. According to the standard definition, a functional $F$ is called local if it is of the form:
\[
F(\ph)=\int\limits_M\dvol(x) f(j_x(\ph))\,,
\]
where $f$ is a function on the jet space over M and $j_x(\ph)=(x,\ph(x),\pa\ph(x),\dots)$ is the jet of $\ph$ at the point $x$. It was already recognized in \cite{DF04,BDF,BFR} in the context of perturbative algebraic quantum field theory that the property of locality can be reformulated using the notion of \textit{additivity}\index{additivity of functionals} of a functional. The concept itself dates back to the works of Chac\'on and Friedman \cite{ChF}. See also the survey of Rao \cite{Rao}. We say that $F$ is additive if for all fields $\varphi_1,\varphi_2,\varphi_3\in\E(M)$ such that $\textrm{supp}(\varphi_1)\cap\textrm{supp}(\varphi_3)=\varnothing$ we have:
\be\label{add}
F(\varphi_1+\varphi_2+\varphi_3)=F(\varphi_1+\varphi_2)-F(\varphi_2)+F(\varphi_2+\varphi_3)\,.
\ee
One shows that a smooth compactly supported functional is local if it is additive and
the wave front sets of its derivatives are orthogonal to the tangent bundles of the thin diagonals $\Delta^k(M)\doteq\left\{(x,\ldots,x)\in M^k:x\in M\right\}$, considered as subsets of the tangent bundles of $M^k$, i.e.: $\textrm{WF}(F^{(k)}(\ph))\perp T\Delta^k(M)$. In particular $F^{(1)}(\ph)$ is a smooth section for each fixed  $\ph$. 

The space of compactly supported smooth local functions $F:\E(M)\to\RR$ is denoted by $\F_\loc(M)$. The algebraic completion of $\F_\loc(M)$ with respect to the pointwise product\index{product!pointwise}
\be\label{prod}
F\cdot G(\ph)=F(\ph)G(\ph) \,,
\ee
is a commutative algebra $\F(M)$ consisting of sums of finite products of local functionals. We call it \textit{the algebra of multilocal functionals}\index{multilocal!functional}.
 $\F$ becomes a (covariant) functor by setting $\F\chi(F)=F\circ \E\chi$, i.e. $\F\chi(F)(\ph)=F(\ph\circ\chi)$ (see figure \ref{functionals}).
 \begin{figure}[!tb]
 \begin{center}
 	\begin{tikzpicture}[scale=0.6]
\shadedraw[shading=radial, inner color=honey!50!yellow,outer color=see!40!lighthoney,scale=0.65] plot[smooth cycle] coordinates{(-1,6)(-3,7)(-4,9)(-3,10.7)(-1,11)(1,10)(2.5,9.8)(3.5,10)(4.9,10)(5.3,8.5)(3,6)(1,5.7)};
\draw(0.2,4.2) node {\small$\E(M)$};
\shadedraw[shading=axis,shading angle=90, left color=see!80!white,right color=white] plot[smooth cycle] coordinates{(-1,-1) (-2.2,0)(-2.1,1)(-1,2)(1,2.3)(2,2.5)(3,2)(3,1)(2,0)(1,-1)(0,-1.2)};
\draw(-1,0.5) node {$M$};
\path[->] (-1.5,1) edge  [bend left]  node[left] {$\E$} (-2,4.6);
\draw (1,0.5) node[name=s,minimum size=1.5cm,draw,diamond] {};
\draw (7,0.5) node[shading=axis,shading angle=90, left color=see!80!white,right color=white,name=t,minimum size=1.5cm,draw,diamond] {}
node {$\Ocal$};
\path[->] (t.north west) edge  [bend right,thick]  node[above] {$\chi$} (s.north east);
\shadedraw[shading=radial, inner color=honey!50!yellow,outer color=see!40!lighthoney,scale=0.6] plot[smooth cycle] coordinates{(12,6)(10,7)(9,9)(10,10.7)(12,11)(14,10)(15.5,8)(14,6.2)};
\path[->] (t.north east) edge  [bend right]  node[left] {$\E$} (8.4,3.7);
\draw(7,4.2) node {\small$\E(\Ocal)$};
\path[->] (.2,6.4) edge  [bend left,thick]  node[above] {$\chi^*$} (6.8,6.4);
\draw(3.4,10) node [name=real] {$\RR$};
\path[->] (8,5.8) edge  [out=90, in=0]  node[right] {$F$} (real);
\path[->] (-1.7,5.8) edge  [out=90, in=180]  node[left] {$\F\chi F$} (real);
\filldraw  (8,5.8) circle (1pt)
		        (-1.7,5.8) circle (1pt);
\draw node[below]  at (8,6) {\small$\chi^*h$};
\draw node[below]  at (-1.7,6) {\small$h$};
\end{tikzpicture}
\end{center}
\caption{Definition of $\F\chi F$.\label{functionals}}
\end{figure} 
  In the context of quantum field theory one needs to enlarge this algebra with more singular objects, since it doesn't contain for example the Wick polynomials. Besides, even in classical theory, $\F(M)$ is not enough to build a Poisson algebra, because it turns out not to be closed under the Poisson bracket. We shall come back to this issue in section \ref{scal}.
 \section{Dynamics}\label{dyn}
Now we have the playground ready and we can start with something more physical. All the nice kinematical structures from the previous section don't describe any physics yet. A specific model is specified by introducing the \textit{dynamics}.
This can by done by means of the \textit{generalized Lagrangian}\index{generalized Lagrangian}. As the name suggests the idea is motivated by the Lagrange mechanics. Indeed, we can think of this formalism as a way to make precise the variational calculus
in field theory. Note that since we deal with globally hyperbolic spacetimes, they are not compact. Moreover we cannot restrict ourselves to compactly supported field configurations, since the nontrivial solutions of globally hyperbolic equations don't belong to this class. Therefore we cannot identify the action with a function on $\E(M)$ obtained 
by integrating the Lagrangian density over the whole manifold. Instead we follow \cite{BDF} and define a Lagrangian $L$ as a natural transformation between the functor of test function spaces $\D$ and the functor $\F_\loc$ such that it satisfies $\supp(L_M(f))\subseteq \supp(f)$ and the additivity rule 
\footnote{We do not require linearity since in quantum field theory the renormalization flow does not preserve the linear structure; it respects, however, the additivity rule (see \cite{BDF}).}
\[
L_M(f+g+h)=L_M(f+g)-L_M(g)+L_M(g+h)\,,
\]
for $f,g,h\in\D(M)$ and $\supp\,f\cap\supp\,h=\varnothing$.  
The action $S(L)$ is now defined as an equivalence class of Lagrangians  \cite{BDF}, where two Lagrangians $L_1,L_2$ are called equivalent $L_1\sim L_2$  if
\be\label{equ}
\supp (L_{1,M}-L_{2,M})(f)\subset\supp\, df\,, 
\ee
for all spacetimes $M$ and all $f\in\D(M)$. 
This equivalence relation allows us to identify Lagrangians differing by a total divergence.

The equations of motion are to be understood in the sense of \cite{BDF}. Concretely, the Euler-Lagrange derivative\index{derivative!Euler-Lagrange} of $S$ is a natural transformation $S':\E\to\D'$ defined as
\be\label{ELd}
\left<S'_M(\ph),h\right>=\left<L_M(f)^{(1)}(\ph),h\right>\,,
 \ee
 with $f\equiv 1$ on $\supp h$. The field equation is now a condition on $\ph$:
\be
 S_M'(\ph)=0\,.\label{eom}
\ee
The space of solutions of (\ref{eom}) is a subspace of $\E(M)$ , denoted by $\E_S(M)$. In the on-shell setting of classical field theory one is interested in the space $\F_S(M)$ of multilocal functionals on $\E_S(M)$. This space can be understood as the quotient $\F_S(M)=\F(M)/\F_0(M)$, where $\F_0(M)$ is the space of multilocal functionals that vanish on $\E_S(M)$ (on-shell). This observation will be crucial for the construction of the Koszul complex\index{Koszul!complex}, performed in \ref{K}.
\section{Axioms of locally covariant field theory}\label{axioms}
It is time to make a short stop in our journey and summarize what we already obtained and what we are heading for. Since at the end we want to obtain a satisfactory conceptual framework, it is important to formulate some guidelines and postulates which we want to fulfill. If you start a long journey, looking for something you haven't seen before, you always take with you some travel guides or maps that will help you recognize the object of your search. A set of axioms in mathematical physics plays a similar role. You want to identify the object of your search and make it clear, that as soon as all the axioms are fulfilled, you would know, that you reached the end of your journey. In previous sections we made a lot of effort to prepare a convenient language to formulate the classical field theory in a locally covariant way. Now it's time to state what this theory really \textit{is}. To this end we formulate now a set of axioms.

A local, generally covariant classical field theory\index{locally covariant!field theory!classical} is defined as a covariant functor $\fA$ between the category $\Loc$ of spacetimes and $\Obs$ of observables. 
The construction of this functor can be seen as a generalization of the local net of observables \cite{Haag}. The covariance property reads:
$$ \fA\psi' \circ \fA\psi = \fA(\psi' \circ \psi)\,,\quad \fA({\id}_M) = {\id}_{\fA(M,g)}\,,$$
for all morphisms $\psi'$ from $\Hom_{\Loc}((M_2,g_2),(M_3,g_3))$, $\psi \in \Hom_{\Loc}((M_1,g_1),(M_2,g_2))$ and all objects of $\Loc$. Defining locally covariant field theory as a functor
corresponds to the axioms of \textbf{locality}\index{axiom!of locality}, \textbf{subsystems}\index{axiom!of subsystems} and  \textbf{covariance}\index{axiom!of covariance}. The two
remaining axioms (\textbf{time-slice axiom} and \textbf{causality}) can be formulated as certain properties of the functor $\fA$.

The \textbf{causality}\index{axiom!of causality} axiom is related to the tensorial structure of the underlying categories \cite{FR2}. The tensorial structure of $\Loc$ is defined in terms of disjoint unions, namely the objects in $\Loc^{\otimes}$ are elements $M$ that can be written as $M_1\otimes\ldots\otimes M_N:=M_1\coprod \ldots\coprod M_n$. The unit is provided by the empty set $\varnothing$. By admissible embeddings we mean maps $\chi: M_1\coprod \ldots\coprod M_n\rightarrow  M$ such that each component is a morphism of $\Loc$ and all images are spacelike to each other, i.e., $\chi(M_1) \perp\ldots\perp\chi(M_n)$. 
Now we turn to the tensorial structure of the category $\Obs$. 
Since there is no unique tensor structure on general locally convex vector spaces, one has to either restrict to nuclear spaces\footnote{Details on the nuclear spaces and tensorial structure are given in the subsection \ref{nuclear}.} or make a  choice of the tensor structure based on some physical requirements. In the present context we shall use the first of those
possibilities, since the spaces we are working with are nuclear.
 The functor $\fA$ can be then extended to a functor $\fA^\otimes$ between the categories $\Loc^\otimes$ and $\Obs^\otimes$. The requirement for $\fA^\otimes$ to be tensorial is a condition that
\begin{align}
\fA^\otimes\left(M_1\coprod M_2\right)&=\fA(M_1)\otimes\fA(M_2)\,,\\
\fA^\otimes(\chi\otimes\chi')&=\fA^\otimes(\chi)\otimes\fA^\otimes(\chi')\,,\\
\fA^\otimes(\varnothing)&=\CC\,.
\end{align}
Now we show (see \cite{FR2}) that if $\fA$ is a tensor functor, then the theory is causal. Consider the natural embeddings $\iota_{i}:M_i\rightarrow M_1\coprod M_2$, $i=1,2$ for which $\fA\iota_1(A_1)=A_1\otimes\1$, $\fA\iota_2(A_2)=\1\otimes A_2$, $A_i\in\fA(M_i)$. Now let $\chi_i:M_i\rightarrow M$ be admissible embeddings such that the images of  $\chi_1$ and  $\chi_2$ are causally disjoint in $M$. We define now an admissible embedding $\chi:M_1\coprod M_2\rightarrow M$ as:
\be
\chi(x)=\left\{
\begin{array}{lcl}
\chi_1(x)&,&x\in M_1\\
\chi_2(x)&,&x\in M_2
\end{array}
\right.
\ee
Since $\fA^\otimes$ is a covariant tensor functor, it follows:
\be
\{\fA\chi_1(A_1),\fA\chi_2(A_2)\}=\fA\chi\{\fA\iota_1(A_1),\fA\iota_2(A_2)\}=\fA\chi\{A_1\otimes\1,\1\otimes A_2\}=0
\ee
This provides the notion of causality. It can be shown that also the opposite implication holds, i.e. the causality axiom implies that the functor $\fA$ is tensorial. The proof will be provided in \cite{BFR}.

The \textbf{time-slice}\index{axiom!time-slice} axiom basically means that one can reconstruct the full algebra of observables associated to a given region $\Ocal$ knowing only the algebra of a causally convex neighbourhood of a Cauchy surface contained in $\Ocal$. Precisely, if the morphism $\psi \in \Hom_{\Loc}((M,g),(M',g'))$ is such that $\psi(M)$ contains a Cauchy-surface in $(M',g')$, then $ \fA\psi$ is an isomorphism.
For the classical field theory it amounts to the statement that the Cauchy problem is well defined for the field equations (in the sense of (\ref{eom})). Since we are always dealing with the hyperbolic (system of) equations, this is obviously satisfied. In case of the anti-commuting fields the problem is less trivial and we shall come back to it in section \ref{fer}. 

For the sake of completeness we recall here the precise definition of the algebra that we can associate to the Cauchy surface \cite{BFV,BF1,FR,Chi}.
Let $\Sigma\subset M$ be a Cauchy surface. We can associate to it a family of algebras $\{\fA(N)\}_{{N}\in I}$ indexed by the (admissibly embedded) subspacetimes $N$ of 
${M}$, such that $\Sigma\subset N$. The inclusion of spacetimes provides an order relation on the index set. For ${N}_i,N_j\in I$, such that ${N}_i\subset{N}_j\in I$ we introduce a notation ${N}_i\geq{N}_j$. The Cauchy surface $\Sigma$ is the upper limit with respect to $\geq$, so we obtain a directed system of algebras  $(\{\fA({N})\}_{{N}\in I},\geq)$.

Let $\chi_{ij}:{N}_i\hookrightarrow{N}_j$ be the canonical isometric embedding of ${N}_i\geq{N}_j$. Using the covariance property we obtain a morphism of algebras $\fA{\chi_{ji}}:\fA({N}_i)\hookrightarrow\fA({N}_j)$. Consider a family of all such maps between the elements of the directed system  $(\{\fA({N})\}_{{N}\in I},\geq)$. It can be easily checked that the family of mappings $\fA{\chi_{ij}}$ provides the transition morphisms for this directed system and we can define the corresponding projective (inverse) limit:
\be
\fA(\Sigma):=\!\varprojlim\limits_{{N}\supset\Sigma}\fA({N})= \Big\{\mbox{ germ of }(a)_I\!\! \in\!\! \prod_{{{N}\in I}}\fA({N}) \;\Big|\; a_{{N}_i}= \fA{\chi_{ij}}( a_{{N}_j})\ \forall\ {N}_i \leq {N}_j \Big\}. 
\ee
Having defined the algebra associated to a Cauchy surface we can now provide a notion of dynamics on the  functorial level \cite{BF1,FR2}. This will be especially important in section \ref{grav} in the context of classical gravity, where we go beyond the fixed manifold and work on the level of categories.

Consider natural embeddings of Cauchy surfaces $\Sigma$ into  ${M}$. The time-slice axiom implies that  $\alpha_{{M}\Sigma}$ are isomorphisms. It follows that the propagation from $\Sigma_1$ to  $\Sigma_2$ is described by:
\be\label{alphaM}
\alpha^{{M}}_{\Sigma_1\Sigma_2}:=\alpha_{{M}\Sigma_1}^{-1}\alpha_{{M}\Sigma_2} \ .
\ee
\section{Poisson structure for the free scalar field}\label{scal}
After formulating the general formalism we give now a concrete example of a construction of the locally covariant classical field theory. We follow the approach of Peierls \cite{Pei} (see also the work of Marolf \cite{Mar}), where the Poisson structure is defined in the Lagrangian formalism, without the need to introduce a Hamiltonian. This definition of a Poisson algebra was used also by deWitt in \cite{Witt} and now is a basic tool in the functional approach to classical field theory \cite{DF02,DF,Duetsch:2000nh,BV09,Rej}. The main assumption for the existence of the Peierls bracket\index{Peierls bracket} is the hyperbolicity of  the equations of motion. To avoid the notational complications we take as an example  the free minimally coupled scalar field. The generalized Lagrangian\index{generalized Lagrangian! of the scalar field} is given by:
\be\label{Lscalar}
L_M(f)(\ph)=\frac{1}{2}\int\limits_M \dvol(\nabla_\mu\ph\nabla^\mu\ph-m^2\ph^2)f
\ee
The second variational derivative of this Lagrangian induces an operator $\E(M)\rightarrow\E(M)$ given by $P=\Box+m^2$ (the Klein-Gordon operator). Clearly it is normally hyperbolic. It was proved in \cite{Baer} that for this class of operators on a generic globally hyperbolic Lorentzian manifold there exist retarded and advanced Green's functions $\Delta^{R/A}$, i.e. the inverses of $P$ satisfying in addition:
\begin{align}
\supp(\Delta^R)&\subset\{(x,y)\in M^2| y\in J^-(x)\}\label{retarded}\,,\\
\supp(\Delta^A)&\subset\{(x,y)\in M^2| y\in J^+(x)\}\label{advanced}\,.
\end{align}
The causal propagator is defined as $\Delta\doteq\Delta^R-\Delta^A$.
The Poisson structure on the space of functionals $\F(M)$ (off-shell!) can be defined as:
\begin{equation}
\{F,G\}_S=-\left<F^{(1)},\Delta*G^{(1)}\right>\,.\label{peierls1}
\end{equation}
From the wave front set properties of the elements of $\F(M)$ follows that the above expression is well defined. It turns out however that the space of multilocal functionals is not closed under $\{.,.\}_S$. Therefore, in order to obtain a Poisson algebra, one has to extend $\F(M)$ with more singular objects. A natural choice is provided by the space of \textit{microcausal} functionals\index{microcausal!functionals} \cite{BFK95,BDF,BFR}. Let $\Xi_n$ be an open come defined as $\Xi_n\doteq\{(x_1,...,x_n,k_1,...k_n)| (k_1,...k_n)\notin (\overline{V}_+^n \cup \overline{V}^n_-)\}$. We denote by $\F_\mc(M)$ the space of smooth compactly supported functionals, such that their derivatives at each point are distributions with wave front set contained in the open cone $\Xi_n$:
 \be\label{mlsc}
\WF(F^{(n)}(\ph))\subset \Xi_n,\quad\forall n\in\NN,\ \forall\ph\in\E(M)\,.
\ee
The space of such distributions, denoted by $\Ecal_{\Xi_n}(M)$ can be equipped with the H\"ormander topology \cite{Hoer}. We recall now its definition.
Let $C_n\subset \Xi_n$ be a closed cone contained in $\Xi_n$. We introduce (after \cite{Hoer,BaeF,BDF}) the following family of seminorms on $\Gamma'_{C_n}(M^n)$: 
\[
p_{n,\ph,\tilde{C},k} (u) = \sup_{\xi\in V}\{(1 + |\xi|)^k |\widehat{\ph u}(\xi)|\}\,,
\]
where the index set consists of $(n,\ph,\tilde{C},k)$ such that $k\in \NN_0$, $\ph\in \D(U)$ and $\tilde{C}$ is a closed cone in $\RR^n$ with $(\supp ( \ph ) \times \tilde{C}) \cap C_n = \varnothing$. These seminorms, together with the seminorms of the weak topology provide a defining system for a locally convex topology denoted by $\tau_{C_n}$. To control the wave front set properties inside open cones, we take an inductive limit. It can be shown that, to form this inductive limit one can choose a family of closed cones contained inside $\Xi_n$ to be countable. The resulting topology, denoted by $\tau_{\Xi_n}$, doesn't depend on the choice of this family. It was shown in \cite{BDF} that the H\"ormander topology $\tau_{C_n}$ can be equivalently defined as an initial topology with respect to the family of maps:
\begin{align}
u&\mapsto \left<u,f\right>,\quad f\in\Gamma(M)\,,\\
u&\mapsto Au\in\Gamma(M,V)\,,
\end{align}
where $A$ is a pseudodifferential operator with characteristics containing $C_n$. This is a family of linear maps into nuclear spaces that satisfies assumptions of corollary \ref{ninit}, so we can conclude that the topology $\tau_{C_n}$ is nuclear. According to \ref{nuc} a countable inductive limit of  nuclear vector spaces is nuclear so it holds also for $\tau_{\Xi_n}$. We can now equip $\F_\mc(M)$ with the initial topology with respect to the mappings:
\be\label{tauF}
\F(M)\ni F\mapsto F^{(n)}(\ph)\in\Ecal'_{\Xi_n}(M^{n})\quad n\geq0\,,
\ee
where $\ph$ runs through all elements of $\E(M)$. We denote this initial topology on $\F_\mc(M)$ by $\tau_\Xi$. Using  again the result \ref{ninit} we conclude that this topology is nuclear.

It can be shown that $\F_\mc(M)$ is closed under the Poisson bracket (\ref{peierls1}) and the assignment
$M\mapsto (\F_\mc(M),\{.,.\})$ is a covariant functor from $\Loc$ to $\Obs$. Moreover it turns out that the space of microcausal functionals will be useful for the construction of the quantum theory in section \ref{quantalg}.
 \section{Interlude: vector valued functions on $\E(M)$}\label{vvf}
 After introducing various definitions concerning the functionals this seems to be a good moment to make some generalizations that will be useful later on. In the same way as one defines local and microcausal functionals one can also introduce those notions for functions on $\E(M)$ with values in more general locally convex vector spaces. 
 
We start with the notion of locality. Consider functions on $\E(M)$ with values in a locally convex vector space of distributional sections $\W(M):=\prod\limits_{n=0}^\infty\Gamma'(M^n,W_n)$ (equipped with the strong topology) for arbitrary finite dimensional vector spaces $W_n$. 
The $k$-th functional derivative of an element $F\in\Ci(\E(M),\W(M))$, at each point $\ph\in\E(M)$  can be identified with an element of:
\[
F^{(k)}(\ph)\in\prod\limits_{n=0}^\infty\Gamma'(M^{k+n},V^{\otimes k}\otimes W_n)\,.
\]
We define a support of $F$ as a closure of the sum of the supports of its derivatives (c.f. (\ref{support2})), i.e. $\supp\, F=\overline{\bigcup\limits_{{\ph\in\E(M)}}\supp(F^{(1)}(\ph))}$.

We say that $F$ is local\index{local!vector-valued functions} if it is compactly supported and all the distributions $F^{(k)}(\ph)$ have their supports on the thin diagonal and their wavefront sets are orthogonal to the tangent bundles of the thin diagonals $\Delta^{n+k}(M)$. The subspace of $\Ci(\E(M),\W(M))$ consisting of all the local functions is denoted by $\Ci_\loc(\E(M),\W(M))$. The space of multilocal functions  $\Ci_\ml(\E(M),\W(M))$ is again the algebraic completion of $\Ci_\loc(\E(M),\W(M))$ as an $\Ci_\loc(\E(M),\RR)$-module. We define microcausal vector valued functionals\index{microcausal!vector-valued functions} $\Ci_\mc(\E(M),\W(M))$, as the smooth compactly supported ones, for which the functional derivatives at each point satisfy the condition (\ref{mlsc}), i.e.
\[
F^{(k)}(\ph)\in\prod\limits_{n=0}^\infty\Gamma'_{\Xi_{n+k}}(M^{n+k},V^{\otimes k}\otimes W_n)\,.
\]
Topologies $\tau$ and $\tau_\Xi$ generalize to vector-valued functions on $\E(M)$ in a natural way and
multilocal functions  $\Ci_\ml(\E(M),\W(M))$ are dense in  $\Ci_\mc(\E(M),\W(M))$ with respect to the topology $\tau_{\Xi}$.

\section{Free fermionic fields}\label{fer}
Up to now we discussed only fields that are commuting. We know however, that in the quantum field theory also anicommuting ones appear. On may ask, if it is necessary to consider the classical counterpart of fermionic fields, since physically they exist only in the quantum context. There is however a good motivation to do it, since the quantization scheme we use involves the deformation quantization. This means, that we start first with some classical structure and then deform it to obtain the quantum algebra of observables. Therefore it is necessary to give a meaning to the classical theory also in case of fermions. Besides, we have to keep in mind that we aim for the locally covariant formulation of the BV formalism and in this formalism anticommuting fields are intrinsically present. Of course they are not to be understood as physical quantities, but rather as auxiliary objects. Nevertheless a consistent formalism that treats both the fermionic and bosonic fields on an equal footing is in this context necessary.

In this section we discuss the construction of the classical field theory of fermionic fields in the functional approach. We work with the example of Dirac fields, but the framework can be applied to arbitrary anticommuting variables, like for example the ghost fields. To formulate the theory of Dirac fields in the locally covariant framework we first need some geometrical preliminaries. For details we refer to \cite{Ko-Dirac}. By spin structure on $M$ we mean a pair $(SM,\pi)$
, where $SM$ is a
principal $Spin^0_{1,3}$-bundle over $M$
and $\pi:SM\rightarrow FM$ is
the \emph{spin frame projection}: a base-point preserving bundle
homomorphism of $SM$ and the frame bundle, compatible with the universal covering map. By the frame bundle $FM$ we mean the bundle whose fibre at each point $x\in M$ consists of orthonormal bases of $T_xM$ in the metric $g$.

A spacetime endowed with the spin structure $(SM,\pi)$ is called  a \emph{spin spacetime}\index{spin!spacetime}. Every globally hyperbolic spacetime admits a spin structure\index{spin!structure} but it doesn't have to be unique \cite{Ko-Dirac}. We define the (standard) locally covariant spinor bundle of the spin spacetime $M$ to be the associated vector bundle 
$DM=SM\times_{Spin^0_{1,3}}\CC^4$ of $SM$. The cospinor bundle is defined to be $D^*M$, the dual of $DM$.
We take the Whitney sum $DM\oplus D^*M$ and define the configuration space to be $\Ecal(DM\oplus D^*M)$, the set of smooth sections. The kinematical space of the theory consists of antisymmetric multilinear functionals on the exterior algebra of $\Ecal(DM\oplus D^*M)$ satisfying certain regularity conditions. Let $(DM\oplus D^*M)^{\boxtimes n}$ denote the exterior tensor product
of vector bundles. We define $\F_\loc(M)\doteq\prod\limits_{n=0}^\infty \F_\loc^{\,n}(M)$, where $\F_\loc^{\,n}(M)$ consists of antisymmetric distributions $\Ecal'((DM\oplus D^*M)^{\boxtimes n})$ 
with wave front sets orthogonal to the tangent bundles of the thin diagonals $\Delta^n(M)\doteq\left\{(x,\ldots,x)\in M^k:x\in M\right\}$, considered as subsets of the tangent bundles of $M^n$. Similarly to the bosonic case we can also define the multilocal functionals $\F(M)$ as products of such distributional sections. The condition of ``asymmetry'' means that:
$u_{a_1,...,a_k,a_{k+1},...,a_p}(x_1,..., x_k,x_{k+1},..., x_p)=-u_{a_1,...,a_{k+1},a_k,...,a_p}(x_1,..., x_{k+1},x_k,..., x_p)$.
We define a support of a functional $F\in\F(M)$ to be a sum of supports of all its components: $\supp(F):=\overline{\bigcup\limits_{p}\supp(F^p)}$.
 
 The construction  described above is in fact functorial. Let $\SLoc$ be defined as the subcategory of the fiber bundle category which has spin spacetimes as objects and the morphisms are 
 bundle morphisms covering those of $\Loc$, which are compatible with the right action and projections.
 Since Dirac fields carry an additional geometrical information concerning the spin structure, it is justified to choose as an underlying category $\SLoc$ and distinguish between spacetimes carrying different spin structures. We fix the Dirac structure to be the standard one, defined above.
 The details concerning the non-uniqueness of this structure are discussed in \cite{Ko-Dirac}. The map assigning the configuration space $\Ecal(DM\oplus D^*M)$ to a spin spacetime $M$ can be made into a contravariant functor $\E$ from  $\SLoc$ to $\Vect$.  For future convenience we introduce one more functor which assigns to $M\in\obj(\SLoc)$ the space $\E^p(M)=\Ecal((DM\oplus D^*M)^{\boxtimes p})$ of antisymmetric configurations of degree $p$. The morphisms are mapped to pullbacks. The space of all antisymmetric configurations is assigned by the functor $\Cc(M)=\bigoplus\limits_{p=0}^\infty\E^p(M)$.
 Clearly $\F(M)\subset \Cc'(M)$.
 
We define also a covariant functor $\E_c:\SLoc\rightarrow\Vect$ that assigns to $M\in\obj(\SLoc)$ the space of test sections $\Dcal(DM\oplus D^*M)$ and to morphisms, their pushforwards, similarly we define the functor of test sections $\D:\SLoc\rightarrow\Vect$. The assignment of the space of antisymmetric functionals to a manifold is in a natural way a covariant functor $\F:\SLoc\rightarrow \Vect$ into the category of nuclear topological locally convex vector spaces. Furthermore we introduce functors $\F^p:\SLoc\rightarrow\Vect$ that assign to $M\in\obj(\SLoc)$ the spaces of functionals of a fixed grade $p$.

The dynamics is introduced similarly to the scalar case by means of a generalized Lagrangian. We already discussed in \cite{Rej} that only in case of quadratic interaction we can understand the equations of motion (\textsc{eom}'s)  as a system of (differential) equations on the configuration space $\E(M)$. In general we cannot define the ``on-shell'' functionals as those on the solution space, but we can still keep the algebraic definition.  The graded (left) derivative\index{derivative!graded} of an antysymetric functional $F\in\F(M)$ is defined as:
\be\label{d1}
F^{(1)}(u)[h]:=F(h\wedge u)\quad F\in \F^{p+1}(M),\ u\in\E^p(M),\ h\in\E(M),\ p>0\,.
\ee
It is clear that $F^{(1)}$ is a vector-valued distribution (see section \ref{vvalued}), more precisely, an element of $\E'(M)\widehat{\otimes}\F(M)$, where $\widehat{\otimes}$ is the completed tensor product\footnote{The completion is meant with respect to the weak topology on the space $\prod_{n=0}^\infty\Ecal'((DM\oplus D^*M)^{\boxtimes n})\supset \E'(M)\otimes\F(M)$.}. The Euler-Lagrange derivative of the action $S$ provides on each $M$ a distribution $S'_M\in\E_c'(M)\widehat{\otimes}\F(M)$ defined by $\left<S'_M(u),h\right>=\left<L_M(f)^{(1)}(u),h\right>$, with $f\equiv 1$ on $\supp h$. To implement the dynamics we define an ideal $\F_0(M)\subset\F(M)$ as the one generated (in the algebraic and topological sense) by the set $\left\{\left<S'_M(.),h\right>\right\}_{h\in \E_c(M)}\subset\F(M)$. Then the on-shell algebra of functionals is the quotient $\F_S(M):=\F(M)/\F_0(M)$. The assignment $M\rightarrow \F_S(M)$ is a covariant functor from $\SLoc$ to $\Vect$.

After this short introduction we discuss as an example  the construction of locally covariant classical theory of the free Dirac field. Like in the case of the scalar field we have to extend the space of functionals by more singular objects, in order to obtain a structure closed under the Peierls bracket. Therefore we weaken the condition on the wave front sets of the distributional sections that constitute antisymmetric functionals. We define the space of microcausal functionals as $\Fcal_\mc\doteq\prod\limits_{n=0}^\infty \Fcal^n_\mc$, where $\Fcal^n_\mc$ consists of antisymmetric distributional sections $\Ecal'_{\Xi_n}((DM\oplus D^*M)^{\boxtimes n})$ with the wavefront set contained in the open cone  $\Xi_n$.
Elements of $\F^{\,p}_\mc(M)$ can be formally\footnote{See discussion in section \ref{distr}.} represented as distributional kernels:
\begin{align*}
F(u)&\form\int\limits_M \dvol \, f(x_1,...,x_p,y_1,...,y_q) \Psi_{x_1}\wedge...\wedge\Psi_{x_p}\wedge \overline{\Psi}_{y_1}\wedge...\wedge\overline{\Psi}_{y_p}(u)=\\
&\stackrel{\phantom{\mathrm{formal}}}{=}\int\limits_M \dvol\, f(x_1,...,y_q)u(x_1,...,y_q),\quad u\in\E^p(M)
\,,\label{Fp}
\end{align*}
where $\Psi_{x_i}$, $\overline{\Psi}_{y_j}$ are evaluation functionals, $\ \dvol$ is an invariant measure on $M$ and the spinor indices were suppressed. Consider the generalized Lagrangian of the free Dirac field:
\be
L_M(f)(u)=\int \dvol(x)\, f(x)(\overline{\Psi}_x\wedge(i\partial\!\!\!/-m)\Psi_x)(u)\quad f\in\D(M)\,,\label{dirac}
\ee
The second derivative of the action possesses retarded and advanced Green's functions $\Delta_M^{R/A}$. 
It is convenient to write $S(1)^{(2)}(x,y)$ as a block matrix in the basis $(u_A,\overline{u}^{\dot A})$:
\begin{equation}
S_M(1)^{(2)}(x,y)=\delta(x-y)\left(\begin{array}{cc}
0&D^{*^T}(x)\\
-D(x)&0
\end{array}\right)\,,
\end{equation}
where ${}^T$ denotes the transpose of a matrix. We can construct the retarded and advanced Green's functions $\Delta_0^{R/A}$ using the fact that $DD^*=D^*D=\Box+m^2$. Let $G^{R/A}$ be retarded/advanced Green's function for $(\Box+m^2)$. It can be shown that\footnote{We drop the subscript $M$ whenever it is clear from the context, in which spacetime we are working}:
\begin{equation}
\Delta^{R/A}(x,y)=\left(\begin{array}{cc}
0&-D^{*}(x)G^{R/A}(x,y)\\
D^T(x)G^{R/A}(x,y)&0
\end{array}\right)\doteq\left(\begin{array}{cc}
0&{K}_*^{R/A}\\
-K^{R/A}&0
\end{array}
\right)\,,\label{free}
\end{equation}
We define the causal propagator as $\Delta=\Delta^R-\Delta^A$. The free field theory is defined by introducing on $\F_\mc(M)$ the following graded Poisson structure:
\begin{equation}
\{F,G\}_S=(-1)^{|F|+1}\left<F^{(1)},\Delta*G^{(1)}\right>\,.\label{peierls2}
\end{equation}
Let $\PgAlg$ denote the category of graded Poisson algebras. The assignment $M\mapsto (\F_{\mc,S}(M),\{.,.\}_S)$ defines a covariant functor $\fA:\SLoc\rightarrow~\PgAlg$. 
\section{Interaction}\label{inter}
Since the interaction is essentially introduced in the same way for the bosonic and fermionic fields, we treat both cases together and from now on by $\F(M)$ we mean a space of graded functionals that can depend both on commuting and anticommuting variables. We use the perturbative construction (see \cite{DF,DF02}) similar to the methods of the scattering theory in quantum mechanics. Let us fix the interaction Lagrangian, i.e. the natural transformation $F$. We assume, that for all $M\in\Loc$ (or $\SLoc$) and $f\in\D(M)$, $F_M(f)$ is even and of degree higher than 2. Without the loss of generality, in the fermionic case we can take $F$ to be homogenous, i.e. $F_M(f)\in\F^p(M)$, $p>2$. Let us now fix the spacetime $M$. Following the spirit of algebraic field theory we want to construct a local net of algebras of interacting fields associated with the spacetime $M$. To do it we first switch on the interaction only in a neighbourhood of a causally closed region $\Ocal$ and construct the local interacting algebra associated to $\Ocal$. Let $g$ be a test function supported in some causally closed neighborhood of  $\Ocal$. It will play a role of a coupling constant and a cut-off function for the interaction. Consider a functional of the form: $L_{\inte, M}(f)=L_M(f)+F_M(fg)$, where $L$ is the generalized Lagrangian (\ref{dirac}). The interacting equation of motion on $M$ is given by:  $\left<L_M(f)^{(1)}+F_M(fg)^{(1)},h\right>$, $f\equiv1$ on $\supp h$, $h\in\E_c(M)$. At this point there is again a difference between the bosonic and fermionic fields. In the first case one can understand the equations of motion as equations for the C-valued functions. Then one can directly define 
the on-shell functionals as functions on the solution space $\E_{S+F(g)}(M)$. A recent result of R.~Brunetti, K.~Fredenhagen and P.~Lauridsen Ribeiro announced in \cite{Pedro} (to be published in \cite{BFR})  shows that a map from the free solutions $\E_{S}(M)$ to the interacting ones $\E_{S+F(g)}(M)$ can be constructed non-perturbatively by means of the Nash-Moser-H\"ormander theorem. Then there arises an interesting question: for which theories the ideal of on-shell functionals is generated only by the elements of the form  $\left<L_M(f)^{(1)}+F_M(fg)^{(1)},h\right>$? Until now the problem is not solved in a full generality, but one can check this property in specific examples. 

In case of the fermionic fields the notion of on-shell functionals is \textit{defined} as an algebraic quotient, so the question doesn't arise. Since the anticommuting case will not be discussed in \cite{BFR}, we shall provide it here. Let now $\F(M)$ denote the space of antisymmetric microcausal functionals. Note that  the variation of $L_M(f)+F_M(fg)$ defines a distribution in $\E_c'(M)\widehat{\otimes}\F(M)$ which we denote by $\left<L_M(1)^{(1)}+F_M(g)^{(1)},.\right>$. The ideal $\F_{g,0}(M)$ is generated by elements of the form: $\left<L_M(1)^{(1)}+F_M(g)^{(1)},h\right>$, $h\in \E_c(M)$. 
The interacting algebra corresponding to the interaction $F_M(g)$ is defined as the quotient: $\F_{g}(M):=\F(M)/\F_{g,0}(M)$.
Following \cite{DF02,Rej} we can now define the M{\o}ller map\index{Moller map@M{\o}ller map!classical} $ r_{F(g)}\equiv r_{S+F(g),S}:\F_{g}(M)\rightarrow\F_S(M)$ intertwining the interacting algebra with the free one.
\be\label{retarded:class}
r_{F(g)}(G)=\sum\limits_{k=0}^\infty\frac{\lambda^k}{k!}R_{S,k}(F(g)^{\otimes k},G)\,,
\ee
where $R_{S,0}=\id$, the first order term is just the retarded product\index{retarded product!classical} with the free field Green's function $\Delta^R$,
\[
 R_{S,1}(F(g),G)=-\left<F(g)^{(1)},\Delta^R*G^{(1)}\right>\,,
 \]
and for higher order terms there is a recursive formula:
\begin{equation}
  R_{S,n+1}(F_1^{\otimes n}\otimes F_2,G)=-\sum_{l=0}^n \binom{n}{l}
R_{S,l}\Bigl( F_1^{\otimes l},(-1)^{|F_2|+1}\left<F_{2}^{(1)}, 
        \Delta_{S+F_1}^{A\,(n-l)}*
       G^{(1)}\right>\Bigr),\label{nreta}
\end{equation}
with $F_1,F_2\in\F^+(M)$ and $  \Delta_{S+F_1}^{A\,(k)}$ defined as:
\begin{gather}
  \Delta_{S+F_1}^{A\,(k)}\doteq
\frac{d^k}{d\lambda^k}\Big|_{\lambda =0}
\Delta_{S+\lambda F_1}^{A}=(-1)^k k!\, \Delta^A*
 F_{1}^{(2)}*\Delta^{A}*\ldots*F_{1}^{(2)}*\Delta^{A}\>\,.\label{nretb}\end{gather}
The map $r_{F(g)}$ is invertible in the sense of a formal power series, so we can construct the interacting fields from the free ones. An important property of the retarded M{\o}ller maps is the causality: $r_{F(g)}G=G$,  if $\supp(G)\cap(\supp(g)+\overline{V}_+)=\varnothing$.  The Poisson structure on $\F_g(M)$ can be introduced in the similar way as for the free case, since the inverses of $L_M(1)^{(2)}+F_M(g)^{(2)}$ can be constructed from $\Delta^{R/A}$ by 
means of a formal power series. We denote this Poisson structure by $\{.,.\}_{F(g)}$. It is related to $\{.,.\}_S$ by:
\be\label{intertwine:class}
\{r_{F(g)}(H),r_{F(g)}(G)\}_{S}=r_{F(g)}(\{H,G\}_{F(g)})\,.
\ee
The same definitions can be applied in the case of bosonic fields, but then the formal power series is not enough and one has to prove the convergence. The proof was announced in \cite{Pedro} and will be published in \cite{BFR}.

To summarize, we have now defined the algebra $\F_g(M)$ of interacting on-shell functionals. Let $\Ocal$ be  a causally convex region. In both bosonic and fermionc case we associate to $\Ocal$ the subspace of $\F_g(\Ocal)$, consisting of the interacting fields supported inside $\Ocal$. We equip it with the Poisson bracket $\{.,.\}_{F(g)}$ and denote the resulting local algebra by $\fA_g(\Ocal)$. The next proposition shows that this structure doesn't depend on the behavior of $g$ outside of $\Ocal$. 
\begin{prop}\label{local}
Let $\Ocal$ be a relatively compact causally convex region of spacetime $M$  and $f\in\D(M)$ a test function with $\supp(f)\cap\Ocal=\varnothing$. Then the ideal $\F_{g,0}(\Ocal)$ is equal to   $\F_{g+f,0}(\Ocal)$ and the Poisson structure induced by $F(g+f)$ on  $\F_{g}(\Ocal)$ concides with the one induced by  $F(g)$ 
\end{prop}
\begin{proof}
The ideal   $\F_{g+f,0}(\Ocal)$ is  generated by elements of the form 
\[\left<L_M(1)^{(1)}+F_M(g+f)^{(1)},h\right>,\quad h\in \E_c(\Ocal)\,.
\]
We can decompose $g$ into $g_1$ and $g_2$ such that $\supp g_1\cap\supp(f)=\varnothing$ and $g_1+g_2=g$. From the additivity of $F$ follows now that  $F_M(g+f)=F_M(g_2+f)+F_M(g)-F_M(g_2)$. Since  $\left<F_M(g_2+f)-F_M(g_2),.\right>\big|_{\Ocal}\equiv 0$, we have:  $\left<L_M(1)^{(1)}+F_M(g+f)^{(1)},h\right>=\left<L_M(1)^{(1)}+F_M(g)^{(1)},h\right>$, $h\in \E_c(\Ocal)$. A similar reasoning can be now applied to $L_M(1)^{(2)}+F_M(g+f)^{(2)}$, by using the formula for the formal inverse \cite{DF,Rej}.
\end{proof}
Now we can construct the local net of algebras following the prescription of \cite{BF0} (see also \cite{Duetsch:2000nh}). 
Let $\Theta ({\Ocal})$
be the set of all functions $g\in {\D}(M)$ which
are identically equal to $1$ in a causally convex open neighbourhood of
${\Ocal}$ 
and consider the bundle
\be\label{bundlealg}
\bigcup_{g\in \Theta ({\Ocal})}\{g\}\times {\fA}_
{g}({\Ocal}).
\ee
Now let $A\in\F(\Ocal)$. For a given test function $g\in\Theta ({\Ocal})$ and an element $A\in\F(\Ocal)$ we can consider the equivalence class $[A]_g\in\F_g(\Ocal)=\F(\Ocal)/\F_{g,0}(\Ocal)$ and the interacting algebra ${\fA}_{S+F}({\Ocal})$ is defined as the algebra of sections of (\ref{bundlealg}) that are of the form:
\be
{\fA}_{S+F}({\Ocal})\ni ([A]_g)_{g\in \Theta ({\Ocal})}
\quad\quad [A]_g\in {\F}_{g}({\Ocal})\label{3.6}
\ee
and the Poisson structure is the one discussed in proposition \ref{local}.
The sections are covariantly constant in the sense, that the identity map on $\F(\Ocal)$ induces a canonical map $\alpha_{\id}:\F_g(\Ocal)\rightarrow \F_{g'}(\Ocal)$ and it holds $\alpha_{\id}([A]_g)=[A]_{g'}$, where $g,g'\in\Theta ({\Ocal})$.
The embeddings $\iota_{21}:
{\fA}_{S+F}({\Ocal}_1)\hookrightarrow {\fA}_
{S+F}({\Ocal}_2)$ for ${\Ocal}_1\subset {\Ocal}_2$ are  defined
by restricting the sections, since $\Theta ({\Ocal}_2)\subset\Theta ({\Ocal}_1)$. They satisfy the compatibility
relation $i_{12}\circ i_{23}=i_{13}$ for ${\Ocal}_3\subset {\Ocal}_2
\subset {\Ocal}_1$ and define therefore an inductive system. The 
global interacting algebra is defined as the inductive limit of local algebras
\begin{equation}
{\fA}_{S+F}(M):= \bigcup\limits_{\Ocal\subset M}{\fA}_{S+F}({\Ocal}).\label{global}
\end{equation}
This assignment defines a covariant functor ${\fA}_{S+F}$ between categories $\Loc$ (or $\SLoc$) and $\Obs$.
\chapter{Batalin-Vilkovisky formalism}\label{BVform}
\vspace{-5ex}
\begin{flushright}
 \begin{minipage}{8cm}
 \begin{flushright}
\textit{
Und wenn Natur dich unterweist,\\
Dann geht die Seelenkraft dir auf,\\
Wie spricht ein Geist zum andern Geist.\\
(...)\\
Ihr schwebt, ihr Geister, neben mir;\\
Antwortet mir, wenn ihr mich h\"ort! \\}
$\ $\\
J.W. Goethe, \textit{Faust}
\end{flushright}
 \end{minipage}
\end{flushright}
\vspace{3ex}
\noindent\rule[2pt]{\textwidth}{1pt}
\vspace{1ex}\\
\section{Historical remarks}\label{history}
The story of the BV formalism\index{BV!formalism} began already in 1962 with R.P. Feynman at the conference held in Warszawa-Jab{\l}onna 25-31 July. This event, orgnised by Leopold Infeld, gathered together many of the most admired scientists of that time, including Paul Dirac, John Wheeler, Herman Bondi, Vitalij Ginzburg, Vladimir Fock and Subramanyan Chandrasekhar. The conference was concerned with gravity and Feynman gave a talk there entitled \textit{Quantum theory of gravitation}. Notes from this lecture, taken by Marek Demia\'nski (at that time a student) and John Stachel, were published later in \textsl{Acta Physica Polonica} \cite{Fey}.
This seems to be the first mention of the problem present in Yang-Mills theories and gravity, known today as gauge invariance\index{gauge!invariance}. Feynman goes around this difficulty using some heuristic arguments and tricks. He proposes a method to modify the rules for calculating the one loop diagrams, which eventually gives the desired results. Only four years later those modified rules were understood in a more fundamental way, using the path integral approach. This was done by L. D. Faddeev and V. N. Popov \cite{FadPop} and auxiliary fields they were using in their construction are now called Fadeev-Popov ghosts\index{ghosts}. The gauge invariance 
started to be an important topic of investigation. Conditions for the gauge invariance of observables, given as certain relations between Green's functions were formulated between 1971 and 1972 in independent works of A.A. Slavnov \cite{Slav} and  J.C. Taylor \cite{Tay} and are known today as Slavnov-Taylor identities. Based on these identities, Zinn-Justin together with Lee proved renormalizability of Yang-Mills theories in 1972 \cite{LZJ}. A proof based on Feynman rules was presented in paralel by 't Hooft and Veltman.

This was the general state of knowledge, when the seminal paper of Becchi, Rouet and Stora \cite{BRST0} was published in 1974\index{BRST}. It was recognized that a gauge symmetry can be replaced by a more general one, involving also the ghosts and auxiliary fields and this symmetry is not broken by the quantization. This paper and a following one \cite{BRST1}, published in 1975 were concerned with the quantization of the Higgs-Kibble model. The result on Yang-Mills theories appeared one year later in \cite{BRST2}. At approximately the same time similar ideas appeared independently in the work of Tyutin \cite{Tyu}. The new symmetry, called today \textit{the BRST symmetry} was implemented on the infinitesimal level by the \textit{BRST operator}. The  invariance of the quantum action under the BRST was used to derive the Slavnov-Taylor identities which govern unitarity and the gauge-independence of the S-matrix. The discovery of BRST symmetry also led  to a general proof of renormalizability of non-abelian gauge theories by Zinn-Justin \cite{Zinn}. This was also the first time when the so called \textit{antifields}\index{antifields} appeared in the literature. In the paper of Zinn-Justin they were just sources conjugate to the BRST variations.

Parallel to the developments in Lagrangian formalism, also the Hamiltonian approach gained a lot from the discovery of the BRST symmetry. A group at Lebedev, including I.A. Batalin, E.S. Fradkin and G.A. Vilkovisky
was working at that time on the phase space integral quantization of gauge theories. In 1977 the problem was solved \cite{Batalin:1977pb} for the so called \textit{closed algebras}. More general systems with symmetries associated to \textit{open algebras} were treated later in a work of Fradkin and Fradkina \cite{Fradkin:1978xi}. We discuss our interpretation of the notions of open and closed algebras in section \ref{open}. The complete construction of the BRST operator for general irreducible I class constraints in the Hamiltonian formalism was done in 1985 by Henneaux in \cite{Henn85}. Reducible constraints were treated by \cite{HFTS} four years later.

Coming back to the Lagrangian formulation we finally reach the moment when the Batalin-Vilkovisky formalism was born. In the seminal series of papers \cite{Batalin:1981jr,Batalin:1983wj,Batalin:1983jr} between 1981 and 1983 Batalin and Vilkovinski extended methods of BRST to Lagrangian formalism. This is where the \textit{antibracket}\index{antibracket}, \textit{antifields}\index{antifields} and the \textit{master equation}\index{master equation} were for the first time presented in a complete formulation.

Parallel to the development in physics also mathematicians started to be interested in the BRST symmetry. The geometrical interpretation of the ghosts was already suggested in the original paper of Becchi, Rouet and Stora and was further developped in 1983 by Bonora and Cotta-Ramusino \cite{CR}. In this paper ghost fields are understood as Maurer-Cartan forms on the gauge group. The first recognition of the cohomological aspects of the BRST formalism was put forward a year later in a paper of McMullan \cite{McM}, which appeared only as an Imperial College preprint. In 1985, those aspects were again stressed in the already mentioned report of Henneaux \cite{Henn85}. It was shown that results of the Fradkin group can be formulated in the language of homology. Henneaux was making use of the acyclicity of a certain complex. McMullan and Browning \cite{BrMM} showed later (1987) that it was the \textit{Koszul complex}\index{Koszul!complex}, a construction well known in homological algebra since the seminal work of Koszul \cite{Kosz}\index{Koszul!J.-L.} from 1950. Its generalization, the \textit{Koszul-Tate complex}\index{Koszul-Tate!complex} was proposed by Tate \cite{Tate} in 1957.

 Also in 1987 McMullan published a paper on Yang-Mills theories \cite{McM2}. The cohomological aspects of Batalin-Vilkovisky formalism started to become apparent at that time and mathematicians were trying to understand the underlying structure using simple models. By replacing the infinite dimensional phase space with a finite dimensional symplectic manifold one can understand the Batalin-Fradkin-Vilkovisky formalism using purely geometrical considerations. First attempts in this direction were done by Stasheff \cite{Stasheff1,Stasheff2} and Kostant, Starnberg \cite{KostSte}. The first comprehensive treatment of the BRST cohomology was published in 1987 by Dubois-Violette \cite{DV}. Similar ideas were presented independently by Figueroa-O'Frarril and Kimura \cite{FFK} in 1988. All those results are rigorous but concern only finite dimensional systems, which doesn't really reflect the physical situation in field theories, where we have infinitely many degrees of freedom.

Coming back to the physical side, the development of homological techniques opened new perspectives for the applications of the BV and the BRST formalism. The application of homological perturbation theory was first put forward by Henneaux and Fisch in a paper \cite{FH} form 1990. Also the algebraic structure of the antifield formalism was thoroughly discussed. It was recognized that the Koszul (and Koszul-Tate) complex plays an important role in understanding the structure of both classical and quantum field theories. A cohomological point of view on the antifield formalism was presented by Henneaux in the lecture notes from 1989 \cite{Henneaux:1989jq}. A very comprehensive monography on BRST formalism in the Hamiltonian and antifield approach was published by Henneaux and Teitelboim in 1992 \cite{Henneaux:1992ig}, with the title \textit{Quantization of gauge systems}. A year later another monograph appeared, authored this time by J. Gomis, J. Par\'is and S. Samuel \cite{GP}.

The BRST method was also used in a more axiomatic framework. The first complete treatment of the operator covariant quantization of Yang-Mills theories with the use of BRST was due to Kugo and Ojima \cite{KuOji0,KuOji}, although there were also some earlier attempts by G. Curci and R. Ferrari \cite{CurFer}. 
In \cite{KuOji0,KuOji} authors  investigated the BRST quantization of Yang-Mills theory from the operator algebraic view point. They defined the physical Hilbert space of the theory as the cohomology of the BRST charge. A similar approach was also applied later to gravity and appeared in the book of N. Nakanishi and I. Ojima \textit{Covariant operator formalism of gauge theories and quantum gravity} \cite{NaOji}.

Nowadays the cohomological methods of the BRST and BV formalism are commonly used in classical and quantum field theories. Nevertheless we think that there are still some issues that remain  open. For example the infinite dimensional character of the configuration space was up to now neglected in the mathematical literature concerning BV formalism. It is particularly important in the context of quantum field theory. As we know, quantum observables cannot be described with local functionals only (for example the Wick powers \cite{BFK95}). In this case, the description in the jet space language \cite{Sard2,Sard,SardGia} is less natural and quite complicated. On the other hand, the infinite dimensional approach allows to treat a very general class of functionals. Moreover the topological and functional analytic aspects arise in a natural way in this formulation. Another important point is the issue of compact and noncompact support. Since globally hyperbolic spacetimes are non-compact, one has to be careful with distinguishing between compactly and noncompactly supported configurations. For example, a normally hyperbolic system of equations doesn't have any compactly supported solutions. There are also some subtle points when one distinguishes between \textit{large} and \textit{small} gauge transformations. We discuss it in section \ref{ChEil}. We show that the conscious application of infinite dimensional calculus makes may of the constructions in BV formalism clear and better motivated. We believe that the mathematical structures in physics should be made as conceptually simple as possible. Clearly infinite dimensional calculus is a relatively new branch of mathematics and it takes time to get acquainted with its techniques. Nevertheless, the gain we are obtaining is worth the effort. 

We argue in this thesis that the infinite dimensional viewpoint on the BV formalism solves many conceptual problems and can be seen as a natural generalization of the finite dimensional structure.  Very clear and mathematically appealing presentation of the finite dimensional case can be found for example in \cite{Froe,Witten,DV}. While reading these papers, one has an overwhelming feeling that the structure is really simple and beautiful. It is somehow disappointing to find out that this treatment cannot be generalized to the infinite dimensional case, especially in the context of quantization (due mainly to the fact that the path integral is ill defined). This was one of the motivations to look more closely at the structures appearing in the BV formalism and try to find a conceptually satisfactory setting. We argue that the infinite dimensional differential geometry provides such a language. Moreover in this formalism one sees a clear way to perform the quantization in a mathematically precise way. 
While looking back on the rich history of the BV formalism it is a pleasure to be able to add one more page to this fascinating book.
\section{Koszul complex}\label{K}
We start our discussion of the BV formalism with describing the \textit{Koszul complex}\index{Koszul!complex}. In homological algebra this was introduced in \cite{Kosz} in the context of finite dimensional vector spaces to define a cohomology theory for Lie algebras. Since then it has become a standard tool in homological algebra (see \cite{Weibel,Hilton,Eisen} for a systematic introduction into these concepts). 
 From the physical point of view this structure turns out to be useful in the off-shell formulation of field theory since it allows one to have a control of the ``off-shell'' quantities. Here we give a geometrical interpretation of this structure along the lines of \cite{FR}. 

The goal of the Koszul construction is to find a homological interpretation for the space of on-shell functionals $\F_S(M)$, for a given action $S$. We already mentioned in section \ref{dyn} that this space can be written as a quotient $\F_S(M)=\F(M)/\F_0(M)$, where $\F_0(M)$ is the space of functionals that vanish on-shell. 
\subsection{Vector fields on a configuration space}\label{Vfonconfspace}
In the first step of the construction we want to find a characterization of $\F_0$. In order to do this we introduce first one more geometrical structure on $\E(M)$. In section \ref{idc} we gave a definition of vector fields on a locally convex manifold. Since  $\E(M)$ is in particular a trivial manifold, smooth vector fields\index{infinite dimensional!vector fields} on this space can be identified with smooth maps from $\E(M)$ to itself. As in the finite dimensional case, vector fields act on the space of smooth functionals $\Ci(\E(M))$ as derivations,
\be\label{derivation}
\pa_XF(\ph):=\langle F^{(1)}(\ph),X(\ph)\rangle\,.
\ee
The spacetime support of a vector field $X$\index{spacetime support!of a vector field} is defined as follows:
\be\label{suppder}
\begin{split}
\supp\, X=\{x\in M|\forall \text{ neigh. }U\text{ of }x\ & \exists F\in\Ci_c(\E(M)), \supp\,F\subset U\ \text{ such that }\partial_XF\neq 0 \\
\text{or } \exists\ \ph,\psi\in\E(M),\supp\,\psi\subset U & \text{ such that }X(\ph+\psi)\neq X(\ph)\}\ .
\end{split}
\ee
Equivalently we can think of a vector field $X$ as a $\E(M)$-valued function on $\E(M)$ and define its support using (\ref{support2}) (see the discussion in \ref{vvf}). Note that the  definition (\ref{suppder}) has two parts. First part characterizes the support of $X$ understood as a derivation, i.e. a map from $\F(M)$ to $\F(M)$. The second part concerns the support of $X$ as a map from $\E(M)$ to $\E(M)$. Those two notions are not equivalent and the requirement of compact support concerns both of them. 
The definition we adopted takes into account both roles played by vector fields. It agrees with the idea to understand the spacetime support as the subset of spacetime where the presence of a given vector field is ``not felt'' this definition agrees with a more general one given in \ref{vvf}, which can be applied to a wider class of vector-valued functions on $\E(M)$. 

After this short explanation we 
can come back to the main thread of our discussion. Going a little bit ahead we can already reveal that vector fields will serve as infinitesimal transformations of the configuration space. Since we are interested in the local structure, we consider only compactly supported vector fields.
We can require them to be local, i.e. $X(\ph)(x)$\index{local!vector field} would depend only on the jet of $\ph$ at the point $x$. 
This is equivalent to the notion of locality for vector-valued functions, that we introduced in section \ref{vvf} if we consider $X$ as an element of $\Ci(\E(M),\E(M))$\footnote{Precisely speaking, one has $X^{(k)}(\ph)\in\Gamma'(M^k,V^{\otimes k})\widehat{\otimes}\E(M)\subset\Gamma_c'(M^{k+1},V^{\otimes k+1})$}. Clearly the space of local vector fields is not an $\F(M)$-module, so we have to take its algebraic completion. We call the elements of this completion \textit{multilocal vector fields}\index{multilocal!vector field}.
To make the assignment of the space of the vector fields functorial one needs to make a further restriction on the space of vector fields. The functoriality is of course desirable if we use the local covariance as our guiding principle. 
We define $\V(M)$ to be a subspace of $\Gamma_c(T\E(M))$ consisting of vector fields that can be considered as maps $\Ci_\ml(\E(M),\E_c(M))$. In other words, the restriction concerns the image and allows now to do push-forwards. We define a  \textit{covariant} functor $\V:\Loc\rightarrow\Vect$ which assigns to a manifold $M$ the space $\Ci_\ml(\E(M),\E_c(M))$ and maps morphisms $\chi:M\rightarrow N$ to
\be\label{funct}
\V\chi(X)= \E_c\chi\circ X\circ \E\chi \ .
\ee
The action of vector fields on functions and the Lie bracket of vector fields can be extended to a graded bracket (called the  Schouten bracket) on the space of alternating multi-vector fields. In our case these are smooth, compactly supported maps from $\E(M)$ into $\Lambda(\E_c(M))$, with
\[
\La(\E_c(M))=\bigoplus \La^n(\E_c(M))\,,  
\]
where $\La^n(\E_c(M))$ is the space of compactly supported sections on $M^n$ which are totally antisymmetric under permutations of arguments (with $\Lambda^0(\E_c(M))=\RR$).
The alternating multi-vector fields with the regularity properties discussed above form a graded commutative algebra $\La\V(M)$ with respect to the product
\be\label{wedge}
(X\wedge Y)(\ph)\doteq X(\ph)\wedge Y(\ph)\,,
\ee
The Schouten bracket is an odd graded Poisson bracket on this algebra with the following properties:
\begin{enumerate}
\item it is a map $\{\cdot,\cdot\}:\La^n\V(M)\times\La^m\V(M)\to\La^{n+m-1}\V(M)$,
\item is graded antisymmetric:
\[
\{Y,X\}=-(-1)^{(n-1)(m-1)}\{X,Y\}\,,
\]
\item it satisfies the graded Leibniz rule
\be\label{leibniz}
\{X,Y\wedge Z\}=\{X,Y\}\wedge Z+(-1)^{nm}\{X,Z\}\wedge Y\,,
\ee
where $n$ is the degree of $Y$ and $m$ the degree of $Z$.
\item for $X\in\La^1\V(M)\equiv\V(M)$ and $F\in\La^0\V(M)\equiv \F(M)$ it coincides with the action of $X$ as a derivation
\[
\{X,F\}=\pa_XF\,,
\]
\item for $X,Y\in\La^1\V(M)$ it coincides with the Lie bracket
\[
\pa_{\{X,Y\}}=\pa_X\pa_Y-\pa_Y\pa_X \ .
\]
\item it satisfies the graded Jacobi rule
\be\label{Jacid}
\{X,\{Y,Z\}\}-(-1)^{(n-1)(m-1)}\{Y,\{X,Z\}\}=\{\{X,Y\},Z\} \ , \ n=\mathrm{deg}(X),m=\mathrm{deg}(Y)\,.
\ee
\end{enumerate}
\subsection{Ideal ``generated by the equations of motion''}
Using the vector fields on $\E(M)$ one can now characterize elements of $\F_0(M)$. Indeed, it is easy to see that this space includes elements of the form 
\[
\ph\mapsto\left<S_M'(\ph),X(\ph)\right>=:\delta_S(X)(\ph)\,,
\]
where $S_M'$ is the Euler-Lagrange derivative of the action, defined in (\ref{ELd}). The notation introduced above stresses the point that one can associate with the action $S$ a map $\delta_S:\V(M)\rightarrow \F(M)$. Clearly this map is a differential and its image is contained in $\F_0(M)$. It was already mentioned in section \ref{inter} that the opposite inclusion doesn't hold in general. Indeed, we can give a simple example of a finite dimensional model where it is not the case.
\begin{exa}
Let the configuration space be simply $\RR$. The action is in this case just a function $S\in\Ci(\RR)$ and vector fields are also identified with elements of $\Ci(\RR)$. In particular we can take $S=x^4$. The equation of motion reads $x^3=0$, so the solution space is one point $x=0$. Now we see that the function $f=x^2$ vanishes on-shell and since the vector fields can be written as $X(x)\pa_x$, one obtains: $X(x)\pa_xS=x^2$ and therefore $X(x)=\frac{1}{4x}$. Clearly this function doesn't belong to $\Ci(\RR)$, so there is no vector field that can allow us to write $f$ as $\delta_S(X)$. 
\eex\end{exa}
This shows that one has to be extremely careful with this characterization of the ideal $\F_0(M)$. If it indeed holds that $\F_0(M)=\delta_S(\V(M))$ we say that $\F_0(M)$ is \textit{generated by equations of motion} (\textsc{eom}'s). To prove this property one needs some regularity conditions imposed on the action. In the absence of symmetries these conditions can be in principle reduced to the requirement that \textbf{$S''_M$ is a normally hyperbolic operator}. For the free scalar field it is easy to see that this condition is sufficient to show that $\F_0(M)$ is generated by {\eom}'s. We will show it in example \ref{ftc1} and at the end of this section we outline the idea how this method can be generalized to the case when symmetries are present. Again it turns out that  $\F_0(M)$ is generated by \textsc{eom}'s if \textbf{$S''_M$ is a normally hyperbolic operator \textit{after the gauge fixing}}.

From a more practical point of view one can show (see for example \cite{Henneaux:1992ig,HennBar}) that $\F_0(M)$ is generated by the \textsc{eom}'s in a wide range of physical models, including Yang-Mills and gravity by applying the jet space methods. This reduces the problem to the finite dimensional one. The key idea behind this approach is the use of locality. Since $\F(M)$ is generated by finite products of local functionals, it suffices to characterize the \textit{local} functionals that vanish on-shell. They, in turn, can be written as an integration $\int\limits_M \dvol (j^\infty_x)^*(\ph)f(x)$ of a function on the jet space, that depends only on a finite number of derivatives $j^k_x(\ph)$ of configuration fields $\ph\in\E(M)$ at a given point $x\in M$. Therefore locally the problem reduces to a finite dimensional one. More singular functionals can be characterized using the continuity arguments. Nevertheless, from the point of view of the general structure, it is worth to look at the problem also from the point of view of infinite dimensional analysis. We undertake this task in the present section. Before we state a general result we want to give one more example which can be thought of as a guideline in the construction that follows. 
\begin{exa}[Free scalar field]\label{ftc1}
For the free scalar field the equations of motion read: $P\ph=0$, where $P=\Box+m^2$ is the Klein-Gordon operator. For an arbitrary configuration $\ph\in\E(M)$ we can can use the (for example retarded) inverse of $P$ to solve the equation $P\ph=\psi$, i.e.: $\ph=\Delta^R\psi$. Now let $F\in\F_{0}$. 
We can write it as: $F(\ph(\psi))$. Therefore it induces a functional $\tilde{F}(\psi)$ and the condition $F\in\F_{0}$ can be reformulated as:
\[
\tilde{F}(0)=0\,.
\]
Using the fundamental theorem of analysis for calculus in locally convex vector spaces, we can write $\tilde{F}$ as:
\[
\tilde{F}(\psi)=\int\limits_0^1 d\tilde{F}(t\psi)[\psi]dt
=\int\limits_0^1 \frac{\delta\tilde{F}}{\delta \psi}(t\psi)[\psi]dt\,.
\]
Setting $\tilde{f}(\psi)=\int\limits_0^1\frac{\delta\tilde{F}}{\delta \psi}(t\psi)[{}_\bullet]dt$ (where the integral is understood in a weak sense) we obtain: $\tilde{F}(\psi)=\left<\tilde{f}(\psi),\psi\right>$. Rewriting it again in terms of variables $\ph$ gives a vector field $X(\ph)\doteq \tilde{f}(P\ph)$ and since $\delta_S(X)(\ph)=\left<\tilde{f}(P\ph),P\ph\right>=F(\ph)$, it follows that $\F_0(M)$ is indeed generated by equations of motion.
\eex\end{exa}
On this simple example we see that one can characterize $\F_0(M)$ as generated by \textsc{eom}'s if the action is quadratic and the Euler-Lagrange derivative can be considered as a map $S'_M:\E(M)\rightarrow\E(M)$ (in general it maps into $\E_c'(M)$). This can be generalized also to the nonlinear case, with the use of the M{\o}ller operators discussed in \ref{inter}.
\begin{prop}
Let $S=S_0+\lambda F(g)$ be an action functional such that $S'_M,\,{S_0}'_M :\E(M)\rightarrow\E(M)$, $S_0$ is quadratic and ${S}''_M$ is a normally hyperbolic operator.   If  the intertwining M{\o}ller map $r_{S_0+\lambda F(g),S_0}:\E(M)\rightarrow\E(M)$ exists\footnote{There are indications that the existence of $r_{S_0+\lambda F(g),S_0}$ can be actually proved in this case under some technical assumptions on $F$. The proof was announced in \cite{Pedro} and will be published in \cite{BFR}. The argument uses the \textsl{a priori} estimates on the inverse $\Delta^R$ of  $S''_M$ to apply the Nash-Moser theorem \cite{Ham}. We don't discuss the details here since this problem doesn't lie 
in the scope of the present work.} in every open neighborhood of the solution space $\E_S(M)$, then the ideal $\F_0(M)$ for the action $S$ is generated by the \textsc{eom}'s.
\end{prop}
\begin{proof}
From the properties of the M{\o}ller operator we obtain
\[
{S}'_M\circ r_{S_0+\lambda F(g),S_0} ={S_0}'_M\,.
\]
Now let $\Delta^R$ be the fundamental solution for the free action. Since 
\[
{S}'_M(r_{S_0+\lambda F(g),S_0}\circ\Delta^R)={S_0}'_M\circ\Delta^R=\id\,,
\]
it follows that $r_{S_0+\lambda F(g),S_0}\circ\Delta^R$ is an inverse of ${S}'_M$ and we can solve the equation ${S}'_M\ph=\psi$ for $\psi$. Now we rewrite a functional $F\in\F_0(M)$ in the new variables $\psi$ and the condition $F\in\F_{0}$ can be again reformulated as $\tilde{F}(0)=0$. The rest follows exactly along the lines of \ref{ftc1}.
\end{proof}

Let us now assume that we are given an action $S$ such that $\F_0(M)=\delta_S(\V(M))$ holds.
Then we have:
\[
\F_S(M)=\F(M)/\F_0(M)=\F(M)/\im\delta_S\,.
\]
This can be easily translated into the language of homological algebra. Consider a chain complex:
\be\label{Kshort}
\begin{array}{c@{\hspace{0,2cm}}c@{\hspace{0,2cm}}c@{\hspace{0,2cm}}c@{\hspace{0,2cm}}c@{\hspace{0,2cm}}c@{\hspace{0,2cm}}c}
0&\xrightarrow{}&\V(M)&\xrightarrow{\delta_S}&\F(M)&\rightarrow &0\\
 &&1&&0&&
\end{array}
\ee
The numbers below indicate the chain degrees. The $0$-order homology of this complex is equal to: $\F(M)/\F_0(M)=\F_S(M)$. This completes the first step in finding the homological interpretation of $\F_S(M)$. Next we construct a \textit{resolution} of $\F_S(M)$. 
\begin{defn}
In homological algebra a resolution of an algebra $A$ is a differential graded algebra $(\A,\delta)$, such that $H_0(\delta)=A$ and  $H_n(\delta)=0$ for $n>0$. 
\end{defn}
We can start constructing the resolution of  $\F_S(M)$ from the chain complex (\ref{Kshort}). We said before that the space of multivector fields  $\La\V(M)$ is a graded commutative algebra with respect to the product (\ref{wedge}). Moreover it is equipped with the natural bracket $\{.,.\}$. Since $\delta_S(X)$ is just $\{X,L_M(f)\}$ for $f\equiv 1$ on $\supp X$, $X\in \V(M)$, we can extend $\delta_S$ to $\La\V(M)$ by means of this graded bracket and obtain the complex:
\be\label{Kos}
\begin{array}{c@{\hspace{0,2cm}}c@{\hspace{0,2cm}}c@{\hspace{0,2cm}}c@{\hspace{0,2cm}}c@{\hspace{0,2cm}}c@{\hspace{0,2cm}}c@{\hspace{0,2cm}}c@{\hspace{0,2cm}}c}
\ldots&\rightarrow&\La^2\V(M)&\xrightarrow{\delta_S}&\V(M)&\xrightarrow{\delta_S}&\F(M)&\rightarrow& 0\\
 &&2&&1&&0&&
\end{array}\ ,
\ee
where $\delta_S$ is called the Koszul map\index{Koszul!map}. Now we want to calculate $H_1(\La\V(M),\delta_S)$.
First we identify the elements of  $\ke(\delta_S)_{\V(M)\rightarrow\F(M)}$. In the BV formalism they are called
\textit{symmetries}. We discuss them in detail in section \ref{symm}. 
\subsection{Antifields and antibracket}\label{aa}
To end this section we want to add a remark on the relation to the standard approach to the BV formalism. Vector fields $\V(M)$ correspond to objects called in physics literature: functionals of \textit{the antifields}\index{antifields}. Originally they were interpreted only as formal generators of some graded algebra. A geometrical interpretation was first given in \cite{Witten}, but it applies only to the case where the configuration space is finite dimensional. We generalize these ideas to the field-theoretic context and also simplify the structure.
To understand the relation to the traditional approach we write formally the action of a vector field on a functional in the ``integral'' notation:
\[
\partial_XF(\ph)=\left<F^{(1)}(\ph),X(\ph)\right>\stackrel{\mathrm{formal}}{=}\int_M\dvol X(\ph)(x)\frac{\delta F(\ph)}{\delta\ph(x)}\,.
\]
In this sense  we can identify the functional derivatives $\frac{\delta}{\delta\ph(x)}$ with the antifields $\ph^\ddagger(x)$ and write elements of $\V(M)$ as:
\[
X(\ph)\stackrel{\mathrm{formal}}{=}\int_M\dvol X(\ph)(x)\frac{\delta}{\delta\ph(x)}\equiv \int_M\dvol X(\ph)(x)\ph^\ddagger(x)\,.
\]
The algebra of alternating multivector fields is then the algebra ``generated'' by fields and antifields, and the \textbf{\textit{antibracket}}\index{antibracket} used in the physics literature is in our interpretation just the Schouten bracket on the space of multivector fields. To fix the sign convention we can write the antibracket in a slightly formal notation used commonly in the literature:
\be\label{antibracketformal}
\{X,Y\}=-\int dx\left(\frac{\delta X}{\delta\ph(x)}\frac{\delta Y}{\delta\ph^\ddagger(x)}+(-1)^{|X|}\frac{\delta X}{\delta\ph^\ddagger(x)}\frac{\delta Y}{\delta\ph(x)}\right)\,.
\ee
The derivative with respect to the antifield has to be understood as the left derivative (see section \ref{fer}). The extra ``-'' sign comes from the fact that in the literature one also introduces the right derivative.
\section{Symmetries}\label{symm}
Now it's time for symmetries\index{symmetries}! As we already mentioned, they should characterize the kernel of $\delta_S$ at the $1^{\mathrm{st}}$ order of (\ref{Kshort}). Therefore
the proper understanding of symmetries is the very heart and essence of the BV formalism.
In our formulation we define them as certain vector fields on $\E(M)$, that describe directions in the configuration space in which the action $S$ is constant. This can be expressed as the condition that
\be\label{sym0}
0=\delta_SX(\ph)=\left<S_M'(\ph),X(\ph)\right>=:\partial_X(S_M)(\ph)\,,
\ee 
for all $\ph\in\E(M)$ (see also \cite{Urs,FR}). A symmetry $X$ is called \textit{trivial}\index{symmetries!trivial} if it vanishes on-shell, i.e. $X(\ph)=0$ for all $\ph\in\E_S(M)$. It was recognized in \cite{Henneaux:1992ig} that the trivial symmetries play an important role in the BV construction. This issue appeared in the discussion  of
the so called \textit{open symmetry algebras}\index{open symmetry algebra}. We explain what this notion means in our formalism in section \ref{open}. The Lie subalgebra of  $\V(M)$ consisting of all symmetries would be denoted by $\sm(M)$. The trivial ones are the elements of $\sm_0(M)$. Now we show that the assignment of these spaces is indeed functorial.
\begin{prop}\label{syms}
Let $S$ be fixed. The map which assigns to every $M\in\Loc$ the corresponding algebra of symmetries $\sm(M)$ and to each morphism $\chi$, a morphism $\sm_\chi={\V}\chi$ is a contravariant functor $\Loc\rightarrow\Vect$.
\end{prop}
\begin{proof}
We have to show that $\sm_\chi={\V}\chi$ maps symmetries into symmetries. Let $X \in \sm(M)$, then $\forall f\in\D(M)$ with $\supp(X)\subset f^{-1}(1)$, $\ph\in\E(N)$ we have 
\begin{align*}
\left<L_N(\chi_*f)^{(1)}(\ph),(\V\chi X)(\ph)\right>&=\left<\chi_*(L_M(f)^{(1)}(\chi^*\ph)),\chi_*(X(\chi^*\ph))\right>=\\
&=\left<L_M(f)^{(1)}(\chi^*\ph),X(\chi^*\ph)\right>=0\,.
\end{align*}
\end{proof}
The trivial symmetries form a subalgebra  $\sm_0(M)\subset\sm(M)$. Alternatively we can characterize them as:
\begin{equation}
\sm_0(M)\doteq\{X\in\sm(M)|\,\partial_XF\in\F_{0}(M)\ \forall F\in\F(M)\}\,,\label{fieldsinv0}
\end{equation}
The space of nontrivial symmetries is defined as the quotient: 
\[
\smp(M)\doteq\sm(M)/\sm_0(M)\,,
\]
so it is a set of equivalence classes of vector fields on $\E(M)$, where the equivalence relation is given by:
\begin{equation}\label{eqv1}
X\sim Y\quad\textrm{if }\ X-Y\textrm{ vanishes on-shell}
\end{equation}
The algebra of symmetries $\sm(M)$ has a natural action on $\F(M)$ by derivations: $\sm(M)\ni X:F\mapsto \pa_XF$. Its quotient  $\smp(M)$ acts faithfully on $\F_S(M)$. 

The aim of the BV formalism is to analyze the structure of the spaces $\sm(M)$ and $\smp(M)$.
One can show that trivial symmetries are contained in the image of $\delta_S$, $(\im\delta_S)_{\La^2\V(M)\rightarrow\V(M)}$, so they don't contribute to $H_1(\La\V(M),\delta_S)$. We conclude that the first homology of the Koszul complex (\ref{Kos}) is trivial if the action $S$ doesn't possess any \textit{nontrivial local symmetries}. This condition can be formulated as follows \cite{FR}:
\begin{equation}\label{sym}
X(\ph)\perp S_M'(\ph)\quad \forall\ph\in\E(M)\Rightarrow X(\ph)=0\ \forall\ph\in\E_S(M)\ .
\end{equation}
From (\ref{sym}) follows a sufficient condition for an action $S$ to be free of nontrivial symmetries.
Let $\langle S_M'(\ph),X(\ph)\rangle=0$ then differentiating with respect to $\ph$ results in
\begin{equation} \label{no nontrivial symmetries}
\langle S_M''(\ph),X(\ph)\otimes\psi\rangle+\langle S'(\ph),X^{(1)}(\ph)\psi\rangle=0 \ .
\end{equation}
where the second derivative $S''$ of an action $S\equiv S(L)$ is a natural transformation $\E\to \E_c'\otimes  \E_c'$ defined by
\[S''_M= L_M(f)^{(2)} \text{ on } \E(K)\otimes  \E(K)\]
with any compact subset $K\subset M$ and with $f\equiv 1$ on $K$. 
From the locality it follows that $S_M''(\ph)$ induces a differential operator $\E(M)\rightarrow \E_c'(M)$,  defining the linearized equation of motion around the field configuration $\ph$. Now for $\ph\in \E_S(M)$ the second term in \eqref{no nontrivial symmetries} vanishes, hence
\[
\langle S_M''(\ph),X(\ph)\otimes\psi\rangle=0\qquad\forall\psi\in\E_c(M) \ .
\]
This means that for all $\ph\in \E_S(M)$, $X(\ph)$ has to be a solution of the linearized equation of motion. Since $X(\ph)\in\E_c(M)$, the action $S$ possesses no nontrivial symmetries if the linearized equation of motion doesn't have any nontrivial compactly supported solutions. In particular this is the case when $S_M''(\ph)$  is a normally hyperbolic differential operator. If an action doesn't possess any nontrivial symmetries, then $H_1(\La\V(M),\delta_S)=0$ and we get
\begin{align*}
H_0\left(\La\V(M),\delta_S\right)&=\F_S(M)\,,\\
H_k\left(\La\V(M),\delta_S\right)&=0,\ k>0\,.
\end{align*}
In this case the complex $(\La\V(M),\delta_S)$ is a resolution of $\F_S(M)$, called the \textit{Koszul resolution}\index{Koszul!resolution}. 

It was already stressed in \cite{FR}, that in the present setting, where the configurations are not compactly supported and the manifold $M$ is non-compact, the operator
 $\delta_S$ is not an inner derivation with respect to the antibracket. This is a major difference with respect to other approaches. A reason for this is the fact that the action itself is not an element of $\F_{\loc}(M)$, but rather an equivalence class of natural transformation between the functors $\D$ and $\F_{\loc}$. Nevertheless, locally $\delta_S$ can be written in terms of inner derivations, since 
$\delta_{S(L)}X=\{X,L_M(f)^{(1)}\}$ for $f\equiv 1$ on $\supp\, X$, $X\in \V(M)$.

To end this section we provide a simple finite dimensional example, which demonstrates the interplay between the trivial and non trivial symmetries.
\begin{exa}[Finite dimensional system with symmetries]\label{findimex}
 Let the configuration  space $E$ be a finite dimensional manifold. An action $S\in\Ci(E,\RR)$ is a functional on $E$. Let $d$ be the exterior derivative. Let $E_S$ be the set of all critical points of $S$, i.e.
\begin{equation}
E_S=\{x\in P|\ dS(x)\equiv 0\}.
\end{equation}
The condition $dS(x)=0$ can be written in local coordinates (with respect to a chart ($U_\alpha,\varphi_\alpha$)  as a system of $n$ equations for $n$ variables: $\sigma_i((\varphi_\alpha^{-1})^1(x),\ldots,(\varphi_\alpha^{-1})^n(x))=0$, $i=1,\ldots n$. These correspond to ``equations of motion''. The set $E_S$ corresponds to the space of solutions. A critical point is called nondegenerate if at this point the (local) Hessian matrix $H_S(\varphi_\alpha^{-1}(x))$ is nondegenerate. In this case we have a system of independent equations. This is a case with no nontrivial symmetries. Trivial ones can be written in local coordinates (in the following we denote the local coordinates by $x^1\doteq(\varphi_\alpha^{-1})^1(x)$ and we keep the local chart implicit) as $X(x)^j=\si_iM^{ij}(x)$, for some antisymmetric matrix $M$. In general only $k<n$ of the equations are independent. This is the case of a system with symmetries. As an example we take $E=\RR^3$ and we choose  the action $S$ to be:
\begin{equation}
S(x,y,z)=z\prod\limits_{k=1}^N(x^2+y^2-k^2)
\end{equation}
The corresponding equations of motion take the form:
\begin{align}
\sigma_1(x,y,z)&=2xz\sum\limits_{k=1}^N\prod\limits_{{l=1}\atop{l\neq k}}^N(x^2+y^2-k^2)\\
\sigma_2(x,y,z)&=2yz\sum\limits_{k=1}^N\prod\limits_{{l=1}\atop{l\neq k}}^N(x^2+y^2-k^2)\\
\sigma_3(x,y,z)&=\prod\limits_{k=1}^N(x^2+y^2-k^2)
\end{align}
It is easy to check, that the solution space $E_S$ is the disjoint union of concentric circles of radii $1\ldots N$, lying in the $z=0$ plane.
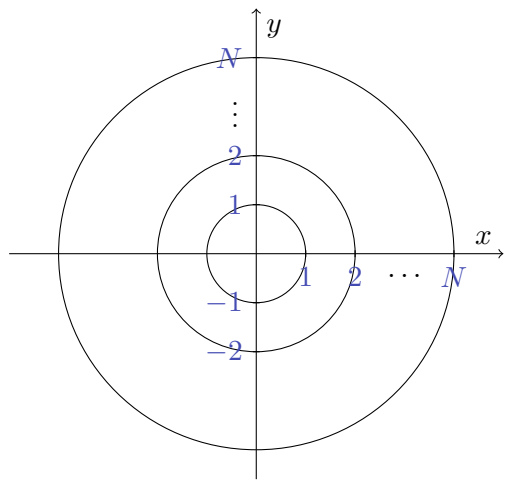
\begin{figure}[htb]
\begin{center}
\begin{tikzpicture}[scale=1.3]
 \useasboundingbox (-4,-1.8) rectangle (4,2.8); 
\draw[->] (-2.5,0) -- (2.5,0);	
\draw[->] (0,-2.3) -- (0,2.5); 
\draw (0,0) circle (1cm);
\draw (0,0) circle (0.5cm);
\draw (0,0) circle (2cm);
\foreach \x in {1,2} \draw (0.5*\x cm,1pt) -- (0.5*\x cm,-1pt) node[anchor=north,color=see] {$\x$};
\draw (2cm,1pt) -- (2cm,-1pt) node[anchor=north,color=see] {$N$};
\draw (1.5,-0.1) node[anchor=north] {$\ldots$};
\draw (-0.07,1.5) node[anchor=east] {$\vdots$};
\foreach \y in {-2,-1,1,2} \draw (1pt,0.5*\y cm) -- (-1pt,0.5*\y cm) node[anchor=east,color=see] {$\y$};
\draw (1pt, 2cm) -- (-1pt,2cm) node[anchor=east,color=see] {$N$};
\draw (2.3,0) node[anchor=south] {$x$};
\draw (0,2.3,0) node[anchor=west] {$y$};
\end{tikzpicture}
\end{center}
\caption[Finite dimensional system with symmetries: the solution space]{Solution space for the action functional $S(x,y,z)=z\prod\limits_{k=1}^N(x^2+y^2-k^2)$}
\end{figure}
The functional $S$ is invariant under the action of group $G=\Ci(\RR^3,SO(2))$. W can represent elements of $G$  by matrices acting on $\RR^3$ by matrix multiplication:
\begin{equation}
\alpha_\theta(x,y,z)=\left(\begin{array}{ccc}
\cos(\theta(x,y,z))&\sin(\theta(x,y,z))&0\\
-\sin(\theta(x,y,z))&\cos(\theta(x,y,z))&0\\
0&0&1
\end{array}
\right)\,,
\end{equation}
where $\theta\in\Ci(\RR^3)$. The group of trivial transformations $G_0$ consists of those elements of $G$ for which $\theta$ vanishes on $E_S$:
\begin{equation}
G_0=\{\alpha_\theta\in  G |\theta(x,y,z) =0\ \forall (x,y,z)\in E_S\}
\end{equation}
An example of such element can be $\alpha_\theta$, for $\theta(x,y,z)=z\prod\limits_{k=1}^N(x^2+y^2-k^2)$. The group of physical transformations $G_\p$ is in this case:
\begin{equation}
G_\p=G/G_0=\{\alpha_\theta\in  G |\theta\in\Ci(E_S)\}
\end{equation}
There are $N$ gauge orbits on $E_S$. A single gauge orbit is a circle with radius $1,\ldots,N$. Gauge invariant functions on $E_S$ have to be constant on each of the circles so we find that:
\begin{equation}
\Ci_\inv(E_S)\cong\RR^N
\end{equation} 
The algebra of on-shell symmetries is just
\begin{equation}
\smp=\mathfrak{so}(2)\otimes\Ci(E_S)
\end{equation}
It can be parametrized by:
\begin{equation}
X_\theta(x,y,z)=\theta(x,y,z)\left(\begin{array}{cc}
0&1\\
-1&0
\end{array}\right), \qquad \theta\in\Ci(E_S)
\end{equation}
Action of $\smp$ on $\Ci(E_S)$ can be written as:
\begin{equation}
\Ci(E_S)\ni f\mapsto X_\theta f(x,y,z)=\theta(x,y,z)(y\partial_x-x\partial_y)f(x,y,z)\label{che1}
\end{equation}
\eex\end{exa}
\subsection{General theories with symmetries}
In this section we provide a general discussion of theories possessing nontrivial local symmetries. The first issue that arises is the formulation of the regularity conditions on the action, that allow us to characterize $\F_0(M)$ as the image of the Koszul map $\delta_S$. Clearly $S''_M(\ph)$ is not an invertible operator and therefore the argument from the example \ref{ftc1} will not work. We need a weaker notion of ``invertibility'' of $S''_M$. Before we put it in mathematical terms let us discuss the underlying physical intuition. If the symmetries of $S$ arise from a Lie group action, then 
$S'_M$ is a well defined, nondegenerate map on the space of gauge orbits. Unfortunately in general 
this space is not a vector space and not even a manifold \cite{Rudolph,Abb,Kon,Kon2}. The way out is to parametrize it locally by choosing representants of the gauge equivalence classes. This amounts to performing \textit{local gauge fixing}. It was shown by Singer for the case of Yang-Mills theories that this cannot be done globally \cite{Sing,Sing2}, see also \cite{Rudolph,Abb,Kon,Kon2} for a detailed description of the structure of the gauge orbit space\index{gauge!orbits}. This problem is related to the so called Gribov ambiguities \cite{Grib}. In our setting this choice of parametrization will be made by defining a projection $R$ in each open neighborhood of $\E_S(M)$. This formulation applies also to the situation, where the symmetries don't necessarily arise from a Lie-group action.

Let $\{\Ocal_\alpha\}$ be an open covering of the solution space $\E_S(M)\subset\E(M)$. We say that the action \textit{doesn't have nontrivial local symmetries after the gauge fixing}\index{gauge!fixing} if for each $\Ocal_\alpha$ there exists an operator $R_\al:\Ocal_\alpha\rightarrow \E(M)$ such that $R_\al^2=R_\al$,  $R_\al(\Ocal_\alpha)$ and $(1-R_\al)\Ocal_\alpha$ are linearly independent and it holds:
\begin{equation}\label{sym2}
X(\ph)\perp S_M'(\ph)\quad \forall\ph\in R_\al(\Ocal_\al)\Rightarrow X(\ph)=0\ \forall\ph\in R_\al(\Ocal_\al)\,,
\end{equation}
where $X$ is a multilocal vector field on $R_\al(\Ocal_\al)$ (i.e. $X\in\Gamma(TR_\al(\Ocal_\al))$) with image in the compactly supported sections. We also require the operators $R_\al$, $R_{\al'}$ to be compatible on the intersection of $\Ocal_\al$ and $\Ocal_{\al'}$. In other words $R_\al$ projects locally to a subset of $\E(M)$ such that there are no symmetries tangent to it. It follows now that $S''_M$ is invertible on the tangent space of  $R_\al(\Ocal_\al)$  in the sense that for $\ph\in R_\al(\Ocal_\al)\cap\E_S(M)$
\[
\langle S_M''(\ph),X(\ph)\otimes\psi\rangle=0\quad\forall\psi\in \E_c(M)\Rightarrow X(\ph)=0\,,
\]
where $X$ is a vector field as in (\ref{sym2}). Reversing this reasoning results in a sufficient condition for an action to be free of nontrivial local symmetries after the gauge fixing. We phrase it in the following way:
\begin{ass}[Regularity condition]\label{reg2}
Assume that there exists an open covering $\{\Ocal_\alpha\}$ of $\E_S(M)$ and a family of operators $R_\al$, such that the linearized equations of motion $S''_M$ don't have nontrivial compactly supported solutions in $T_\ph(R_\al(\Ocal_\al))$, $\ph\in R_\al(\Ocal_\al)$.
\end{ass}
Under this assumption, (\ref{sym2}) holds and if the gauge-fixed $S''_M$ is normally hyperbolic, one can show that the ideal $\F_0(M)$ is generated by the \textsc{eom}'s. The above condition basically means that we can divide the equations of motion into dependent and independent ones and the independent can be chosen as new variables. Indeed, let us first consider the case when $S$ is quadratic. Let
 $\ph\in\Ocal_\alpha$ and $\psi=R_\al\circ S''_M\ph+(1-R_\al)\ph$. Assuming that $S''_M$ is hyperbolic as an operator acting on $R_\al(\Ocal_\al)$, we obtain the retarded solution $\Delta^R$, that fulfills 
 \[
 R_\al\circ S''_M\circ R_\al\circ\Delta^R\circ R_\al=R_\al\,.
 \]
 It follows that $\ph=R_\al\circ\Delta^R\circ R_\al \psi+(1-R_\al)\psi$ and we can write a functional $F\in\F(\Ocal_\al)$ in terms of the new variables as: $F(\ph)=F(R_\al\circ\Delta\circ R_\al \psi+(1-R_\al)\psi)=\tilde F(R_\al \psi,(1-R_\al)\psi)$. On-shell we have $R_\al \psi=R_\al\circ S''_M\ph=0$ and $(1-R_\al)\psi=(1-R_\al)\ph$. So if $F\in\F_0(M)$, then
\[
\tilde{F}(0,(1-R_\al)\psi)=0
\]
Denoting $\psi_1\doteq R_\al\psi$, $\psi_2\doteq(1-R_\al)\psi$ we can write $\tilde{F}$ as:
\[
\tilde{F}(\psi_1,\psi_2)=\int\limits_0^1 d\tilde{F}(t\psi_1,\psi_2)[\psi_1]dt
=\int\limits_0^1 \frac{\delta\tilde{F}}{\delta \psi_1}(t\psi_1,\psi_2)[\psi_1]dt\,.
\]
Setting $\tilde{f}(\psi_1,\psi_2)=\int\limits_0^1\frac{\delta\tilde{F}}{\delta \psi_1}(t\psi_1,\psi_2)[{}_\bullet]dt$ we obtain: $\tilde{F}(\psi)=\left<\tilde{f}(\psi_1,\psi_2),\psi_1\right>$. Rewriting it again in terms of  $\ph\in\Ocal_\alpha$ results in $\left<\tilde{f}(R_\al\circ S''_M\ph,(1-R_\al)\ph),R_\al\circ S''_M\ph\right>=F(\ph)$. We can repeat it for all $\Ocal_\alpha$ from the open cover and it follows that $\F_0(M)$ is indeed generated by equations of motion in the open neighborhood of $\E_S(M)$.

The above result is independent of the choice of $\{\Ocal_\alpha\}$ and corresponding projections and it provides a description of $\F_{0}(M)$ in the open neighborhood of $\E_S(M)$. A global argument can be obtained using the fact, that the configuration space $\E(M)$ is a Montel space and therefore paracompact \cite{Jar,Koe,Bou}. The following proposition is an infinite dimensional version of a formal argument provided in appendix 1.A of \cite{Henneaux:1992ig}.
\begin{prop}\label{idealth2}
Let $S$ be a generalized action fulfilling regularity conditions formulated above. Then each element of $\F_0$ can be written as
\begin{equation*}
F(\ph)=\left<f(\ph),S''_M\ph\right>,\qquad \ph\in \E(M)\,.
\end{equation*}
\end{prop}
\begin{proof}
Let $\{\Ocal_\rho\}$ be the open covering of $\E_S(M)$ constructed above. We complete it to the open covering of $\E(M)$ by choosing the family of open sets $\{\Vcal_\alpha\}$ that don't intersect $\E_S(M)$ such that $\{\Wcal_\tau\}\doteq\{\Ocal_\rho\}\cup\{\Vcal_\alpha\}$ covers $\E(M)$  and for each index $\alpha$ 
\begin{equation*}
S''_M(\ph)\neq0,\quad\forall\ph\in \Vcal_\alpha\,.
\end{equation*}
Now we can write each $F\in\F_{0}(M)$ in the form $F(\ph)=\left<f_{\alpha}(\ph), S''_M(\ph)\right>$ for $\ph\in \Vcal_\alpha$, where $f_{\alpha}(\ph)\doteq hF(\ph)/S''_M(\ph)$ and $h\in\D(M)$ is an arbitrary test function such that $\int\limits_M\dvol h(x)dx=1$. Together with the result proved before we obtain for each $\Wcal_\tau$ an expression: $F(\ph)=\left<f_{\tau}(\ph), S''_M(\ph)\right>$, $\ph\in \Wcal_\tau$. Coefficients $f_{\tau}$ are defined only locally, but the global ones can be constructed using the partition of unity. Since $\E(M)$ is paracompact, we can choose a partition of unity $\sum\limits_\tau \chi_\tau=1$ such that $\supp \chi_\tau\subset \Wcal_\tau$ and at each point $\ph\in\E(M)$ only finitely many functions $\chi_\tau$ don't vanish. The global coefficients can be defined by $f(\ph)\doteq\sum\limits_\tau \chi_\tau f_{\tau}(\ph)$. There is no problem with the convergence, since at each configuration $\ph$ the sum has only finitely many nonvanishing terms. It is important that we can choose the partition of unity in such a way that functions $\chi_\tau$ are local functionals. This is possible, since the topology of the configurations space is initial with respect to local functionals\footnote{See \cite{Michor} for the discussion of the existence of smooth partitions of unity in infinite dimensional analysis.}. 
\end{proof}
The case of a nonlinear action can be handled in a similar way as in the case without the symmetries, provided that the projections $R_\alpha$ can be chosen to be the same in the free and interacting case. We illustrate the general construction presented above on two examples. First we give a finite dimensional one, and next we describe the case of  Yang-Mills theories.
\begin{exa}[Finite dimensional system with symmetries]\label{findimex2}
We continue analysing the example \ref{findimex}. The regularity conditions imposed on $S$ in the finite dimensional case take the form: for each point $x\in E_S$ there exists an open neighborhood with the corresponding chart $(U_\alpha,\varphi_\alpha)$ such that  $\sigma_i((\varphi_\alpha^{-1})^1(x),\ldots,(\varphi_\alpha^{-1})^k(x))=0$, $i=1,\ldots k$ are independent, i.e. the Hessian matrix $H(\varphi_\alpha^{-1}(x))$ is of rank $k$ for all $x\in E_S$. Under this regularity condition we can choose $(\sigma_1\ldots\sigma_k, x_{k+1},\ldots, x_n)$ as new local coordinates in the neighbourhood $U_\alpha$ of each point of $E_S$. The projection $R_\alpha$ can be explicitly written as $R_\alpha(x)(\si_1(x),...,\si_k(x),0,...0)$ and $(1-R_\alpha)(x)=(0,...,0,x_{k+1},..., x_n)$. For the Lagrangian from example \ref{findimex} the Hessian matrix on the solution space has the form:
\begin{equation}
H(x,y,z)=2\sum\limits_{k=1}^N\prod\limits_{{l=1}\atop{l\neq k}}^N(x^2+y^2-k^2)\left(\begin{array}{ccc}
0&0&x\\
0&0&y\\
x&y&0
\end{array}\right)
\end{equation}
The rank of $H(x)$ is equal to 2 for $x\in E_S$. We can choose $\sigma_2$ and $\sigma_3$ as independent equations. The local projections $R_\alpha$ are best written in cylindric coordinates. As an example we take the open neighborhood $\Ocal_1$ (see figure \ref{coord}). 
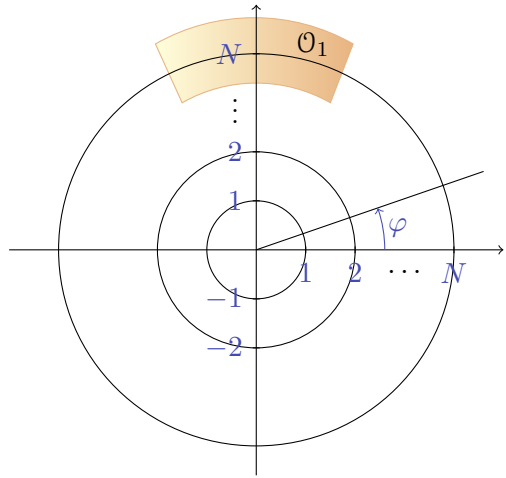
\begin{figure}[htb]
\begin{center}
\begin{tikzpicture}[scale=1.3]
 \useasboundingbox (-4,-1.8) rectangle (4,2.8); 
\shadedraw[left color=lighthoney,right color=honey, draw=honey] (-0.75,1.5) arc (120:60:1.5cm) -- (0.98,2.1) arc (60:120:2cm) -- cycle;
\draw[->] (-2.5,0) -- (2.5,0);	
\draw[->] (0,-2.3) -- (0,2.5); 
\draw(0,0) -- (2.3,0.8); 
\draw (0,0) circle (1cm);
\draw (0,0) circle (0.5cm);
\draw (0,0) circle (2cm);
\draw[->,color=see] (1.3cm,0mm) arc (0:19:1.3cm) node[anchor=north west] {$\ph$};
\foreach \x in {1,2} \draw (0.5*\x cm,1pt) -- (0.5*\x cm,-1pt) node[anchor=north,color=see] {$\x$};
\draw (2cm,1pt) -- (2cm,-1pt) node[anchor=north,color=see] {$N$};
\draw (0.3,2.1) node[anchor=west] {$\Ocal_1$};
\draw (1.5,-0.1) node[anchor=north] {$\ldots$};
\draw (-0.07,1.5) node[anchor=east] {$\vdots$};
\foreach \y in {-2,-1,1,2} \draw (1pt,0.5*\y cm) -- (-1pt,0.5*\y cm) node[anchor=east,color=see] {$\y$};
\draw (1pt, 2cm) -- (-1pt,2cm) node[anchor=east,color=see] {$N$};
\end{tikzpicture}
\end{center}
\caption[Finite dimensional system with symmetries: the choice of coordinates.]{Local choice of coordinates for the action functional $S(x,y,z)=z\prod\limits_{k=1}^N(x^2+y^2-k^2)$.\label{coord}}
\end{figure}
The corresponding projection takes the form:
$R_1(\ph,r,z)=(0,\sigma_2/(r\sin(\ph)),\sigma_3)=:(0,u,v)$, $(1-R_1)(\ph,r,z)=(\ph,0,0)$. This change of coordinates can be inverted on $\Ocal_1$, resulting in $z=u$, $r=\sqrt{v+1}$, so each functional $F(\ph,r,z)$ on $\Ocal_1$ induces a functional $\tilde{F}(\ph,\sigma_2,\sigma_3)=F(\ph,\sqrt{\sigma_3+1},\frac{\sigma_2}{\sin(\ph)\sqrt{\sigma_3+1}})$ and the condition of vanishing on-shell can be expressed as $\tilde{F}(\ph,0,0)=0$. 

Note that the gauge orbits are closed and the change of variables we performed works only locally. A similar problem arises in the infinite dimensional case, when the gauge fixing cannot be done globally. This is a toy model for the so called Gribov problem \cite{Grib}.
\eex\end{exa}
\begin{exa}[Yang-Mills theory]
After a simple finite-dimensional example we can now move to a more physical one. Consider the Yang-Mills theory with the Lagrangian:
\[
L_M(f)(A)=-\frac{1}{2}\int_M f\tr(F\wedge *F)\,,
\]
where $F=dA+[A,A]$ is the field strength corresponding to the gauge potential $A\in\E(M)=\Omega^1(M,g)$, $g$ is a finite dimensional Lie group and $*$ is the Hodge operator. The equations of motion are: $\hinv \!D\!*\!DA=0$, where $D$ is the covariant derivative (the geometry of Yang-Mills theories will be discussed in details in section \ref{geom}). A choice of a gauge fixing defines a surface in $\E(M)$ and local projections $R_\alpha$ are projections to this surface. One can also see it as providing a split in the space of gauge equivalence classes. In the literature concerning the structure of gauge orbit space \cite{Rudolph,Abb,Kon,Kon2} it is referred to as the \textbf{\textit{choice of the gauge slice}}.
\eex\end{exa}
\subsection{Koszul-Tate resolution}
From now on, we shall assume that we are given an action $S$ such that $S''_M$ is invertible after the gauge fixing. The Koszul complex is modified to the Koszul-Tate complex\index{Koszul-Tate!complex} \cite{Tate}. Its underlying graded module is just $\sm(M)\oplus\V(M)\oplus\F(M)$ and the differential is defined as $\delta\doteq\iota\oplus\delta_S\oplus0$, where $\iota$ is the natural inclusion map. We obtain the following short exact sequence:
\[
0\rightarrow \sm(M)\xhookrightarrow{\iota}\V(M)\xrightarrow{\delta_S}\F(M)\rightarrow 0\,.
\]
The space of on-shell functionals is characterized as $\F_S(M)=H_0(\delta)$. Like in the case without nontrivial symmetries, we can construct from the graded module $\sm(M)\oplus\V(M)\oplus\F(M)$ a graded algebra by extending it with graded-symmetric tensor powers. The resulting structure is:
\be\label{KoTa}
\KT(M)\doteq S^\bullet_{\sst\F(M)}\,\sm(M)\underset{{\sst\F(M)}}{\otimes}\bigwedge_{\sst\F(M)} \V(M)\,.
\ee
The differential $\delta$ is extended to $\KT(M)$ by requiring the graded Leibniz rule. We obtain a sequence:
\[
\dots\rightarrow \bigwedge_{\sst\F(M)}^2\V(M)\oplus\sm(M)\xhookrightarrow{\delta_S\oplus\iota}\V(M)\xrightarrow{\delta_S}\F(M)\rightarrow 0\,.
\]
By construction, the kernel of $\delta_S$ is the image of $\delta_S\oplus\iota$, so the above sequence is exact in degree 1. In the next step one has to check if it is also exact in degree 2. If this is not the case, one has to further extend the graded module above to $\sm_1(M)\oplus\sm(M)\oplus\V(M)\oplus\F(M)$ and the differential  to $\delta\doteq\iota_1\oplus\iota\oplus\delta_S\oplus0$, in such a way that the cohomology in degree 2 is ``killed'' by these new generators. The procedure continues until one obtains a resolution, which is called the \textit{Koszul-Tate resolution}\index{Koszul-Tate!resolution}. In general, there is no guarantee that this procedure terminates after finitely many step. Fortunately, in case of local symmetries, one can use the locality to reduce the problem of constructing the Koszul-Tate resolution to a finite dimensional one. In this work we assume that \eqref{KoTa} is already a resolution. This is justified in examples, which we discuss in chapters \ref{YM} and \ref{grav}.
\subsection{Chevalley-Eilenberg cohomology}\label{CheE}
We discuss now in detail the structure of the space of symmetries. The space of all the diffeomorphisms of $\E(M)$ is $\Ci(\E(M),\E(M))$. This space is ``badly'' infinite dimensional and introducing a Lie group structure on it is not possible. The space of diffeomorphisms that leave the action invariant can be too big as well. This, however, doesn't pose a problem for the application in physics, since we are interested only in the ``infinitesimal symmetries'', i.e. vector fields on $\E(M)$, which build a Lie algebra.

It is well known in homological algebra \cite{Weibel,Hilton,Eisen} that with the Lie-algebra action one can associate in a canonical way a co-chain complex, called \textit{the Chevalley-Eilenberg complex}. It turns out however, that in our case a slightly different structure would be more appropriate. From the definition, $\sm(M)$ is the subalgebra of the algebra of vector fields $\V(M)$. This in turn is in a trivial way a so called \textit{Lie algebroid}. In section \ref{cate} we describe it in more detail. Essentially this is a generalization of a Lie algebra.
 From this point of view it is more natural to use the algebroid Chevalley-Eilenberg construction for the action of $\sm(M)$ on $\F(M)$. In section \ref{cate} we discuss in details how it covers many of the physical examples and how it relates to the Lie algebra picture. For now the only consequence of this interpretation is the use of a following definition of the algebraic Chevalley-Eilenberg complex: $\left(\bigwedge_{\sst\F(M)}\sm^*(M),\gamma\right)$, where now the tensor products are over the ring $\F(M)$.  The space  $\sm^*(M)$ is defined as $\sm^*(M)\doteq\Ci(\E(M),\E'(M))/\mathcal{J}$, where $\mathcal{J}\subset\Ci(\E(M),\E'(M))$ is the ideal consisting of forms that vanish on $\sm(M)$ and the duality between $\sm^*(M)$ and $\sm(M)$ is given by:
 \[
 \left<\omega,\xi\right>(\ph)\doteq \left<\omega(\ph),\xi(\ph)\right>\qquad \omega\in\sm^*(M), \xi\in\sm(M)\,.
 \]
The assignment of  $\sm^*(M)$ to a spacetime can be made into a covariant functor. Morphisms $M\rightarrow N$ are mapped to $\sm^*\chi$, defined as:
\[
(\sm^*\chi)(\omega)\doteq \E'\chi\circ \omega\circ \E\chi\,.
\]
The grading of $\bigwedge_{\sst\F(M)}\sm^*(M)$ is called the \textit{pure ghost number} and we denote it by $\#\pg$.
The differential $\gamma$ is given by the formula:
\begin{align}
\gamma:&\ \bigwedge_{\sst\F(M)}^q\sm^*(M)\rightarrow\bigwedge_{\F(M)}^{q+1}\sm^*(M)\,,\nonumber\\
(\gamma \omega)(\xi_0,\ldots, \xi_q)&\doteq\sum\limits_{i=0}^q(-1)^i\partial_{\xi_i}(\omega(\xi_0,\ldots,\hat{\xi}_i,\ldots,\xi_{q+1}))+\nonumber\\
&+\sum\limits_{i<j}(-1)^{i+j}\omega\left([\xi_i,\xi_j],\ldots,\hat{\xi}_i,\ldots,\hat{\xi}_j,\ldots,\xi_{q+1}\right)\,,\label{gaM}
\end{align}
Note that if $F\in\F_S^\inv(M)$ is an on-shell functional  invariant under the action of $\sm(M)$, then $\gamma F\equiv 0$. Therefore the $0$-order cohomology of the Chevalley-Eilenberg complex characterizes the gauge invariant on-shell functionals.\index{gauge!invariant functionals}

The Chevalley-Eilenberg complex can be assigned to a spacetime in a functorial way. Let $\dgA$ be the category with differential graded algebras as objects and differential graded algebra homomorphisms as morphisms. We define a covariant functor $\CE$ from $\Loc$ to $\dgA$, by setting $\CE(M)\doteq\bigwedge\limits_{\sst\F(M)}\sm(M)^*$ for objects. Morphisms are mapped in the way discussed above. Similarly we define $\CE_S(M)\doteq\bigwedge\limits_{\sst\F(M)}\sm(M)^*\Ftens\F_S(M)$. 
\subsection{Going off-shell}\label{off}
The Chevalley-Eilenberg complex constructed in the previous subsection allows us to characterize the gauge invariant functionals as a certain cohomology but we still have to work \textit{on-shell}. We would like to avoid it, since the functional approach to field theory is based on the \textit{off-shell} setting. One would be tempted to simply take the full symmetry algebra $\sm(M)$, construct the Chevalley-Eilenberg complex  $\left(\bigwedge_{\sst\F(M)}\sm^*(M),\gamma\right)$ corresponding to its action on $\F(M)$, calculate the 0-cohomology and go on-shell at the end.
 This, however, is not the optimal solution, since the corresponding cohomology doesn't capture really the gauge invariant on-shell functionals. This can be seen on a simple example 
\begin{exa}\label{findimex3}
Consider the finite dimensional system from \ref{findimex}. An example trivial symmetry is given by:
\[
X(p)=\sigma_3(p)\partial_x-\sigma_1(p)\partial_z\,.
\]
We already showed in \label{findimex2} that the gauge invariant on-shell functionals are those constant on the concentric circles. Let $F(r)$ be such a functional. Acting  on it with $X$ we obtain:
\[
\partial_XF(p)=2r\cos(\ph)\sigma_3(p)F'(r)\,,
\]
therefore $F$ is invariant under $X$ only if $F$ is constant, so the Chevalley-Eilenberg cohomology corresponding to the action of $\sm(M)$ on $\F(M)$ is just $\RR$, whereas the space of gauge-invariant on shell functionals is characterized by $\RR^N$.
\eex\end{exa}
From this example and the discussion of the previous subsection we see that what we need is not the cohomology of $\left(\bigwedge_{\sst\F(M)}\sm^*(M),\gamma\right)$, but of $\left(\bigwedge_{\sst\F(M)}\sm^*(M)\Ftens\F_S(M),\gamma\right)$. To go off-shell we have to replace $\F_S(M)$ by its Koszul-Tate resolution (\ref{KoTa}). We obtain the algebra:
\be\label{BV00}
\BV(M)\doteq\bigwedge_{\sst\F(M)}\sm^*(M)\Ftens S^\bullet_{\sst\F(M)}\,\sm(M)\underset{{\sst\F(M)}}{\otimes}\bigwedge_{\sst\F(M)} \V(M)\,,
\ee
with a differential $\delta$ acting on $\bigwedge_{\sst\F(M)}\sm^*(M)$ as the identity.  We extend the grading $\#\pg$ of $\CE(M)$ to a grading $\#\gh$ (called total ghost number) on $\BV(M)$ by 
\[
\#\gh=\#\pg-\#\af\,.
\]
The antifield number $\#\af=1$ is assigned to the vector fields, the antifield number $\#\af=2$ to the elements of $\sm(M)$, whereas elements of $\bigwedge_{\sst\F(M)}\sm^*(M)$ have $\#\af=0$. 
The Chevalley-Eilenberg differential acts on $H_0(\delta)=\bigwedge\limits_{\sst\F(M)}\sm^*(M)\Ftens\F_S(M)$ and we have:
\[
H^0(H_0(\delta),\gamma)=\F_S^{\inv}(M)\,.
\]
The differential $\gamma$ can be in a natural way extended to vector fields $\V(M)$. Note that $\sm(M)$ has a natural Lie algebroid action on $\V(M)$ via the commutator, so $\gamma$ on $\V(M)$ is just the Chevalley-Eilenberg differential constructed for this action. Explicitly we can write it as:
\begin{align}
\gamma:&\ \bigwedge_{\sst\F(M)}^q\sm^*(M)\Ftens\V(M)\rightarrow\bigwedge_{\sst\F(M)}^{q+1}\sm^*(M)\Ftens\V(M)\,,\nonumber\\
(\gamma \omega)(\xi_0,\ldots, \xi_q)&\doteq\sum\limits_{i=0}^q(-1)^i[\xi_i,\omega(\xi_0,\ldots,\hat{\xi}_i,\ldots,\xi_{q+1})]+\nonumber\\
&+\sum\limits_{i<j}(-1)^{i+j}\omega\left([\xi_i,\xi_j],\ldots,\hat{\xi}_i,\ldots,\hat{\xi}_j,\ldots,\xi_{q+1}\right)\,,\label{gaM}
\end{align}
A similar reasoning applies to the action of $\sm(M)$ on itself. By requiring the graded Leibniz rule we extend $\gamma$ to the full graded algebra $\BV(M)$. This way we obtain a double complex $(\BV(M),\delta,\gamma)$. We can picture this on a diagram\footnote{Since all the constructions are functorial, for the simplicity of notation we don't write the dependence on $M$ explicitly.}:
\be\label{BVshort}
\begin{tikzpicture} \matrix(a)[matrix of math nodes, row sep=2em, column sep=1.5em, text height=2.3ex, text depth=0.2ex] {
  &  0            &0          &0                                                      &\\
...&\bigwedge\limits^2_{\sst\F}\V\oplus\sm &\V&\F                                               &0\\
 &     ...          &{\bigwedge\limits^1_{\sst\F}\sm^*\Ft\V}          &{\bigwedge\limits^1_{\sst\F}\sm^*}&0\\
 &               &          &{\ldots}&\\
 };
 \path[->,font=\scriptsize](a-2-1) edge node[above]{} (a-2-2); 
 \path[->,font=\scriptsize](a-2-2) edge node[above]{$\delta_S\oplus\iota$} (a-2-3); 
 \path[->,font=\scriptsize](a-3-2) edge node[above]{} (a-3-3); 
 \path[->,font=\scriptsize](a-2-3) edge node[above]{$\delta_S$} (a-2-4); 
 \path[->,font=\scriptsize](a-3-3) edge node[above]{$\delta_S$} (a-3-4); 
 \path[->,font=\scriptsize](a-2-4) edge node[above]{} (a-2-5); 
 \path[->,font=\scriptsize](a-3-4) edge node[above]{} (a-3-5); 
 \path[->,font=\scriptsize](a-1-4) edge node[above]{} (a-2-4); 
  \path[->,font=\scriptsize](a-1-3) edge node[above]{} (a-2-3); 
  \path[->,font=\scriptsize](a-1-2) edge node[above]{} (a-2-2); 
 \path[->,font=\scriptsize](a-2-2) edge node[right]{$\gamma$} (a-3-2); 
 \path[->,font=\scriptsize](a-2-4) edge node[right]{$\gamma$} (a-3-4); 
 \path[->,font=\scriptsize](a-2-3) edge node[right]{$\gamma$} (a-3-3); 
 \path[->,font=\scriptsize](a-3-4) edge node[right]{$\gamma$} (a-4-4); 
\end{tikzpicture}
\ee
The full BV differential is in our case defined as $s=\delta+\gamma$. In general, if $\gamma^2$ is not 0 off-shell, $s$ has to contain also higher order terms, but in the present work we study the examples where $\delta+\gamma$ is already a differential. Since $(\BV(M),\delta)$ is a resolution, one can use the main theorem of homological perturbation theory\footnote{See \cite{FH,Henneaux:1992ig}.} to conclude that
\[
H^0(\BV(M),s)=H^0(H_0(\BV(M),\delta),\gamma)\,.
\]
The right hand side of the above relation is called the \textit{cohomology of $\gamma$ modulo $\delta$} and is denoted by  $H^*(\gamma|H_*(\delta))$. We can identify the gauge invariant functionals\index{gauge!invariant functionals} with the $0$-cohomology of $s$, i.e. $\F_\inv^S(M)=H^0(\BV(M),s)$. 
 This is the most general situation in the BV formalism. Although we already obtained the homological interpretation of $\F_\inv^S(M)$, from the practical point of view it is worth to work on this structure a little bit more, in order to find more explicit formulas for the objects involved.
\subsection{Open and closed algebras}\label{open}
In this subsection we will analyze in detail the structure of $\smp(M)$.
We want to make contact with the standard approach \cite{Batalin:1981jr,Batalin:1983wj,Batalin:1983jr} and show how the structures we are using relate to those used in the literature. Therefore, we shall now discuss the interpretation of \textit{open} and \textit{closed} algebras.
In section \ref{off}, we showed that one can interpret $\F_\inv^S(M)$ as a cohomology of $\gamma$ modulo $\delta$, where the Chevalley Eilenberg complex was constructed for the full symmetry algebra $\sm(M)$. This algebra is ``much bigger'' than $\smp(M)$ since it contains all the trivial symmetries. In praxis one wants to work with a smaller space. For concrete examples this can be done, but involves a great deal of arbitrary choices. Nevertheless, we think it is worth to follow the underlying reasoning to see at which points the choices have to be made.

The goal is to find a resolution of the graded differential algebra $\CE_S(M)$. The difficulty
 is that one has to bring off-shell not only the observables, but also the symmetries. This is a little bit problematic from a conceptual point of view, since in a general case the physical symmetries $\smp(M)$ are an intrinsic property of the solution subspace $\E_S(M)$. Their extension to $\E(M)$ is in principle an extension problem for vector fields defined on a submanifold and need not be possible in general. Nevertheless in a wide class of examples one can find a 
subset $\sr(M)\subset\sm(M)$ such that we can write each element $X\in\sm(M)$ in a form: $X=Y+I$, where $Y\in\sr(M)$ and $I\in\sm_0(M)$ is a trivial symmetry. The space $\sr(M)$ is in general not a Lie subalgebra of   $\sm(M)$. We shall call it the \textit{reduced space of symmetries}\index{symmetries!reduced}. This leads to a classification usually used in physics:
\begin{itemize}
\item The space of reduced symmetries $\sr(M)$ is a subalgebra of $\sm(M)$. In physics terminology this situation corresponds to a \textit{\textbf{closed algebra}}.
\item The space of reduced symmetries $\sr(M)$ is \textit{not} a subalgebra of $\sm(M)$. We deal with a \textit{\textbf{open algebra}} situation.\index{open symmetry algebra}
\end{itemize} 
It was already stressed in \cite{Henneaux:1992ig} that this terminology doesn't really concern the algebra of symmetries $\sm(M)$, which is always a Lie algebra, but rather a specific parametrization of it. In the book \cite{Henneaux:1992ig} this parametrization is related to the choice of generating sets, whereas in our setting it corresponds to the choice of a split (as a vector space not as a Lie algebra) of the quotient $\smp(M)$. Although $\sr(M)$ is not a Lie-subalgebra of $\sm(M)$, we can still define a graded algebra $\CE(M)\doteq\bigwedge_{\sst\F(M)}\sr(M)^*$ and a derivation $\gamma$ on it is given by (\ref{gaM}). In general it is not a differential on $\CE(M)$, since it is not nilpotent of degree 2, but becomes a differential on-shell. In homological algebra this notion has a precise meaning as:
\begin{defn}
Let $\delta \in \Der(A)$ be a differential of degree 1 on the graded algebra $A$. We consider its homology $H^*(\delta)$. Let $d$ be a derivation also with degree $1$, satisfying:
$d\delta + \delta d = 0$ and $d^2$ is $\delta$-exact, i.e., there is a derivation $D$ such that
$d^2 =[D,\delta]$. Then $d$ induces a differential (which we still call $d$) on $H^*(\delta)$. We denote the cohomology of $d$ on $H^*(\delta)$ by $H^*(d|H_*(\delta))$ and call $d$ a \textit{differential modulo $\delta$}.
\end{defn}
In case of an ``open algebra'' $\gamma$ can be defined as a differential modulo $\delta$ and the main theorem of homological perturbation theory can be again applied to conclude, that $\F_\inv^S(M)=H^0(\gamma|H_0(\delta))$.
\subsection{Appendix: Interpretation in terms of category theory}\label{cate}
In this section we want to provide a motivation for the definition of the Chevalley-Eilenber complex given in \ref{CheE}. First we introduce some definitions from category theory \cite{MacLane}.
\begin{defn}
A Lie algebroid $\fa$ is a triple $(E, [\cdot,\cdot], \rho)$ consisting of a vector bundle $E$ over a manifold $\Mcal$, together with a Lie bracket $[\cdot,\cdot]$ on its module of sections $\Ga(E)$ and a morphism of vector bundles $\rho: E\rightarrow T\Mcal$ called the \textit{anchor}. The anchor and the bracket are to satisfy the Leibniz rule:
\[
    [X,fY]=\rho(X)f\cdot Y + f[X,Y]
\]
where $X,Y \in \Gamma(E), f\in C^\infty(\Mcal)$ and $\rho(X)f$ is the derivative of $f$ along the vector field $\rho(X)$. It follows that:
\[
    \rho([X,Y]) = [\rho(X),\rho(Y)]
\]
for all $X,Y \in \Gamma(E)$.
\end{defn}
The corresponding Chevalley-Eilenberg algebra is defined as $\Big(\bigwedge_{\Ci(\Mcal)}\Gamma(E)^*,\gamma\Big)$, where the tensor products and the dualization are over the ring $\Ci(\Mcal)$  and the differential $\gamma$ is given by:
\begin{align*}
(\gamma\omega)(\xi_0,\ldots, \xi_q)&\doteq\sum\limits_{i=0}^q(-1)^i\rho(\xi_i)(\omega(\xi_0,...,\hat{\xi}_i,...,\xi_q))+\\
&+\sum\limits_{i<j}(-1)^{i+j}\omega\left([\xi_i,\xi_j],...,\hat{\xi}_i,...,\hat{\xi}_j,...,\xi_q\right)\,.
\end{align*}
An ``integrated'' version of Lie algebroids are Lie groupoids. They play a role similar as the Lie groups play to Lie algebras.
\begin{defn}
A \textit{groupoid} is a small category in which every morphism is an isomorphism.
\end{defn}
In other words a groupoid $G\rightrightarrows B$ consists of a set $B$ of objects (usually called the base), and a set $G$ of morphisms, usually called the arrows. Each arrow has an associated source object and an associated target object. This means that there are two maps
$s, t : G  \rightrightarrows B$ called the source and target, respectively. 
A Lie groupoid is a groupoid where the set $B$ of objects and the set $G$ of morphisms are both manifolds, the source and target operations are submersions, and all the category operations are smooth.
More explicitly a Lie groupoid is given by a following data:
\begin{itemize}
\item two smooth manifolds $G$ (arrows) and B (objects), 
\item two smooth maps (source and target): $s, t : G  \rightrightarrows B$,
\item a smooth embedding $\iota : B\rightarrow  G$ (the identities or constant arrows),
\item a smooth involution $I : G \rightarrow G$, also denoted $x \mapsto x^{\minus}$,
\item a multiplication $m:G_2 \rightarrow G$, $(x,y) \mapsto x\cdot y$,
where $G_2 =G_s \times_t G=\{(x,y)\in G\times G|s(x)=t(y)\}$, such that the source map and target map are surjective submersions (hence $G_2$ is a smooth manifold), the multiplication is smooth and
\begin{itemize}
\item $s(x\cdot y)=s(y)$, $t(x\cdot y)=t(x)$, 
\item $x\cdot(y\cdot z)=(x\cdot y)\cdot z$, 
\item $\iota(t(x))\cdot x=x\cdot\iota(s(x))$,
\item $s(x^{\minus}) = t(x)$, $t(x^{\minus}) = s(x)$,
\item $x \cdot x^{\minus} = \iota(t(x))$, $x^{\minus}\cdot x = \iota(s(x))$,
\end{itemize}
whenever $(x, y)$ and $(y, z)$ are in $G_2$.
\end{itemize}
In particular with a left action of a Lie group $G$ on a manifold $\Mcal$ one can associate a groupoid called \textit{the action groupoid} in a following way:
\begin{itemize}
\item $G\times \Mcal$ is the set of morphisms,
\item an object is an element of $\Mcal$,
\item a morphism from $x\in \Mcal$ to $y\in \Mcal$ is a group element $g\in G$ with $gx=y$. A general morphism is a pair $(g,x):x\rightarrow gx$.
\item The multiplication of $(g,x):x\rightarrow gx=y$ and $(h,y):x\rightarrow hy$ is $(hg,x):x\rightarrow hgx$.
\item $s(g, x)=g^{\minus}x$,
\item $t(g, x)=x$,
\item $I(g,x)=(g^{\minus},g^{\minus}x)$,
\item $\iota(x)=(e,x)$, where $e$ is the unit of $G$.
\end{itemize}
As we mentioned before one can associate to a Lie groupoid a Lie algebroid in a canonical way. In case of the action groupoid we obtain the action algebroid. Explicitely it is given by the Lie algebra $\frakg$ of $G$ acting on the manifold $\Mcal$. This induces a morphism of Lie algebras
$\frakg \rightarrow \Gamma(T\Mcal)$, $\xi\mapsto X_{\xi}$. Consider now $\fa = \frakg \times \Mcal$ a trivial vector bundle over $\Mcal$. We identify sections of $\fa$ with maps $\Mcal \rightarrow \frakg$ and define a bracket on sections by:
\[
[a, b](x) = [a(x), b(x)] + \pa_{X_{b(x)}}a  - \pa_{X_{a(x)}}b\,.
\]
The anchor map $\rho:\fa\rightarrow T\Mcal$ is defined by:
\[
(\xi,x) \mapsto X_{\xi}(x)\,.
\]
Coming back to the BV formalism we see that the action of a Lie group on the configuration space $\E(M)$ (considered as a trivial manifold) in a natural way provides us a Lie algebra $\Ci(\E(M),\frakg)$ and a morphism $\rho$. We can identify $\Ci(\E(M),\frakg)$ with the reduced symmetry algebra $\sr(M)$ and the Chevalley-Eilenberg complex of the Lie algebroid $(\sr(M),[.,.],\rho)$ is exactly the complex we defined in section \ref{CheE}.
\section{Yang-Mills theories}\label{YM}
In the previous section,  we discussed the general structure of the BV construction. Now we turn to a particular example, where it can be applied, namely to Yang-Mills theory. In this case one has a simple characterization of symmetries as being induced by the action of a Lie group on the configuration space. We start this section with some geometrical preliminaries concerning the structure of gauge theories. Next we construct the Batalin-Vilkovisky complex and perform the gauge fixing. We end this section with providing a construction of the Poisson algebra of classical Yang-Mills theories. 
\subsection{Geometrical preliminaries}\label{geom}
To put gauge theories into the category theory setting we recall first some basic definitions concerning principal bundles (see for example \cite{Kob,Gan,Michor2,Nak}). First, we introduce a category of fibered manifolds $\FM$. A fibered manifold is defined as:
\begin{defn}
A triple $(N, \pi, M )$, where $\pi : N \rightarrow M$ is a surjective submersion, is called a \textit{fibered manifold}\index{fibered manifold}. $N$ is called the total space, $M$ is called the 
base.
\end{defn}
In this definition one point needs few words of comment. A mapping $f : N \rightarrow M$ between manifolds is called a \textit{submersion}\index{submersion} at $x\in  N$, if the rank of $T_x f : T_x N \rightarrow T_{f (x)} M$ equals dim $M$ . The mapping $f$ is 
said to be a submersion, if it is a submersion at each $x\in N$. 

To have a category, beside objects, we still need morphisms. In case of fibered manifolds they are defined in quite a straightforward way. Given another fibered manifold
$(N',\pi',M')$, a morphism $(N,\pi,M)\rightarrow (N',\pi',M')$ means a smooth map $f:N\rightarrow N'$ transforming each fiber $N_x\doteq \pi^{-1}(x)$ of $N$ into a fiber of $N'$. The relation $f(N_x) \subset N'_{x'}$ defines a map $\underline{f} : M\rightarrow M'$, which is characterized by the property: 
$\pi' \circ f = \underline{f} \circ \pi$. We say that $f$ \textit{covers} $\underline{f}$. Using this definitions one can introduce a following category:
\begin{itemize}
\item[$\FM$] 
\begin{itemize}
 \item[$\ $] $\obj(\FM)$: fibered manifolds over $\obj(\Loc)$\index{category! of fibered manifolds}
  \item[$\ $] {\bf Morphisms}: fibered manifolds' morphisms that cover those of $\Loc$
\end{itemize}
\end{itemize}
There is an important functor from $\FM$ into the category of spacetimes $\Loc$. It generalizes the notion of a base space to the level of categories.
\begin{defn}
A \textit{base functor}\index{base functor} is a functor $\Ba:\FM\rightarrow  \Loc$ which assigns to every fibered manifold $(N, \pi, M )$  its base $M$ and to every fibered manifold morphism $f : (N, \pi, M ) \rightarrow ( N',\pi',M')$ the induced map $\underline{f} : M\rightarrow M'$.
\end{defn}
A special case of a fibered manifold is a principal bundle. This structure is of particular interest in the present work, since we want to consider the example of Yang-Mills theories. In the rest of this mathematical warm-up we want to revise some basic facts concerning principal bundles, but we will do it already in the more abstract language. At the same time we want to show the relation to more standard definitions used in the physics literature. One can consider this section as a dictionary, which makes it easier to jump from very abstract concepts to practical calculations. We start with recalling the definition of the $G$-bundle structure. This is the first step on the way to define principal bundles. Let $G$ be a fixed finite dimensional Lie group and $g$ it's Lie algebra. 
\begin{defn}[\cite{Michor2}, 10.1.]\label{10.1.}
Let $(N, \pi, M)$ be a fiber bundle with a standard fiber $S$. A $G$-bundle structure on the fiber bundle consists of the following data:
\begin{enumerate}
\item A left action: $\ell:G\times S\rightarrow S$ of the Lie group on the standard fiber. 
\item A fiber bundle atlas $(U_\alpha,\psi_\alpha)$ whose transition functions $\psi_{\alpha\beta}$ act on $S$ via the $G$-action: There is a family of smooth mappings $(\varphi_{\alpha\beta}: U_{\alpha\beta}\rightarrow G)$ which satisfies the cocycle condition $\varphi_{\alpha\beta}(x)\varphi_{\beta\gamma}(x) = \varphi_{\alpha\gamma}(x)$ for $x\in U_{\alpha\beta\gamma}$ and $\varphi_{\alpha\alpha}(x) = e$, the unit in the group, such that $\psi_{\alpha\beta}(x,s)=\ell(\varphi_{\alpha\beta}(x), s) =\varphi_{\alpha\beta}(x).s$.
A fiber bundle with a G-bundle structure is called a G-bundle.
\end{enumerate}
\end{defn}
To summarize, a $G$-bundle is a fiber bundle equipped with a left action of the Lie group $G$ on the standard fiber. This action has to be of course compatible with the bundle structure. A special case of a $G$-bundle is a fibered manifold, where the group $G$ itself is the standard fiber. This structure is called a \textit{principal bundle}.
\begin{defn}[\cite{Michor2}, 10.2.]\label{10.2.}
A	principal	(fiber)\index{principal!bundle}	bundle	$P(G,M,\pi)$ over the base $M$	is	a	$G$-bundle	with typical fiber a Lie group $G$, where the left action of $G$ on $G$ is just the left translation.
\end{defn}
Group $G$ is called \textit{the structure group}\index{structure group}. Each principal bundle admits a unique right action\index{principal!right action} $r : P \times G \rightarrow P$ , called the \textit{principal right action}, given by
\[
\ph (r(\ph^{-1}(x, a), g)) = (x, ag)\,.
\]
In order to define a category of principal bindles we have to specify the morphisms. 
Let $P(G,M,\pi)$ and  $P'(G,M',\pi')$ be principal bundles over manifolds $M$ and $M'$ respectively with $\pi$ and $\pi'$ being projections to the base. A pair $(\chi, \underline{\chi})$, $\chi: P\rightarrow P'$, $\underline{\chi}:M\rightarrow M'$ is called a bundle mapping of principal bundles $P(G,M,\pi)$  and $P'(G,M',\pi)$  if $\chi$, $\underline{\chi}$ are smooth, $\chi(p\cdot g)=\chi(p)\cdot g$ holds for all $p\in P$, $g\in G$ (equivariance) and we have a following diagram:
 \begin{equation*}
\begin{CD}
P(G,M) @>\chi>>P(G',M')\\ 
@V{\pi}VV     @VV{\pi'}V\\
M@>{\underline{\chi}}>>M'
\end{CD}
\end{equation*}
Given $\chi$, the map $\underline{\chi}$ is uniquely determined by the requirement, that the above diagram commutes. If this is an identity map, $\chi$ is a bundle isomorphism. We define the following category:
\begin{itemize}
\item[$\PB(G)\!\!\!\!\!\!$] 
\begin{itemize}
 \item[$\ $] $\obj(\PB(G))$: principal $G$-bundles over manifods $M\in\Loc$.\index{category!of principal bundles}
 \item[$\ $] {\bf Morphisms}: bundle mappings $\chi$
\end{itemize}
\end{itemize}
With those definitions we can now describe quantities that are of importance in gauge theories\index{gauge!theories}, namely the \textit{\textbf{connections}}. We start with a more general geometrical definition and at the end we make contact with the more practical approach, used in the physics literature.

Consider the tangent bundle $TP$ of the principal bundle $P$. One defines the vertical bundle $V$ that consists of all vectors which are tangent to the fibers as $V=\ke( d\pi)$, where $d\pi:TP\rightarrow \pi^*TM$ and  $\pi^*TM$ is the pullback bundle. An \textbf{\textit{Ehresmann connection}}\index{connection} on E is a smooth subbundle $H$ of $TP$, called the horizontal bundle of the connection, which is complementary to $V$, i.e. $TP = H\oplus V$. It is a \textbf{\textit{principal connection}}\index{connection!principal} if it is additionaly $G$-equivariant in the sense that $dr_g(H)=H$, where $g\in G$ and $r$ is the principal right action.

A connection allows us to decompose each vector $X\in T_pP,\ p\in P$ into horizontal and vertical part: $X=X^H+X^V$. Now let $\phi\in\Omega^m(P)\otimes V$ be a vector-valued $m$-form on $P$ ($V$ is a $k$-dimensional vector space with basis $\{e_a\}$). A connection provides the \textbf{\textit{covariant derivative}}, defined as:
\begin{equation*}
D\phi(X_1,\ldots, X_{m+1})\doteq d_P\phi(X_1^H,\ldots, X^H_{m+1})\,,
\end{equation*}
where $d_P\phi\doteq d_P\phi^a\otimes e_a$, $\phi=\sum\limits_{a=1}^k\phi\otimes e_a$ and $d_P$  is the exterior derivative.

Connections on principal bundles are uniquely defined by the so called  \textbf{\textit{connection 1-forms}}\index{connection!1-form}. This notion is especially important in the physics context.
\begin{defn} A (gauge) connection\index{gauge!connection} 1-form on $P(G,M,\pi)$ is a $g$-valued form $\alpha\in \Omega^1(P,g)$ such that:
\begin{enumerate}
\item $\alpha(Z_{\xi})=\xi$, for $Z_{\xi}$ the fundamental vector field generated by $\xi\in g$,
\item $R^*_g\alpha=\ad_{g^{-1}}\alpha$, where $R_g$ is the right action of $G$ on itself, $\ad$ is the adjoint representation of $G$ on $g$.
\end{enumerate}
\end{defn}
For the completeness we recall that the fundamental vector field\index{fundamental vector field} $Z_{\xi}\in\Gamma(TP)$ is defined by the relation $Z_\xi(x)\doteq \frac{d}{d t}(e^{t\xi}.x)\big|_{t=0}= T_e(l^x).\xi$, where $l^x$ is the map $G\rightarrow P$ given by $l^x(a)=l(a,x)=a.x$, $a\in G$.

Now we make contact with the usual physics formulation by writing the object defined above in  local coordinates. Let $\{U_i\}_{i\in I}$ be an open covering of $M$. Using local sections $\sigma_i$ we can define a pullback of a connection 1-form $\alpha$ on each open set $U_i$:
\begin{equation*}
\euA_i\equiv\sigma_i^*\alpha\in\Omega^1(U_i,g)\,.
\end{equation*}
When $P(M,G)$ is the trivial bundle, then we can define this pullback globally. Otherwise it exists only locally. Passing between two different trivializations results in identities between corresponding $\euA_i$ and $\euA_j$, namely:
\be\label{gaug}
\euA_i=\ad_{g_{ij}}\circ \euA_j+g_{ji}^*\theta\,,
\ee
where $g_{ij}:U_i\cap U_j=:U_{ij}\rightarrow G$ are the transition functions and $\theta$ is the Maurer-Cartan form. Given a family of such 1-forms $\euA_i\in\Omega^1(U_i,g)$ satisfying condition (\ref{gaug}) on overlaps $U_{ij}$ one can construct a globally defined $\alpha\in\Omega^1(P,g)$.

Another important structure in gauge theories is the \textbf{\textit{field strenth}}. It arises as a local expression for the curvature 2-form of the gauge connection.
\begin{defn}
\textit{The curvature 2-form}\index{curvature 2-form} $\Omega$ is the covariant derivative of the connection 1-form $\alpha$:
\begin{equation*}
\Omega\doteq D\alpha\,.
\end{equation*}
\end{defn}
Given a local trivialization $\{(U_i\sigma_i)\}_{i\in I}$ one defines the pullback:
\begin{equation*}
\euF_i\doteq\sigma^*\Omega\,.
\end{equation*}
This is called a \textit{field strength}\index{field strength}. From Cartan's structure equations follows that it is related to the local form $\euA_i$ by
\begin{equation*}
\euF_i=d\euA_i+\euA_i\wedge \euA_i\,,
\end{equation*} 
where $d$ is the exterior derivative on $M$. After presenting the local description of the connection 1-form we would like now to describe the structure of the space of all such forms. It turns out that it is an affine space. Before we get into details, we recall one more definition important in the gauge theory, namely the \textit{\textbf{associated bundle}}\index{associated bundle}.
\begin{defn}
Let $P(G,M,\pi)$ be a principal bundle and let $\ell:G\times S\rightarrow S$ be a left action of $G$ on a manifold $S$. We can construct an associated bundle $(P\times_G S)\rightarrow M$, denoted by $P[S,\ell]$ as the quotient  $P\times_G S:= (P \times S) /G$ by the action $(p,v)g =(pg,l_{g^{\minus}}v)$.
\end{defn}
In particular we can choose $S$ to be $g$ and the action $\ell$ to be the adjoined action $\ad$ of $G$ on $g$. We call the resulting associated bundle $P[S,\ell]=:\adP$, the \textit{adjoint bundle}. 

We can now come back to our discussion of the space of all the affine connections. Since it is not a vector space, but an affine space, to get a vector space structure one has to fix a reference connection $\alpha_0$. Now we can characterize all the gauge connections using the following result: Let $\alpha$ be another connection 1-form on a principal bundle $P$, then the difference  $A=\alpha-\alpha_0$ is a $G$-equivariant g-valued 1-form on $P$ which is horizontal in the sense that it vanishes on any section of the vertical bundle $V$ of $P$. Hence it is determined by a 1-form on $M$ with values in the adjoint bundle. Conversely, any element of $\Omega^1(M,\adP)$, defines a $G$-equivariant horizontal 1-form on $P$ and the space of principal connections is an affine space for this space of 1-forms. We call elements of $\Omega^1(M,\adP)$ \textit{\textbf{gauge potentials}}\index{gauge!potential}.
Locally $A$ is given by a family of $g$-valued 1-forms $A_i\in \Omega^1(U_i,g)$, called  \textit{\textbf{ local gauge potentials}}. Those are the objects familiar in physics. The field strength corresponding to $\alpha$ can be locally written as:
\[
\euF_i=\euF_{0,i}+dA_i+[A_i,A_i]\,,
\]
We identify the configuration space of the gauge theory with $\Omega^1(M,\adP)$. It is justified, since our setting is perturbational, i.e. we consider as dynamical variables only the perturbation around the given background connection. Therefore we can consider as a physical configuration space the vector space underlying the affine space of connections, instead of the affine space itself.

We want to make this assignment functorial. There is however a problem, since for a given spacetime $M$ with a nontrivial topology one can construct different principal bundles $P$. There are two ways to overcome the problem. Either we assume that the bundle is always trivial, or we can replace the category $\Loc$ with the category of principal bundles $\PB$ as the underlying structure for the locally covariant field theory. The second option is justified if there are some physical effects that depend on the topology and distinguish between different topologies. The Dirac monopole can be considered as such an effect. Nevertheless, for now we settle for the first option, namely we assume the bundle to be trivial. This provides us a functor from $\Loc$ to $\PB$ and composition of this functor with the functor of sections going from $\PB$ to $\Vect$ results in a contravariant functor $\E:\Loc\rightarrow \Vect$. It assigns to a spacetimes $M$ the configuration space $\E(M)=\Omega^1(P,g)^G\cong \Omega^1(M,g)$. 
Now let $\chi:M\rightarrow N$ be a causal isometric embedding. We can define a morphism from $\E(N)$ to $\E(M)$ in a natural way by setting: $\E\chi(\omega\otimes a):=\chi^*\omega\otimes a$, where $\omega\in\Omega^1(M)$, $a\in g$ and the pullback of a differential form is defined as:
 $(\chi^*\omega)_x:=\omega_{\chi(x)}\circ d_x\chi$. In this way $\E$ becomes a contravariant functor between the categories $\Loc$ and $\Vect$. One can also define a covariant functor $\E_c$ by assigning to a spacetime the space of compactly supported $g$-valued forms $\Omega_c^1(M,g)$. In this case ${\E_c}\chi$ maps forms to their push-forwards.

Now we introduce the generalized Lagrangian\index{generalized Lagrangian! of Yang-Mills theory}
 \be\label{LagrYM}
 L_M(f)=-\frac{1}{2}\int_M f\,\tr(F \wedge * F)\,,
 \ee
 where $F=dA+[A,A]$ is the field strength corresponding to the gauge potential $A$ and $*$ is the Hodge operator. One can check that for this action the linearized equation of motion might possess nontrivial compactly supported solutions. Therefore from the criterion (\ref{sym}) follows that $S(L)$ has nontrivial local symmetries. Actually these symmetries can be easily characterized. They arise as a consequence of the group action of the so called gauge group\index{gauge!group}. This is an infinite dimensional space consisting of vertical $G$-equivariant compactly supported diffeomorphisms of $P$:
 \[
 \Gcal:=\{\alpha\in \Diff_c(P)|\alpha(p\cdot g)=\alpha(p)\cdot g, \pi(\alpha(p))=\pi(p),\ \forall g\in G, p\in P\}\,.
 \]
This space can be also characterized by $\Gcal\cong\Gamma_c(M\leftarrow(P\times_G G))$. For a trivial bundle $P$ this is just $\Gcal(M)\cong\Ci_c(M,G)$. It was shown (\cite{Neeb04,Gloe,Michor}, see also \cite{Neeb,Wock}) that $\Ci_c(M,G)$ can be equipped with a structure of an infinite dimensional Lie group modeled on its Lie algebra $\frakgc(M):=\Ci_c(M,g)$. The exponential mapping can be defined and it induces a local diffeomorphism at 0. While the gauge group is a subgroup of $\Diff(P)$, it has a natural action on $\Omega^1(P,g)^G$ by the pullback:
\[\rho_M(\alpha)A=(\alpha^{-1})^*A, \quad\alpha\in\mathcal{G},\quad A\in\Omega^1(P,g)^G\] 
The derived action of the Lie algebra $\frakgc$ on $\Omega^1(P,g)^G$ is therefore defined as:
\be\label{rho}
\rho_M(\xi)A\doteq\frac{d}{dt}\Big|_{t=0}\rho_M(\exp t\xi)A=\frac{d}{dt}\Big|_{t=0}(\exp(-t\xi))^*A=\pounds_{Z_{\xi}}A=d\xi+[A,\xi]\,,
\ee
where $Z_{\xi}$ is the fundamental vector field on $P$ associated to $\xi$. We see that $\rho_M$ induces a map from $\frakgc$ to the space of vector fields on $\E(M)$, which assigns to a gauge parameter $\xi\in\frakgc$ a vector field $\rho_M(\xi)$. This is indeed an element of $\V(M)$, since $\rho_M(\xi)$ associates to the field configuration $A\in\E(M)$ the compactly supported configuration $d\xi+[A,\xi]$, i.e.
\be\label{rhoM}
\partial_{\rho_M(\xi)}F(A)=\left<F^{(1)}(A),\rho(\xi)A\right>=\left<F^{(1)}(A),d\xi+[A,\xi]\right>\,.
\ee
 Clearly, $\rho$ is a natural transformation from $\frakgc$ to $\V$ as may be seen from the relation
\[
\rho_N(\chi_*\xi)A=\chi_*(\rho_M(\xi)(\chi^*A))
\]
for a causal embedding $\chi:M\to N$.

Moreover one can check that the full algebra of nontrivial local symmetries can be characterized by $\fG(M)\doteq\Ci_\ml(\E(M),\frakgc(M))$, the $\F$-module of smooth compactly supported multilocal\footnote{The notions of support and locality are defined as in \ref{vvf} with $\W(M)=\frakgc$, i.e. $\Xi^{(k)}(\ph)\in\Gamma'(M^{k+1},V^{\otimes k}\otimes g)$.} functions from $\E(M)$ to $\frakgc(M)$. The map from $\fG(M)$ to $\V(M)$, denoted by the same symbol as in ($\ref{rhoM}$), is defined as:
\be\label{rhoM2}
\partial_{\rho_M(\Xi)}F(A)=\left<F^{(1)}(A),\rho(\Xi(A))A\right>\,,
\ee
where $\Xi\in\fG(M)$. $\fG$ becomes a covariant functor after setting: $\fG\chi(\Xi):=\frakgc\chi\circ\Xi\circ\E\chi$. With this definition (\ref{rhoM2}) yields a natural transformation $\rho$ between $\fG$ and $\V$. The space $\fG$ completely characterizes the nontrivial local symmetries, since any symmetry may be obtained by a sum of a trivial one and one of the form $\rho_M(\Xi)$ with $\Xi\in\fG(M)$.
\subsection{Chevalley-Eilenberg complex} \label{ChEil}
We start this section with the construction of the off-shell Chevalley-Eilenberg cochain complex\index{Chevalley-Eilenberg!complex} $\CE(M)$.  The space of on-shell symmetries $\smp(M)$ from section \ref{CheE} is in this case identified with $\Ci_\ml(\E_S(M),\frakg_c)$. Since the symmetries arise from an action of an abstract Lie group, they are well defined also off-shell and we can construct an off-shell version of the Chevalley-Eilenberg complex by means of the reduced symmetry algebra $\sr(M)=\Ci_\ml(\E(M),\frakg_c)$. We recall that the algebraic Chevalley-Eilenberg complex\index{Chevalley-Eilenberg!complex!algebraic} is given by $\left(\bigwedge_{\sst\F(M)}\sr^*(M),\gamma\right)$. The topological one can be defined by restricting the underlying algebra to
\[
\CE(M)\doteq\Ci_\ml\left(\E(M),\bigwedge\frakg'(M)\right)\,,
\]
equipped with the pointwise convergence topology. 

The assignment of $\frakg(M)$ to a manifold $M$ is a contravariant functor. It associates to a morphism $\chi:M\rightarrow N$ a map $\frakg\chi$ acting on functions as a pullback: $\frakg\chi(f\otimes a):=\chi^*f\otimes a$ for $f\in\Ci(N)$, $a\in g$.
From the naturality of $\rho$ and $[.,.]$ it follows that $\CE(M)$ becomes a covariant functor from $\Loc$ to $\dgA$ if we set $\CE\chi(\omega):=(\frakg'\chi)^k\circ\omega\circ\E\chi$, for $\omega\in \Ci(\E(M),\Lambda^{k}\frakg'(M))$. 

To end this subsection we want to make a comment on the relation to the ``standard'' approach. 
In physics literature one often writes distributions in the integral notation (see the discussion in section \ref{distr}). This way the elements of 
$\CE(M)$ would be formally written as:
\[
F(\ph)\stackrel{\textrm{formal}}{=}\int \dvol (x^1)\ldots\ \dvol(x^n) f_{a_1...a_n}(\ph)(x_1,\ldots, x_n) C(x_1)^{a_1}\wedge...\wedge C(x_n)^{a_n}\,,
\]
where $f_{a_1...a_n}(\ph)$ is a compactly supported distribution and $C^a(x)\in  {\frakg}'$ are coefficients of the Maurer-Cartan form $C$ on $\Gcal(M)$. In physics one calls them the \textit{ghost fields}\index{ghosts}. They can be seen as formal ``generators'' of the algebra $\CE(M)$. In the present setting they appear naturally as elements of the Chevalley-Eilenberg complex.
\subsection{Batalin-Vilkovisky complex}\label{KTcom}
We have finally come to the main point of this section, namely to the construction of the Batalin-Vilkovisky  complex. 
First we want to present one more interpretation of its underlying algebra (\ref{BV00}). Just as in \ref{K} we extended the algebra of functionals with its local compactly supported derivations, we can extend the graded algebra  $\bigwedge_{\F(M)}\sr(M)^*$ with graded symmetric powers of its graded derivations. One can easily check, that
\[
\bigwedge_{\sst\F(M)}\sr(M)^*\Ftens S^\bullet_{\sst\F(M)}\,\sr(M)\underset{{\sst\F(M)}}{\otimes}\bigwedge_{\sst\F(M)} \V(M)\subset S^\bullet_{\sst\F(M)}\Der\Big(\bigwedge_{\sst\F(M)}\sr(M)^*\Big)\,.
\]
This means that elements of the algebraic BV complex\index{BV!complex} can be treated as graded tensor powers of derivations of the Chevalley-Eilenberg algebra. To introduce some topology into the structure, we restrict ourselves to those derivations which are smooth compactly supported multilocal maps. From now on we will use the following definition of the underlying topological algebra of the BV complex: 
\be\label{BV}
\BV(M)=\Ci_\ml(\E(M),\A(M))\,,
\ee
where we denoted $\A(M):=\Lambda\E_c(M)\widehat{\otimes}S^\bullet \frakg_c(M)\widehat{\otimes}\Lambda {\frakg}'(M)$\footnote{We choose the completion $\widehat{\otimes}$ of the tensor product to be the space of compactly supported distributions $\prod\limits_{\mathclap{\sst{\stackrel{n=0}{n=k+l+m}}}}^\infty \Gamma'(M^n,\Lambda^k V\otimes S^l g\otimes\Lambda^m g)$, smooth in the first $k+l$ arguments.}. The topology on $\BV(M)$ is again the topology of pointwise convergence.
 The space $\BV(M)$ contains in particular $\V(M)$ and $\fG(M)=\Ci_\ml(\E(M),\frakg_c(M))$.
A pair $(X,\xi)\in\V(M)\oplus\fG(M)$ acts on $\CE(M)$ in the following way:
\be \label{interior}
(\partial_{(X,\xi)}F)(\ph):=(\partial_XF)(\ph)+i_{\xi(\ph)} F(\ph)\,,
\ee
where $i_{\xi(\ph)}$ is the interior product, i.e. the insertion of $\xi(\ph)\in\frakg_c(M)$ into
$\La\frakg'(M)$. The action of a general derivation of the form (\ref{BV}) can be now defined by imposing the graded distributive rule and the Leibniz rule. 

The graded commutator $[.,.]$ on $\Der(\CE(M))$ and the evaluation of a derivation on an element of $\CE(M)$ can be extended to the Schouten bracket $\{.,.\}$ on  $\BV(M)$. Like in the scalar case, this structure is called \textit{the antibracket}\index{antibracket}.
\begin{rem}
 Analogously to section \ref{aa} we can formaly write the antibracket in the form (\ref{antibracketformal}). The only difference is that the field configurations can be of fermionic type, so we get additional sign rules:
\be\label{antibracketformal2}
\{X,Y\}=-\!\sum\limits_\alpha(-1)^{(1+|X|)|\ph_\al|}\!\int\! dx\!\left(\!\frac{\delta X}{\delta\ph_\al(x)}\frac{\delta Y}{\delta\ph_\al^\ddagger(x)}+(-1)^{|X|}\frac{\delta X}{\delta\ph_\al^\ddagger(x)}\frac{\delta Y}{\delta\ph_\al(x)}\!\right)\,,
\ee
where $\alpha$ runs through different field types (ghosts, physical fields, etc.), $|.|$ denotes the ghost grading $\#\gh$ and we used the fact that $\#\gh(\ph_\al^\ddagger)=-\#\gh(\ph_\al)-1$.
\end{rem}
Using the fact that we restricted ourselves to derivations with compact support, one can show that $\BV(M)$ is a covariant functor from $\Loc$ to $\Vect$. The Chevalley-Eilenberg derivation $\gamma$ itself is not an element of $\BV(M)$, since it is not compactly supported. Compare this with a similar situation encountered in Section \ref{K}, where we showed that $\delta_S$ is not an inner derivation with respect to the antibracket. Nevertheless, locally it can be written with the use of inner derivations. First note that $\gamma$ decomposes into the sum $\gamma=\gamma^{(0)}+\gamma^{(1)}$, where:
\begin{align}
(\gamma^{(0)}F)(\xi)&\doteq\partial_{\rho_M(\xi)}F,\quad\quad F\in\F(M),\ \xi\in\frakg(M)\,,\label{gamm}\\
(\gamma^{(1)}\omega)(\xi_1,\xi_2)&\doteq-\omega([\xi_1,\xi_2]),\quad\ \omega\in \frakg'(M),\ \xi_1,\xi_2\in\frakg(M)\nonumber
\end{align}
and $\gamma^{(0)}$, $\gamma^{(1)}$ are extended to the whole space $\CE(M)$ by applying the graded Leibniz rule. Although they are not inner with respect to $\{.,.\}$ we can consider a following family of mappings $\theta_M$ from $\D(M)$ to $\BV(M)$
\begin{align*}
(\theta_M^{(0)}(f)F)(\xi)&\doteq\partial_{f\rho_M(\xi)}F,\quad\quad F\in\F(M),\ \xi\in\frakg(M)\,,\\
(\theta_M^{(1)}(f)\omega)(\xi_1,\xi_2)&\doteq\omega(f[\xi_1,\xi_2]),\quad\ \omega\in\frakg'(M),\ \xi_1,\xi_2\in\frakg(M)\,,
\end{align*}
where $f\in\D(M)$ is a test function. It follows now that $\{\omega,\theta_M(f)\}=\gamma(\omega)$ if $\supp(\omega)\subset f^{-1}(1)$, $\omega\in\CE(M)$. The family of maps $\theta$ constructed in this way can be identified with a natural transformation between the functors $\D$ and $\BV$. One sees immediately the analogy with the generalized Lagrangians. 
Finally we obtain an equivalent definition of the BV-differential\index{BV!operator!classical}\index{BV!differential!classical}:
\be\label{s}
sF=\{F,L_M(f)+\theta_M(f)\}\,,
\ee
where $f\equiv1$ on $\supp\, F$, $F\in \BV(M)$. This differential can be expanded in the antifield number and it contains two terms: $s=s^{({\minus})}+s^{(0)}$. 
We have the following structure of a bicomplex\footnote{We omit the dependence on $M$, since all the maps are natural and can be written on the level of functors.} (compare with (\ref{BVshort}):
\[
\begin{CD}
\ldots@>s^{({\minus})}>>\La^2\V\oplus\fG @>s^{({\minus})}>>\V@>s^{({\minus})}>>\F@>s^{({\minus})}>>0\\ 
@.     @VV{s^{(0)}}V@VV{s^{(0)}}V@VV{s^{(0)}}V@.\\
\ldots@>s^{({\minus})}>>{\Ci_\ml\big(\E,(\La^2\E_c\oplus\frakg_c)\widehat{\otimes}\frakg'}\big)@>s^{({\minus})}>>{\Ci_\ml\big(\E,\E_c\widehat{\otimes}\frakg'\big)}@>s^{({\minus})}>>{\Ci_\ml\big(\E,\frakg'\big)}@>s^{(-1)}>>0
\end{CD}
\]
The first row of this bicomplex corresponds to the resolution of $\F_S(M)$. This can be easily seen, since  $s^{(-1)}$ on $\fG(M)$ is just $\rho_M$ and $\im(\rho_M)_{\fG(M)\rightarrow\V(M)}=\ke(\delta_S)_{\V(M)\rightarrow\F(M)}$. Moreover  $\F_0(M)=\im(\delta_S)_{\V(M)\rightarrow\F(M)}$.
Explicitly the first row of the bicomplex can be written as:
 \[
\ldots\rightarrow\La^2\V\oplus\fG\xrightarrow{\delta_S\oplus\rho}\V\xrightarrow{\delta_S}\F\rightarrow 0
\]
The 0-order homology of this complex is just $\F_S(M)$ and the higher homology groups are trivial. We can therefore recognize two terms in the decomposition of $s$ as:
\begin{enumerate}
\item $s^{({\minus})}$ is the Koszul-Tate differential\index{Koszul-Tate!differential} providing the resolution of $\CE_S(M)$, it contains the information on the off-shell terms.
\item $s^{(0)}$ on $\CE(M)$ is the off-shell
Chevalley-Eilenberg differential.
\end{enumerate}
Since all the steps were done in the covariant way, $\BV$ can be made into a functor from $\Loc$ to $\dgA$.   
The main theorem of homological perturbation theory says that the 0-cohomology of $s$ is given by:
\[
H^0(\BV(M),s)=H^0(\CE_S(M),s^{(0)})=\F^\inv_S(M)
\]
\subsection{Classical master equation}\label{CME}
A very important aspect of the BV-construction is the so called \textit{classical master equation}\index{master equation!classical}. It is extensively discussed for example in \cite{Henneaux:1992ig,BarHenn}, where it is used as a guiding principle to derive the right form of the extended action. The underlying idea is the observation that the BV operator can be written formally as an inner derivation with respect to the antibracket. 
This has to be understood differently in our setting, since we are working on non-compact manifold and typical configurations are not compactly supported. Therefore, as we pointed out before, the BV operator can only locally be written in terms of the antibracket. Nevertheless, the general reasoning is similar. We start with the action $S$ that induces locally the Koszul map\index{Koszul!map} $\delta_S$ through its Euler Lagrange derivative $S'_M$. Next, we include in our construction the Chevalley-Eilenberg map that can also be written locally as the antibracket with a derivation $\theta_M(.)$ (formula (\ref{s})) and we can interpret this as modifying the original Lagrangian $L$ by an additional term $\theta$. This way of looking at the BV differential turns out to be very fruitful, since it allows to obtain the gauge fixed action at the end of the construction (we discuss this in detail in section \ref{gaugefixing}). It also allows to systematically analyze possible deformations of the action and allowed couplings \cite{BarHenn}. In the standard treatment of BV formalism, the classical master equation (CME) is just the condition that the antibracket of the action with itself is equal to 0. This is in turn equivalent to the nilpotency of the BV operator. In the present setting this has to be understood in a different way since Lagrangians are no more fixed functionals but rather natural transformations. It was already shown in \cite{FR} that this can be done in a consistent way. We give a recap of that argument in the present section. 

First, recall that the generalized Lagrangians are natural transformations between the functors $\D$ and $\F_\loc$. In the BV construction we extended the space of functionals $\F(M)$ to the BV complex $\BV(M)$ so the space of Lagrangians should also be extended. Let $\BV_\loc(M)$ denote the linear subspace of  $\BV(M)$ consisting of local functions. We can extend the notion of a Lagrangian to a natural transformation between the functors $\D$ and $\BV_\loc$. Let $\Nat( \D,\BV_\loc )$ denote the set of natural transformations. To obtain a structure that will be closed under the ``pointwise'' product we make one more generalization. Let $\D^k$ be a functor from the category $\Loc$ to the product category ${\Vect}^k$, that assigns to a manifold $M$ a $k$-fold product of the test section spaces $\D(M)\times\ldots\times \D(M)$. Let $\Nat(\D^k,\BV_\loc)$ denote the set of natural transformations from $\D^k$ to $\BV_\loc$. We define extended Lagrangians $L\in Lgr$ to be elements of the space $\bigoplus_{k=0}^\infty \Nat(\D^k,\BV_\loc)$ satisfying: $\supp(L_M(f_1,...,f_n))\subseteq \supp f_1\cup...\cup\supp f_n$ and the additivity rule in each argument. We can introduce on $Lgr$ an equivalence relation similar to (\ref{equ}). We say that $L_1\sim L_2$, $L_1,L_2\in\Nat(\D^k,\BV_\loc)$ if:
\be\label{equ2}
\supp((L_1-L_2)_M(f_1,...,f_k))\subset \supp(df_1)\cup...\cup\supp(df_k),\qquad\forall f_1,...,f_k\in\D^k(M)
\ee
The natural transformation $L^{\ex}:=L+\theta$ is an example of a generalized Lagrangian in $Lgr$. We call the corresponding equivalence class $S^{\ex}$ the \textit{extended action}. As noted before $sF=\{F,L_M^\ex(f)\}$, for $f\equiv 1$ on the support of $F$, $F\in \BV(M)$.  To make a contact with the standard approach we write  $L^{\ex}$ formally as:
\begin{multline*}
L^{\ex}_M(f)\form-\frac{1}{2}\int\limits_M f\tr(F\wedge *F)+\\
+\int\limits_M\dvol f\big(dC+\frac{1}{2}[A,C]\big)^I_\mu(x)\frac{\delta}{\delta A^I_\mu(x)}+\frac{1}{2}\int\limits_M\dvol f [C,C]^I(x)\frac{\delta}{\delta C^I(x)}\,.
\end{multline*}
This is the standard extension of the Yang-Mills action in the BV-formalism. The antibracket can be lifted to a bracket on $Lgr$ by the definition:
\be\label{ntbracket}
\{L_1,L_2\}_M(f_1,...,f_{p+q})=\frac{1}{p!q!}\sum\limits_{\pi\in P_{p+q}}\{{L_1}_M(f_{\pi(1)},...,f_{\pi(p)}),{L_2}_M(f_{\pi(p+1)},...,f_{\pi(p+q)})\}\,,
\ee
where $P_{p+q}$ denotes the permutation group.
%
The classical master equation\index{master equation!classical} extended to the natural transformations (ECME) can now be formulated as:
\be\label{CME0}
\{L^{\ex},L^{\ex}\}\sim 0\,,
\ee 
with the equivalence relation defined in (\ref{equ2}). It guarantees the nilpotency of $s$ defined by $sF=\{F,L^{\ex}_M(f)\}$, where $f\equiv1$ on $\supp F$, $F\in \BV(M)$.
\subsection{Gauge fixing}\label{gaugefixing}
The structure of the BV complex presented up to this point was done on a very general level, but still the main question remains not answered: how to implement the gauge fixing\index{gauge!fixing}? This is indeed the most important issue in BV formalism, since from the physical point of view, this is exactly the reason to invoke it.
In the present section we summarize the present understanding of the gauge fixing procedure in the BV setting. Those issues are not so often addressed in mathematical literature, since the gauge fixing procedure involves some amount of arbitrariness. In physics literature one should mention the book by M.~Henneaux and C.~Teitelboim \cite{Henneaux:1992ig} that covers also this issue, among many others. There is also a paper of G.~Barnich, M.~Henneaux, T.~Hurth and K.~Skenderis \cite{Barnich:1999cy} that adreses the problem of comparing the gauge-fixed and non-gauge-fixed cohomology of the BV complex.

We presented a more algebraic interpretation of the gauge fixing procedure that fits into the present formalism in \cite{FR}. Here we shall follow this approach. Recall that in the functional approach to classical field theory the dynamical structure is encoded in the Peierls bracket\index{Peierls bracket} \cite{Pei,Mar}. We recalled its construction in section \ref{scal}. Note that in order to construct the causal propagator we need first a normally hyperbolic system of equations of motion. This is of course not the case when the action has symmetries (see criterion \ref{sym0}), so to introduce the Peierls bracket we are forced to fix the gauge.
This can be done systematically with the help of the Batalin-Vilkovisky complex. 
 In the BV framework gauge fixing means eliminating the antifields by setting them equal to some functions of fields \cite{Batalin:1981jr,Henneaux:1989jq,Henneaux:1992ig,Froe}.
 
The gauge fixing is usually done in two steps. Firstly, one performs a transformation of  $\BV(M)$, that leaves the antibracket $\{.,.\}$ invariant\footnote{In literature \cite{Henneaux:1992ig} this is dubbed \textit{a canonical transformation} by analogy to the Hamiltionan formalism.}.
This way we obtain a new extended action $\tilde{S}^{\ex}$ and a new differential $\tilde{s}$. The cohomology classes of the new and old BV differential are isomorphic $H^0(\BV(M),s)\cong H^0(\BV(M),\tilde{s})$. In the second step we want to set the antifields to 0. This is done in a systematic way by introducing a new grading on $\BV(M)$, the so called \textit{total antifield number} $\#\ta$. It has value $0$ on fields and value $1$ on all antifields. We expand the differential $\tilde{s}$ with respect to this new grading: $\tilde{s}=\delta^g+\gamma^g+\ldots$ (for Yang-Mills theories this expansion has only two terms). From the nilpotency of $\tilde{s}$ it follows that $\delta^g$ is a differential and $\gamma^g$ is a differential modulo $\delta^g$. Moreover $\delta^g$ can be interpreted as the Koszul map corresponding to the so called ``gauge-fixed action'' $S^g$. We have to choose the canonical transformation of $\BV(M)$ in such a way that this extended action doesn't have nontrivial symmetries. In this case the Koszul map  $\delta^g$ provides a resolution and using the main theorem of homological perturbation theory \cite{Henneaux:1992ig,Barnich:1999cy} one can conclude that:
\be\label{gfixh}
H^0(\BV(M),\tilde{s})=H^0(H_0(\BV(M),\delta^g),\gamma^g)\,.
\ee
The r.h.s of (\ref{gfix}) is called the \textit{gauge-fixed cohomology}\index{gauge!-fixed cohomology}. It was discussed in details in \cite{Barnich:1999cy}. Now we shall describe this construction in details for the example of the Yang-Mills theory. We shall follow closely the presentation given in \cite{FR}.
\subsubsection{Nonminimal sector}
To implement the usual gauge fixing (for example the Lorenz one) we need first to introduce Lagrange multipliers. In the spirit of classical Lagrangian field theory, these are auxiliary, non-physical fields, that have to be eliminated at the end by performing a quotient of the field algebra. In the homological framework we have to introduce them in a way that would not change the cohomology classes of $s$.
In other words we replace the original BV complex by  $\Ci_\ml(\E(M),\Nm(M)\widehat{\otimes}\A(M))$, where $\Nm(M)$ is a certain graded algebra called \textit{\textbf{the nonminimal sector }}\index{nonminimal sector} \cite{Henneaux:1992ig}. The natural way to do it is to extend $\BV(M)$ by contractible pairs. 
\begin{defn}
Two elements $a$, $b$ of the cochain complex with a differential $d$ form a contractible pair if $a=db$ and $a\neq0$. 
\end{defn}
Let $a$ be of degree $n$ and $b$ of degree $n-1$. Since $H^n(d)=\ke(d_n)/\im(d_{n-1})$, $a$ and $b$ are mapped to the trivial elements  $[0]$ of the cohomology classes. 
This observation provides us with a method to add Lagrange multipliers to the BV-complex. For concreteness we take the Lorenz gauge $G(A)=\hinv d*\!A$, where $*$ denotes the Hodge dual and $G$ is a map from $\E(M)$ to $\frakg(M)$. The Lagrange multipliers are therefore elements of the dual $\frakg'(M)$. Hence we extend the BV-complex by tensoring with the space: $S^\bullet\frakg'(M)$ and its elements have grade $\#\gh=0$. Together with this space we introduce the space $\La\frakg'(M)=\bigoplus\limits_{k=0}^\infty\La^k\frakg'(M)$. These are the so called \textit{antighosts}\index{antighosts} and have $\#\gh=-k$. They form trivial pairs with Nakanishi-Lautrup fields\index{Nakanishi-Lautrup fields} if we define
\begin{align}
sF&=0,\qquad\qquad\quad\, F\in S^1\frakg'(M)\,,\label{nonm}\\
sG&=\Pi G\circ m_i,\qquad G\in\La^1\frakg'(M)\,,\nonumber
\end{align}
where $\Pi$ denotes the grade shift by $+1$ and $m_i$ the multiplication of the argument by $i$, i.e.
 $G\circ m_i(B)\doteq G(iB)$, for $B\in\frakg(M)$. The appearance of the imaginary unit $i$ is just a convention used in physics to make antighosts hermitian. Together with antighosts and Nakanishi-Lautrup fields we can introduce the corresponding antifields. Just like in the minimal sector those are the derivations of  $\Ci_\ml(\E(M),S^\bullet\frakg'(M)\widehat{\otimes}\La\frakg'(M))$, that can be written as smooth compactly supported multilocal maps on the configuration space. In other words the full nonminimal sector is of the form:
\[
\Nm(M)=\La \frakg' [-1]\widehat{\otimes}S^\bullet\frakg'(M)[0]\widehat{\otimes}S^\bullet{\frakgc}(M)[0]\widehat{\otimes}\La {\frakgc}[-1]\,,
\]
where we indicated the grades explicitly in brackets. The new BV-complex consists of compactly supported multilocal maps $\Ci_\ml(\E(M),\Nm(M)\widehat{\otimes}\A(M))$ with the BV-differential $s$ defined above.

Again it is worth to write $s$ as a locally inner derivation with respect to the antibracket. 
On the nonminimal sector $\{.,.\}$ is again just the Schouten bracket. We require $sF=\{F,L^\ex(f)\}$, where $f\equiv 1$ on $\supp F$, $F\in \BV(M)$. From equation (\ref{nonm}) we can read off that $L^\ex=L+\theta+L^{\textrm{nm}}$, where $L^{\textrm{nm}}$ is a natural transformation between functors $\D$ and $\BV$. Concretely $L^{\textrm{nm}}_M(f)$ is a graded derivation of the graded $\F(M)$-module $\Ci_\ml(\E(M),\frakg'(M)[0]\oplus\frakg'(M)[-1])$ defined as:
 \[
\partial_{L^{\textrm{nm}}_M(f)}(F\oplus G)(\ph)=G(\ph)\circ m_{if}\oplus 0\,,
 \]
 where $m_{if}$ denotes multiplication of the argument with a function $if$. To make contact with the standard approach we discuss now how the formal expressions appearing in the literature can be interpreted in our formulation. Let us denote by $B^I(x)$ the evaluation functional on the space of Lagrange multipliers ${\frakg}(M)$, i.e. $B^I(x)\in S^1\frakg'(M)$. In this sense elements of $S^\bullet\frakg'(M)$ can be seen formally as integrals 
 \[
 F\form \int dx_1...dx_n f_{a_1...a_n}(x_1,...,x_n) B^{a_1}(x_1)\otimes...\otimes B^{a_n}(x_n)\,.
 \]
  Similarly we can write $G\in \La\frakg'(M)$ as 
 \[
 G\form\int dx_1...dx_n g_{a_1...a_n}(x_1,...,x_n) \bar{C}^{a_1}(x_1)\wedge...\wedge  \bar{C}^{a_n}(x_n)\,,
 \] 
 with the evaluation functionals $\bar{C}^{a_n}(x_n)$. Moreover we identify $\frac{\delta}{\delta \bar{C}^I(x)}$ with the derivation that acts on $\La\frakg'(M)$ as the left derivative (see \cite{Rej} for the detailed discussion). In this notation we can write the extended Lagrangian $L^\ex=L+\theta+L^{\textrm{nm}}$ as:
\begin{multline}\label{Snm}
L^{\ex}_M(f)\form-\frac{1}{2}\int\limits_M f\tr(F\wedge *F)+\int\limits_M\dvol\, f\big(dC+\frac{1}{2}[A,C]\big)^I_\mu(x)\frac{\delta}{\delta A^I_\mu(x)}+\\
+\frac{1}{2}\int\limits_M\dvol\, f [C,C]^I(x)\frac{\delta}{\delta C^I(x)}-i\int\limits_M\dvol\, f B_I(x)\frac{\delta}{\delta \bar{C}_I(x)}\,.
\end{multline}
The last term corresponds to $L_M^{\textrm{nm}}(f)$.
Moreover the action of $s$ on the non-minimal  sector can be formally written as: $sB^I(x)=0$, $s\bar{C}^I(x)=iB^I(x)$. The expression (\ref{Snm}) is what is usually meant under the name \textit{extended action}.
 \subsubsection{Gauge fixing for the Yang-Mills theory}\label{GFYM}
 With the structure extended by the nonminimal sector we can now turn back to the gauge-fixing\index{gauge!fixing!in Yang-Mills theories}. 
Let $\psi\in\BV(M)$ be a fixed algebra element of degree $\#\gh=-1$ and  $\#\af=0$. 
Using $\psi$ we define a linear transformation $\alpha_\psi$ on $\BV(M)$ 
by
\be\label{gfermion}
\alpha_\psi(X):=\sum_{n=0}^{\infty}\frac{1}{n!}\underbrace{\{\psi,\dots,\{\psi}_n,X\}\dots\}\,,
\ee
The antibracket with $\psi$ preserves the ghost number $\#\gh$ and lowers the antifield number $\#\af$ by 1. Hence the sum in \eqref{gfermion} is finite and $\alpha_\psi$ preserves the grading with respect to the ghost number. Moreover, it preserves as well the product and the antibracket itself. 
Let now $\Psi$ be a natural transformation from $\D$ to $\BV$ such that $\Psi_M(f)$ satisfies the conditions which were stated on $\psi$ above. $\Psi$ is called \textit{the gauge fixing fermion}\index{gauge!fixing!fermion}. 
We define an automorphism on $\BV(M)$ by
\[
\alpha_{\Psi}(X)\doteq\alpha_{\Psi_M(f)}(X),\qquad X\in\BV(M)\,
\]
where $f\equiv1$ on the support of $X$. Let $\tilde{L}^{\ex}=\alpha_\Psi\circ L^\ex$ be the transformed generalized Lagrangian. 
We define a new BV-operator as $\tilde{s}F:=\{F,\tilde{L}^\ex_M(f)\}$, for $f\equiv 1$ on $\supp F$, $F\in\BV(M)$ . We have:
\[
\F_S^{\inv}(M)\cong H^0(\BV(M),\tilde{s})\,,
\]
where the isomorphism is given by means of $\alpha_\Psi$. 

For the Lorenz gauge we choose the gauge-fixing fermion of the form:
\be\label{Lorenz}
\Psi_M(f)=i\int\limits_M f\left(\frac{\alpha}{2}\kappa(\bar{C},B)+\left<\bar{C},*d*\!A\right>_g\right)\dvol\,,
\ee
where $\kappa$ is the Killing form on the Lie algebra $g$ and $\left<.,.\right>_g$ is the dual pairing between $g$ and $g'$. It is interesting to compare this formulation with the standard definition  of the gauge-fixing fermion \cite{Henneaux:1992ig}. In the standard approach it is required to be a local functional. This fits very well in the present formalism, since the condition of locality is basically replaced with the condition to be a natural transformation. Explicitly the transformed Lagrangian takes the form:
\[
\alpha_{\Psi}(L^\ex_M(f))=L_M(f)+\theta_M(f)+\{\Psi_M(f_1),\theta_M(f)+L_M^{\textrm{nm}}(f)\}\,,
\]
where $f_1\equiv 1$ on the support of $f$. Formally this can be written with the more familiar expression:
\begin{multline}\label{extendYM}
L^{\ex}_M(f)\form-\frac{1}{2}\int_M f\tr(F\wedge *F)-i\int_M f\tr\left(*DC\wedge d\bar{C}\right)+\int_M f\tr\Big(*DC\wedge\frac{\delta}{\delta A}\Big)+\\
+\frac{1}{2}\int_M\dvol\, f [C,C]\frac{\delta}{\delta C}-i\int_M\dvol\, f B\frac{\delta}{\delta \bar{C}}+\int_M\dvol\, f\left(\frac{\alpha}{2}\kappa(B,B)+\left<B,\hinv d*\!A\right>_g\right)\,,
\end{multline}
where we suppressed all the form and Lie-algebra indices. 

In the second step of the gauge fixing procedure we expand the differential $\tilde{s}$ with respect to the total antifield number: $\tilde{s}=\delta^g+\gamma^g$, where $\delta^g$ lowers $\#\ta$ by 1 and $\gamma^g$ preserves it. Let $X\in\BV(M)$ be a derivation of total antifield number $\#\ta=1$. The action of $\delta^g$ on $X$ is given by:
\be
\delta^gX=\{X,\tilde{L}^\ex_M(f)\}\Big|_{\#\ta=0\atop \textrm{terms}}=\{X,L^g_M(f)\},\qquad f\equiv 1\textrm{ on }\supp X\,,\label{deltag}
\ee
where $L^g$ is the so called gauge-fixed Lagrangian\index{gauge!-fixed!Lagrangian} and is obtained from $\tilde{L}$ by putting all antifields to 0. The corresponding equivalence class $S^g$ is the gauge-fixed action\index{gauge!-fixed!action}. The ideal of $\BV(M)$ generated by all terms of the form (\ref{deltag}) is the graded counterpart of the ideal of $\F(M)$ generated by the equations of motion. In the next section we shall see that one can introduce a notion of a derivative on $\BV(M)$ which makes this correspondence precise. In this sense the $0$-order homology of $\delta^g$ is the algebra of on-shell functions for the gauge-fixed action $S^g$. For Yang-Mills theory the gauge-fixed Lagrangian reads:
 \[
L_M^g(f)=S_M(f)+\gamma^g\Psi_M(f)\,.
\]
In case of the Lorenz gauge we obtain:
 \be
L_M^g(f)\form-\frac{1}{2}\int\limits_M f\tr(F\wedge *F)-i\int\limits_M f\tr(d\bar{C}\wedge*D C)-\int\limits_M\dvol\, f\left(\frac{\alpha}{2}\kappa(B,B)+\left<B,\hinv d*\!A\right>_g\right)\,.\label{fixed}
\ee

The differential $\gamma^g$ is called \textit{the gauge-fixed BRST differential}\index{BRST!gauge-fixed differential}. 
The action of the gauge-fixed BRST differential on the functions in $\BV(M)$ is summarized in the table below.\\
\begin{center}
{\setlength{\extrarowheight}{2.5pt}
\begin{tabular}{ll}
\toprule%
& $\gamma^g$\\\otoprule%
$F\in\F(M)$&$\left<F^{(1)},dC+[.,C]\right>$\\
 $C$&$-\frac{1}{2}[C,C]$\\
 $B$& $0$\\
  $\bar{C}$& $iB$\\\bottomrule
\end{tabular}}
\end{center}
\subsection{Peierls bracket}\label{Peierls}
We come finally to the discussion of the dynamical structure. As already mentioned in Section \ref{gaugefixing}, we first need to fix the gauge, before we can define the Peierls bracket\index{Peierls bracket}, that implements the dynamics. As discussed in section \ref{scal}, the space of multilocal functionals is not closed under the Peierls bracket. To fix this we replace it by the space $\F_\mc(M)$ of \textit{microcausal functionals}, equipped with the topology $\tau_\Xi$.  Multilocal functionals are dense in $\F_\mc(M)$ with respect to this topology. Using this space as a starting point one can repeat the construction of the BV-complex given in section \ref{KTcom} with some technical changes. The extended BV graded algebra $\BV_\mc(M)$ is defined as the space of microcausal vector-valued functions. Now we want to extend the BV differential to those more singular objects.
Since the Lagrangian $L_M(f)$ is a local functional and its functional derivative is a smooth test section, the Koszul operator can be extended from multilocal vector fields to $\V_\mc(M)=\Ci_\mc(\E(M),\E'(M))$ and the resolution of  $\F_\mc(M)$ in the scalar case is provided by the differential graded algebra $\La\V_\mc(M)$. Similar reasoning applies also to the case when symmetries are present. The extended Chevalley-Eilenberg complex $\CE_\mc(M)$ consists of microcausal vector-valued functions with $W_n=\Lambda^n g$ (see section \ref{vvf}). Since the map $\partial_{\rho_M(.)}F$ is an element of 
 $\CE_\mc(M)$, for all $F\in\F_\mc(M)$, we can conclude that:
 \[
 H^0(\CE_\mc(M),\gamma)=\F^\inv_\mc(M)\,.
 \]
The full BV complex is equipped with the topology $\tau_{\Xi}$ and since multilocal functionals lie dense in $\BV_\mc(M)$, the differential  $s$ can be extended to the full complex by continuity.
From the above discussion it follows that $H^0(\BV_\mc(M),s)$ is the space of microcausal gauge invariant functionals on-shell. We want to point out however, that the antibracket itself is not well defined on the whole $\BV_\mc(M)$. This is because the commutator of vector fields $\V(M)$ can be extended only to those elements of the space $\V_\mc(M)$ that have smooth first derivative. 

In Yang-Mills theories the elements of the extended BV complex are now microcausal vector-valued functions 
\be\label{BVYM}
\BV_\mc(M)=\Ci_\mc(\E(M),\A(M))\,, 
\ee
where $\A(M)$ is of the form:
\[
\A(M)=\prod\limits_{k,l,m=0}^\infty\Gamma'_{\Xi_n}(M^n,S^kg\otimes\Lambda^{l}g\otimes\Lambda^{m}g\otimes\textrm{Antifields})\,.
\]
The first three factors in the space $W_n$ correspond accordingly to the Lagrange multipliers, antighosts and ghosts. The subsequent factors are the antifields. The fact that gauge fixing is possible implies that if we keep all the unphysical fields with fixed values, then the initial value problem is well posed for the physical fields. This is however not enough. Since the differential $\delta^g$ is the Koszul map for the gauge fixed action, we should understand the equations of motion as equations for the full field multiplet with the auxiliary fields included. 

To make precise what we mean by the \textsc{eom}'s for odd variables we use similar definitions as in section \ref{fer}. The functional derivative of a function $F\in\Ci_\mc(\E(M),\A(M))$ at the point $A_0\in\E(M)$ is an $\A(M)$-valued distribution: $F^{(1)}_A(A_0):=\frac{\delta F}{\delta A}(A_0)\in \E'(M)\widehat{\otimes}\A(M)$\footnote{The completed tensor product is to be understood as the space of distributions $\prod\limits_{k,l,m=0}^\infty\Gamma'_{\Xi_n}(M^n,g\otimes S^kg\otimes\Lambda^{l}g\otimes\Lambda^{m}g\otimes\textrm{Antifields})$ equipped with the H\"ormander topology.}. The derivatives with respect to the odd variables are defined pointwise.  For example for the functions of ghost fields we have: $F^{(1)}_C(A_0):= \big(F(A_0)\big)^{(1)}_C$, $F(A_0)\in\A(M)$, where the derivative  on the graded algebra $\A(M)$ is defined as in section \ref{fer}:
\be\label{d1}
F^{(1)}(a)[h]:=F(h\wedge a)\quad F\in \La^p{\frakg}'(M),\ a\in\La^{p-1}\frakg(M),\ h\in\frakg(M)\ p>0
\ee
Note that $F^{(1)}_C(A_0)\in{\frakg}'(M)\widehat{\otimes}\A(M)$. Now to implement the equations of motion we take the quotient of $\BV(M)$ by the ideal generated by graded functions of the form:
\be\label{eomf}
A_0\mapsto\langle {S^g}_\alpha'(A_0), \beta(A_0)\rangle\,,
\ee
where $A_0\in\E(M)$, $\alpha$ is  $A,B,C$ or $\bar{C}$ and $\beta(A_0)$ is the appropriate test section. Note that in the graded case, when $S^g$ is of degree higher than $1$ in anticommuting variables, we don't have an interpretation of the equations of motion as equations on the configuration space. Instead, the algebraic definition on the level of functionals can still be applied. We can compare this situation to the purely bosonic case, when we had to show that the ideal $\F_0(M)$ is generated by the equations of motion for a given action functional. In the fermionic case, we reason differently and \textit{define} this ideal as generated by elements of the form (\ref{eomf}). See section  \ref{fer} for details. Equivalently we can say that $\F_0(M)$ is the image of the map $\delta^g$ acting on  derivations with $\#\ta=1$. Hence the on-shell functionals for the action $S^g$ are characterized by $H_0(\A(M),\delta^g)$.

We conclude that after the gauge fixing the full dynamics is described by the action $S^g$ and therefore this generalized Lagrangian is the starting point for the construction of the Peierls bracket\index{Peierls bracket}. The off-shell formalism \cite{DF04,DF02} is to be understood with respect to $S^g$ and going on-shell means taking the quotient by the ideal generated by the equations of motion. The construction of the Peierls bracket is a straightforward generalization of the construction done in the scalar case \cite{BDF,BFR}. The only subtle point is the grading. We discussed the general case of Peierls bracket for anticommuting fields in sections \ref{vvalued} and \ref{fer} (see also \cite{Rej}). All the distributional operations have to be generalized to distributions with values in a graded algebra. The distributional operations like convolution, contraction and pointwise product generalized to the $\A(M)$-valued distributions are now graded commutative. As in the case of $\RR$-valued distributions, the pointwise product is well defined only when the sum of the wave front sets of the arguments does not intersect the zero section of the cotangent bundle. We point out that the use of $\A(M)$-valued distributions already accounts for the grading, so there is no need to introduce additional Grassman algebras by hand.

Since $S^g$ has at most quadratic terms with respect to the anticommuting variables, its second derivative can be again treated as an operator on the extended configuration space \cite{Rej}. To construct the Peierls bracket, we need this operator to be normally hyperbolic. Therefore we need to find a gauge-fixing fermion which makes the linearized equations of motion of $S^{g}$ into a normally hyperbolic system in variables $A,B,C$ and $\bar{C}$. The existence of such a fermion in a general case is an open question. In case of Yang-Mills theory, it suffices to take the Lorenz gauge with $\Psi$ given by (\ref{Lorenz}). Taking the first functional derivative of $S^g$ results in a following system of equations\footnote{These equations should be understood as relations in the algebra $\A(M)$, that we have to quotient out. For example (\ref{1YM}) means that we quotient out the ideal generated (in the algebraic and topological sense) by evaluation functionals $(\hinv D\!*\!F-dB-i[d\bar{C}, C])^I_\mu(x)$.\label{eqs}}:
\begin{align}
\hinv D\!*\!F=\hinv D\!*\!DA&=-dB-i[d\bar{C}, C]\label{1YM}\,,\\
\hinv d\!*\!A+\alpha B&=0\,,\label{gaugecond}\\
\hinv d\!*\!DC&=0\,,\nonumber\\
\hinv D\!*\!d\bar{C}&=0\,,\nonumber
\end{align}
where $D\omega=D+[A,\omega]$ denotes the covariant derivative. Acting with $\hinv D*$ on equation (\ref{1YM}) we obtain an evolution equation for $B$:
\begin{equation*}
\hinv D*dB=-i*[d\bar{C}, *DC]\,.
\end{equation*}
For every field from configuration space the second variational derivative of (\ref{fixed}) is an  integral kernel of a normal hyperbolic differential operator. Indeed, in the linearized system of equations the only terms containing second derivatives in (\ref{1YM}) are of the form $\hinv d*dA=\Box A-d\hinv d*A$. From the gauge fixing condition (\ref{gaugecond}) it follows that $\hinv d*A=-\alpha B$ and therefore the only contributions containing the second derivatives are of the form $\Box \phi^\alpha$, where $\phi^\alpha=A,C,\bar{C}$ or $B$. This means that ${S^g}''_M$ provides a hyperbolic system of equations and one can construct the advanced and retarded Green's functions $\Delta^R_{S^g}$, $\Delta^A_{S^g}$. 
We define the Peierls bracket\index{Peierls bracket} by:
\begin{align}
\{F,G\}_{S^g}&\doteq R_{S^g} (F,G)-A_{S^g}(F,G)\,,\label{peierls}\\
R_{S^g} (F,G)&\doteq\sum\limits_{\alpha,\beta}(-1)^{(|F|+1)|\phi_\alpha|}\left<F^{(1)}_\alpha,(\Delta^R_{S^g})_{\alpha\beta}*G_\beta^{(1)}\right>\,,\label{ret}\\
A_{S^g} (F,G)&\doteq\sum\limits_{\alpha,\beta}(-1)^{(|F|+1)|\phi_\alpha|}\left<F^{(1)}_\alpha,(\Delta^A_{S^g})_{\alpha\beta}*G_\beta^{(1)}\right>\,,\label{adv}
\end{align}
where $\Delta^{R/A}_{S^g}$ has to be understood as a matrix, $\phi^\alpha=A,C,\bar{C}$ or $B$ and $|.|$ denotes the ghost number. The sign convention chosen here comes from the fact that we use only left derivatives. One can show that $\{.,.\}_{S^g}$ is a well defined graded Poisson bracket on $\A(M)$. Moreover the algebra  $\A(M)$ is closed under this bracket. The next proposition shows that there is a relation between this dynamical structure and the BRST symmetry.
\begin{prop}
The BRST operator $\gamma^g$ satisfies the graded Leibniz rule with respect to the Peierls bracket:
\begin{equation}
\gamma^g\{F,G\}_{S^g}=(-1)^{|G|}\{\gamma^gF,G\}_{S^g}+\{F,\gamma^gG\}_{S^g}\,.\label{BRSideal}
\end{equation}
\end{prop}
\begin{proof}
From the definition of the BRST operator we know that it is a graded derivation on the algebra $\BV(M)$, acting from the right. Therefore it holds:
\begin{multline*}
\gamma^g\left<F^{(1)}_\alpha,({\Delta^R_{S^g}}_{\alpha\beta})*G_\beta^{(1)}\right>=(-1)^{|\phi_\alpha|+|G|}\left<\gamma^g\left(F^{(1)}_\alpha\right),({\Delta^R_{S^g}}_{\alpha\beta})*G_\beta^{(1)}\right>+\\
+(-1)^{|\phi_\beta|+|G|}\left<F^{(1)}_\alpha,\left(\gamma^g({\Delta^R_{S^g}})_{\alpha\beta}\right)*G_\beta^{(1)}\right>+\left<F^{(1)}_\alpha,({\Delta^R_{S^g}})_{\alpha\beta}*\gamma^g\left(G_\beta^{(1)}\right)\right>\,.
\end{multline*}
Now, using the fact that $S^g$ is $\gamma^g$-invariant (follows from the nilpotency of $\tilde{s}$), we obtain:
\[
\gamma^gR_{S^g}(F,G)=(-1)^{|G|}R_{S^g}(\gamma^gF,G)+R_{S^g}(F,\gamma^gG)\,.
\]
The same holds for $A_{S^g}(F,G)$, so the result follows from the definition (\ref{peierls}).
\end{proof}
Now we want to show that the Peierls bracket\index{Peierls bracket} $\{.,.\}_{S^g}$ is well defined on the algebra of gauge invariant observables.
We recall that $\F^{\inv}_S(M)\cong H^0(H_0(\BV(M),\delta^g),\gamma^g)$ and $H_0(\delta^g,\BV(M))$ is the on-shell algebra of the gauge-fixed action $S^g$. For the scalar field the subalgebra of functionals that vanish on-shell is a Poisson ideal with respect to $\{.,.\}_S$. A similar reasoning can be applied also to the graded case and one shows that the image of $\delta^g$ in degree $\#\ta=0$ is a Poisson ideal with respect to
$\{.,.\}_{S^g}$. This means that the Peierls bracket is well defined on-shell, i.e. on the homology classes  $H_0(\delta^g,\BV(M))$. To see that it is also compatible with the differential $\gamma^g$ we consider $F,G\in \ke(\gamma^g)$ and from (\ref{BRSideal}) we conclude that
\begin{multline}\label{Poii}
\{F,G+\gamma^g H\}_{S^g}=\{F,G\}_{S^g}+\{F,\gamma^gH\}_{S^g}=\\
=\{F,G\}_{S^g}+\gamma^g\{F,H\}_{S^g}-(-1)^{|H|}\{\gamma^gF,H\}_{S^g}\\
=\{F,G\}_{S^g}+\gamma^g\{F,H\}_{S^g}\,.
\end{multline}
This shows that $\{.,.\}_{S^g}$ is compatible with the cohomology classes of $\gamma^g$. Using this result and the previous one, concerning the $0$-th homology of $\delta^g$, we conclude that the Peierls bracket is well defined on $\F^{\inv}_S(M)$. As a final remark, we note that for Yang-Mills theories the Poisson structure on $\F^{\inv}_S(M)$ defined by the gauge-fixed action doesn't depend on the gauge-fixing fermion $\Psi$. Indeed, let $S^{g}_1=S+\gamma^g\Psi_1$, whereas  $S^{g}_2=S+\gamma^g\Psi_2$. Therefore $S^{g}_2=S^{g}_1+\gamma^g(\Psi_1-\Psi_2)$. It follows now that for $F,G\in \ke(\gamma^g)$ we have:
\begin{equation*}
\{F,G\}_{S^{g}_2}=\{F,G\}_{S^{g}_1}+\gamma^g(\ldots)\,.
\end{equation*}
It means that $\{F,G\}_{S^g_1}$ and $\{F,G\}_{S^g_2}$ are in the same cohomology class.

To end this section we discuss the functoriality of the construction presented above. Let $\PgAlg$ denote the category of graded topological Poisson algebras with continuous faithful graded Poisson algebra morphisms as morphisms. It can be shown that the assignment of $(\BV(M),\{.,.\}_{S^{g}})$ to $M$ is a covariant functor from $\Loc$ to $\PgAlg$. 
\section{Classical gravity}\label{grav}
We finally come to the most interesting example, where the locally covariant construction of the BV complex provides some new insights. As we know, the diffeomorphism invariance of the general relativity is one of the conceptual problems standing on the way of quantization. We think that a better understanding of this feature in the classical theory will be a first step to get a grasp of it also in the quantum case. The fundamental question that arises here is the definition of diffeomorphism invariant observables. If we can characterize them in a systematic way in the classical theory, then we can take this structure as an input to some quantization procedure. In the present thesis, we propose such a definition and show how it fits into the general scheme of locally covariant field theory.

The treatment of quantum gravity in the framework of local covariance was proposed in \cite{F,BF1}. The first step towards this program is the proper understanding of structures that appear already on the classical level. 
We have to identify what corresponds to diffeomorphism invariant physical quantities in our framework. This a very subtle question.
We already anticipate that objects we need have to be global. Therefore, a single element of a local algebra of observables is not a good candidate for such an invariant quantity. Fortunately, the framework of locally covariant field theory provides a very natural solution to the problem of observables. It was proposed in \cite{F,BF1} to consider gauge invariant \textit{fields} as physical objects.
 Here the fields are understood as natural transformations \cite{BFV,Few} between the functors $\E_c$ and $\F$. Let $\Nat( \E_c,\F )$ denote the set of natural transformations.
In analogy to Section \ref{CME} we define the set of fields as $\bigoplus\limits_{k=0}^\infty \Nat(\E_c^k,\F)$. It was shown in \cite{FR} that this is the right structure to consider as a starting point for the BV construction. Indeed in general relativity one always uses objects that are natural, for example the scalar curvature. Although it doesn't make sense to consider it at a fixed spacetime point, it is still meaningful to treat it as an object defined in all spacetimes in a coherent way. This is the underlying idea of identifying the physical quantities with natural transformations. We discuss those issues in detail in section \ref{difftrans}.
\subsection{Geometrical preliminaries}
For the classical gravity the configuration space is $\E(M)=(T^*M)^{2\otimes}\doteq T^0_2M$, the space of rank $(0,2)$ tensors. The Einstein-Hilbert action reads\footnote{In this chapter we use the metric signature $(-+++)$ and the conventions for the Riemann tensor agreeing with \cite{Wald}.}:

\be\label{EH}\index{generalized Lagrangian!Einstein-Hilbert}
S_{(M,g)}(f)(h)\doteq \int R[\tilde{g}]f\,\textrm{d vol}_{(M,\tilde{g})}\,,
\ee
where $g$ is the background metric, $h$ the perturbation and $\tilde{g}=g+h$. For every $g$ the local functional $S_{(M,g)}(f)(h)$ is defined in some open neighborhood $U_g\subset\E(M)$. We can make this neighborhood small enough to guarantee that $\tilde{g}$ is a Lorentz metric with the signature $(-+++)$. Since we are interested only in the perturbation theory, we don't need $S_{(M,g)}(f)(h)$ to be defined on the full configuration space. The diffeomorphism invariance of (\ref{EH}) means that the symmetry group of the theory is the diffeomorphism group $\Diff(M)$. Since we are interested only in local symmetries, we can restrict our attention to $\Diff_c(M)$. It is an infinite dimensional Lie group modeled on $\X_c(M)$, the space of compactly supported vector fields on $M$ \cite{Michor80,Michor,Gloe06,Gloe07}. We can now define the action of $\Diff_c(M)$ on $\E(M)$ or more generally on arbitrary tensor fields. Let $\Tens(M)$ denote the space of smooth sections of the vector bundle $\bigoplus\limits_{k,l}T^k_lM$, where $T^k_lM\doteq \underbrace{TM\otimes\ldots\otimes TM}_{k}\otimes \underbrace{T^*M\otimes\ldots\otimes T^*M}_{l}$. We define a map $\rho_M:\Diff_c(M)\rightarrow L(\Tens(M),\Tens(M))$ as a pullback, namely:
\begin{equation}
\rho_M(\phi)=(\phi^{-1})^*t,\quad \phi\in \Diff_c(M),\ t\in\Tens(M)
\end{equation}
The corresponding derived representation of $\X_c(M)$ on $\Tens(M)$  is the Lie derivative:
\begin{equation}
\rho_M(X)t\doteq\frac{d}{dt}\Big|_{t=0}(\exp(-tX))^*t=\pounds_{X}t,\label{repre}
\end{equation}
where $X\in\X_c(M)$ and the last equality follows from the fact that the exponential mapping of the diffeomorphism group is given by the local flow. The most general nontrivial symmetries of the action (\ref{EH}) can be written as elements of $\Ci_\ml(\E(M),\X_c(M))$ so we can identify the space of reduced symmetries with $\sm_{\textrm{r}}(M)=\Ci_\ml(\E(M),\X_c(M))$. Like in gauge theories one can define the action of $\X_c(M)$ on $\F(M)$, the space of functionals on the configuration space. It is given by: 
\be\label{rhoM3}
\partial_{\rho_M(X)}F(h)=\left<F^{(1)}(h),\pounds_{X}\tilde{g}\right>,\qquad F\in\F(M), X\in\X_c(M)\,.
\ee
\subsection{Observables and diffeomorphism invariance}\label{difftrans}
To understand what the diffeomorphism invariance means in our framework we have to bring the discussion from the previous section to the level of natural transformations. The intuitive idea behind this is actually very simple. If we think about an experiment that locally probes  the geometric structure of spacetime, we can 
associate to our setup a causally convex spacetime region $\Ocal$ of spacetime $M$
and an observable $\Phi$ localised in $\Ocal$, which we measure. Since the experiment has a finite resolution, 
we don't really measure values of the geometric data at a point. There is always some smearing involved.
 For example, in case of the Ricci curvature $R$ we can model it by 
defining our observable quantity as $\Phi(f)=\int f(x)R(x)$, where $f$ is the smearing function with $\supp(f)\subset \Ocal$. In certain situations, 
we can think of the measured observable as a perturbation of the fixed background metric. This is for example the case if we want to observe gravitational waves. In other words, we can do 
 a tentative split into: $\tilde{g}_{\mu\nu}=g_{\mu\nu}+h_{\mu\nu}$. The situation is pictured on the figure \ref{exper}. To formulate what the diffeomorphism invariance means, we first have to answer the question: what happens if we move our experimental setup to a different region $\Ocal'$?
\begin{figure}[htb]
\begin{center}
 \begin{tikzpicture}
\shadedraw[shading=axis,shading angle=90, left color=see!80!white,right color=white] plot[smooth cycle] coordinates{(-1,-1) (-2,0)(-2.1,1)(-1,2)(1,2.3)(2,2.5)(3,2)(3,1)(2,0)(1,-1)(0,-1.2)};
\draw (0,0.5) node[diamond,draw,minimum size=2cm,name=s] {};
\draw (1.8,1.2) node[diamond,draw,minimum size=2cm] {$\Ocal'$};
\draw(-1,1.4) node {\small{$(M,g)$}};
\draw(0,-2) node [name=F]{$\Phi_{(\Ocal,g)}(f)(h)$};
\shade[shading=radial, inner color=darksee!80!white,outer color=see!80!white,scale=0.18] plot[smooth cycle] coordinates{(0,2)(-2,3)(-3.2,4)(-2.4,5.5)(1,6)(3,5)(3.9,3.2)(3,2.2)};
\draw(0,0.8) node [color=lighthoney] {$f$};
\draw(0,0) node {$\Ocal$};
node {};
\path[->] (s) edge  [color=red, bend right]  node[right] {} (F);
\end{tikzpicture}
\end{center}
\caption{Experimental situation while probing the spacetime geometry.\label{exper}}
\end{figure}
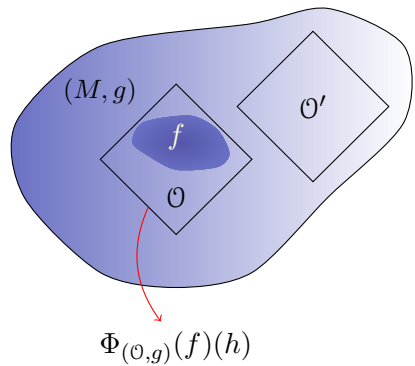
Intuitively thinking, we have to transform the metric and the test function by the pushforward (see figure \ref{mov}). Now to compare $\Phi_{(\Ocal,g)}(f)$ and $\Phi_{(\Ocal',\alpha_*g)}(\alpha_*f)$ we need to know what does it mean to have ``the same observable in a different region''. This is where the notion of fields comes in. It was proposed in \cite{BFV}, that fields should allow us to compare the results of experiments performed in different regions of a spacetime in the absence of symmetries. Moreover, they should be compatible with the net structure given by the isometric embeddings. In this sense, the family $\{\Phi_M\}_{M\in \obj(\Loc)}$ defines a \textit{field}\index{locally covariant!field!classical} if it satisfies certain naturality conditions. It can be made precise in the language of category theory. The idea to treat fields as natural transformations goes back to R. Brunetti and appeared first in \cite{BFV}. One defines fields as natural transformations between functors $\D$ and $\F$.
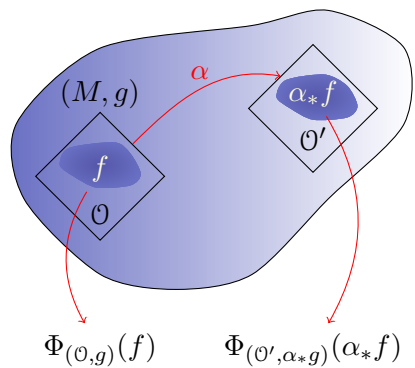
\begin{figure}[!b]
\begin{center}
\begin{tikzpicture}
\shadedraw[shading=axis,shading angle=90, left color=see!80!white,right color=white] plot[smooth cycle] coordinates{(-1,-1) (-2,0)(-2.1,1)(-1,2)(1,2.3)(2,2.5)(3,2)(3,1)(2,0)(1,-1)(0,-1.2)};
\draw(-1,0.3) node[diamond,draw,minimum size=1.7cm,name=s] {};
\draw(1.8,1.2) node[diamond,draw,minimum size=1.7cm,name=s2] {};
\draw(-1,1.4) node {$(M,g)$};
\draw(-1,-2) node [name=F]{$\Phi_{(\Ocal,g)}(f)$};
\draw(1.8,-2) node [name=F2]{$\Phi_{(\Ocal',\alpha_*g)}(\alpha_*f)$};
\shade[shading=radial, inner color=darksee!80!white,outer color=see!80!white,scale=0.15] plot[smooth cycle] coordinates{(-7,1)(-9,2)(-10.2,3)(-9.4,4.5)(-5.4,5)(-4,4)(-3.1,2.2)(-4,1.2)};
\draw(-1,0.4) node [name=f,color=lighthoney]{$f$};
\shade[shading=radial, inner color=darksee!80!white,outer color=see!80!white,scale=0.15] plot[smooth cycle] coordinates{(12,7)(10,8)(8.8,9)(9.6,10.5)(13,11)(15,10)(15.9,8.2)(15,7.2)};
\draw(1.8,1.4) node [name=f2,color=lighthoney]{$\alpha_*f$};
\draw(-1,-0.2) node {$\Ocal$};
\draw(1.8,0.8) node {$\Ocal'$};
node {};
\path[->] (f) edge  [color=red, bend right]  node[right] {} (F);
\path[->] (s.north east) edge  [color=red,out=45,in=170]  node[above] {$\alpha$} (s2.north west);
\path[->] (f2) edge  [color=red, bend left]  node[right] {} (F2);
\end{tikzpicture}
\caption{Moving the experimental setup to a different region of spacetime.\label{mov}}
\end{center}
\end{figure}
The condition for $\Phi$ to be a natural transformation reads
\be\label{nat}
\Phi_\Ocal(f)(\chi^*h)=\Phi_{M}(\chi_*f)(h) \,.
\ee
The interpretation of physically relevant quantities as natural transformations fits very well with the intuitive picture. In classical gravity we understand geometrical data not as pointwise objects but rather as something defined on all the spacetimes in a coherent way. On the practical side, we can think of natural transformations $\Nat(\D,\F)$ as elements of $\ \prod\limits_{\mathclap{\sst{M\in\textrm{Obj}(\Loc)}}}L(\D(M),\F(M))$. In QFT one can drop the assumption of linearity, i.e. compose functors $\F$ and $\D$ with a forgetful functor to the category $\Sts$ of sets or $\TA$ of topological algebras.

Now we want to understand how the fields transform under diffeomorphisms. Firstly we need a notion of transforming all the spacetimes in a coherent way compatible with the embedding structure. Since our discussion is local, we can concentrate on infinitesimal diffeomorphisms, i.e. vector fields in $\X(M)\doteq \Gamma(TM)$. Let $\chi:(M,g')\rightarrow (N,g)$. $\X$ can be made into a \textit{contravariant functor} to the category $\Vect$ as follows: to $\xi\in\X((N,g))$ we associate a 1-form $i_\xi g$ and pull it back to a form $\chi^*(i_\xi g)\in\Gamma(T^*M)$. This again corresponds to a vector field by means of the $g'$-induced duality between $\Gamma(T^*M)$ and $\Gamma(TM)$. We denote the resulting map by $\X\chi:\X(N)\rightarrow\X(M)$. If $\chi$ is a diffeomorphism, this is just the pullback map $\chi^*$.

We can now define a Lie algebra $\Xcal$, which provides us with a notion of transforming all the spacetimes in a coherent way.
	\[
	\mathcal{X}\doteq\prod\limits_{M\in \obj(\Loc)}\!\!\!\!\X(M)\,.
	\]
Let now $\vec{\xi}\in\Xcal$ be an element of this algebra with all the components compactly supported. In this case we can apply the exponential map and define a diffeomorphism: $\alpha_{\sst{M}}\doteq \exp(\xi_{\sst{M}})$.
Let $\vec{\alpha}$ denote the corresponding sequence of diffeomorphisms. It acts on  $\Nat(\D,\F)$ as:
 \[
(\vec{\alpha}\Phi)_{(M,g)}[h]\doteq\Phi_{(\alpha_{M}(M),{\alpha_{M}}_*g)}[h+g-{\alpha_{\sst{M}}}_*g]\,.
\]
From the naturality condition (\ref{nat}) follows that:
\[
 \Phi_{(\alpha_{M}(M),{\alpha_{M}}_*g)}({\alpha_{\sst{M}}}_*\,f)[h+g-{\alpha_{\sst{M}}}_*g]=\\
	 =\Phi_{(M,g)}(f)[\alpha_{\sst{M}}^*\tilde{g}-g]\,.
\]
Therefore it holds \textit{always}: 
\[
(\vec{\alpha}\Phi)_{(M,g)}(f)[h]=\Phi_{(M,g)}({\alpha_{\sst{M}}^{\minus}}_*\,f)[\alpha_{\sst{M}}^*\tilde{g}-g]\,.
\] 
This is the diffeomorphism \textit{\textbf{covariance}}\index{diffeomorphism!covariance}. It is a consequence of the locally covariant formulation of the theory. Now the diffeomorphism \textit{\textbf{invariance}}\index{diffeomorphism!invariance} is a stronger condition, that states:
\[
(\vec{\alpha}\Phi)_{(M,g)}(f)[h]=\Phi_{(M,g)}(f)[h]\,.
\]
Now we come back to the discussion of infinitesimal diffeomorphisms. The derived action of compactly supported elements of $\Xcal$ on fields reads:
\be\label{diffeo}
(\vec{\xi}\Phi)_{(M,g)}(f)[h]=\left<(\Phi_{(M,g)}(f))^{(1)}(h),\pounds_{\xi_M}\tilde{g}\right>+\Phi_{(M,g)}(\pounds_{\xi_M}\,f)[h]\,.
\ee
The right hand side is well defined also if we drop the compact support condition on $\vec{\xi}$, so we can adapt the above formula as the definition of the action of $\Xcal$ on $\Nat(\D,\F)$. Diffeomorphism invariance is now the requirement that: $\vec{\xi}\Phi=0$.
\begin{exa}
As an example of a diffeomorphism invariant field we can take
\[
\Phi_{(M,g)}(f)[h]=\int R[\tilde{g}]f\,\textrm{dvol}_{(M,\tilde{g})}\,.
\]
One can check that $\xi\Phi=0$ for all $\xi\in\Xcal$. Note that both the scalar curvature and the volume form depend on the full metric $\tilde{g}$. However if we take a field defined as $\int R[\tilde{g}]f\,\textrm{dvol}_{(M,g)}$ it is still diffeomorphism \textit{covariant}, but no longer invariant, i.e. $\qquad\vec{\xi}\Phi\ /\!\!\!\!\!\!\equiv0$.
\eex\end{exa}
After discussing the generalities concerning the diffeomorphism invariance, we can now reformulate it in the homological language. In other words we have to build the BV complex that will describe the infinitesimal symmetries acting on $\Nat(\D,\F)$ according to (\ref{diffeo}).
\subsection{BV complex on the level of natural transformations}\label{BVnat}
In this section we define the BV-complex of general relativity, basing on the concept proposed by R. Brunetti to treat fields as natural transformation. The construction follows the one we proposed in \cite{FR}. Firstly we note that in (\ref{diffeo}) the action of infinitesimal symmetries on the elements of $\Nat(\D,\F)$ has two terms, where the first one is analogous to the infinitesimal transformation in gauge theories (\ref{rhoM2}). The second term is characteristic to the theories where symmetries are a consequence of diffeomorphism invariance. 
Since we are now working on the level of natural transformations, the underlying algebra of the BV complex
also has to be built from a set of natural transformations. It is useful to make one more generalization. Instead of test functions, we can consider more general objects,  such as arbitrary compactly supported tensors $\Tens_c(M)$. From now on we set $\E_c=\Tens_c$ and use this functor as a functor associating to a manifold the space of test field configurations.

The Chevalley-Eilenberg complex on the level of natural transformations is defined as $\bigoplus\limits_{k=0}^\infty \Nat(\E_c^k,\CE)$, where $\CE$ is a functor constructed analogously as in Yang-Mills theories, i.e.  $\CE(M)$ is the graded algebra of smooth compactly supported multilocal maps $\CE(M)\doteq\Ci_\ml(\E(M),\La{\X}'(M))$ with the differential $\gamma$ defined as:
\begin{align*}
\gamma:&\ \Lambda^q\X'(M)\otimes\F(M)\rightarrow\Lambda^{q+1}\X'(M)\otimes\F(M)\,,\nonumber\\
(\gamma \omega)(\xi_0,\ldots, \xi_q)&\doteq\sum\limits_{i=0}^q(-1)^i\partial_{\rho_M(\xi_i)}(\omega(\xi_0,\ldots,\hat{\xi}_i,\ldots,\xi_q))+\nonumber\\
&+\sum\limits_{i<j}(-1)^{i+j}\omega\left([\xi_i,\xi_j],\ldots,\hat{\xi}_i,\ldots,\hat{\xi}_j,\ldots,\xi_q\right)\,,
\end{align*}
and extended by continuity.

On the level of natural transformations $\bigoplus\limits_{k=0}^\infty \Nat(\E_c^k,\CE)$ we can define the Chevalley-Eilenberg differential $\gamma=\gamma^{(0)}+\gamma^{(1)}$ by:
\begin{align*}
\gamma^{(0)}\Phi&=\rho(.)\Phi,\qquad \Phi\in\Nat(\E_c,\F)\,,\\
\gamma^{(1)}\Phi&=-\Phi\circ[.,.],\qquad\Phi\in\Nat(\E_c,\X')\,.
\end{align*}
In the above formula $\rho(.)\Phi$ is a natural transformation between the functors $\D$ and $\Ci_\ml(\E,\X')$
 given by (\ref{diffeo}), i.e.
 \[
(\rho(.)\Phi)_{(M,g)}(f)(h,X)=\left<(\Phi_{(M,g)}(f))^{(1)}(h),\pounds_{X}\tilde{g}\right>+\Phi_{(M,g)}(\pounds_{X}\,f)[h]\,.
\]
The requirement of the graded Leibniz rule allows us to extend $\gamma$ to the whole $\bigoplus\limits_{k=0}^\infty \Nat(\E_c^k,\CE)$. We define the extended algebra of fields\index{extended algebra of fields} as:
\[
Fld=\bigoplus\limits_{k=0}^\infty \Nat(\E_c^k,\BV)\,,
\]
where $\BV$ is a functor from $\Loc$ to $\PgAlg$ defined in analogy to the gauge theories:
\be
\BV(M)=\Ci_\ml\big(\E(M),\La\E_c(M)\widehat{\otimes}\La{\X}'(M)\widehat{\otimes}S^\bullet \X_c(M)\big)\label{BVfix}
\ee
The set $Fld$ becomes a graded algebra if we equip it with addition defined pointwise and a graded product defined as:
\be\label{ntprod}
(\Phi\Psi)_M(f_1,...,f_{p+q})=\frac{1}{p!q!}\sum\limits_{\pi\in P_{p+q}}\mathrm\Phi_M(f_{\pi(1)},...,f_{\pi(p)})\Psi_M(f_{\pi(p+1)},...,f_{\pi(p+q)})\,,
\ee
where the product on the right hand side is the product of the algebra $\BV(M)$.
We can also introduce on $Fld$ a graded bracket using definition (\ref{ntbracket}), i.e.:
\[
\{\Phi_1,\Phi_2\}_M(f_1,...,f_{p+q})=\frac{1}{p!q!}\sum\limits_{\pi\in P_{p+q}}\{{\Phi_1}_M(f_{\pi(1)},...,f_{\pi(p)}),{\Phi_2}_M(f_{\pi(p+1)},...,f_{\pi(p+q)})\}\,.
\]
This bracket is graded antisymmetric, satisfies the graded Jacobi identity (\ref{Jacid}) and the Leibniz rule (\ref{leibniz}) with respect to the product (\ref{ntprod}), so $(Fld,\{.,.\})$ is a graded Poisson algebra.
 As in the Yang-Mills case we can define natural transformations $\theta^{(0)},\theta^{(1)}\in Fld$ that locally implement the Chevalley-Eilenberg differential on the fixed background:
 \begin{align*}
(\theta_M^{(0)}(f)F)(X)&\doteq\left<F^{(1)}(h),\pounds_{fX}\tilde{g}\right>,\quad\quad F\in\F(M),\ X\in\X(M)\,,\\
(\theta_M^{(1)}(f)\omega)(X_1,X_2)&\doteq\omega(f[X_1,X_2]),\quad\ \omega\in\X'(M),\ X_1,X_2\in\X(M)\,,
\end{align*}
Let $\theta=\theta^{(0)}+\theta^{(1)}$.
The full Chevalley-Eilenberg differential can be written as:
\[
(\gamma \Phi)_{M}(f)=\{\Phi,\theta\}_M(f,f_1)+(-1)^{|\Phi|}\Phi_{M}(\pounds_{(.)}\,f)\,
\]
where $\Phi\in Fld$ and  $f_1\equiv1$ on $\supp f$. The BV differential $s$ is now defined as:
\[
(s\Phi)_M(f)\doteq\{\Phi,L+\theta\}_M(f,f_1)+(-1)^{|\Phi|}\Phi_M(\pounds_{(.)}f)\quad f_1\equiv1\,\mathrm{on}\,\supp f\,.
\]
The $0$-cohomology of $s$ is nontrivial, since it contains for example the Riemann tensor contracted with itself, smeared with a test function:
\[
\Phi_{(M,g)}(f)(h)=\int\limits_M R_{\mu\nu\alpha\beta}[\tilde{g}]R^{\mu\nu\alpha\beta}[\tilde{g}] fd\textrm{vol}_{(M,\tilde{g})},\qquad \tilde{g}=g+h\,.
\]
Based on the discussion in \ref{difftrans}, we claim that the physical quantities should be identified with the elements of $H^0(Fld,s)$. As natural transformations they define what it means to have the same physical objects in all spacetimes. In this sense we get a structure that is completely covariant. We can now introduce dynamics on $H^0(Fld,s)$ by defining the Poisson bracket.
\subsection{Peierls bracket}\label{GRPeierls}
To obtain an algebra closed under the Peierls bracket\index{Peierls bracket} we 
have to replace the multilocal with the microcausal functionals and use the topology $\tau_\Xi$.  The corresponding space of natural transformation is denoted by $Fld_{mc}$. We can construct the Peierls bracket analogously as in Section \ref{Peierls}. Our starting point is a Lagrangian $L^\ex=L+\theta$. Formally it can be written as
\be
L^\ex_M(f)\form\int fR[\tilde{g}]\,d\textrm{vol}_{(M,\tilde{g})}+\int\!\!d\textrm{vol}_{(M,g)}( f\pounds_C \tilde{g}_{\mu\nu})\frac{\delta}{\delta h_{\mu\nu}}+\frac{1}{2}\int\!\!d\textrm{vol}_{(M,g)}(f[C,C]^\mu)\frac{\delta}{\delta C^{\mu}}\,.
\ee
To impose the gauge fixing we introduce the nonminimal sector. We shall do it already on the level of natural transformations. The functions of Nakanishi-Lautrup\index{Nakanishi-Lautrup fields} fields will be the elements of $\Nat(\E_c,S^\bullet\X')$ and functions of antighosts\index{antighosts} will belong to $\Nat(\E_c,\La\X')$. We can define the BV operator on the nonminimal sector simply as: $s\Phi_1:=\Pi\Phi_1\circ m_i$, $s\Phi_2=0$ for $\Phi_1\in \Nat(\E_c,\La^1\X')$, $\Phi_2\in \Nat(\E_c,S^1\X')$. To impose the gauge fixing we shall use a gauge fixing fermion $\Psi\in Fld_{mc}$. It induces a transformation of $Fld_{mc}$ given by (\ref{gfermion}). This transformation is an isomorphism on the cohomology groups, since:
 \[
 (\tilde{s}\tilde{X})_M(f_1)=\widetilde{\{X_M(f_1),L_M(f_2)\}}+(-1)^{|X|}\tilde{X}_M(\pounds_{(.)}f_1)= \widetilde{(sX)}_M(f_1)\,,
 \]
where $X$ is a natural transformation with values in derivations, $\tilde{X}:=\alpha_\Psi(X)$ and $f_2\equiv 1$ on the support of $f_1$. The above result can be written more compactly as:
 \[
 \tilde{s}\tilde{X}=\widetilde{sX}
 \]
 To fix the gauge we have to choose $\Psi$. Since it has to be covariant, the most natural choice is the background gauge (see \cite{Ichi,NiOk}), i.e.:
 \[
 \Psi_{(M,g)}(f)=i\int\!\!\textrm{dvol}_{(M,g)}\left(\frac{\alpha}{2}\bar{C}_\mu B^\mu+\frac{1}{\sqrt{-g}}\bar{C}_\mu\nabla_\nu (\gt{\nu}{\mu})\right)=i\int\!\!g\big(\bar{C}, \frac{\alpha}{2}B+K(h)\big)\textrm{dvol}_{(M,g)}\,,
 \]
 where the indices are lowered in the background metric $g$, $\nabla$ is the covariant derivative on $(M,g)$ and we denoted $\tilde{\textfrak{g}}^{\nu\lambda}:=\sqrt{-\tilde{g}}\tilde{g}^{\nu\lambda}$, $K^\mu(h)=\frac{1}{\sqrt{-g}}\nabla_\nu \gt{\nu}{\mu}$. For $\alpha=0$ this is just the harmonic gauge. After putting antifields to $0$ we obtain the following form of the gauge-fixed Lagrangian:
 \[
 L^g_M(f):=\int fR[\tilde{g}]\,\textrm{dvol}_{(M,\tilde{g})}+\int\!\!\textrm{dvol}_{(M,g)}f\left(ig\big(\bar{C},\frac{\delta K}{\delta h}[\pounds_C \tilde{g}]\big)-g\big(B, \frac{\alpha}{2}B+K(h)\big)\right)\,.
 \]
The differential $\tilde{s}$  can be expanded with respect to the total antifield number as $\tilde{s}=\delta^g+\gamma^g$, where $\delta^g$ is the Koszul differential of the gauge fixed action and $\gamma^g$ is the gauge-fixed BRST differential given by:
 \[
(\gamma^g\Phi)_M(f)=\gamma_0^g(\Phi_M(f))+(-1)^{|\Phi|}\Phi_M(\rho(.)f)\,,
 \]
 where:
 \begin{center}
{\setlength{\extrarowheight}{2.5pt}
\begin{tabular}{ll}
\toprule%
& $\gamma_0^g$\\\otoprule%
$F\in\F(M)$&$\left<F^{(1)},\pounds_C(g+.)\right>$\\
 $C$&$-\frac{1}{2}[C,C]$\\
 $B$& $0$\\
  $\bar{C}$& $iB$\\\bottomrule
\end{tabular}}
\end{center}
The algebra of physical microcausal fields can be recovered as $Fld_{ph}:=H^0(H_0(Fld_{mc},\delta^g),\gamma^g)$. To introduce the Poisson structure on it, we shall first do it on $Fld_{mc}$. We start with finding the field equations for the action $S^g$. We use the fact that in local coordinates:
 \[
 \pounds_C\gt{\mu}{\nu}=-\gt{\alpha}{\nu}\nabla_\al C^\mu-\gt{\alpha}{\mu}\nabla_\al C^\nu+\nabla_\alpha (C^\alpha\gt{\mu}{\nu})\,.
\]
The field equations\footnote{The field equations have to understood in the algebraic sense, see footnote \ref{eqs}.} take the form:
\begin{align*}
R_{\nu\lambda}[\tilde{g}]&=-i\Big(\nabla_{(\nu}\bar{C}_{|\mu|}\nabla_{\la)}C^\mu+\nabla_\mu\bar{C}_{(\nu}\nabla_{\la)}C^\mu+(\nabla_\al\nabla_{(\nu}\bar{C}_{\la)})C^\al\Big)-\nabla_{(\la}B_{\nu)}\,,\\%
\tilde{\textfrak{g}}^{\nu\lambda}\nabla_{\nu}\nabla_{\lambda} C^\mu&=\tilde{\textfrak{g}}^{\alpha\nu}R_{\lambda\nu\alpha}^{\quad\ \mu}[g]C^\lambda+\alpha\sqrt{-g}\Big(B^\lambda\nabla_\lambda C^\mu-\nabla_\lambda(B^\mu C^\lambda)\Big)\,,\\%
\tilde{\textfrak{g}}^{\nu\lambda}\nabla_{\nu}\nabla_{\lambda} \bar{C}_\mu&=-\tilde{\textfrak{g}}^{\nu\lambda}R_{\nu\mu\lambda}^{\quad\ \alpha}[g]\bar{C}_\alpha+\alpha\sqrt{-g}\Big(B^\lambda\nabla_\lambda \bar{C}_\mu+B^\lambda \nabla_\mu \bar{C}_\lambda\Big)\,,
\\%
\nabla_\nu \tilde{\textfrak{g}}^{\nu\mu}&=-\alpha\sqrt{-g} B^\mu\,.
\end{align*}
The details of the calculation are given in the appendix \ref{variables}. This system is gauge-fixed but not normally hyperbolic in all the variables, since we have second derivatives of the ghosts in the first equation. To find a solution of this problem we rewrite the first equation using the Leibniz rule:
\[
(\nabla_\al\nabla_{(\nu}\bar{C}_{\la)})C^\al=-\nabla_{(\nu}(C^\al\nabla_{|\al|}\bar{C}_{\la)})-\nabla_\al\bar{C}_{(\la}\nabla_{\nu)}C^\al-\frac{1}{2}\Big(R_{\al\nu\ \ \la}^{\quad\bet}\bar{C}_\bet C^\al+R_{\al\la\ \ \nu}^{\quad\bet}[g]\bar{C}_\bet C^\al\Big)\,.
\]
It follows that
\[
-R_{\nu\lambda}[\tilde{g}]=i\nabla_{(\nu}\bar{C}_{|\mu|}\nabla_{\la)}C^\mu+\nabla_{(\nu}\Big(B_{\la)}-iC^\al\nabla_{|\al|}\bar{C}_{\la)}\Big)-iR_{\al\nu\bet\la}[g]\bar{C}^{(\bet} C^{\al)}\,.
\]
It is now evident that the system can be made hyperbolic by a suitable variable change. Before setting antifields to $0$ we perform a canonical transformation of the algebra $\BV(M)$ by setting $b_\la=B_{\la}-iC^\al\nabla_{\al}\bar{C}_{\la}$. The antifields have to transform in such a way that the antibracket remains conserved, i.e.: 
\begin{align}\label{vartr}
b_\la&=B_{\la}-iC^\al\nabla_{\al}\bar{C}_{\la}&\frac{\delta}{\delta b_\la}&=\frac{\delta}{\delta B_\la}\,,\nonumber\\
c_\nu&=C_\nu&\frac{\delta}{\delta c_\la}&=\frac{\delta}{\delta C_\la}+i\nabla^{\la}\bar{C}_{\bet}\frac{\delta}{\delta B_\bet}\,,\\
\bar{c}_\nu&=\bar{C}_\nu&\frac{\delta}{\delta \bar{c}_\la}&=\frac{\delta}{\delta \bar{C}_\la}-iC^\al\nabla_{\al}(.)\circ\frac{\delta}{\delta B_\la}\,.\nonumber
\end{align}
We used a short-hand notation: $\big(C^\al\nabla_{\al}(.)\circ\frac{\delta}{\delta B_\la}\big)(x)\doteq\int dz C^\al(z)\nabla_\al^z\delta(x-z)\frac{\delta}{\delta B_\la(z)}$. The new gauge-fixed Lagrangian takes the form
 \begin{multline}
 L^g_M(f)=\int fR[\tilde{g}]\,d\textrm{vol}_{(M,\tilde{g})}+\\
 +\int\!\!d\textrm{vol}_{(M,g)}f\left(ig\big(\bar{c},\frac{\delta K}{\delta h}[\pounds_c \tilde{g}]\big)-g\Big(b+(ic^\al\nabla_{\al})\bar{c}, \frac{\alpha}{2}(b+(ic^\al\nabla_{\al})\bar{c})+K(h)\Big)\right)\,.
 \end{multline}
The gauge-fixed BRST differential $\gamma^g_0$ is now defined as:
 \begin{center}
{\setlength{\extrarowheight}{2.5pt}
\begin{tabular}{ll}
\toprule%
& $\gamma_0^g$\\\otoprule%
$F\in\F(M)$&$\left<F^{(1)},\pounds_c(g+.)\right>$\\
 $c$&$-\frac{1}{2}[c,c]$\\
 $b$& $i(c^\bet\!\wedge c^\al\nabla_\bet\nabla_\al)\bar{c}+c^\al\nabla_\al b$\\
  $\bar{c}$& $ib-c^\la\nabla_\la \bar{c}$\\\bottomrule
\end{tabular}}
\end{center}
The equations of motion in the new variables can be written as:
\begin{align}
\tilde{R}_{\nu\lambda}&=-i\nabla_{(\nu}\bar{c}_{|\mu|}\nabla_{\la)}c^\mu-\nabla_{(\nu}b_{\la)}+iR_{\al\nu\bet\la}\bar{c}^{(\bet} c^{\al)}\label{sys2}\,,\\%
\tilde{\textfrak{g}}^{\nu\lambda}\nabla_{\nu}\nabla_{\lambda} c^\mu&=\tilde{\textfrak{g}}^{\alpha\nu}R_{\lambda\nu\alpha}^{\quad\ \mu}c^\lambda+\alpha\sqrt{-g}\Big((b^\la+ic^\al\nabla_\al \bar{c}^\la)\nabla_\lambda c^\mu-\nabla_\lambda((b^\mu+ic^\al\nabla_\al \bar{c}^\mu) c^\lambda)\Big)\,,\nonumber\\%
\tilde{\textfrak{g}}^{\nu\lambda}\nabla_{\nu}\nabla_{\lambda} \bar{c}_\mu&=-\tilde{\textfrak{g}}^{\nu\lambda}R_{\nu\mu\lambda}^{\quad\ \alpha}\bar{c}_\alpha+\alpha\sqrt{-g}\Big((b^\la+ic^\al\nabla_\al \bar{c}^\la)\nabla_\lambda \bar{c}_\mu+(b^\la+ic^\al\nabla_\al \bar{c}^\la) \nabla_\mu \bar{c}_\lambda\Big)\,,\nonumber
\\%
\nabla_\nu \tilde{\textfrak{g}}^{\nu\mu}&=-\alpha\sqrt{-g} (b^\mu+ic^\al\nabla_\al \bar{c}^\mu)\,,\nonumber
\end{align}
where we denoted $\tilde{R}_{\nu\la}:=R_{\nu\la}[\tilde{g}]$ and $R_{\al\bet\gamma\la}:=R_{\al\bet\gamma\la}[g]$. The equation for $b$ can be obtained from the first equation by means of the Bianchi identity. One can already see that after linearization we obtain a normally hyperbolic system of equations since from the antisymmetry of ghost fields follows that $c^\la c^\al\partial_\lambda\partial_\al \bar{c}^\mu=0$ and all the other second order terms are of the d'Alembertian-like form.
For such a system, retarded and advanced solutions of the linearized equations exist and one can define the Peierls bracket\index{Peierls bracket} on $Fld_{mc}$. Like in case of Yang-Mills theories it is well defined also on $Fld_{ph}$ and we obtain a Poisson algebra $(Fld_{ph},\{.,.\}_{S^g})$. Although the Poisson structure on $Fld_{mc}$ can depend on the choice of variables in the extended algebra, this doesn't affect the structure induced on $Fld_{ph}$. Note that  for the harmonic gauge ($\al=0$) and the Minkowski background the system (\ref{sys2}) simplifies to (compare with \cite{NaOji}):
\begin{align*}
\tilde{R}_{\nu\lambda}&=-i\partial_{(\nu}\bar{c}_{|\mu|}\partial_{\la)}c^\mu-\partial_{(\la}b_{\nu)}\,,\\%
\Box_{\tilde{g}}c^\mu&=0\,,\\%
\Box_{\tilde{g}}\bar{c}_\mu&=0\,,\\%
\Box_{\tilde{g}}b_\mu&=0\,,\\%
\partial_\nu \tilde{\textfrak{g}}^{\nu\mu}&=0\,.
\end{align*}
\subsection{Appendix: calculation of the equations of motion}\label{variables}
The gauge fixed Lagrangian is equivalent to: $L^g\sim L^{\sst{\mathrm{EH}}}+L^{\sst{\mathrm{FP}}}+L^{\sst{\mathrm{GF}}}$, where:
\begin{align*}
L^{\sst{\mathrm{EH}}}_{\sst{M}}(f)&=\int d\textrm{vol}_{(M,\tilde{g})} R[\tilde{g}]f\,,\\
L^{\sst{\mathrm{FP}}}_{\sst{M}}(f)&=i\int d\textrm{vol}_{(M,g)}\frac{1}{\sqrt{-g}}\nabla_\nu\bar{C}_\mu\Big(\gt{\lambda}{\nu}\nabla_\lambda C^\mu+\gt{\lambda}{\mu}\nabla_\lambda C^\nu-\nabla_\lambda(C^\lambda\gt{\mu}{\nu})\Big)f\,,\\
L^{\sst{\mathrm{GF}}}_{\sst{M}}(f)&=-\int d\textrm{vol}_{(M,g)}\Big(\frac{\alpha}{2}B^\mu B_\mu+\frac{1}{\sqrt{-g}}B_\mu\nabla_\nu(\gt{\nu}{\mu})\Big)f\,.
\end{align*}
We used the fact that in local coordinates:
 \[
 \pounds_C\gt{\mu}{\nu}=-\gt{\alpha}{\nu}\nabla_\al C^\mu-\gt{\alpha}{\mu}\nabla_\al C^\nu+\nabla_\alpha (C^\alpha\gt{\mu}{\nu})\,.
\]
Now we calculate the field equations. We obtain following formulas for the Euler-Lagrange derivatives of the Fadeev-Popov term:
\begin{align}
(S^{\sst{\mathrm{FP}}}_{\sst{M}})'_{\bar{C}_\mu}&=-\frac{i}{\sqrt{{\sst{-}}g}}\nabla_\nu\Big(\gt{\la}{\nu}\nabla_\la C^\mu+\gt{\la}{\mu}\nabla_\la C^\nu-\nabla_\la(C^\la\gt{\mu}{\nu})\Big)\,,\label{ghost}\\
(S^{\sst{\mathrm{FP}}}_{\sst{M}})'_{C^\mu}&=-\frac{i}{\sqrt{{\sst{-}}g}}\Big(\nabla_\la(\gt{\la}{\nu}\nabla_\nu\bar{C}_\mu)+\nabla_\la(\gt{\la}{\nu}\nabla_\mu\bar{C}_\nu)-\gt{\al}{\bet}\nabla_\mu\nabla_\al\bar{C}_\bet\Big)\,,\label{antigh}\\
(S^{\sst{\mathrm{FP}}}_{\sst{M}})'_{h^{\lambda\nu}}&=\frac{\sqrt{{\sst{-}}\tilde{g}}}{\sqrt{{\sst{-}}g}}\bigg(i\Big(\nabla_{(\nu}\bar{C}_{|\mu|}\nabla_{\la)}C^\mu+\nabla_\mu\bar{C}_{(\nu}\nabla_{\la)}C^\mu+(\nabla_\al\nabla_{(\nu}\bar{C}_{\la)})C^\al\Big)+\nonumber\\
&-\frac{i}{2}\tilde{g}_{\la\nu}\Big(\tilde{g}^{\bet\al}\nabla_\al\bar{C}_\mu\nabla_\bet C^\mu+\tilde{g}^{\bet\mu}\nabla_\al\bar{C}_\mu\nabla_\bet C^\al+\tilde{g}^{\mu\al}(\nabla_\bet\nabla_\al\bar{C}_\mu)C^\bet\Big)\bigg)\,.\label{FPvar}
\end{align}
Similarly for the gauge fixing term:
\begin{align}
(S^{\sst{\mathrm{GF}}}_{\sst{M}})'_{B_\mu}&=-\al\sqrt{{\sst{-}}g} B^\mu-\nabla_\nu(\gt{\nu}{\mu})\,,\label{gfix}\\
(S^{\sst{\mathrm{GF}}}_{\sst{M}})'_{h^{\lambda\nu}}&=\frac{\sqrt{{\sst{-}}\tilde{g}}}{\sqrt{{\sst{-}}g}}\Big(\nabla_{(\la}B_{\nu)}-\frac{1}{2}\tilde{g}_{\nu\la}\tilde{g}^{\al\bet}\nabla_\al B_\bet\Big)\,.\label{GFvar}
\end{align}
The variation of $L_{\mathrm{EH},M}(f)$ gives simply the Einstein's equation for the full metric:
\be
(S^{\sst{\mathrm{EH}}}_{\sst{M}})'_{h^{\lambda\nu}}=\frac{\sqrt{{\sst{-}}\tilde{g}}}{\sqrt{{\sst{-}}g}}\,\tilde{G}_{\nu\la}\,,\label{EHvar}
\ee
where we denoted $\tilde{G}_{\nu\la}:=G[\tilde{g}]_{\nu\la}$. Using the fact that $\tilde{G}_{\nu\lambda}=\tilde{R}_{\nu\lambda}-\frac{1}{2}\tilde{g}_{\nu\lambda}\tilde{g}^{\al\bet}\tilde{R}_{\al\bet}$ we can combine equations (\ref{FPvar}), (\ref{GFvar}) and (\ref{EHvar}) into
\be\label{riem2}
-\tilde{R}_{\nu\lambda}=i\Big(\nabla_{(\nu}\bar{C}_{|\mu|}\nabla_{\la)}C^\mu+\nabla_\mu\bar{C}_{(\nu}\nabla_{\la)}C^\mu+(\nabla_\al\nabla_{(\nu}\bar{C}_{\la)})C^\al\Big)+\nabla_{(\la}B_{\nu)}\,.
\ee
The equation for the ghost (\ref{ghost}) can be rewritten as:
\begin{multline*}
0=\gt{\la}{\nu}\nabla_\nu\nabla_\la C^\mu+\gt{\la}{\mu}\nabla_\nu\nabla_\la C^\nu-(\nabla_\nu\nabla_\la C^\la)\gt{\mu}{\nu}+\\+(\nabla_\nu\gt{\la}{\nu})\nabla_\la C^\mu+(\nabla_\nu\gt{\la}{\mu})\nabla_\la C^\nu-C^\la \nabla_\nu\nabla_\la \gt{\mu}{\nu}-(\nabla_\nu C^\la )\nabla_\la\gt{\mu}{\nu}-(\nabla_\la C^\la )\nabla_\nu\gt{\mu}{\nu}\,.
\end{multline*}
Using the formula for the commutator of covariant derivatives we obtain:
\begin{multline*}
0=\gt{\la}{\nu}\nabla_\nu\nabla_\la C^\mu+\gt{\la}{\mu}(\nabla_\nu\nabla_\la C^\nu-R_{\nu\la\al}^{\quad\ \nu}C^\al)-\gt{\la}{\mu}(\nabla_\la\nabla_\nu C^\nu)+(\nabla_\nu\gt{\la}{\nu})\nabla_\la C^\mu+\\-(\nabla_\la C^\la )\nabla_\nu\gt{\mu}{\nu}-C^\la (\nabla_\la\nabla_\nu \gt{\mu}{\nu}+R_{\la\nu\al}^{\quad\ \mu}\gt{\al}{\nu}+R_{\la\nu\al}^{\quad\ \nu}\gt{\al}{\mu})=\\
=\gt{\la}{\nu}\nabla_\nu\nabla_\la C^\mu-\gt{\la}{\mu}R_{\nu\la\al}^{\quad\ \nu}C^\al+(\nabla_\nu\gt{\la}{\nu})\nabla_\la C^\mu-(\nabla_\la C^\la )\nabla_\nu\gt{\mu}{\nu}+\\-C^\la (\nabla_\la\nabla_\nu \gt{\mu}{\nu}+R_{\la\nu\al}^{\quad\ \mu}\gt{\al}{\nu}-R_{\nu\al\la}^{\quad\ \nu}\gt{\al}{\mu})=\\
=\gt{\la}{\nu}\nabla_\nu\nabla_\la C^\mu+(\nabla_\nu\gt{\la}{\nu})\nabla_\la C^\mu-(\nabla_\la C^\la )\nabla_\nu\gt{\mu}{\nu}-C^\la (\nabla_\la\nabla_\nu \gt{\mu}{\nu}+R_{\la\nu\al}^{\quad\ \mu}\gt{\al}{\nu})\,.
\end{multline*}
Terms containing the derivatives of the full metric can be eliminated with the gauge fixing condition (\ref{gfix}). The result is:
\begin{align}
\gt{\la}{\nu}\nabla_\nu\nabla_\la C^\mu&=C^\la R_{\la\nu\al}^{\quad\ \mu}\gt{\al}{\nu}+\al\sqrt{{\sst{-}}g}\Big( B^\la\nabla_\la C^\mu-(\nabla_\la C^\la ) B^\mu-C^\la \nabla_\la B^\mu\Big)=\nonumber\\
&=C^\la R_{\la\nu\al}^{\quad\ \mu}\gt{\al}{\nu}+\al\sqrt{{\sst{-}}g}\Big( B^\la\nabla_\la C^\mu-\nabla_\la (C^\la B^\mu)\Big)\,.\label{ghost2}
\end{align}
Similarly one obtains from the equation for the antighost (\ref{antigh}):
\be\label{antigh2}
\gt{\la}{\nu}\nabla_\la\nabla_\nu \bar{C}_\mu=-\gt{\la}{\nu}R_{\la\mu\nu}^{\quad\ \al}\bar{C}_\al+\al\sqrt{{\sst{-}}g}\Big(B^\nu\nabla_\nu \bar{C}_\mu+B^\nu\nabla_\mu\bar{C}_\nu\Big)\,.
\ee
Together equations (\ref{gfix}), (\ref{riem2}),  (\ref{ghost2}) and (\ref{antigh2}) constitute the following system:
\begin{align}
-\tilde{R}_{\nu\lambda}&=i\Big(\nabla_{(\nu}\bar{C}_{|\mu|}\nabla_{\la)}C^\mu+\nabla_\mu\bar{C}_{(\nu}\nabla_{\la)}C^\mu+(\nabla_\al\nabla_{(\nu}\bar{C}_{\la)})C^\al\Big)+\nabla_{(\la}B_{\nu)}\,,\nonumber\\%
\tilde{\textfrak{g}}^{\nu\lambda}\nabla_{\nu}\nabla_{\lambda} C^\mu&=\tilde{\textfrak{g}}^{\alpha\nu}R_{\lambda\nu\alpha}^{\quad\ \mu}C^\lambda+\alpha\sqrt{-g}\Big(B^\lambda\nabla_\lambda C^\mu-\nabla_\lambda(B^\mu C^\lambda)\Big)\,,\label{sys1}\\%
\tilde{\textfrak{g}}^{\nu\lambda}\nabla_{\nu}\nabla_{\lambda} \bar{C}_\mu&=-\tilde{\textfrak{g}}^{\nu\lambda}R_{\nu\mu\lambda}^{\quad\ \alpha}\bar{C}_\alpha+\alpha\sqrt{-g}\Big(B^\lambda\nabla_\lambda \bar{C}_\mu+B^\lambda \nabla_\mu \bar{C}_\lambda\Big)\,,\nonumber
\\%
\nabla_\nu \tilde{\textfrak{g}}^{\nu\mu}&=-\alpha\sqrt{-g} B^\mu\,.\nonumber
\end{align}
For the harmonic gauge $\al=0$ this amounts to
\begin{align*}
-\tilde{R}_{\nu\lambda}&=i\Big(\nabla_{(\nu}\bar{C}_{|\mu|}\nabla_{\la)}C^\mu+\nabla_\mu\bar{C}_{(\nu}\nabla_{\la)}C^\mu+(\nabla_\al\nabla_{(\nu}\bar{C}_{\la)})C^\al\Big)+\nabla_{(\la}B_{\nu)}\,,\\%
\tilde{\textfrak{g}}^{\nu\lambda}\nabla_{\nu}\nabla_{\lambda} C^\mu&=\tilde{\textfrak{g}}^{\alpha\nu}R_{\lambda\nu\alpha}^{\quad\ \mu}C^\lambda\,,\\%
\tilde{\textfrak{g}}^{\nu\lambda}\nabla_{\nu}\nabla_{\lambda} \bar{C}_\mu&=-\tilde{\textfrak{g}}^{\nu\lambda}R_{\nu\mu\lambda}^{\quad\ \alpha}\bar{C}_\alpha\,,\\%
\nabla_\nu \tilde{\textfrak{g}}^{\nu\mu}&=0\,.
\end{align*}
In the Minkowski background this system simplifies to
\begin{align*}
-\tilde{R}_{\nu\lambda}&=i\Big(\partial_{(\nu}\bar{C}_{|\mu|}\partial_{\la)}C^\mu+\partial_\mu\bar{C}_{(\nu}\partial_{\la)}C^\mu+(\partial_\al\partial_{(\nu}\bar{C}_{\la)})C^\al\Big)+\partial_{(\la}B_{\nu)}\,,\\%
\Box_{\tilde{g}}C^\mu&=0\,,\\%
\Box_{\tilde{g}}\bar{C}_\mu&=0\,,\\%
\partial_\nu \tilde{\textfrak{g}}^{\nu\mu}&=0\,.
\end{align*}
%
\part{Quantum field theory}\label{quant} %
In this part, we show that using the classical structures discussed before one can gain some insight into the quantum theory. The interpretation of the BV complex formulated in the language of infinite dimensional geometry allows us to treat the BV quantization as a deformation of the geometrical structure resulting from the deformation of the pointwise product. Using this method, we provide the definition of the renormalized quantum BV operator. The problem in incorporating the renormalization into the BV formalism is present since the first papers of Batalin and Vilkovisky \cite{Batalin:1981jr,Batalin:1983wj,Batalin:1983jr}. The gauge invariance of the path integral measure is in this formalism equivalent to the \textit{quantum master equation}. It is similar to the classical one but contains an additional contribution from a functional differential operator $\Lap=\frac{\delta}{\delta \ph^\ddagger(x)\delta \ph(x)}$, which can be seen as a divergence of vector fields on the configuration space. The quantum master equation (\textsc{qme}) reads: $\frac{1}{2}\{S,S\}-i\hbar\!\Lap\! S=0$. It was clear from the beginning, that the operator $\Lap$ is a very singular object and the \textsc{qme} is not well defined. Nevertheless, it can be used for formal manipulations in the path integral quantization of gauge theories, as long as one gives up on the geometrical meaning of the structure. In \cite{Batalin:1983jr} authors comment on this problem pointing out the existence of divergences in higher loop order in \textsc{qme}:
\begin{center}
\begin{minipage}{12cm}
\it
The solution of Eq. (3.7) \textrm{[quantum master equation]} can be expanded in powers of $\hbar$
\[
W=S+\sum\limits_{p=1}^\infty\hbar^pM_p.\qquad\qquad\qquad\qquad\qquad\qquad\qquad\quad\ (3.11)
\]
This gives 
\begin{align*}
&p =0:&\ \{S,S\}=0,\qquad\qquad\qquad\qquad\qquad\qquad\qquad(3.12)\\ 
&p =1:&\{M_1,S\}=i\Lap\! S,\qquad\qquad\qquad\quad\qquad\qquad\quad (3.13)\\
&p\geq 2:&\{M_p,S\}=i\Lap\! M_{p-1}-\frac{1}{2}\sum\limits_{q=1}^{p-1}\{M_q,M_{p-q}\}\qquad\, (3.14)
\end{align*}
\end{minipage}
\begin{minipage}{12cm}
\it (\ldots) However, in a local basis of the gauge algebra the right-hand sides of Eqs. (3.13) and (3.14) are proportional to $\delta(0)$. In the framework of a regularization which annihilates such divergences one may put $M_p=0, p\geq 1$.
\end{minipage}
\end{center}
As seen from  the above quote, the first idea to deal with the divergences was to apply some regularization scheme that puts all the terms proportional to $\delta(0)$ to 0 and then perform the renormalization. In \cite{Troost} it was proposed to use instead a regularization that puts the divergent terms of \textsc{qme} at finite non-zero values. This approach allowed the analysis of the anomalies in a more systematic way and their relation to obstructions in fulfilling the \textsc{qme}. 
The regularization used in \cite{Troost} is the Pauli-Villars scheme and the discussion is restricted only to the 1-loop order. A method valid for higher loop orders was proposed in \cite{Pa95}, but the regularization  scheme used there is non-local. The dimensional regularization and renormalization in the context of BV formalism were discussed in \cite{Tonin}. The BPHZ renormalization is discussed in \cite{JPT}. All of the mentioned approaches rely on some regularization scheme and involve arbitrary choices. From the conceptual point of view it is still unclear how the \textsc{qme} should be interpreted in the renormalized theory. An alternative treatment of \textsc{qme} which involves certain extension of the field-antifield formalism was presented in \cite{Barnich}.

As seen from the above review, the status of the renormalized \textsc{qme} is still in the focus of ongoing research. Our aim is to show that it arises naturally in the framework of perturbative algebraic field theory, without referring to a specific regularization scheme. As mentioned at the beginning of this section, we need firstly to deform the classical geometric structure of the BV complex. The idea of deformation quantization goes back to Bayen, Flato, Fronsdal, Lichnerowicz and Sternheimer \cite{BFFLS} and the first 
attempt to use these structures in quantum field theory is due to Dito \cite{Dit90}.
Based on these ideas Brunetti, D\"utsch, and Fredenhagen \cite{Duetsch:2000nh,Dutsch:2004hu,DF,DF04,DF02,BreDue,BDF} developed a formalism, which uses a purely algebraic formulation of perturbative QFT and treats the renormalization on a very general level. An appealing feature of this formalism is the idea to treat the renormalization procedure in terms of certain differential operators on the infinite dimensional space of functionals. This approach can be related to the path-integral formalism, but it is formulated on the purely algebraic level.
 In the present work we follow this idea and show in particular that the renormalized time-ordered product
is a full product on a certain domain, which is invariant under the renormalization. The BV differential defined with respect to this product provides a mathematically rigorous interpretation of the quantum BV operator. We will show that the renormalized \textsc{qme} can be related to the Master Ward Indentity (\textsc{mwi}) postulated in \cite{DueBoas} in the context of p\textsc{aqft} and analyzed in detail in \cite{DF02}, where it was proved to hold also in the classical theory.
\chapter{Batalin-Vilkovisky formalism in p\textsc{aqft} }\label{pAQFT}
\vspace{-5ex}
\begin{flushright}
 \begin{minipage}{10cm}
\textit{Zasadniczą cechą dzieł sztuki jest jedność w wielości bez względu na to, jakie są elementy tej wielości i jakim sposobem osiągnięta jest ich jedność. Tę właśnie cechę jedności nazywamy pięknem danego tworu; pojęcie to może stosować się i do dzieł sztuki, i do innych tworów (...).}
\begin{flushright}
S.I. ``Witkacy'' Witkiewicz
\end{flushright}
\textit{The main feature of artistic creations is their unity in multiplicity, no matter what the elements of this multiplicity are and how the unity is achieved. This feature of unity is called the beauty of a given artistic creation; this notion applies not only to artistic creations, but to all human creations in general (...).}
\begin{flushright}
S.I. `Witkacy'' Witkiewicz (translation K.R.)
\end{flushright}
 \end{minipage}
\end{flushright}
\vspace{5ex}
\noindent\rule[2pt]{\textwidth}{1pt}
\vspace{1ex}\\
Here we come to the final chapter of this thesis and to the end of our metaphorical journey. Some questions that we posed at the beginning  will now be answered and after a long struggle through the technical subtleties of mathematics and physics we arrived at the point where it all somehow gets its purpose and justification. It is always gratifying in science to see how a multitude of threads and concepts converge together to a unified idea. This feeling that everything has its right place is certainly one of the features of the mathematical beauty of a theory. In this sense it is impossible not to admire, how the quantum field theory is built on the fundament of rich mathematical and experimental background. After almost 100 years' history it still undergoes development and provides us with surprises. In particular the QFT on curved spacetimes is a relatively new research field and a lot has been done in it in recent years (for example \cite{BFK95,BF0,HW,HW2,HW3,HW5,H,DHP,Thomas,BDF,Chi}). The major step was done in the seminal paper of Radzikowski \cite{Rad}, who recognized, that the microlocal spectrum condition allows us to define a class of states on  generic globally hyperbolic spacetimes, that behave similar to the vacuum in Minkowski spacetime. Based on this observation, the Wick polynomial algebra has been constructed in \cite{BFK95,BF0,HW,HW5}. It was done by directly encoding the Wick theorem in the algebraic structure. This approach, called \textit{deformation quantization} has been introduced in \cite{BF0} and further developed in \cite{DF,BDF}. The interaction is introduced there with the use of Epstein-Glaser \cite{EG} causal perturbation theory and  the whole approach is commonly called \textit{perturbative algebraic quantum field theory} ({\paqft}). 
 
There are of course other frameworks for {\qft} aiming at mathematically rigorous results. Some of them cover only specific problems, without going into the deep conceptual roots, but all these ideas influence each other and constitute together a remarkable construction. Sometimes one finds a link between different approaches and each such result can lead to completely new developments. Concerning causal perturbation theory and the BV formalism, a first example of combining technics of both frameworks is due to Hollands and appeared in \cite{H}, in the context of Yang-Mills theories. There were however still some open questions concerning the general structure. The BV framework was originally associated with the path integral formalism and up to now there were no attempts to develop it intrinsically within {\paqft}. The present thesis proposes a solution of this problem. In the previous chapters we prepared all the mathematical tools and now finally we can show how the BV quantization can be understood from the point of view of {\paqft}. Following the idea of deformation quantization, we start with the classical structure and deform the pointwise product into the noncommutative one. All the classical structures we discussed in the previous chapters can be now deformed into the quantum ones in a systematical way. In particular the geometric construction of the BV complex can be easily translated to the quantized theory. Before we move to this task, we recall in the next section basic definitions and results of perturbative algebraic quantum field theory.
\section{General structure}\label{general}
\subsection{Algebraic formulation of perturbative QFT}\label{algQFT}
We start with a short  review of the {\paqft} formalism. Ideas presented in this subsection come mainly from \cite{BDF,Duetsch:2000nh,DF,DF04,Kai}. In causal perturbation theory one starts with the free action and then introduces the interaction in a perturbative way, following the prescription of Bogoliubov \cite{BS}.  For concreteness we will use the example of the free minimally coupled scalar field with the generalized Lagrangian given by (\ref{Lscalar}). In the algebraic approach one focuses on the abstract structure of the theory. This is what was done in the classical case (chapter \ref{classfunc}), where we fixed the specific theory by 
defining a certain Poisson algebra. Analogously, in the quantum case our aim is to construct an involutive associative algebra ($*$-algebra), or rather a local net of algebras, in the spirit of local quantum physics. We use as a starting point the same topological space as in the classical theory, namely $\F_\mc(M)$. At the beginning, to avoid technical difficulties, we restrict ourselves only to regular functionals. Let $\F_{\reg}(M)\subset\F_\mc(M)$ be the space of functionals whose derivatives are smooth sections. This space is already equipped with the commutative pointwise product $\cdot$, but to obtain a quantized structure we need a noncommutative one. To this end we use the deformation quantization scheme. The space of formal power series in $\hbar$ with coefficients in $\F_{\reg}(M)$ will be denoted by $\F_{\reg}(M)[[\hbar]]$. On this space we can define a noncommutative $\star$-product\index{product!star}:
\be\label{star product}
F\star G\doteq m\circ \exp({i\hbar \DC})(F\otimes G) \ ,
\ee
where $m$ is the pointwise multiplication and $\DC$ is the functional differential operator
\begin{equation}\label{star product2}
\DC\doteq\frac{1}{2}
                  \int dx\, dy \De(x,y)\frac{\de}{\de\varphi(x)}\otimes\frac{\de}{\de\varphi(y)},\qquad \De=\De_R-\De_A\,.
\end{equation}
To simplify the notation we write $dx$ instead of $\ \dvol(x)$ whenever the choice of the integration measure is clear from the context. 
We can view the definition of the $\star$-product as a pullback of the pointwise product $\cdot$ by means of the following diagram:
\begin{center}
\begin{tikzpicture} \matrix(a)[matrix of math nodes, row sep=2em, column sep=0.4em, text height=2.5ex, text depth=0.4ex] {
\F_\reg(M)[[\hbar]]^{\otimes 2}          &&\F_\reg(M)[[\hbar]]^{\otimes 2}\\
&\F_\reg(M)[[\hbar]]&\\};
\path[->,font=\scriptsize] (a-1-1) edge node[above]{$\exp(i\hbar\DC)$} (a-1-3); 
\path[->,font=\scriptsize](a-1-1) edge[dashed] node[left]{$\star$} (a-2-2); 
\path[->,font=\scriptsize](a-1-3) edge node[right]{$\cdot$} node[right]{} (a-2-2);  
\end{tikzpicture}
\end{center}
The complex conjugation satisfies the relation:
\be
\overline{F\star G}=\overline{G}\star\overline{F}\,.
\ee
Therefore we can use it to define an involution  $F^*(\ph)\doteq\overline{F(\ph)}$. The resulting structure is a $*$-algebra $(\F_{\reg}(M)[[\hbar]],\star)$, which provides the quantization of $(\F_{\reg}(M),\{.,.\}_S)$. For the example of a scalar field we obtain the following commutation relations:
\[
[\Phi(f),\Phi(g)]_{\star}=i\hbar \langle f,\De g\rangle \ , \quad f,g\in \D(M)\,,
\]
where $\Phi(f)(\ph)\doteq \int\! f \ph\,\  \dvol $ is a smeared field. Compare it with the corresponding Poisson bracket $\{\Phi(f),\Phi(g)\}_S=\langle f,\De g\rangle$, constructed in \ref{scal}. The algebra $(\F_{\reg}(M)[[\hbar]],\star)$ contains an ideal generated (with respect to $\star$) by the free equations of motion. We can obtain the standard algebra of the free scalar field after performing the quotient by this ideal. Nevertheless it is more convenient not to do it in the beginning and to work in the off-shell formalism instead.

Up to now this was all free field theory. Now we want to introduce the interaction. To this end we need yet
another algebraic structure, namely the \textit{\textbf{time-ordered product}} $\T$\index{product!time ordered}. Following \cite{BDF} we define it by means of the time ordering operator $\TT$.
\[
\TT(F)\doteq e^{i\hbar\DD}(F)\,,
\]
where $\DD=\int dxdy \Delta_D(x,y)\frac{\delta^2}{\delta\ph(x)\delta\ph(y)}$ and $\Delta_D=\frac{1}{2}(\Delta_R+\Delta_A)$ is the Dirac propagator. The antitime-ordering $\overline{\Tcal}$ is the inverse of $\Tcal$ and is defined by replacing $\Delta_D$ with $-\Delta_D$. With these operators we define the time-ordered product $\T$ on $\Dcal_{\TT}(M)\doteq\TT(\F_\reg(M)[[\hbar]])$ by:
\[
F\T G\doteq \Tcal(\Tcal^{\minus}\cdot\Tcal^{\minus}G)
\]
The time-ordered product is associative, commutative and even equivalent to the pointwise product! This is a good news, because it can serve as a mean to carry over the classical structures to the quantum world. We will make use of this feature to bring the BV complex to the quantum level. The time ordered product can be also seen as a pullback by the diagram:
\begin{center}
\begin{tikzpicture} \matrix(a)[matrix of math nodes, row sep=2em, column sep=3em, text height=2.5ex, text depth=0.3ex] {
\F_\reg(M)[[\hbar]]^{\otimes 2}&\Dcal_{\TT}(M)^{\otimes 2}\\
\F_\reg(M)[[\hbar]]&\Dcal_{\TT}(M)\\};
\path[->,font=\scriptsize] (a-1-1) edge node[above]{$\TT^{\otimes 2}$} (a-1-2); 
\path[->,font=\scriptsize] (a-2-1) edge [dashed] node[above]{$\TT$} (a-2-2); 
\path[->,font=\scriptsize](a-1-1) edge node[left]{$\cdot$} (a-2-1); 
\path[->,font=\scriptsize](a-1-2) edge node[right]{$\T$} node[right]{} (a-2-2);  
\end{tikzpicture}
\end{center}
Formally the time ordering operator $\TT$ may be understood as the operator of convolution with the oscillating Gaussian measure with covariance $i\hbar\Delta_D$,
\be\label{path integral}
\TT F(\ph)\form \int d\mu_{i\hbar\Delta_D}(\phi)F(\ph-\phi) \ . 
\ee
To see that $\T$ is indeed the time ordered product for $\star$ we can look at the example of the scalar field. In this case we have:
\begin{align*}
\Phi(f)\T\Phi(g)&=\Phi(f)\cdot\Phi(g)+ \frac{i\hbar}{2}\langle f, (\De_R+\De_A)g\rangle\,,\\
\Phi(f)\star\Phi(g)&=\Phi(f)\cdot\Phi(g) + \frac{i\hbar}{2}\langle f,(\De_R-\De_A)g\rangle\,.
\end{align*}
From the support properties of $\Delta_R$ and $\Delta_A$ it follows that $\Phi(f)\T\Phi(g)=\Phi(f)\star\Phi(g)$ if the support of $f$ is later than the support of $g$, i.e. $\supp f\gtrsim\supp g$. The time ordered product provides us with means to introduce the interaction using the local S-matrices. For localized $V\in\TT(\F_\reg(M))$ the formal S-matrix is defined as the time-ordered exponential:
\[
\Scal(V)\doteq e_{\sst{\TT}}^V=\TT(e^{\TT^{\minus}V})\,.
\]
Note that it is not meant to be understood as an actual scattering matrix, but as a generating functional for the higher order time-ordered products:
\be\label{Smatrix}
\Scal(V)=\sum_{n=0}^\infty\frac{1}{n!}V\T...\T V\equiv\sum_{n=0}^\infty\frac{1}{n!}\TT^n(V^{\otimes n})\,.
\ee
We can now define the relative S-matrix by the formula of Bogoliubov:
\be\label{Bog}
\Scal_V(F)\doteq\Scal(V)^{\star-1}\star \Scal(V+F)\,.
\ee
$\Scal_V$ plays a role of an intertwining M{\o}ller map\index{Moller map@M{\o}ller map!quantum} between the free and the  interacting theory. You can compare it with the classical version of M{\o}ller maps introduced in section \ref{inter}. 
Interacting observables can be now obtained from $\Scal_V(F)$,  which serves as a generating functional. 
We will discuss it in detail in section \ref{renorm}. Up to now we only defined the time-ordered products for regular functionals. This is not enough to introduce a ``sensible'' interaction, since typical interaction terms are local nonlinear functionals and the expression $\Scal(V)$ would be ill-defined. To amend this, one has to carefully extend $\Scal(V)$ to more general objects. It was shown in \cite{EG} that this can be done by means of causal perturbation theory but the extension is not unique. Its ambiguity corresponds to the renormalization freedom. 
We discuss the problem of renormalization in section \ref{renorm} but for now we want to stay for a while in the realm of non-renormalized time-ordered products to look closer at the algebraic structure. Interacting quantum fields are generated by $\Scal_{iV/\hbar}(F)$ and we can write them as formal power series:
\be\label{Rv}
\frac{d}{d\lambda}\Big|_{\lambda=0}\Scal_{iV/\hbar}(\lambda F)=\sum\limits_{n=0}^\infty\frac{1}{n!} R_{n,1}(V^{\otimes n},F)\equiv R_V(F)\,,
\ee
where the coefficients $R_{n,1}$ are called \textit{\textbf{retarded products}}\index{retarded product!quantum} and are obtained from (\ref{Bog})
by differentiation:
\begin{align*}
\Big(\frac{\hbar}{i}\Big)^n\, R_{n,1}(V^{\otimes n},F)&=\frac{d}{d\lambda}\Big|_{\lambda=0}\Scal(V)^{\star-1}\star \Scal(V+\lambda F)=\\
&=\frac{d}{d\lambda}\Big|_{\lambda=0}\sum\limits_{k=0}^n\overline{\TT}(V^{\otimes k})\star \TT((V+\lambda F)^{\otimes (n-k)})
\end{align*}
You can compare this definition with the classical counterpart (\ref{retarded:class}). In the classical case the higher ordered retarded products were coefficients of the expansion of a M{\o}ller map $r_{F}$ in powers of $\lambda$. The same structure appears now in the quantum case. More explicitly the intertwining map $R_V$ can be written as
\[
R_V(F)=\left(e_{\sst{\TT}}^{iV/\hbar}\right)^{\star\minus}\star\left(e_{\sst{\TT}}^{iV/\hbar}\T F\right)\,.
\]
When we switch on the interaction, also the star product has to change. A natural definition can be obtained with the use of the intertwining map $R_V(F)$. Following \cite{F11} we define the interacting star product as (compare it with the classical counterpart (\ref{intertwine:class})):
\be\label{interacting:star}
F\star_V G\doteq R_V^{\minus}\left( R_V(F)\star R_V(G)\right)\,,
\ee
where the inverse of $R_V$ is given by:
\[
R_V^{\minus}(F)=e_{\sst{\TT}}^{-iV/\hbar}\T\left(e_{\sst{\TT}}^{iV/\hbar}\star F\right)\,.
\]

Now we want to relate algebraic formulas we are using to a more ``standard'' formulation of QFT. This is particularly interesting in the context of the BV formalism, which was developed originally in the path integral approach. For the moment let $M=\M$, the Minkowski spacetime.
Before we discuss the relation between the algebraic and path integral formalisms, we want to make a small comment. 
In the definition of the time-ordering operator, instead of $\Delta_D$ we could also use the Feynman propagator $\Delta_F=i\Delta_D+\Delta_1$, where $\Delta_1$ is the symmetric part of the Wightmann 2-point function $\Delta_+$. In fact one can take a Hadamard solution, that satisfies certain conditions, in place of  $\Delta_1$ and the algebraic structure will still be equivalent \cite{BDF}. We discuss this in detail in section \ref{quantalg}. For the comparison with the path integral formalism it is convenient to use for a moment $\Delta_F$ in the definition of time-ordered products and correspondingly $\Delta_+=\frac{i}{2}\Delta+\Delta_1$ for the $\star$-product.
Let $\omega_0(F)\doteq F(0)$, $F\in\F_\reg(M)$. This functional plays a role of the vacuum state. 
As we already mentioned, the expectation value in $\omega_0$ is related to the formal expression of integrating with the Gaussian measure, therefore:
\be\label{omV}
\omega_0(\Scal(iV/\hbar))\form \int d\mu_{i\hbar\Delta_F}e^{:V:}\form \int\! D\phi\, e^{\frac{i}{\hbar}(S+:V:)}\,.
\ee
Note that in the present framework the normal ordering is implemented by applying $\TT^{\minus}$ to $V$. 

Assume that we can take the adiabatic limit, so in this limit $\Scal$ is translation invariant and we have a unique vacuum. Then (\ref{Bog}) relates to the Gell-Man-Law formula, since 
$\omega_0(\Scal(V)\star\Scal_V(F))=\omega_0(\Scal(V))\omega_0(\Scal_V(F))$ and we can write formally:
\[
\omega_0(\Scal_V(F))\form\frac{\int d\mu_{i\hbar\Delta_F}e^{:V:}:F:}{ \int d\mu_{i\hbar\Delta_F}e^{:V:}}\,.
\]
With this little dictionary at hand we can now set to understand the BV quantization in the formalism of \paqft.
\subsection{BV quantization in the algebraic framework}\label{BValg}
In the previous section we were only reviewing the results on algebraic quantization available elsewhere in the literature. Finally it's 
time to use these tools to tackle a new problem. The main result of the present thesis is the formulation of the BV quantization within the framework of {\paqft} and now we can outline the underlying idea of our construction. The ideas presented here are the fundament for a completely new way to look at the BV complex in the context of {\paqft}. We will elaborate on this concept in the rest of the present chapter. 

In order to perform the construction of the BV complex we need to extend the algebra of functionals with its derivations, i.e. vector fields. As in the classical case this extension is provided by the space $\V_\mc(M)$ of microcausal maps from $\E(M)$ to $\E_c(M)$. By a slight abuse of notation we call it \textit{the space of microcausal vector fields}\index{microcausal!vector fields}. Now we want to extend the time ordering operator $\TT$ to vector fields. For the moment we restrict ourselves to the regular ones. We define $\V_\reg(M)$ to be the space of vector fields such that for $n>0$ the functional derivative $X^{(n)}(\ph)$  is a test section in $\Gamma_c(M^{n+1},V^{\otimes(n+1)})$. In particular this excludes all the local non-constant vector fields since in this case $X(\ph)(x)$ is a function of the infinite jet of $\ph$, i.e.: $X(\phi)(x)=f(j^\infty_x(\ph))$ so the functional derivative is proportional to the derivatives of delta distribution.
It was already discussed in section \ref{Vfonconfspace} that vector fields can be seen from two viewpoints: on one hand as derivations of $\F(M)$ and on the other as sections of the tangent bundle, i.e. maps from $\E(M)$ to $\E_c(M)$. These two roles played by vector fields have their consequences for the definition of the time ordering operator on $\V_\reg(M)$. Indeed, if we think of an element $X\in \V_\reg(M)$ as a section, then $\TT$ acts on it simply as a differential operator and we can put forth a following definition:
\be\label{timeordvf}
\TT X\doteq \int\!dx\, \TT (X(x))\frac{\delta}{\delta \ph(x)}
\ee
In section \ref{Vfonconfspace} we denoted by $\partial_X$  the derivation of $\F_\reg(M)$ corresponding to the vector field $X$. It is now natural to ask how this derivation transforms. The guiding principle for all our definitions is to use $\TT$ as a mean to transport the classical structure to the quantum algebra. In this spirit we can associate with $Y\in\TT(\V_\reg(M))$ a derivation of $\TT(\F_\reg(M))$ defined as
\be\label{timeorder}
\partial^{\sst{\TT}}_Y F=\mathcal{T}\langle \mathcal{T}^{-1}Y,\mathcal{T}^{-1}F^{(1)} \rangle\,\qquad F\in\TT(\F_\reg(M))\,.
\ee
From the above formula it is evident that $\partial^{\sst{\TT}}_Y$ is a derivation of $\TT(\F_\reg(M))$ with respect to the time ordered product $\T$:
\be
\partial^{\sst{\TT}}_Y(F\T G)=(\partial^{\sst{\TT}}_Y F)\T G+F\T(\partial^{\sst{\TT}}_Y G)\,,
\ee
Moreover we obtain a following identity:
\be
\partial^{\sst{\TT}}_{\mathcal{T}X}=\mathcal{T}\circ\partial_X\circ\mathcal{T}^{-1}
\ee
The construction we performed shows that we can recover in a natural way all the classical structures of the BV complex in the quantum algebra, but they are defined with respect to the the time-ordered product, not $\star$. The advantage of using $\T$ is that it is still a commutative product (in contrast to $\star$), so the Koszul-Tate complex makes sense in this case. 
The graded algebra of antifields is transformed into $\TT(\Lambda\V_\reg(M))$. This algebra is equipped with the Schouten bracket $\{.,.\}_{\TT}$ defined as:
\be
\{X,Y\}_{\sst{\TT}}=\mathcal{T}\{\mathcal{T}^{-1}X,\mathcal{T}^{-1}Y\}
\ee
Again we see that this is the graded extension of the commutator of derivations and the evaluation of a derivation on a functional in $\TT(\F_\reg(M))$, since it holds:
\begin{align}
\partial^{\mathcal{T}}_{\{X,Y\}_{\sst{\TT}}}&=[\partial^{\sst{\TT}}_X,\partial^{\sst{\TT}}_Y]\\
\{X,F\}_{\sst{\TT}}&=\partial^{\sst{\TT}}_X F
\end{align}
Now we want to see how the ideal generated by the equations of motion is transforming under the time ordering. We identify it as the image of the time-ordered Koszul operator:
\be
\delT=\mathcal{T}\circ\delta_{\mathcal{T}^{\minus}S}\circ\mathcal{T}^{-1}\,,
\ee
where $S\in\TT(\F_\reg(M))$. We have just characterized the classical ideal of equations of motion, but what about the quantum one? There is also a nice way to describe it, but before we turn to this task, we need one more definition. We already defined the $\T$-product of antifields, but we need also the $\star$. The definition is quite natural if we treat vector fields as functions $\E(M)\rightarrow\E_c(M)$ and apply to them the operator $\exp({i\hbar \DC})$ defined by (\ref{star product}) and (\ref{star product2}). We spell out this definition explicitly. Let $X,Y\in\V_\reg(M)$. Then we define
\be\label{star:product:v}
X\star Y\doteq \exp({i\hbar \DC})(X\wedge Y) \ ,
\ee
where $\DC$ is the functional differential operator defined in (\ref{star product2}). We can also write it in terms of the vector field ``coefficients'' as:
\be\label{star:product:v2}
(X\star Y)(x,y)=X(x)\star Y(y)-X(y)\star Y(x)\,.
\ee
From the above formula we can see that $X\star Y$ is an antisymmetric compactly supported function on $M^2$ with values in the algebra $(\F_\reg(M),\star)$. 
Let as now have a closer look at the image of $\delT$. Acting  on a time-ordered vector field $X\in\TT(\V_\reg(M))$ with $\delT$ we obtain
\[
\delT(X)=\TT(\delta_{\mathcal{T}^{\minus}\!S}(\TT^{\minus} X))=m \circ e^{i\hbar\DDp}\left(\int\! dx\, X(x)\otimes \frac{\delta S}{\delta \ph(x)}\right)\,,
\]
where in the second step we used the Leibniz rule. Since $S$ is a functional of second order in $\ph$, the expansion of $e^{i\hbar\DD}$ has only two nontrivial terms and we finally obtain:
\begin{align}
\delT(X)&=\int\! dx\,X(x)\frac{\delta S}{\delta \ph(x)}+i\hbar \int\! dxdydz\,\frac{\delta X(x)}{\delta\ph(y)}\frac{\delta^2 S}{\delta \ph(x)\delta\ph(z)}\Delta_D(y,z)=\nonumber\\
&=\int\! dx\,X(x)\frac{\delta S}{\delta \ph(x)}+i\hbar \int\!dx\, \frac{\delta X(x)}{\delta\ph(x)}=\nonumber\\
&=\delta_S(X) +i\hbar\Lap\! X\label{QMO}\,,
\end{align}
where $\Lap$ is a map that acts on regular vector fields $\V_\reg(M)$ like a divergence\footnote{This operator is in the literature denoted by $\Delta$, but we use here a slightly different symbol $\Lap$, to distinguish it from the causal propagator $\Delta(x,y)$.}:
\[
\Lap Q\doteq \int\! dx\,\frac{\delta^2 Q}{\delta\ph^\ddagger(x)\delta\ph(x)}=\int\! dx\,\frac{\delta Q(x)}{\delta\ph(x)},\qquad Q\in\V_\reg(M)\,.
\]
This operator can be extended also to multi-vector fields $\La\V_\reg(M)$ in such a way that it becomes a differential, i.e. $\Lap^2=0$ is fulfilled. Explicitly we can write $\Lap$ as:
\[
\Lap Q=(-1)^{(1+|Q|)}\int dx \frac{\delta^2 Q}{\delta\ph^\ddagger(x)\delta\ph(x)},\qquad Q\in\Lambda\V_\reg(M)\,.
\]
It has also some nice properties with relation to the antibracket. For example it holds:
\be\label{Delta:bracket}
\{P,Q\}= \Lap(PQ)-\Lap(P)Q-(-1)^{|P|}P\Lap\!(Q)\,,
\ee
where $P,Q\in\La\V_\reg(M)$.
Moreover, using (\ref{Delta:bracket}) and the nilpotency of $\Lap$, one can show that:
\be\label{Delta:bracket2}
\Lap\{P,Q\}=-\{\Lap(P),Q\}-(-1)^{|P|}\{P,\Lap\!(Q)\}\,.
\ee
The graded algebra $\La\V_\reg(M)$ together with the antibracket $\{.,.\}$ and the differential $\Lap$ form a structure, which is called in mathematics the BV-algebra\index{BV!algebra}.

Note that since the time ordering commutes with both derivatives $\frac{\delta}{\delta\ph(x)}$ and $\frac{\delta}{\delta\ph^\ddagger(x)}$, it also commutes with $\Lap$. Hence we obtain
\be\label{Delta:Tbracket}
\{X,Y\}_{\TT}= \Lap(X\T Y)-\Lap(X)\T Y-(-1)^{|X|}X\T \Lap(Y)\,,
\ee
where $ X,Y\in\TT(\La\V_\reg(M))$.
Now we can come back to the problem of comparing the quantum and the classical ideal of {\eom}'s. To see the relation between them, we use the fact that 
\be\label{identity:S:star}
\int\! dx\,X(x)\cdot\frac{\delta S}{\delta \ph(x)}=\int\! dx\, X(x)\star\frac{\delta S}{\delta \ph(x)}\,,
\ee
 and we can rewrite (\ref{QMO}) as:
\be\label{tkoszul}
\delT(X)=\int\! dx\,X(x)\star\frac{\delta S}{\delta \ph(x)}+i\hbar\! \Lap\!(X)\,.
\ee
In this formula both the time-ordered and the $\star$-product appear 
and it is natural to ask, if there is a $\star$-transformed version for the antibracket. In analogy to (\ref{Delta:bracket}) and (\ref{Delta:Tbracket}) we can define it as\footnote{Note that this is \textit{not} a Poisson bracket, essentially because $\star$ is not graded commutative. Nevertheless $\{.,Y\}$ defines a derivation with respect to $\star$ if $\frac{\delta Y}{\delta\ph(x)}$ and $\frac{\delta Y}{\delta\ph^\ddagger(x)}$ are central.}:
\be\label{Delta:star:bracket}
\{X,Y\}_{\star}= \Lap(X\star Y)-\Lap(X)\star Y-(-1)^{|X|}X\star \Lap(Y)\,.
\ee
The sign rule in (\ref{Delta:star:bracket}) was chosen to be consistent with formulas (\ref{antibracketformal}) and (\ref{antibracketformal2}).
This can also be written as:
\be\label{antibracketstar}
\{X,Y\}_{\star}=-\int\! dx\!\left(\!\frac{\delta X}{\delta\ph(x)}\star\frac{\delta Y}{\delta\ph^\ddagger(x)}+(-1)^{|X|}\frac{\delta X}{\delta\ph^\ddagger(x)}\star\frac{\delta Y}{\delta\ph(x)}\!\right)\,,
\ee
In this new notation we can write (\ref{tkoszul}) as:
\be\label{tkoszul2}
i\hbar\! \Lap\!(X)=\{X,S\}_{\TT}-\{X,S\}_\star\,.
\ee
This formula provides us with a nice interpretation of $\Lap$ as an operator describing the difference between the classical ideal of equations of motion represented by the image of $\{.,S\}_\TT$ and the quantum one, characterized as the image of  $\{.,S\}_\star$.
Using the identity (\ref{identity:S:star}) it is easy to see that the operator $\{.,S\}_\star$ is a derivation with respect to the $\star$-product. We can view it as the quantum Koszul map\index{Koszul!map!quantum} of the free action. 
The fact that $\{.,S\}_\TT$  and $\{.,S\}_\star$ differ  by a $\hbar$-order term corresponds to the Schwinger-Dyson type equations. 
This is of course something we expected, since quantum theory should be distinguishable from the classical one. In the following sections we will analyze a similar structure for the renormalized time ordered products. This involves many technical subtleties but the guiding principle is the same as in this section. We start with the classical structure and then construct the quantum algebra, where two products are defined, $\star$-product and the time-ordered product. Various relations between these two  products result also in relations between certain ideals and differential operators arising from the BV structure.
\subsection{Quantum master equation and the quantum BV operator}\label{algQFT}
The really interesting story in the BV quantization starts when we consider an action with symmetries. We shall do it on the example of Yang-Mills theories in section \ref{YMquant}. It will involve some technicalities, so in order to keep
focused on the essential structure we want to discuss some preliminary steps already on the level of nonrenormalized time-ordered product. To this end we consider $\BV_\reg(M)$, a subspace of the BV-complex (\ref{BVYM}) consisting of graded functionals and derivations that are regular. 
\begin{rem}
The algebra $\BV_\reg(M)$ contains also functionals of fermionic fields, so some additional signs appear in formulas used in the previous section. The operator $\Lap$ in the graded case is defined as:
\[
\Lap X=\sum\limits_\alpha(-1)^{|\ph_\al|(1+|X|)}\int dx \frac{\delta^2 X}{\delta\ph^\ddagger(x)\delta\ph(x)}\qquad X\in\BV_\reg(M)\,,
\]
where $|.|$ denotes the ghost number $\#\gh$ and $\alpha$ runs through all the field configuration types of the theory, i.e. physical fields, ghosts, antighosts, etc. The antibrackets $\{.,.\}_{\TT}$ and  $\{.,.\}_{\star}$ are now simply given by formulas (\ref{Delta:Tbracket}), (\ref{Delta:star:bracket}) with the graded version of $\Lap$ defined above.
\end{rem}
With these structures at hand we want now to discuss the gauge fixing. Our starting point is a classical Lagrangian, where a suitable canonical transformation (\ref{gfermion}) was performed, so that the term of $\#\ta=0$, quadratic in fields, induces a normally hyperbolic system of equations. We denote it again by $S$. This is the free part of the Lagrangian and we use it to define the free time-ordered product $\TT$. The interacting part of our action has to be chosen with some cautiousness. We don't want to use the transformed Yang-Mills Lagrangian (\ref{extendYM}) yet, since it is local and therefore its nonlinear part is not an element of $\BV_\reg(M)$. Instead we consider for the moment some other functional $V\in\TT(\BV_\reg(M))$ with ghost number $\#\gh=0$ which also contains antifields. The natural question to ask is, what will happen, if we change the gauge-fixing fermion. In other words we want to perform again a canonical transformation $\alpha_\psi$ and see how the structure is changing. We choose the new gauge-fixing fermion $\psi$ as an element of $\TT(\BV_\reg(M))$ with $\#\gh=-1$. Assume that $\psi$ doesn't contain antifields. Just like in the classical case, first we define an automorphism (\ref{gfermion}) of the algebra $\TT(\BV_\reg(M))$ by
\be\label{gfermionq}
\alpha_{\psi}(X):=\sum_{n=0}^{\infty}\frac{1}{n!}\underbrace{\{\psi,\dots,\{\psi}_n,X\}_{\TT}\dots\}_{\TT}=\TT(\alpha_{\TT^{\minus}\psi}(\TT^{\minus}X))\,,
\ee
In the second step of the gauge fixing procedure we set all the elements with $\#\ta>0$ in $\TT(\BV_\reg(M))$ at 0. Following the standard approach to BV-quantization (see for example \cite{Henneaux:1992ig}) we want now to provide a condition, which assures that the quantum average would be independent of $\psi$, modulo terms that vanish on-shell. This can be formulated as
\be\label{indepobs}
\frac{d}{d\lambda}\,e_{\sst{\TT}}^{i\alpha_{\la\psi}(V)/\hbar}\T \alpha_{\la\psi}(F)=\left\{\psi,e_{\sst{\TT}}^{i\tilde{V}/\hbar}\T \tilde{F}\right\}_{\TT}\os 0\,,
\ee
where ``o.s.'' means ``on shell'' and we denoted $\tilde{F}\doteq\alpha_{\la\psi}(F)$, $\tilde{V}\doteq\alpha_{\la\psi}(V)$ . In particular we require the independence of the S-matrix of the gauge fixing, i.e.
\be\label{indepS}
\left\{\psi,e_{\sst{\TT}}^{i\tilde{V}/\hbar}\right\}_{\TT}\os 0\,.
\ee
We can rewrite the left hand side of this equation using the identity (\ref{Delta:Tbracket}):
\be\label{gaugeindep}
\Lap\left(\psi\T e_{\sst{\TT}}^{i\tilde{V}/\hbar}\right)+\psi\T \Lap(e_{\sst{\TT}}^{i\tilde{V}/\hbar})\os0\,.
\ee
Now on the first term we can use the identity (\ref{tkoszul2}) which corresponds to the Schwinger-Dyson equation and obtain:
\[
-\frac{i}{\hbar}\left(\left\{\psi\T e_{\sst{\TT}}^{i\tilde{V}/\hbar},S\right\}_{\TT}-\left\{\psi\T e_{\sst{\TT}}^{i\tilde{V}/\hbar},S\right\}_\star\right)-\psi\T\Lap(e_{\sst{\TT}}^{i\tilde{V}/\hbar})\os0\,.
\]
Note that the second term is an element of the ideal of equations of motion. Therefore a sufficient condition to fulfill (\ref{indepS}) on-shell can be formulated as
\be\label{suff:condition}
\frac{i}{\hbar}\left\{\psi\T e_{\sst{\TT}}^{i\tilde{V}/\hbar},S\right\}_{\TT}+\psi\T\Lap(e_{\sst{\TT}}^{i\tilde{V}/\hbar})\os0\,.
\ee
The operator $\Lap$ acting on the exponential function produces: $\Lap(e_{\sst{\TT}}^{i\tilde{V}/\hbar})=\frac{i}{\hbar}(\Lap \tilde{V}+\frac{i}{2\hbar}\{\tilde{V},\tilde{V}\})\T e_{\sst{\TT}}^{i\tilde{V}/\hbar}$. Moreover from (\ref{Delta:bracket2}) it follows that $\Lap(\al_{\la\psi}V)=\al_{\la\psi}(\Lap V)$.  Using again the Leibniz rule in (\ref{suff:condition}) and the fact that $\psi$ doesn't contain antifields we arrive finally at a condition of the form:
\[
\psi\T\al_{\la\psi}\Big(\frac{i}{\hbar}\{V,S\}_{\TT}+ \Lap V+\frac{i}{2\hbar}\{V,V\}_{\TT}\Big)\T e_{\sst{\TT}}^{i \tilde{V}/\hbar}\os0\,.
\]
Since the above condition has to be valid for arbitrary $\psi$, we obtain the following sufficient condition
\[
\{ V,S\}_{\TT}+\frac{1}{2}\{V,V\}_{\TT}-i\hbar\Lap V=0\,.
\]
Using  the fact, that $S$ doesn't contain antifields, we can write the above result in the form of the \textbf{\textit{quantum master equation}}\index{master equation!quantum} (\qme).
\be\label{QME0}
\frac{1}{2}\{S+V,S+V\}_{\TT}=i\hbar\Lap (S+V)\,.
\ee
Note that this is exactly the same condition, which is used in the path integral formalism to guarantee the gauge independence of the gauge-fixed ``measure'' \cite{Henneaux:1992ig}. In a similar way, by manipulating expression (\ref{indepobs}) we can conclude that the condition that the expectation value of an observable $F\in\TT(\BV_\reg(M))$ on-shell is independent of the gauge fixing is guaranteed by:
\[
\frac{i}{\hbar}F\cdot \{V,S\}_{\TT}+\frac{i}{\hbar}\{F,S\}_{\TT}+ F\cdot\Lap V+\Lap F+\frac{i}{\hbar}\{F,V\}_{\TT}+\frac{i}{2\hbar}F\cdot\{V,V\}_{\TT}=0\,.
\]
Using the {\qme} we obtain:
\[
\{F,S+V\}_{\TT}-i\hbar \Lap F=0\,. 
\]
We can conclude, that the expectation value of $F$ is independent of the gauge fixing, modulo on-shell terms, if $F$ lies in the kernel of
\[
\hat{s}=\{.,S+V\}_{\TT}-i\hbar \Lap\,.
\]
This operator is called in the literature the \textit{\textbf{quantum BV operator}}\index{BV!operator!quantum}. Using the above considerations we can express $\hat{s}$ off-shell in terms of time-ordered products. This interpretation is completely new and we argue that it allows for a consistent treatment of renormalization. First we note that if $X$ is a regular function or a vector field  (i.e. $X\in\TT(\V_\reg(M)\oplus\F_\reg(M))$), then from the relation (\ref{tkoszul2}) follows that:
\be\label{QBV1}
\left\{ e_{\sst{\TT}}^{iV/\hbar}\T X,S\right\}_{\TT}=\left\{ e_{\sst{\TT}}^{i V/\hbar}\T X,S\right\}_\star+
i\hbar \Lap( e_{\TT}^{iV/\hbar}\T X)\,,
\ee
In particular for $X=1$ we can rewrite this formula as:
\be\label{QME2}
e_{\sst{\TT}}^{iV/\hbar}\T\left(\frac{1}{2}\{S+V,S+V\}_{\TT}-i\hbar\Lap (S+V)\right)=\{e_{\sst{\TT}}^{i V/\hbar},S\}_{\star}\,,
\ee
so the {\qme} is a statement that $\{e_{\sst{\TT}}^{i V/\hbar},S\}_{\star}=0$ also off-shell. This is just invariance of the S-matrix under the free quantum Koszul operator. We can now obtain a formula for the quantum BV operator by using the {\qme} in equation (\ref{QBV1}).
\be\label{QBV0}
\hat{s}X=e_{\sst{\TT}}^{-i V/\hbar}\T\left(\{e_{\sst{\TT}}^{ iV/\hbar}\T X,S\}_{\star}\right)\,.
\ee
In section \ref{renormBV} we will show that this expression for the quantum BV operator can be generalized to renormalized time-ordered products and no divergences appear. Now we will show yet another way to interpret the formula (\ref{QBV0}). To obtain interacting fields one uses the intertwining map $R_V(F)$, so it is natural to ask, what is the transformed version of the free quantum Koszul operator $\{.,S\}_\star$. Using the definition (\ref{Rv}) and (\ref{Bog}) we obtain:
\begin{align*}
(\{.,S\}_\star\circ R_V)(X)&=\{(e_{\sst{\TT}}^{iV/\hbar})^{\star-1}\star (e_{\sst{\TT}}^{iV/\hbar}\T X),S\}_\star=(e_{\sst{\TT}}^{iV/\hbar})^{\star-1}\star \{e_{\sst{\TT}}^{iV/\hbar}\T X,S\}_\star=\nonumber\\
&=(e_{\sst{\TT}}^{iV/\hbar})^{\star-1}\star( e_{\sst{\TT}}^{iV/\hbar}\T e_{\sst{\TT}}^{-iV/\hbar}\T\{e_{\sst{\TT}}^{iV/\hbar}\T X,S\}_\star)=\nonumber\\
&=R_V(e_{\sst{\TT}}^{-iV/\hbar}\T\{e_{\sst{\TT}}^{iV/\hbar}\T X,S\}_\star)=(R_V\circ \hat{s})(X)\,.
\end{align*}
In the first step we used the {\qme}, which guarantees that $\{e_{\sst{\TT}}^{iV/\hbar},S\}_\star=0$ and since $\{.,S\}_\star$ is a derivation it holds $\{(e_{\sst{\TT}}^{iV/\hbar})^{\star-1},S\}_\star=0$ as well. The above result justifies the interpretation of $\hat{s}$ as the interacting quantum BV operator. It is obtained from the free Koszul operator by means of the intertwining map $R_V$:
\be\label{intertwining:s}
\{.,S\}_\star\circ R_V=R_V\circ\hat{s}\,.
\ee
It is now clear that $\hat{s}$ is a derivation with respect to the interacting star product $\star_V$. Moreover we can characterize the (co)homology of $\hat{s}$ knowing the  one of $\{.,S\}_\star$. 

To close this section we want to reflect a while on the equation  (\ref{QME2}). Looking at this formula, it is natural to ask, if we can add to the free Lagrangian a term that  contains antifields? Let us denote it by $\theta_0\in\TT(\BV_\reg(M))$. Of course it has to be linear both in fields and antifields. The full extended action takes the form $S+\theta_0+\lambda(\theta_1+S_1)=S+\theta_0+V$, where $\theta_1$ is linear in antifields, and $S_1\in\TT(\BV_\reg(M))$. We can interpret $\theta_0$ as the free BRST operator. The 0-th order in the coupling constant of equation (\ref{QME0}) is a statement that
\be\label{Stheta}
 \{S,\theta_0\}_{\TT}=0\,.
\ee
From this relation we obtain the following result:
\begin{prop}\label{brackettheta}
Let $S$, $\theta_0$ be defined as above. If (\ref{Stheta}) holds, then:
\[
\left\{e_{\sst{\TT}}^{i V/\hbar}\T X,\theta_0\right\}_\TT=\left\{e_{\sst{\TT}}^{iV/\hbar}\T X,\theta_0\right\}_\star\,.
\]
\end{prop}
\begin{proof}
To prove this identity we first note that:
\begin{multline}\label{theta0:exp}
\int\! dx\,\TT\left( \TT^{\minus}\frac{\delta}{\delta\ph(x)}(e_{\sst{\TT}}^{iV/\hbar}\T X)\cdot \TT^{\minus}\theta_0(x)\right)=
\int\! dx\, m\circ e^{i\hbar\DDp}\left(\frac{\delta}{\delta\ph(x)}(e_{\sst{\TT}}^{iV/\hbar}\T X)\otimes \theta_0(x)\right)=\\
\int\! dx\frac{\delta}{\delta\ph(x)}(e_{\sst{\TT}}^{iV/\hbar}\T X)\cdot\theta_0(x)+i\hbar\int\! dxdydz\ \Delta_D(y,z) \frac{\delta \theta_0(x)}{\delta\ph(y)}\frac{\delta^2}{\delta\ph(x)\delta\ph(z)}(e_{\sst{\TT}}^{iV/\hbar}\T X)\,.
\end{multline}
Now it remains to prove that 
\be\label{theta0:star}
\int dx\frac{\delta}{\delta\ph(x)}(e_{\sst{\TT}}^{iV/\hbar}\T X)\cdot\theta_0(x)=\int dx \frac{\delta}{\delta\ph(x)}(e_{\sst{\TT}}^{iV/\hbar}\T X)\star\theta_0(x)\,,
\ee
and that the second term of the expansion (\ref{theta0:exp}) vanishes. Actually both results can be obtained in a similar way. We start with the second one. From (\ref{Stheta}) it follows that: 
\[
\int\! dxdydz\ \Delta_D(y,z) \frac{\delta \theta_0(x)}{\delta\ph(y)}\frac{\delta^2 F}{\delta\ph(x)\delta\ph(z)}=
m\circ(\DDp)^2\left(\int\! dx\,\theta_0\frac{\delta S}{\delta\ph(x)}\otimes F\right)=0\,,
\]
for an arbitrary  argument $F\in\TT(\BV_\reg(M))$. To show (\ref{theta0:star}) we use a similar reasoning, but this time with the causal propagator:
 \[
\int\! dxdydz\ \Delta(y,z) \frac{\delta \theta_0(x)}{\delta\ph(y)}\frac{\delta^2 F}{\delta\ph(x)\delta\ph(z)}=
m\circ\DDp\circ\DC\left(\int\! dx\,\theta_0\frac{\delta S}{\delta\ph(x)}\otimes F\right)=0\,,
\]
It follows now that
\[
\TT\left( \int\! dx\,\TT^{\minus}\frac{\delta}{\delta\ph(x)}(e_{\sst{\TT}}^{iV/\hbar}\T X)\cdot \TT^{\minus}\theta_0(x)\right)=\int dx\frac{\delta}{\delta\ph(x)}(e_{\sst{\TT}}^{iV/\hbar}\T X)\star\theta_0(x)\,,
\]
so to end the proof we need to check
\[
\TT\left(\int\! dx\, \TT^{\minus}\frac{\delta}{\delta \ph^\ddagger(x)}(e_{\sst{\TT}}^{iV/\hbar}\T X)\cdot  \TT^{\minus}\frac{\delta\theta_0}{\delta \ph(x)}\right)=\int\! dx\,\frac{\delta}{\delta \ph^\ddagger(x)}(e_{\sst{\TT}}^{iV/\hbar}\T X)\star \frac{\delta\theta_0}{\delta \ph(x)}\,,
\]
but this is trivially fulfilled, since $\theta_0$ is linear and  hence $\frac{\delta\theta_0}{\delta \ph(x)}$ doesn't depend on fields anymore.
\end{proof}
Using the proposition \ref{brackettheta} and  equation (\ref{QBV1}) we obtain a following formula:
\be\label{braketfull}
\{e_{\sst{\TT}}^{iV/\hbar}\T X,S+\theta_0\}_{\TT}=\{e_{\sst{\TT}}^{i V/\hbar}\T X,S+\theta_0\}_\star+i\hbar\Lap\left(e_{\sst{\TT}}^{i V/\hbar}\T X\right)\,.
\ee
In particular for $X=1$ we have:
\[
e_{\sst{\TT}}^{iV/\hbar}\T\left(\{V,\theta_0+S\}_{\TT}+\frac{1}{2}\{V,V\}_{\TT}-i\hbar\Lap (V)\right)=\{e_{\sst{\TT}}^{iV/\hbar},S+\theta_0\}_\star\,.
\]
The {\qme} for the free action (0-th order in $\lambda$) states that $\frac{1}{2}\{S+\theta_0,S+\theta_0\}_{\TT}=i\hbar\Lap( S+\theta_0)$, so the {\qme} for the full action $S+\theta_0+V$ guarantees that $\{V,\theta_0+S\}_{\TT}+\frac{1}{2}\{V,V\}_{\TT}-i\hbar\Lap (V)=0$ and we obtain:
\[
\{e_{\sst{\TT}}^{iV/\hbar},S+\theta_0\}_\star=0\,.
\]
Therefore the quantum BV operator can be alternatively written as:
\be\label{QBV2}
\hat{s}X=e_{\sst{\TT}}^{-iV/\hbar}\T\left(\{e_{\sst{\TT}}^{i V/\hbar}\T X,S+\theta_0\}_{\star}\right)\,.
\ee
We also obtain  the formulation of the on-shell gauge invariance of the S-matrix, which is closer to the one given in \cite{H} (we come back to this discussion in section \ref{YMquant}):
\[
\{e_{\sst{\TT}}^{iV/\hbar},\theta_0\}_\star=0\qquad\textrm{on shell}\,,
\]
To summarize, we have shown in this section, that important notions of the BV quantization have a natural interpretation in the language of {\paqft}. All the steps were done in a mathematically precise way and expressions we obtained are well defined. The challenge we have to face right now is the generalization of these structures to more singular objects. As we already pointed out, the operator $\Lap$, which plays an important role in the BV-quantization is not well defined on local vector fields. This pathology results from the fact, that we were using the non-renormalized time-ordered product $\T$. Now we want to amend it, by means of renormalization. In the next section we argue, that one can completely avoid any divergences, if one works with the renormalized time-ordered product from the very beginning. 
\section{Renormalization}\label{renorm}
\subsection{Algebra of observables}\label{quantalg}
In the previous section we considered only very regular objects which allowed us to present the general structure of the quantum theory without going into technical details. Now we want to extend our discussion to more singular objects, namely elements of $\F_\mc(M)$ and $\V_\mc(M)$. We need here a little bit of distribution theory and not to make things too complicated at the beginning, we start our discussion with the Minkowski spacetime $\M$.  We review first the results obtained in \cite{BDF}. The condition (\ref{mlsc}) guarantees that the $\star$-product will be well defined if we replace $\frac{i}{2}\Delta$ with the Wightman 2-point function $\frac{i}{2}\Delta+\Delta_1$ in the definition (\ref{star product}), where $\Delta_1$ is the symmetric part. We denote this product by $\star_{\De_1}$. It is related to the previous one by a following transformation:
\[
\al_{\De_1}\doteq\exp({\hbar\Ga_{\De_1}}) : 
\F_\reg(\M)\rightarrow \F_\reg(\M)
\]
Let us write this relation explicitely:
\be
F\star_{\De_1}G\doteq\al_{\De_1}(\al_{\De_1}^{-1}(F)\star \al_{\De_1}^{-1}(G))\ .
\ee
The new star product can now be extended to the elements of $\F_\mc(M)$. This choice of the star product is called in the literature \textit{the Wick quantization}. It turns out however that it is not optimal in the context of locally covariant field theory.  In \cite{HW2,HW3} the smooth behavior under scaling of all dimensionful parameters at zero was crucial for the renormalization method and $\De_1$ doesn't depend smoothly on mass at $m^2=0$. We can amend this by replacing the mass dependent family $({\De_1}_m){m^2>0}$ by a family of Hadamard distributions  $H=(H_m)_{m^2\in\RR}$, $H_m\in\Dcal^\prime(\M^2)$, such that each $H_m$ is a distributional solution of the Klein-Gordon equation in both arguments, $H_m+i\De_m/2$ satisfies the microlocal spectrum condition (see \cite{Rad,BFK95,BF0}) and for each test function $f\in \Dcal(\M^2)$, $\langle H_m,f \rangle$ is a smooth function of $m^2$. For each such $H_m$ we can define a corresponding transformation $\al_{H_m}\doteq\exp({\hbar\Ga_{H_m}}):\F_\reg(\M)[[\hbar]]\rightarrow \F_\reg(\M)[[\hbar]]$ and a star product equivalent to the original one:
\be
F\star_{H_m} G\doteq\al_{H_m}(\al_{H_m}^{-1}(F)\star \al_{H_m}^{-1}(G)) \ .
\ee
From the mathematical point of view we are already happy with the construction, but how is it related to
more standard approaches used in QFT? Well, it turns out that this is nothing else, but an algebraic version of the Wick ordering. To see it clearly we need again a little bit of abstract reasoning.
Recall the discussion of topologies we had in section \ref{scal}. We introduced at that point the topology $\tau_\Xi$ on the space of microcausal functionals, which is induced by the H\"ormander topology. The space of regular functionals $\F_\reg(\M)$ is dense in $\F_\mc(\M)$ with respect to this topology. One can now equip 
$\F_\reg(\M)$ with the initial topology with respect to $\al_{H_m}$, i.e. the coarsest topology that makes this map continuous. We denote this topology by $\tau_{H_m}$ and the sequential completion of  $\F_\reg(\M)$ with respect to this topology is denoted by  $\overline{\F_\reg(\M)}^{\,\tau_{H_m}}$. Different choices of Hadamard functions lead to equivalent topologies. Now we can consider the following extension of  $\F_\reg(\M)$: we take all the sequences $(f_n)$ in $\F_\reg(\M)$ that converge to an element $\lim\limits_{n\rightarrow \infty}F_n=F\in\F_{\mc}(M)$ and map them back to  $\overline{\F_\reg(\M)}^{\,\tau_{H_m}}$ with the use of  $\al_{H_m}^{-1}$. This way we obtain sequences $(\al_{H_m}^{-1}F_n)$ that converge in the topology $\tau_{H_m}$ to elements $\lim\limits_{n\rightarrow \infty}(\al_{H_m}^{-1}F_n)$. Let $\fA^{(m)}(\M)$ denote the closure of $\F_\reg(\M)$ with respect to all such sequences. We can think of elements of $\fA^{(m)}(\M)$ as Wick products. This can be seen from the following example:
\begin{exa}
Consider a sequence $F_n(\ph)=\int \ph(x)\ph(y)g_n(y-x)f(x)$ with a smooth function $f$ and a sequence of smooth functions $g_n$ which converges to the $\delta$ distribution in the H\"ormander topology. By applying $\al_{H_m}^{-1}$ we obtain a sequence $\al_{H_m}^{-1}F_n= \int (\ph(x)\ph(y)g_n(y-x)f(x)-H_m(x,y)g(y-x)f(f))$. The limit of this sequence is an element of $\fA^{(m)}(\M)$ that can be identified with $\int :\ph(x)^2:f(x)$.
 \eex\end{exa}
Note that in even dimensions $H_m$ is not uniquely determined by the conditions above, but depends 
on an additional mass parameter $\mu>0$ (see \cite{DF04,H04}). The change of $\mu$ amounts to the transition to an equivalent product. The equivalence transformation between the products $\star_{m,\mu_1}$ and $\star_{m,\mu_2}$ is given by the linear continuous isomorphism of $(\F_\reg(\M),\tau_{H^{\mu_1}_m})$ and  $(\F_\reg(\M),\tau_{H^{\mu_2}_m})$ given by
\begin{equation}
\al_{w_{m}^{\mu_1,\mu_2}}\doteq\exp({\hbar\Ga_{w_{m}^{\mu_1,\mu_2}}})\,,
\end{equation}
where
\begin{equation}
w_{m}^{\mu_1,\mu_2}=H_{m}^{\mu_1}-H_{m}^{\mu_2}\,,
\end{equation}
is smooth. As in \cite{BDF} we can use these intertwining maps to  define a projective limit $\F_\reg^{(m)}(\M)$ of locally convex topological vector spaces $(\F_\reg(\M),\tau_{H^{\mu}_m})$. 
\[
\F_\reg^{(m)}(\M)\doteq \varprojlim_{\mu} (\F_\reg(\M),\tau_{H^{\mu}_m}) = \Big\{(F_\mu)_{\mu>0} \in \prod_{\mu\in \RR}\F_\reg(\M) \;\Big|\; F_{\mu_1} = \al_{w_{m}^{\mu_1,\mu_2}}(F_{\mu_2}),\mu_1 \leq \mu_2\Big\}. 
\]
We equip this space with the initial topology, which by definition is $m$-dependent, but $\mu$-independent. 
Elements  of $\F_\reg^{(m)}(\M)$ may be identified with families $(\alpha^{-1}_{H_m^\mu}(F))_{\mu>0}$, $F\in\F_\reg(\M)$. Again we can take the completion of $\F_\reg^{(m)}(\M)$ with respect to all the sequences  $(\alpha^{-1}_{H_m^\mu}(F_n))_{\mu>0}$, such that $F_n$ converges in $(\F_\mc(\M),\tau_\Xi)$. We denote this completion by $\fA^{(m)}(\M)$. We can define a map $\alpha_{H_m}:\fA^{(m)}(\M)\rightarrow\F_\mc(\M)$ by setting
\[
\alpha_{H_m}\Big(\lim\limits_{n\rightarrow\infty}\big(\alpha^{-1}_{H_m^\mu}(F_n)\big)_{\mu>0}\Big)\doteq \lim\limits_{n\rightarrow\infty}(F_n)\,,
\]
where $\lim\limits_{n\rightarrow\infty}(\alpha^{-1}_{H_m^\mu}(F_n))_{\mu>0}$ is a generic element of $\F^{(m)}(\M)$. This map is sequentially continuous from the definition of the topology $\tau_{H_m}$. We introduce now a simplified notation  $\alpha^{-1}_{H_m}(F)\doteq\lim\limits_{n\rightarrow\infty}(\alpha^{-1}_{H_m^\mu}(F_n))_{\mu>0}$, where 
$\lim\limits_{n\rightarrow\infty}(F_n)=F\in\F_\mc(\M)$.  We can equip $\fA^{(m)}(\M)$ with a sequentially continuous product $\star_m$ defined as
\[
\alpha^{-1}_{H_m}(F)\star_m\alpha^{-1}_{H_m}(G)\doteq\lim\limits_{n\rightarrow\infty}(\alpha^{-1}_{H_m^\mu}(F_n\star_{H^\mu_m} G_n))_{\mu>0}\,.
\]
The construction we performed here may seem a little bit abstract, but it will safe us some effort later on, since now the only scale in the algebra is $m$. Following \cite{BDF} we define now the bundle of algebras
\[
\bigsqcup_{m^2\in\RR}\fA^{(m)}(\M)\,.
\] 
We denote by  $\fA(\M)$ the algebra of sections  $A=(A_m)_{m^2\in\RR}$ of this bundle such that
$\al_{H}(A)$, defined as
\[
(\al_{H}(A))_m\doteq\al_{H_{m}}(A_m),\ m^2\in\RR\,,
\]
is a smooth function of $m^2$. Note that $\al_{H_{m}}(A_m)=\al_{H_{m}^\mu}((A_m)_\mu)\in\F_\mc(\M)$.
By $\fA_\reg(\M)$ we denote the subalgebra of $\fA(M)$ consisting of sections taking values in $\F_\reg(\M)$, in the sense that $\al_{H_{m}}(A_m)\in\F_\reg(\M)$. In a similar way we define the subspace $\fA_{\loc}(\M)$ of 
$\fA(\M)$, whose elements will provide the possible interaction terms for a quantum field theory. We also define a subspace  $\fA_{\ml}(\M)$ consisting of multilocal observables. Star products $\star_m$ on $\fA^{(m)}(\M)$ induce a star product on $\fA(\M)$, which we denote by $\star$.
\be\label{star}
(F\star G)_{m}\doteq \alpha^{-1}_{H_m}\Big( \alpha_{H_m}(F_m)\star_m\alpha_{H_m}(G_m)\Big),\qquad F,G\in\fA(\M)\,.
\ee
Functional derivatives on $\F_\mc(\M)$ induce linear mappings on $\fA(\M)$ defined by
\be\label{eq:Aderivative}
\left\langle \frac{\de}{\de\ph}A,\psi\right\rangle=\alpha_H^{-1}\left\langle\frac{\de}{\de\ph}\alpha_H A,\psi\right\rangle\,.
\ee
Following (\cite{BDF}) we associate to every $A\in\fA(\M)$ a compact region (denoted as $\supp(A)$ by abuse of notation) defined as the set 
\[
\supp(A)\doteq\supp(\alpha_H(A))\ .
\]
We can see that all the structures defined on $\F_\mc(\M)$ can be easily brought to the space of observables $\fA(M)$. 
\subsection{Renormalized time ordered product}\label{renprod0}
In the last section we saw that the wave front set of $\frac{i}{2}\Delta +H^\mu_m$ is such that using the microlocal spectrum condition \cite{BFK95,Rad}, the star product $\star_{H^\mu_m}$ can be uniquely extended by sequential continuity to $\F_\mc(\M)$. This allowed us to define the $\star$-product on $\fA(\M)$ by means of (\ref{star}).
The situation for the time-ordered product is more complicated. The time ordering operator $\TT:\F_\reg(\M)[[\hbar]]\rightarrow \fA(\M)[[\hbar]]$ can be defined by a family of maps $H_{F,m}^\mu$, where  $H_{F,m}^\mu\doteq i\Delta_D+H^\mu_{m}$. Explicitly we can write it as:
\[
(\TT F)_m\doteq\big(\alpha^{\minus}_{H_m}\circ\alpha_{H_{F,m}}^{\phantom{\minus}}\big)F=\alpha_{i\Delta_D}(F)
\]
Using $\TT$ we define a time-ordered product on $\fA_\reg(M)[[\hbar]]$ by
\[
F\T G\doteq \TT(\TT^{\minus}F\cdot\TT^{\minus}G)
\]
This can be also written as:
\begin{align}\label{THmn}
(\alpha_{H}^{\minus}F\T \alpha_{H}^{\minus}G)_m&=\alpha^{\minus}_{H_m}\!\!\circ\alpha_{H_{F,m}}^{\phantom{\minus}}\left(\alpha_{H_{F,m}}^{\minus}\!\!\!\circ\alpha^{\phantom{\minus}}_{H_m}((\alpha_{H}^{\minus}F)_m)\cdot \alpha_{H_{F,m}}^{\minus}\!\!\!\circ\alpha^{\phantom{\minus}}_{H_m}((\alpha_{H}^{\minus}G)_m)\right)=\nonumber\\
&=\alpha^{\minus}_{H^\mu_m}(F\cdot_{\TT_{m,\mu}}G)\,,
\end{align}
where $(F_{m,\mu})_{\mu>0}=\alpha^{\minus}_{H^\mu_m}(F)$, $F\in\F_\reg(\M)$ and similar for $G$. In the above formula  $\cdot_{\TT_{m,\mu}}$ is a product on $\F_\reg(\M)$ defined as
\[
F\cdot_{\TT_{m,\mu}}G\doteq \alpha_{H_{F,m}^\mu}(\alpha_{H_{F,m}^\mu}^{\minus}F\cdot \alpha_{H_{F,m}^\mu}^{\minus}G)\,.
\]
In contrast to the star product $\star_{H^\mu_m}$ this product is not continuous with respect to the topology $\tau_{H_m^{\mu}}$, therefore it cannot be extended to the full space of microcausal functionals $\F_\mc(\M)[[\hbar]]$. Nevertheless the $n$-fold time ordered product $\cdot_{{\TT}_{m,\mu}}$ can be defined for functionals $F_1,\ldots,F_n\in\F_\mc(\M)$ as long as their supports are disjoint and for such choice of $F_i$'s we can define an operator $\Tcal_{m,\mu}^{\,n}:\F_\mc(\M)[[\hbar]]^{\otimes n}\rightarrow \F_\mc(\M)[[\hbar]]$. 
\[
\Tcal_{m,\mu}^{\,n}(F_1,\ldots,F_n)=F_1\cdot_{{\TT}_{m,\mu}}\ldots\cdot_{{\TT}_{m,\mu}} F_n\,,
\]
These operators on $\F_\mc(\M)$ induce operators $\Tcal_H^{\,n}:\F_\mc(\M)[[\hbar]]^{\otimes n}\rightarrow \fA(\M)[[\hbar]]$ by means of (\ref{THmn}), i.e.
\[
\Tcal_H^{\,n}(F_1,\ldots,F_n)=\al_H^{\minus} (F_1)\T\ldots\T \al_H^{\minus} (F_n)\,,
\]
for $F_i\in \F_\mc(\M)$ with disjoint supports.
This is already something, but don't forget, that the reason for introducing time-ordered products was implementation of the interaction. Recall from the formula (\ref{Smatrix}) that for the definition of the S-matrix we need  to make sense of expressions like $\Tcal_H^{\,n}(\al_H (F)^{\otimes n})$, $F\in\fA_\loc(\M)$. It's clear that the extension of $\Tcal_H^{\,n}$ to local observables with disjoint supports is not enough. We need to work a little bit more,  since we want to make sense of the time ordered exponential of local interaction terms. This is where the most exciting adventure in perturbative quantum field theory begins. The extension of $\Tcal_H^{\,n}$ to arbitrary elements of $\F_\loc(\M)$ is the central problem of renormalization theory. With a similar construction we can also formulate this question for an arbitrary globally hyperbolic background manifold $M$. See \cite{BF0} for details.

Looking from the outside one can get a wrong impression that QFT has some internal problem, because plenty of divergencies appear everywhere. This is however not the case. The problem of renormalization can be formulated as a well posed mathematical question of distributions' extension. If we are careful enough with defining all our objects, no divergencies would appear.  The extension of $\Tcal_H^n$ can be constructed by means of the inductive procedure of Epstein and Glaser \cite{EG} (developed on the basis of ideas of St\"uckelberg \cite{Stuck} and Bogoliubov \cite{BS}).  There is of course a certain freedom in extending $\Tcal_H^{\,n}$ to functionals with coinciding supports, but this freedom is well understood and under control. In causal perturbation theory the choice of time ordering prescriptions is characterized by the renormalization group in the sense of St\"uckelberg-Petermann \cite{DF04}. Its relation to different notions of renormalization group like the Gell-Mann-Low or Willson renormalization groups was discussed in \cite{BDF}.

Equipped with powerful tools of distribution theory we can now start our quest to find the extension of  $\Scal$ to elements of $\fA_\loc(\M)$. There are some properties of the formal S-matrix that we want to keep. Here we list some of them, following \cite{BDF}:
\begin{enumerate}[{\bf S 1.}]
\item {\bf Causality.} $\Scal(A+B)=\Scal(A)\star \Scal(B)\ $ if $\supp(A)$ is later than $\supp(B)$.
\item {\bf Starting element.} $\Scal(0)=1$, $\Scal^{(1)}=\id\ $.
\item {\bf $\ph$-Locality.} $\al_H\circ S(V)(\ph_0)=\al_H\circ S\circ \al_H^{-1}(\al_H(V)_{\ph_0}^{(N)})(\ph_0)+O(\hbar^{N+1})$. This guarantees that for the computation of a certain coefficient in the $\hbar$-expansion of
$\al_H\circ S(V)$, we may replace $\al_H(V)(\ph)$ by a {\it polynomial} in $\ph$ (see \cite{BDF} for details).
\item {\bf Field Independence.}  $\langle\de \Scal(V)/\de\ph,\psi\rangle=\Scal^{(1)}(V)\langle\de V/\de\ph,\psi\rangle\ ,\ $ with $V\in\fA_{loc}(M)$.
\end{enumerate}
The axioms for $\Scal$ can be easily translated to properties of the renormalized time-ordered products, since
\[
\Scal(F)=\sum\limits_{n=0}^\infty \frac{1}{n!}\TTH^{\,n}(\al_H (F)^{\otimes n})\,.
\]
Usually the renormalized time ordered products $\TTH^{\,n}$ are required to fulfill the following normalization conditions (see for example \cite{H}):
\begin{enumerate}[{\bf T 1.}]
\item {\bf Starting element.} For the lowest order time-ordered products we require $\TTH^{\,0}=0$, $\TTH^1=\al_H^{-1}$.
\item {\bf  Symmetry.}The time ordered products are symmetric under a permutation of factors,
\[
\TTH^{\,n}(F_1 \otimes \cdots \otimes F_n)
= \TTH^{\,n}(F_{\pi (1)} \otimes \cdots \otimes F_{\pi (n)})\,,
\]
for any permutation $\pi$.
\item {\bf Unitarity.}
Let $\TTHb^{\,n}(\otimes_i F_i) =
[\TTH^{\,n}(\otimes_i F_i^*)]^*$ be
the antitime-ordered product. Then we require
\[
\TTHb^{\,n} \bigg( \bigotimes_{i=1}^n F_i \bigg) =
\sum_{I_1 \sqcup \dots \sqcup I_j = \underline{n}}
(-1)^{n + j} \TTH^{\,|I_1|}\bigg(
\bigotimes_{i \in I_1} F_i \bigg) \star \dots \star
\TTH^{\,|I_j|}\bigg(\bigotimes_{j \in I_j} F_j \bigg)\,,
\]
where the sum runs over all partitions of the set $ \underline{n}\doteq\{1, \dots, n\}$ into
pairwise disjoint subsets $I_1, \dots, I_j$.
\item {\bf Causal Factorization.}
If the supports of $F_1\ldots F_i$ are later than the supports of $F_{i+1},\ldots F_n$, then we have
\be\label{CausFact}
\TTH^{\,n}(F_1\otimes \dots \otimes F_n)=
\TTH^{\,i}(F_1\otimes \dots \otimes F_i) \star
\TTH^{\,n-i}(F_{i+1} \otimes \dots \otimes F_n) \, .
\ee
For the case of 2 factors, this means
\[
\TTH^{\,2}(F_1\otimes F_2) =
\begin{cases}
\al_H^{\minus}(F_1) \star \al_H^{\minus}(F_2) & \text{if $\supp(F_1) \gtrsim \supp(F_2) $;}\\
\al_H^{\minus}(F_2) \star \al_H^{\minus}(F_1) & \text{if $ \supp(F_2) \gtrsim \supp(F_1)$.}
\end{cases}
\]
\item {\bf Commutator.}
The commutator of a time-ordered product with a free field is given by
lower order time-ordered products multiplied with functions,
namely
\be
\left[\TTH^{\,n} \Big( \bigotimes_i^n F_i \Big), \Phi(x) \right]_{\star} =
i\hbar\sum_{k=1}^n \TTH^{\,n}\bigg(F_1 \otimes \dots \int \Delta(x,y)
\frac{\delta F_k}{\delta \phi(y)} \otimes \dots F_n\bigg),
\ee
\item {\bf  Field equation.}
The free field equation is implemented in a Schwinger-Dyson type equation:
\be\label{fieldeq}
\TTH^{\,n+1}\bigg(
\frac{\delta S_0}{\delta \phi(x)} \otimes \bigotimes_{i=1}^n F_i
\bigg)
= \sum_i
\TTH^{\,n}
\bigg(
F_1 \otimes \cdots \frac{\delta F_i}{\delta
    \phi(x)} \otimes \cdots F_n)
\bigg)+\frac{\delta S_0}{\delta \phi(x)}\TTH^{\,n}
\bigg(
F_1 \otimes \cdots  F_n
\bigg)
\ee
\end{enumerate}
Furthermore one can impose conditions like scaling, smoothness and fulfilling the action Ward Identity. Time ordered products fulfilling those properties were already constructed for the scalar field \cite{BF0}, Dirac fields \cite{DHP,Thomas} and Yang-Mills theory \cite{H}. The ambiguity in defining maps $\TTH^{\,n}$ is described by the St\"uckelberg-Petermann Renormalization Group $\Rcal$ which is the group of analytic maps of $\fA_{\loc}(\M)[[\hbar]]$ into itself with the properties:
\begin{enumerate}[{\bf Z 1.}]
\item $Z(0)  =0$,\label{Z0}
\item $Z^{(1)}(0)  =\id$,
\item $Z        = \id + O(\hbar)$,
\item $Z(A+B+ C)= Z(A+B)-Z(B)+Z(B+C)$ if $\supp(A)\cap\supp(C)=\emptyset$\label{Zloc1},
\item $\varphi$-locality, see \cite{BDF} for details,
\item $\de Z/\de\varphi = 0$\label{Zindep}.
\end{enumerate}
The main theorem of renormalization \cite{DF04,BDF} states that:
\begin{thm}\label{MTR}
Given two $S$-matrices $\Scal$ and $\widehat{\Scal}$ 
satisfying the conditions Causality, Starting Element, $\varphi$-locality, and 
Field Independence, there exists a unique $Z\in\Rcal$ such that 
\begin{equation}
\widehat{\Scal}=\Scal\circ Z \ .\label{mainthm}
\end{equation}
Conversely, given an $S$-matrix $\Scal$ satisfying the 
mentioned conditions and a $Z\in\Rcal$, Eq. (\ref{mainthm})
defines a new $S$-matrix $\widehat{\Scal}$ satisfying also these conditions.
\end{thm}
If $\Scal$ is replaced by $\hat{\Scal}=\Scal\circ Z$ with a renormalization group element $Z$, then for an observable $F\in\fA(\M)$ we obtain
\[
\hat{\Scal}_V(F)=\Scal_{Z(V)}(Z_V(F))
\]
where $Z_V(F)=Z(V+F)-Z(V)$. From the additivity of $Z$ it follows that 
\begin{equation}
Z_V(F)=Z_{V'}(F)\quad\mathrm{if}\quad\supp(V-V')\cap\supp\, F=\emptyset\ .\label{relationZ}
\end{equation}
The relation (\ref{relationZ}) implies that
\begin{equation} 
\supp\, Z_V(F)\subset\supp\, F \ .\label{suppZ} 
\end{equation}

After this short review of the standard methods of causal perturbation theory we move on to the next section, where some new results concerning the renormalized time-ordered product are proven.
\subsection{Associativity of the renormalized time-ordered \-product}\label{renprod}
Up to now the renormalized time-ordered product was not an algebraic product defined as a binary operation on the suitable domain. Instead the family of multilinear maps $\TTH^{\,n}$ was constructed. With this formulation it was difficult to prove the associativity.
In the present work we solve this problem and prove that the renormalized time-ordered product is an \textit{\textbf{associative product}} on a suitable subspace of $\fA(\M)[[\hbar]]$. We show that $\TTH^{\,n}:\F_\loc(\M)[[\hbar]]^{\otimes n}\rightarrow\fA(\M)[[\hbar]]$ can be pulled back to a map $\TTH:\F(\M)[[\hbar]]\rightarrow \fA(\M)[[\hbar]]$ and therefore the renormalized time ordered product can be really treated as a binary operation on the space $\TTH(\F(\M))[[\hbar]]$.
\begin{thm}\label{factorization}
The renormalized time-ordered product $\TRH$\index{product!time ordered!renormalized} is an associative product on $\TTH(\F(\M))$ given by
\[
F\TRH G\doteq\TTH(\TTH^{\minus}F\cdot\TTH^{\minus}G)\,,
\]
where $\TTH:\F(\M)[[\hbar]]\rightarrow\TTH(\F(\M))[[\hbar]]\equiv\Dcal_{\TTR}(\M)$ is defined by means of the commutative diagram:
\begin{center}
\begin{tikzpicture} \matrix(a)[matrix of math nodes, row sep=3em, column sep=1em, text height=2ex, text depth=1.5ex] {
\F(\M)[[\hbar]]&&S^\bullet\F^{(0)}_\loc(\M)[[\hbar]] \\
&\fA(\M)[[\hbar]]&\\};
\path[->,font=\scriptsize] (a-1-1) edge node[above,color=red]{$\beta$} (a-1-3); 
\path[->,dashed,font=\scriptsize](a-1-1) edge node[left]{$\TTH$} (a-2-2); 
\path[->,font=\scriptsize](a-1-3) edge node[right]{$\ \oplus_n\TTH^{\,n}$}(a-2-2);  
\end{tikzpicture}
\end{center}
In the above diagram $S^\bullet\F^{(0)}_\loc(\M)$ denotes the space of symmetric tensor powers of local functionals satisfying $F(0)=0$ and $\beta$ is the inverse of multiplication $m$.
\end{thm}
\begin{proof}
To prove the theorem it remains to construct the map $\beta:\F(\M)\rightarrow S^\bullet\F^{(0)}_\loc(\M)$ which is the inverse of the multiplication $m$ and provides the factorization of a multilocal functional into local ones. Before we do it, we make a following observation:  let $F_1,...,F_n\in\F_\loc(\M)$, $F=m(F_1,...,F_n)=F_1\cdot...\cdot F_n$ and assume that $F_1(0)=...=F_n(0)=0$, then:
\[
F(\ph)=\int\limits_0^1 \big<F_1^{(1)}(\lambda_1 \ph),\ph\big>d\la_1\cdot...\cdot\int\limits_0^1 \big<F_n^{(1)}(\lambda_n \ph),\ph\big>d\la_n\,.
\]
Since a derivative of a local functional is a smooth section, we can change the order of the integration and write the above formula as:
 \be\label{multilocal:repr}
F(\ph)=\int dx_1...dx_n \int\limits_0^1 d\vec{\lambda}F_1^{(1)}(\lambda_1 \ph)(x_1)\cdot...\cdot F_n^{(1)}(\lambda_n \ph)(x_n)\,\ph(x_1)...\ph(x_n)\,,
\ee
where the integration is performed with respect to $\lambda_1$,\ldots, $\lambda_n$ from $0$ to $1$. We can symmetrize the expression on the right hand side with respect to $x_1$,\ldots, $x_n$ and using the fact that all the factors are local we can write $F$ as integration of the pullback of a function on the jet bundle:
 \[
F(\ph)=\int\! d\vec{x}\,((j^{\infty}\ph)^n)^*f(x_1,...,x_n)\,\ph(x_1)...\ph(x_n)\,,
\]
where $j^{\infty}\ph$ denotes the pullback by the jet prolongation of $\ph\in\Ci(\M)$ and $f$ is a smooth symmetric function on the jet bundle. One can picture it on the diagram:
\begin{center}
\begin{tikzpicture} \matrix(a)[matrix of math nodes, row sep=3em, column sep=1em, text height=2.5ex, text depth=0.25ex] {
J^{\infty}(\M\leftarrow \M\times\RR)^{\oplus n}&\M^n\times\RR \\
\M^n&\\};
\path[->,font=\scriptsize](a-1-1) edge node[above]{$f$} (a-1-2); 
\path[->,font=\scriptsize](a-1-1) edge node[right]{}  (a-2-1);  
\path[<-,font=\scriptsize,bend right](a-1-1) edge node[left]{$(j^\infty\ph)^{\oplus n}$}  (a-2-1);  
\path[->,font=\scriptsize](a-1-2) edge node[right]{}(a-2-1);
  \path[<-,font=\scriptsize,bend left,dashed](a-1-2) edge node[below]{$\quad\quad\quad\quad((j^{\infty}\ph)^n)^*f$}(a-2-1);  
\end{tikzpicture}
\end{center}
Knowing $f$ one can now define a functional $\tilde{F}\in S^n\F_\loc(\M)$ by setting:
 \be\label{Fk}
\tilde{F}(\ph_1,...,\ph_n)\doteq\int\! d\vec{x}\,j^{\infty}\vec{\ph}^{\ *}f(x_1,...,x_n)\,\ph_1(x_1)...\ph_n(x_n)\,,
\ee
where $j^{\infty}\vec{\ph}^{\ *}$ is the symmetrized pullback, i.e. $\vec{\ph}\doteq \frac {1}{n!}\sum\limits_{\pi\in P_n} (\ph_{\pi(1)},...,\ph_{\pi(n)})$. It follows from the construction that $\tilde{F}(\ph,...,\ph)=F(\ph)$. Therefore the problem of finding the factorization map $\beta:\F(\M)\rightarrow S^{\bullet}\F_\loc(\M)$ can be reduced to the construction of the jet bundle function $f$ for a given multilocal functional $F$.

Let $F\in\F(\M)$. 
Since $\F(\M)$ contains only finite sums of products of local functionals, there exists a maximal $k>0$ such that $\exists \ph\in\E(\M)$ for which $\left<F^{(k)}(\ph),\ph_1\otimes...\otimes\ph_k\right>\neq 0$ for $\ph_1\otimes...\otimes\ph_k$  with pairwise disjoint supports of the corresponding infinite jets $j^\infty(\ph_1),..,j^\infty(\ph_k)$. Consider a functional: 
\be\label{Ftild}
\int\limits_0^1 d\vec{\lambda} \left<F^{(k)}(\lambda_1\ph_1+...+\lambda_k\ph_k),\ph_1\otimes...\otimes\ph_k\right>\,.
\ee
Now we use the fact that $F$ is multilocal, so the only term of $F^{(k)}(\lambda_1\ph_1+...+\lambda_k\ph_k)$ that contributes to (\ref{Ftild}) is a smooth section. We denote this test section by $g(\lambda_1\ph_1+...+\lambda_k\ph_k)$. Using the theorem of Fubini-Tonelli we can now change the order of integration in (\ref{Ftild}) and write it as:
\be\label{Ftild2}
\int dx_1...dx_k \int\limits_0^1 d\vec{\lambda} g(\lambda_1\ph_1+...+\lambda_k\ph_k)(x_1,...,x_k)\ph_1(x_1)...\ph_k(x_k)
\ee
Now we want to show that $G_{x_1,...,x_k}(\ph_1,...,\ph_k)=\int\limits_0^1 d\vec{\lambda} g(\lambda_1\ph_1+...+\lambda_k\ph_k)(x_1,...,x_k)\ph_1(x_1)...\ph_k(x_k)$ depends only on jets of $\ph$'s at points $x_1,...,x_k$. We pick a test section $h$ infinitely flat at $x_1,...,x_k$ and calculate the change of $G_{x_1,...,x_k}(\ph_1,...,\ph_k)$:
\begin{multline*}
G_{x_1,...,x_k}(\ph_1,...,\ph_k)-G_{x_1,...,x_k}(\ph_1+h,...,\ph_k)=\\
=\int\limits_0^1 d\mu\!\! \int\limits_0^1 d\vec{\lambda} g^{(1)}(\lambda_1(\ph_1+\mu h)+...+\lambda_k\ph_k)(x_1,...,x_k)[h]\ph_1(x_1)...\ph_k(x_k)=\\
=\int\limits_0^1 d\mu\!\! \int\limits_0^1 d\vec{\lambda}\int\! dz g^{(1)}(\lambda_1(\ph_1+\mu h)+...+\lambda_k\ph_k)(x_1,...,x_k,z)h(z)\ph_1(x_1)...\ph_k(x_k)
\end{multline*}
From the maximality of $k$ follows that $g^{(1)}(\lambda_1(\ph_1+\mu h)+...+\lambda_k\ph_k)(x_1,...,x_k,z)$ differs from $0$ only if $z=x_i$ for some $i\in{1,...,k}$, but then $h(z)=0$, so:
$G_{x_1,...,x_k}(\ph_1,...,\ph_k)-G_{x_1,...,x_k}(\ph_1+h,...,\ph_k)=0$ and we conclude that $G_{x_1,...,x_k}(\ph_1,...,\ph_k)$ depends only on the jets of the arguments at points $x_1,...,x_k$. 
We can write now (\ref{Ftild2})  as:
\[
\int d\vec{x}\int\limits_0^1 d\vec{\lambda} g(j^\infty_{x_1}(\lambda_1\ph_1+...+\lambda_k\ph_k),..., j^\infty_{x_k}(\lambda_1\ph_1+...+\lambda_k\ph_k))(x_1,...,x_k)\ph_1(x_1)...\ph_k(x_k)
\]
Note that from the assumption on functions $\ph_1$,\ldots, $\ph_k$ follows that if $\ph_i(x_i)\neq0$, then $j^\infty_{x_j}\ph=0$ for $j\neq i$. This allows us to rewrite the above formula as:
\[
\int d\vec{x}\int\limits_0^1 d\vec{\lambda}\, \tilde{g}(j^\infty_{x_1}(\lambda_1\ph_1),..., j^\infty_{x_k}(\lambda_k\ph_k))(x_1,...,x_k)\ph_1(x_1)...\ph_k(x_k)
\]
We define a jet bundle function  $f$ by:
\[
f(j^\infty_{x_1}(\ph_1),..., j^\infty_{x_k}(\ph_k))(x_1,...,x_k)\doteq \int\limits_0^1 d\vec{\lambda}\, \tilde{g}(j^\infty_{x_1}(\lambda_1\ph_1),..., j^\infty_{x_k}(\lambda_k\ph_k))(x_1,...,x_k)
\]
Locally this function depends only on the finite jets, so it is just a finite dimensional function
and from the surjectivity of the jet projection follows that it is uniquely determined by the above definition. 

It remains to extend $f$ to arbitrary arguments. Note that $g$ is a smooth compactly supported function on the jet space. In particular it is bounded, together with all its derivatives. Therefore the function $f$ and its derivatives are bounded as well and smooth outside the thin diagonal. We can therefore extend $f$
 by the smooth extension to a smooth function on the whole jet bundle. 
Using $f$ we define $\tilde{F}_k$ by means of (\ref{Fk}).
Now we take $F-\tilde{F}_k$. This functional has now maximal degree $l <k$ and we can now repeat all the steps to define $(\widetilde{F-\widetilde{F}_k})_l$. Repeating this until degree $0$ we construct inductively the desired map $\beta$. It follows from the construction that it is indeed the right inverse of the multiplication $m$. To show that it is also the left inverse it remains to show that $m$ is injective on $S^\bullet\F_\loc^{(0)}(\M)$. Let $F=\bigoplus\limits_{k=0}^nF_k \in S^\bullet\F_\loc^{(0)}(\M)$ and $m(F)=0$. It follows that also the $n$-fold derivative of $m(F)$ is equal to $0$. Let us take $x_1,\ldots, x_n\in\M$ such that $x_i\neq x_j$ for $i\neq j$. Then it follows that
\[
\frac{\delta^n(m(F))}{\delta\ph_1(x_1)\ldots\delta\ph_n(x_n)}(\ph)=n!\frac{\delta^n F_n(\ph,\ldots,\ph)}{\delta\ph_1(x_1)\ldots\delta\ph_n(x_n)}=0\,.
\]
We know from the assumption that $F_n$ depends on $\ph$ only via its jet prolongation. Therefore we can replace $\ph$'s in the argument of $F_n$ by different functions $\ph_1\ldots\ph_n$ such that $j^\infty_{x_i}(\ph_i)=j^\infty_{x_i}(\ph)$, $i=1,\ldots,n$. 
Let us now take arbitrary $\ph_1,\ldots,\ph_n$ and define a smooth partition of unity $1=\sum\limits_{i=1}^n\chi_i$, where $\chi_i(x_i)=1$ and $\supp\,\chi_i\cap\{x_j;j\neq i\}=\varnothing$. Now we set $\ph=\sum\limits_{i=1}^n\chi_i\ph_i$ and it follows that:
\be\label{multilical:0}
\frac{\delta^n F_n(\ph_1,\ldots,\ph_n)}{\delta\ph_1(x_1)\ldots\delta\ph_n(x_n)}=0\,.
\ee
This holds for arbitrary field configurations $\ph_1,\ldots,\ph_n$ and for distinct points $x_i$. From the assumption follows that the only potentially non-vanishing contribution to the above derivative is smooth and bounded in $x_i$. This is exactly the smooth function that appears on the right hand side of the formula  (\ref{multilocal:repr}). By smooth extension we can conclude from (\ref{multilical:0}) that this function vanishes everywhere and therefore $F_n=0$. A similar argument can be used for $F_{n-1}$ and if follows that $F=0$, hence $m$ is injective on $S^\bullet\F_\loc^{(0)}(\M)$ and $\beta$ is its inverse.

The existence of the map $\beta$ allows us to treat the renormalized time ordered product $\TRH$ as a binary operation on the space $\TTH(\F(\M))$.
\[
F\TRH G\doteq \TTH(\TTH^{\minus}F\cdot\TTH^{\minus} G),\qquad F,G\in\TTH(\F(M))\,.
\]
The associativity can be now easily shown with the above formula, since:
\begin{multline}
A\TRH(B\TRH C)=\TTH(\TTH^{\minus}A\cdot \TTH^{\minus}\circ\TTH(\TTH^{\minus}B\cdot \TTH^{\minus}C))=\\
=\TTH(\TTH^{\minus}A\cdot\TTH^{\minus}B\cdot \TTH^{\minus}C)=(A\TRH B)\TRH C\,.
\end{multline}
\end{proof}
Note that $\TRH$ is well defined not on the full $\fA(\M)$, but on a smaller space (which is invariant under the renormalization group action) namely $\Dcal_{\TTR}(\M)\doteq\TTH(\F(\M))[[\hbar]]$. 
Similarly as in section \ref{general} we can use the renormalized time ordering operator $\TTH$ to bring classical structures to the quantum world. To avoid the notational inconvenience, we drop from now on the subscript $H$, so $\TTH$ will be now written as $\TTR$. In particular we can define  the time ordering of multilocal vector fields. Let $X\in\V(\M)$, then we define
\be\label{rprodv}
\TTR X\doteq \int dx \TTR( X(x))\frac{\delta}{\delta\ph(x)}\,.
\ee
Since $\TR$ is now defined as a full product on $\Dcal_{\TTR}(\M)$, we can repeat the reasoning from section \ref{general} and define the  $\TTR$-transformed Koszul operator with the renormalized time-ordered product in place of $\T$. Let $S\in\F_\loc(\M)$ be the free action functional. The renormalized time ordered Koszul map is defined as
\[
\delTR\doteq\TTR\circ\delta_{\mathcal{T}^{\minus}S}\circ\TTR^{-1}
\]
Clearly it is a well defined object and no divergences are present. We can also define the time-ordered antibracket:
\[
\{X,Y\}_{\sst{\TTR}}=\TTR\{\TTR^{-1}X,\TTR^{-1}Y\}
\]

Definitions introduced above allow us to provide a mathematically rigorous interpretation of the renormalized quantum BV operator and the renormalized {\qme}. Before we turn to this task we want to remark on the problems encountered in other approaches to the BV quantization. Note that the source of divergences in expression (\ref{QMO}) is the operator $\Lap$, which is ill defined on local vector fields.  In the standard approach this is solved by using an appropriate regularization scheme. Instead we argue, that this problem can be completely avoided if  we work with renormalized time ordered products $\TTR$ from the very beginning. We shall follow now all the steps outlined in \ref{general} and see what is changing when we take the renormalization into account.
%
\section{The renormalized quantum BV operator and the quantum master equation}\label{renormBV}
Now we have all the tools needed to introduce the interacting \textbf\textit{renormalized quantum BV operator}\index{BV!operator!quantum renormalized}. We start with the classical algebra $\BV(\M)$ underlying the BV-complex. It consists of functionals (elements with $\#\ta=0$) and derivations ($\#\ta>0$). The main difference with respect to the scalar case is the appearance of a grading. This implies that axioms for the time ordered products have to be modified by introducing appropriate sign rules. For example time ordered products of ghosts are antisymmetric instead of symmetric. Time ordered products of graded derivations are defined in the same way as time-ordered products of vector fields, i.e. by means of (\ref{rprodv}).

With these considerations in mind we can now set to define the renormalized BV operator. We can do it in a similar way to the non-renormalized case by using the expression (\ref{QBV0}) or  (\ref{QBV2}) with $\TT$ replaced by $\TTR$, namely:
\begin{align}\label{QBVr}
\hat{s}(X)&=e_{\sst{\TTR}}^{-iV/\hbar}\TR\left(\left(\frac{\delta}{\delta\ph^\ddagger(x)}( e_{\sst{\TTR}}^{iV/\hbar}\TR X)\right)\star\frac{\delta S}{\delta\ph(x)}\right)=\nonumber\\
&=e_{\sst{\TTR}}^{-iV/\hbar}\TR\left(\{e_{\sst{\TTR}}^{iV/\hbar}\TR X,S\}_{\star}\right)\,,
\end{align}
where $V, X\in\TTR(\BV_\loc(\M))$. To understand better this expression we shall use the anomalous Master Ward Identity ({\mwi}) \cite{BreDue,DueBoas}. In our formalism it takes the following form:
\begin{prop}
Let $X\in\TTR(\V_\loc(\M))$, $F,S\in\TTR(\F_\loc(\M))$, then it holds:
\be\label{MWI}
\int dx (e_{\sst{\TTR}}^{iF/\hbar}\TR X(x))\star\frac{\delta S}{\ph(x)}=e_{\sst{\TTR}}^{iF/\hbar}\TR(\{X,F+S\}_{\TTR}-\Lap_F(X))\,,
\ee
where $\Lap_F(X)$ is the anomaly\footnote{In the original paper \cite{BreDue} the anomaly term is denoted by $\Lap_X(F)$. We use an opposite convention since it resembles more the notation used for the Laplacian operator $\Lap$ defined on the regular vector fields.}. It is of order $\mathcal{O}(\hbar)$ and can be written in the form $\Lap_F(X)=\sum\limits_{n=0}^\infty \Lap^{(n)}(F^{\otimes n};X)$, where each $\Lap^{\!(n)}$ is local, linear in $X$ and  $\Lap^{\!(n)}(F^{\otimes n};X)=0$ if $\supp X\cap \supp F=\varnothing$.  
\end{prop}
\begin{proof}
The proof bases on \cite{BreDue,DueBoas,H}.
Firstly we rewrite (\ref{MWI}) in a slightly different form:
\be\label{anom1}
e_{\sst{\TTR}}^{iF/\hbar}\TR \!\bar{\Lap}_F(X)=\int dx \frac{\delta S}{\ph(x)}\cdot (e_{\sst{\TTR}}^{iF/\hbar}\TR X(x))-\delTR(e_{\sst{\TTR}}^{ iF/\hbar}\TR X)\,,
\ee
From (\ref{anom1}) it is clear that $\bar{\Lap}_F(X)$ describes the difference between the classical ideal generated by the equations of motion and the image of the time-ordered Koszul operator. Equivalently one can write it as:
\be\label{anom2}
\bar{\Lap}_F(X)=e_{\sst{\TTR}}^{-iF/\hbar}\TR\int dx \frac{\delta S}{\ph(x)}\cdot (e_{\sst{\TTR}}^{iF/\hbar}\TR X(x))-\delTR(X)\,.
\ee
In order to show that $\bar{\Lap}_F(X)$ is local we expand equation (\ref{anom1}) in powers of $F$. It follows that 
$\bar{\Lap}^{(0)}(1;X)=-\delta_S(X)$ and for $n>0$ we obtain:
\[
\bar{\Lap}^{(n)}(F^{\otimes n};X)=\Big(\frac{i}{\hbar}\Big)^n\int\! dx\,F^{n,\sst{\TTR}}\TR X(x)\cdot\frac{\delta S}{\ph(x)}-\sum\limits_{k=1}^n\Big(\frac{i}{\hbar}\Big)^k\binom{n}{k}F^{k,\sst{\TTR}}\TR\bar{\Lap}^{(n-k)}(F^{\otimes (n-k)};X)\,,
\]
where $F^{n,\sst{\TTR}}\doteq \underbrace{F\TR\cdots \TR F}_{n}$. From this formula it is already clear that all the operators $\bar{\Lap}^{(n)}(F^{\otimes n};X)$ are well defined. It remains to show that they are local. To this end we consider functionals $F_1,\ldots, F_n$ such that supports of $F_1,\ldots, F_k$ are later than supports of $F_{k+1},\ldots, F_n$ and the support of $X$. The proof will be done by induction. Assume that the hypothesis holds  up to order $n-1$, so $\bar{\Lap}^{(m)}(F_1\otimes\ldots\otimes F_m;X)=0$, $m\leq n-1$. We want to show that
\begin{multline*}
\bar{\Lap}^{(n)}(F_1\otimes\ldots\otimes F_n;X)=\Big(\frac{i}{\hbar}\Big)^n\int\! dx\,F_1\TR\ldots\TR F_n \TR X(x)\cdot\frac{\delta S}{\ph(x)}+\\-\sum\limits_{I\subset\underline{n}\atop I\neq\varnothing}\Big(\frac{i}{\hbar}\Big)^{|I|}\bigg(\,{\prod_{i\in I}}^{\sst{\TTR}}F_i\bigg)\TR\bar{\Lap}^{(n-|I|)}\big(\bigotimes_{j\leq n\atop j\notin I} V_j;X\big)=0\,,
\end{multline*}
where ${\prod\limits_{i\in I}}^{\sst{\TTR}}$ is the product with respect to $\TR$, indexed by $I$. Now we can use the causal factorization property (\ref{CausFact}) to both terms on the r.h.s. of the above formula. It follows from the induction hypothesis that
\begin{multline*}
\bar{\Lap}^{(n)}(F_1\otimes\ldots\otimes F_n;X)=\Big(\frac{i}{\hbar}\Big)^nF_{1}\TR\ldots\TR F_k\star\bigg(\int\! dx\,F_{k+1}\TR\ldots\TR F_n\TR X(x)\cdot\frac{\delta S}{\ph(x)}\bigg)+\\
-\Big(\frac{i}{\hbar}\Big)^kF_{1}\TR\ldots\TR F_k\star\bigg(\sum\limits_{J\subset\{k+1,...,n\}}\Big(\frac{i}{\hbar}\Big)^{|J|}\bigg(\,{\prod_{i\in J}}^{\sst{\TTR}}F_i\bigg)\TR\bar{\Lap}^{(n-k-|J|)}\big(\bigotimes_{k<l\leq n\atop l\notin J} F_l;X\big)\bigg)=\\
\Big(\frac{i}{\hbar}\Big)^kF_{1}\TR\ldots\TR F_k\star\bigg(\Big(\frac{i}{\hbar}\Big)^{n-k}\int\! dx\,F_{k+1}\TR\ldots\TR F_n\TR X(x)\cdot\frac{\delta S}{\ph(x)}+\\
-\sum\limits_{J\subset\{k+1,...,n\}\atop J\neq\varnothing}\Big(\frac{i}{\hbar}\Big)^{|J|}\bigg(\,{\prod_{i\in J}}^{\sst{\TTR}}F_i\bigg)\TR\bar{\Lap}^{(n-k-|J|)}\big(\bigotimes_{k<l\leq n\atop l\notin J} F_l;X\big)-\bar{\Lap}^{(n-k)}(F_{k+1}\otimes\ldots\otimes F_{n};X)\bigg)=\\
=\Big(\frac{i}{\hbar}\Big)^kF_{1}\TR\ldots\TR F_k\star\bigg(\bar{\Lap}^{(n-k)}(F_{k+1}\otimes\ldots\otimes F_{n};X)-\bar{\Lap}^{(n-k)}(F_{k+1}\otimes\ldots\otimes F_{n};X)\bigg)=0\,.
\end{multline*}
In the first term we also made use of the identity $\int dx\, G\star \big(Y(x)\cdot \frac{\delta S}{\ph(x)}\big)=\int dx\big(G\star Y(x))\cdot \frac{\delta S}{\ph(x)}\big)$ for arbitrary $Y\in\V_\mc(M)$, $G\in\F_\mc(M)$. This proves the induction step. We can now separate from $\bar{\Lap}_F$  the contribution from $\{X,F\}_{\TTR}$, which is also local and vanishes if $X$ and $F$ have disjoint supports. This way we obtain $\Lap_F(X)=\bar{\Lap}_F(X)+\{X,F\}_{\TTR}$. It remains to proof that $\Lap_F(X)$ is of order $\Ocal(\hbar)$. We use the argument of \cite{BreDue,DueBoas,DF02} that bases on the fact that \textsc{mwi} is satisfied in the classical theory. Firstly we can multiply both sides of the equation (\ref{MWI}) with the inverse of $e_{\sst{\TTR}}^{iF/\hbar}$ with respect to the $\star$-product:
\begin{multline}
(e_{\sst{\TTR}}^{iF/\hbar})^{-1,\star}\star e_{\sst{\TTR}}^{iF/\hbar}\TR\delTR( X)=\int dx \frac{\delta S}{\ph(x)}\cdot (e_{\sst{\TTR}}^{iF/\hbar})^{-1,\star}\star\TTR(e^{iF/\hbar}\cdot X(x))+\\-(e_{\sst{\TTR}}^{iF/\hbar})^{-1,\star}\star e_{\sst{\TTR}}^{iF/\hbar}\TR(\{X,F\}_{\TTR}-\!\Lap_F(X))
\end{multline}
Using the Bogoliubov's formula relating the time-ordered product with the retarded product $R_F(G)=(e_{\sst{\TTR}}^{iF/\hbar})^{-1,\star}\star (e_{\sst{\TTR}}^{iF/\hbar}\TR G)$ we obtain:
\be\label{MWIR}
R_F(\delTR( X)+\{X,F\}_{\TTR}+\Lap_F(X))=\int dx \frac{\delta S}{\ph(x)}\cdot R_F(X(x))\,.
\ee
  This is the \textsc{mwi} expressed in terms of retarded products. It was shown in \cite{DF02,BreDue} that in the classical theory (\ref{MWIR}) holds without the anomaly, so $\Lap_F(X)$ must be of order $\Ocal(\hbar)$.
\end{proof}
This identity can be also easily generalized for the case $F\in\BV_\loc(\M)$ (see \cite{H} for details). 
Let us take the interaction of the form $F=F_0+\int dx F_1(x)\frac{\delta}{\delta\ph(x)}$ with an even ghost number. It follows from (\ref{MWI}) that:
\begin{align*}
\hat{s}e_{\sst{\TTR}}^{iF/\hbar}&=e_{\sst{\TTR}}^{iF/\hbar}\TR(\{F,F+S\}_{\TTR}-\Lap_{F}(F_1))=\\
&=e_{\sst{\TTR}}^{iF/\hbar}\TR\left(\frac{1}{2}\{F+S,F+S\}_{\TTR}-\Lap_{F}(F_1)\right)\,,
\end{align*}
Comparing this formula with (\ref{QME0}) we see that \textsc{mwi} provides us with means to formulate the \textit{\textbf{renormalized quantum master equation}}\index{master equation!quantum renormalized}. The divergent term $\Lap$ is now replaced with the finite ``anomaly'' term. To keep a uniform notation we can trivially extend  $\Lap_{(.)}(F_0)$ to the case when $F_0$ is an element with $\#\ta=0$. We just set $\Lap_{(.)}(F_0)=0$. We obtain the renormalized {\qme} in the form:
\be\label{QMEr}
\frac{1}{2}\{F+S,F+S\}_{\TTR}=\Lap_{F}(F)\,.
\ee 
Just like in the non-renormalized case, fulfilling the {\qme} (\ref{QMEr}) is equivalent to the invariance of  the extended S-matrix under the quantum Koszul operator. This guarantees that the equation (\ref{intertwining:s}) is fulfilled also in the renormalized case, i.e.:
\be\label{intertwining:s:r}
\{.,S\}_\star\circ R_F=R_F\circ\hat{s}\,,
\ee
where  $R_F(G)=(e_{\sst{\TTR}}^{iF/\hbar})^{-1,\star}\star (e_{\sst{\TTR}}^{iF/\hbar}\TR G)$. We see that the interpretation of $\hat{s}$  as the $R_F$-transformed $\{.,S\}_\star$ carries over also to the renormalized theory. 

Using (\ref{QMEr}) we obtain for an arbitrary element $\TTR X\in\TTR(\BV(\M))$ a following simple expression for the renormalized quantum BV-operator:
\[
\hat{s}X=\{X,F+S\}_{\TTR}-\Lap_F(X)\,.
\]
We want to stress once again that in the construction we performed no divergences appear in the intermediate steps, since we work all the time with well defined objects, i.e. renormalized time-ordered products $\TTR$. It is maybe a  good moment to stop for a while and think back about the guiding principle of the {\paqft}. By going to a little bit abstract level we got a fresh look at the BV quantization. We used the path integral approach only as a guideline and an heuristic principle to relate abstract constructions with objects commonly used in practical calculations. The main technical tool of our approach is the theory of distributions and in particular microlocal analysis. By using these mathematical methods one can formulate the renormalization as the distributions' extension problem. In the same spirit we defined the BV operator and the {\qme} first on regular objects and then we extended them consistently to more singular ones. The ambiguity involved in this process is just the well understood renormalization freedom resulting from extension of time-ordered products to functionals with coinciding supports. The definition of {\qme} we provided contains in place of a singular operator $\Lap$, a finite expression, which is commonly called \textit{the anomaly term} of the anomalous master Ward identity. Clearly the appearance of $\Lap_F(F)$ in equation (\ref{QMEr}) is expected, since the same structure was present in case of the non-renormalized products. From this point of view we don't speak of any ``anomalous'' behavior here. Moreover the term $\Lap_F(X)$ in {\mwi} is just the reflection of the fact that certain structures transform nontrivially if we go from the classical to the quantum case. The appearance of such behavior is not surprising at all.
Also the fact that certain classical symmetries cannot be recovered in the quantum case is quite natural. Indeed, it is rather exceptional, that some of them can be obtained. Therefore we shouldn't be too much worried about the fact, that the quantum master equation differs from the classical one. Instead we will now take a closer look at the nature of the anomaly term and try to understand it better, by formulating certain consistency conditions.

First we note that in equation (\ref{QBVr}) the star product with $\frac{\delta S}{\delta\ph(x)}$ or $\frac{\delta S}{\delta\ph^\ddagger(x)}$ can be also replaced by the pointwise product and therefore:
\[
\{\{e_{\sst{\TTR}}^{iV/\hbar}\TR X,S\}_{\star},S\}_\star=0
\]
This implies that:
\[
\hat{s}^2( X)=0
\]
From this condition and the classical master equation (\ref{CME0}) for $S+V$ it follows that
\[
\{\Lap_V(X),V+S\}_{\TTR}+\Lap_V(\{X,V+S\}_{\TTR})+\Lap_V(\Lap_V(X))=0
\]
Note that $\{.,V+S\}_{\TTR}$ is just the classical BV operator $s=\{.,\TTR^{\minus}(V+S)\}$ transported to the quantum algebra by means of $\TTR$. Therefore, if $\TTR^{\minus}X$ is invariant under $s$, then also $\{X,V+S\}_{\TTR}=0$ and we obtain a condition analogous to the Wess-Zumino consistency condition \cite{Wess}:
\be\label{constcond2}
\{\Lap_V(X),V+S\}_{\TTR}=-\Lap_V(\Lap_V(X))
\ee
\subsection{Quantum field theory on globally hyperbolic spacetimes}
One of the great advantages of the algebraic formalism is the relatively simple generalization from $\M$ to general globally hyperbolic spacetimes $M\in\obj(\Loc)$. In this section we recall basic definitions and theorems concerning quantum field theory on curved spacetimes, following \cite{BFV}. Construction of the quantum algebra performed in section \ref{quantalg} is done in a similar way, but one has to choose a reference Hadamard state $\omega$. This doesn't affect the algebraic structure, since the resulting quantum algebras are $*$-isomorphic.
 Also Wick products can be constructed as in section \ref{renprod}, but there is an additional renormalization freedom in defining them (i.e. additional freedom in constructiong $\TTH^1$). The detailed discussion of this fact is provided in \cite{BF0,HW}. An abstract quantum algebra $\fA(M)$ assigned in this way to a spacetime $M$ has to fulfill the axioms of locally covariant quantum field theory postulated in \cite{BFV}. The structure is similar to the classical case discussed in section \ref{axioms}. Now the category of observables is the category of $*$-algebras:
\begin{itemize}
\item[$\Obs$]
\begin{itemize}
 \item[$\ $] $\obj(\Obs)$: topological $*$-algebras \index{category!of algebras@of $*$-algebras}
 \item[$\ $] {\bf Morphisms}: continuous $*$-homomorphisms
  \end{itemize}
\end{itemize}
A \textbf{\textit{locally covariant quantum field theory}}\index{locally covariant!field theory!quantum} $\fA$ is defined to be a covariant functor between the categories $\Loc$ and $\Obs$ satisfying the properties of causality ($\fA^\otimes$ is a tensor functor) and the time-slice-axiom.
In the framework of locally covariant quantum field theory a concept of a \textit{\textbf{field}} also needs generalisation. Following \cite{BFV} we define a locally covariant quantum field\index{locally covariant!field!quantum}  $\Phi$ as a natural transformation between the functors $\D$ and $\fA$. This means, that to any isometric embedding $\chi: M \longrightarrow N$ it associates the following commutative diagram:
$$\begin{CD}
\D(M_1) @>\Phi_{(M_1)}>> \fA(M_1)\\
@V\chi_* VV  @VV\alpha_\chi V\\
\D(M_2) @>\Phi_{(M_2)}>>\fA(M_2)
\end{CD}$$
where $\chi_*$ is the push forward under $\D$. Clearly the field $\Phi\equiv(\Phi_{M})_{M\in\Loc}$ has the following covariance property:
\begin{equation}
\fA\chi \circ \Phi_{M_1} = \Phi_{M_2} \circ \chi_\ast\ .
\end{equation}
It was shown in \cite{BFV} that \textit{\textbf{local Wick polynomials}} are examples of such locally covariant quantum fields if certain cohomological condition is fulfilled. We recall briefly this argument here. The enlarged local algebras formed by the Wick polynomials defined in \cite{BF0} can be constructed with quasifree Hadamard states $\omega$ of the free field over globally hyperbolic spacetimes. It can be proved that it satisfies the condition of local covariance. However, it was shown in \cite{BFV,HW} that the Wick polynomials $:\ph^k:_\omega$ themselves are in general not locally covariant quantum fields in the sense of the above definition. Nevertheless  construction of local Wick monomials is possible along the lines of \cite{BFV}. Following this reference we give here an example of the construction of the Wick-square.
\begin{exa}
Let $\omega=\omega_{(M',g')}$ and $\omega'=\omega'_{(M',g')}$ be two
quasifree Hadamard states over the spacetime $(M',g')$.Then there is a smooth
function $B_{\omega,\omega'}$ on $M'$ such that $:\!\Phi^2\!:_{\omega}\!(x')
-\,:\!\Phi^2\!:_{\omega'}\!(x') \, = B_{\omega,\omega'}(x')$ and the 
covariance and cocycle condition
\[ B_{\omega,\omega'} + B_{\omega',\omega''} + B_{\omega'',\omega} =
0\,.\]
are satisfied. The problem of constructing locally covariant Wick monomials can be now reduced to finding a family of smooth functions $f_{\omega_{(M',g')}} \in C^{\infty}(M')$ that trivialize the above cocycle, i.e.
\[ B_{\omega,\omega'}(x') = f_{\omega}(x') - f_{\omega'}(x')\,, \quad x'
\in M'\,.
\]
If this cohomological problem can be solved, then a locally covariant Wick-square can be defined by setting
$ :\!\Phi^2\!:_{(M,g)}\!(x) = \, :\!\Phi^2\!:_{\omega_{(M,g)}}\!(x) -
f_{\omega_{(M,g)}}(x)$ 
for an arbitrary quasifree Hadamard state $\omega_{(M,g)}$ over $(M,g)$. It was shown in \cite{BFV} that the solution to this cohomological problem is provided by $H_{\omega}$, the smooth,
non-geometrical term of the two-point function of a quasifree Hadamard state $\omega$, i.e.
$  f_{\omega}(x) = H_{\omega}(x,x) \,, \quad x \in M$.
\eex\end{exa}
The notion of locally covariant quantum fields recalled here agrees also with the notion used by Hollands and 
Wald \cite{HW,HW2}.
\subsection{Algebraic adiabatic limit}\label{adiablim}
In the previous section we discussed how the free field theory can be formulated on generic spacetimes $M\in\obj(\Loc)$.
Now it's time to introduce the interaction. We will review now basic definitions and theorems from locally covariant quantum field theory, which will be useful later on in the context of the {\qme}. The main point of this section is to go one level of abstraction further and introduce the interaction on the level of natural Lagrangians. In the spirit of locally covariant quantum field theory also the S-matrix and the renormalization group can be understood in this setting. Recall that in the present framework Lagrangians themselves are natural transforamations and therefore, it is justified to ask, how the renormalization group acts on such objects. Besides, we have already shown in section \ref{CME} that the right framework for the formulation of the classical master equation involves the notion of natural transformations. Therefore one expects that also the {\qme} should be investigated in this context. Moreover the formulation in terms of natural Lagrangians allows us to take the so called \textit{\textbf{algebraic adiabatic limit}}. The construction provided here follows  \cite{BDF} and is similar to the classical one, which we recalled in section \ref{inter}. The only difference lies now in underlying categories.
\begin{defn}\label{df:af}
A generalized Lagrangian $L$\index{generalized Lagrangian} is a natural transformation between functors $\D$ and $\fA_\loc$, both seen as functors into the category of topological spaces (we dropped the linearity condition!), satisfying:
\begin{itemize}
\item $\supp(L(f))\subset\supp(f)\ $;
\item $L(0)=0$;
\item $L(f+g+h) =L(f+g) -L(g)+ L(g+h)\ ,$ if $\ \supp(f)\cap\supp(h)=\emptyset\ $;
\end{itemize}
\end{defn}
Just like in the classical case (section \ref{dyn}) we define the action $S$ to be an equivalence class of Lagrangians $S(L)$ in the sense of (\ref{equ}). We recall, that the proposition 6.2 in \cite{BDF} states that the space of generalized Lagrangians is invariant under the action $L\mapsto Z\circ L$ of the renormalization group, $Z\in\Rcal$. Moreover, if $L\sim L'$ in the sense of (\ref{equ}), then the same holds for $Z\circ L$ and $Z\circ L'$, so the renormalization group has a well defined action on the space of actions (equivalence classes of Lagrangians). We set
\[
Z_{S_M(L)}(A)=Z_{L_M(f)}(A)
\] 
with $f\equiv 1$ on a neighbourhood of $\supp\, A$. By the support properties of $L$ and by (\ref{suppZ})
the right hand side does not depend on $f$. 
Let $\Ocal$ be a relatively compact open subregion of spacetime $M$. 
The abstract algebra of observables associated to $\Ocal$ should be independent of an interaction swiched on outside
 of $\Ocal$.  Let us define
\[
\mathscr{V}_{S_1}(\Ocal) \doteq \{ V\in\fA_{loc}(M)\ |\ \supp(V-{L_1}_M(f))\cap\overline{\Ocal}=\varnothing,
\text{ if } [L_1] =S_1\text{ and } f\equiv 1 \text{ on }\Ocal \}\,,
\]
Note hat $Z_V(F)=Z_{S_1}(F)$ if $V\in \mathscr{V}_{S_1}(\Ocal)$ and $\supp\, F\subset \Ocal$ and moreover
\[
V\in\mathscr{V}_{S_1}(\Ocal) \Leftrightarrow Z(V)\in\mathscr{V}_{Z(S_1)}(\Ocal)\,, 
\]
We shall use the notation $Z(S_1)$ for expressing the equivalence class of the transformed Lagrangian $Z(L_1)$.
The relative $S$-matrix in the algebraic adiabatic limit is defined by
\[
\Scal^\Ocal_{S_1}(F)=(\Scal_V(F))_{V\in\mathscr{V}_{S_1}(\Ocal)}
\] 
for $F\in\fA(M)$ with $\supp\, F\subset \Ocal$. The $S$-matrix defined this way is a covariantly constant section  in the sense that for any $V_1,V_2\in\mathscr{V}_{S_1}(\Ocal)$ there exists an automorphism
$\beta$ of $\fA(M)$ such that
\[
\beta(\Scal_{V_1}(F))=\Scal_{V_2}(F)\quad\ \forall F\in\fA_{loc}(M)\ ,\ \supp\, F\subset \Ocal\,.
\]
Compare this with the classical structure defined in section \ref{inter}. The interpretation of $\Scal^\Ocal_{S_1}(F)$ as the algebraic adiabatic limit is justified, since the abstract algebra generated by $\Scal_V(F)\ ,\ \supp\, F\subset \Ocal$ is independent of the choice of $V\in \mathscr{V}_{S_1}(\Ocal)$. The interacting local algebra $\fA_{S_1}(\Ocal)$ of observables in the algebraic adiabatic limit is generated by the elements $\Scal^\Ocal_{S_1}(F)$, $\supp\, F\subset\Ocal$. This assignment can be made into a covariant functor. Let $\chi:M\rightarrow N$ be an isometric embedding, then the embedding $\fA_{S_1}\chi:\fA_{S_1}(M)\hookrightarrow \fA_{S_1}(N)$
is induced by $\fA_{S_1}\chi\,(\Scal^{M}_{S_1}(F))=\Scal^{N}_{S_1}(F)$.

Now, following \cite{BDF} we define the action of the renormalization group on the level of natural transformations.  We want to see how the S-matrix in the adiabatic limit is behaving under the renormalization. Let $Z\in\Rcal$ and $\hat{\Scal}=\Scal\circ Z$. Then
\[
\hat{\Scal}^\Ocal_{S_1}(F)=(\hat{\Scal}_V(F))_{V\in\mathscr{V}_{S_1}(\Ocal)}=
(\Scal_{Z(V)}(Z_V(F)))_{V\in\mathscr{V}_{S}(\Ocal)}
=\Scal^\Ocal_{Z(S_1)}(Z_{S_1}(F)) \,.
\]
The renormalization group flow on the level of natural Lagrangians can be understood by means of the following theorem  \cite{BDF}:
\begin{thm}[Algebraic Renormalization Group Equation]
Let $\hat{\fA}_{S_1}(\Ocal)$ and $\fA_{S_1}(\Ocal)$ 
denote the algebra of observables obtained by using $\hat{\Scal}$ and $\Scal$, respectively. 
The renormalization group element $Z\in\mathcal{R}$ induces an isomorphism 
$\al_Z=(\al_Z^\Ocal)$ of the nets,
\begin{equation} \label{RGE}
\al_Z^\Ocal: \hat{\fA}_{S_1}(\Ocal)\to \fA_{Z(S_1)}(\Ocal)\ , 
\end{equation}
such that $\iota_{\Ocal_2\Ocal_1}\circ\al_Z^{\Ocal_1}=\al_Z^{\Ocal_2}\circ\hat\iota_{\Ocal_2\Ocal_1}$ for $\Ocal_1\subset\Ocal_2$.
The isomorphism is given by 
\begin{equation}
\al^\Ocal_Z(\hat{\Scal}^\Ocal_{S_1}(F))=\Scal^\Ocal_{Z(S_1)}(Z_{S_1}(F))\ .
\end{equation} 
In particular, if $L$ and $Z(L)$ induce the same interaction, $\alpha_Z$ is an automorphism.
\end{thm}
\subsection{Quantum master equation in the algebraic adiabatic limit}\label{adiablimqme}
The idea to generalize the renormalization group to the level of natural transformations may seem a little bit abstract at the beginning. It is however very useful, if we want to have control on the cutoff needed to localize the interaction. In this section we will show that also the quantum master equation appears naturally in this setting. The idea is similar to the classical case discussed in \ref{CME}. Working on the level of natural transformations we avoid problems with boundary terms arising from the cutoff function. The interpretation of the {\qme} which we provide here is new and wasn't discussed in the literature before.

Let $S_0$ be the free action and $S_1$ the interaction term. Both are now to be understood as equivalence classes of natural transformations. The quantum BV operator is defined as
\be\label{eQBVr}
\hat{s}(X)=e_{\sst{\TTR}}^{-i{S_1}_M(f_1)/\hbar}\TR\left(\{e_{\sst{\TTR}}^{i{S_1}_M(f_1)/\hbar}\TR  X,{S_0}_M(f)\}_{\star}\right)\,,
\ee
where $\supp\, X\subset\Ocal$ and $f,f_1\equiv 1$ on $\Ocal$. The quantum master equation is a statement that the S-matrix in the algebraic adiabatic limit is invariant under the quantum BV operator, i.e.:
\[
\supp\left(e_{\sst{\TTR}}^{-i{S_1}_M(f_1)/\hbar}\TR\left(\{e_{\sst{\TTR}}^{i{S_1}_M(f_1)/\hbar},{S_0}_M(f)\}_{\star}\right)\right)\subset \supp\, df\cup \supp\, df_1\,.
\]
Using {\mwi} we can see that this expression is again an element of $\fA_\loc(M)$, so the condition above can be also formulated on the level of natural transformations:
\be\label{eQMEr0}
e_{\sst{\TTR}}^{-iS_1/\hbar}\TR\left(\{e_{\sst{\TTR}}^{i S_1/\hbar},S_0\}_{\star}\right)\sim 0\,,
\ee
This is the extended quantum master equation. We can write it in a more explicit form using (\ref{MWI}). Note that the anomaly term $\Lap_{{S_1}_M(f)}({S_1}_M(f))$ is a natural transformation as well, so (\ref{eQMEr0}) is equivalent to:
\be\label{eQMEr1}
\frac{1}{2}\{S_0+S_1,S_0+S_1\}_{\TTR}-\Lap_{S_1}(S_1)\sim 0\,.
\ee
Note the resemblance of this condition to the classical master equation {\cme} (\ref{CME0}). The quantum BV operator can be now written as
\[
\hat{s}X=\{X,{S_0}_M(f)+{S_1}_M(f)\}_{\TTR}-\Lap_{{S_1}_M(f)}(X)\,,
\]
where $f\equiv 1$ on the support of $X\in\TTR(\BV(M))$. Now we want to see how the {\qme} and the quantum BV operator are transforming under the renormalization group. We start with the  {\qme}.
\begin{prop}\label{renorm:qme}
Let $L_1$ be a natural Lagrangian that solves the {\qme} (\ref{eQMEr0}) for the renormalized time-ordered product $\TTR$. Let $Z\in\Rcal$ be the element of the renormalization group, which transforms between the $S$-matrices corresponding to $\TTR$  and ${\TTR}'$, i.e. $e_{\sst{\TTR}}^{{L_1}_M(f)}=e_{\sst{{\TTR}'}}^{Z({L_1}_M(f))}$. Then $Z(L_1)$ solves the {\qme} corresponding to ${\TTR}'$.
\end{prop}
\begin{proof}
From the equation  (\ref{eQMEr0})  follows that there exists a local element $A(f,f_1)\in\fA_\loc(M)$, depending on test functions $f$, $f_1$, such that $\supp(A(f,f_1))\subset \supp\, df)\cup \supp\,df_1$, such that:
\[
\{e_{\sst{\TTR}}^{i {L_1}_M(f_1)/\hbar},{L_0}_M(f)\}_{\star}=e_{\sst{\TTR}}^{-i{L_1}_M(f_1)/\hbar}\TR A(f,f_1)
\]
We can now transform both sides with the renormalization group element $Z$ to obtain:
\[
\{e_{\sst{{\TTR}'}}^{i Z({L_1}_M(f_1))/\hbar},{L_0}_M(f)\}_{\star}=e_{\sst{{\TTR}'}}^{-iZ({L_1}_M(f_1))/\hbar}\cdot_{{\TTR}'} \langle Z^{(1)}({L_1}_M(f_1)),A(f,f_1)\rangle\,.
\]
Using the property (\ref{suppZ}) of the renormalization group, we can conclude that 
\[
\supp( \langle Z^{(1)}({L_1}_M(f_1)),A(f,f_1)\rangle)\subset\supp(A(f,f_1))\,.
\]
Hence:
\[
\supp\left(e_{\sst{{\TTR}'}}^{-iZ({L_1}_M(f_1))/\hbar}\cdot_{{\TTR}'}\{e_{\sst{{\TTR}'}}^{i Z({L_1}_M(f_1))/\hbar},{L_0}_M(f)\}_{\star}\right)\subset \supp\, df\cup \supp\, df_1\,.
\]
\end{proof}
We can see from the above proposition that the {\qme} is indeed an universal notion and transforms correctly under the renormalization group. A similar property can be shown for the $BV$ operator. To distinguish between operators corresponding to different interaction terms we denote by $\hat{s}_{\sst{S_1}}$ the quantum $BV$ operator defined for the action $S_1$ with respect to the time-ordering operator $\TTR$. 
For a different time ordering $\TTR'$ we obtain a corresponding  operator $\hat{s}'_{\sst{S_1}}$ in the form:
\be\label{s:hat:strich}
\int dx (e^{iS_1/\hbar}\cdot_{\TTR'} X(x))\star\frac{\delta S_0}{\ph(x)}= e_{\sst{\TTR'}}^{iS_1/\hbar}\cdot_{\TTR'}\hat{s}'_{\sst{S_1}}(X)\,,
\ee
On the other hand we know from the main theorem of renormalization \ref{MTR} that there exists an element $Z\in\Rcal$ such that the left hand side of the above formula can be written as:
\begin{multline*}
\int\! dx\, (e_{\sst{\TTR'}}^{i S_1/\hbar}\cdot_{\TTR'}X(x))\star\frac{\delta S_0}{\ph(x)}=\int\! dx\, e_{\sst{\TTR}}^{iZ(S_1)/\hbar}\TR \langle Z^{(1)}(S_1),X(x)\rangle\star\frac{\delta S_0}{\ph(x)}=\\
=e_{\sst{\TTR}}^{iZ(S_1)/\hbar}\cdot_{\TTR}(\hat{s}_{\sst{Z(S_1)}}\langle Z^{(1)}(S_1),X\rangle)\,,
\end{multline*}
Similarly we can rewrite the right hand side of (\ref{s:hat:strich}) as
\[
e_{\sst{\TTR'}}^{iS_1/\hbar}\cdot_{\TTR'}\hat{s}'_{\sst{S_1}}(X)=e_{\sst{\TTR}}^{iZ(S_1)/\hbar}\cdot_{\TTR}\langle Z^{(1)}(S_1),\hat{s}'_{\sst{S_1}}(X)\rangle\,.
\]
By comparing above formulas we obtain:
\[
\hat{s}_{\sst{Z(S_1)}}\langle Z^{(1)}(S_1),X\rangle=\langle Z^{(1)}(S_1),\hat{s}'_{\sst{S_1}}(X)\rangle\,.
\]
Since it holds for arbitrary $X$, we can write the above relation as:
\be
\hat{s}_{\sst{Z(S_1)}}\circ Z^{(1)}(S_1)=Z^{(1)}(S_1)\circ\hat{s}'_{\sst{S_1}}\,.
\ee
This means that also the quantum BV operator transforms under the renormalization group in the natural way.

To end this section we want to discuss the problem of finding a solution to the {\qme}.
We start with a classical action $\TTR^{\minus}(S_0+S_1)$, which satisfies the {\cme}, i.e. $\{S_0+S_1,S_0+S_1\}_{\TTR}=0$. For our renormalized time-ordering operator $\TTR$ we calculate the corresponding anomaly term $\Lap_{S_1}(S_1)$. In general it doesn't vanish, so the {\qme} will not be fulfilled. There are basically two possibilites to proceed. We can either redefine $\TTR$ using the renormalization freedom, or try to absorb $\Lap_{S_1}(S_1)$ into the action, by adding higher loop order terms. 
 The second way is more in the spirit of the original formulation of the BV formalism \cite{Batalin:1981jr,Batalin:1983wj,Batalin:1983jr}, so we follow this path first. The cohomological problem can be formulated in the following way: we look for natural transformations $W_n$ such that $W=\sum_n \hbar^n W_n$, $W_0=S_1$ and it holds: 
\be\label{W}
\{W+S_0,W+S_0\}_{\TTR}-\Lap_{W}(W)\sim0
 \ee
 Let us expand $\Lap_{W}(W)$ as a power series in $\hbar$. Since $\Lap_X$ is linear in $X$, it follows that the lowest order term is just $\Lap_{S_1}(S_1)$. Therefore in the first order $\hbar$ this gives a condition:
 \be\label{firstorder}
 2\{W_1,S_0+S_1\}_{\TTR}-\frac{1}{\hbar}\Lap^1_{S_1}(S_1)\sim 0\,.
 \ee
From the consistency condition (\ref{constcond2}) in the first order in $\hbar$, we know that 
\[
\{\Lap^1_{S_1}(S_1),S_0+S_1\}_{\TTR}\sim0\,.
\]
Therefore the solution $W_1$ to (\ref{firstorder}) is given by the cohomology of $s$ on the space of actions. To understand better this cohomological problem, recall that action is an equivalence class of Lagrangians and these are in turn characterized by maps from $\D(M)$ to $\fA_\loc(M)$. Therefore calculating the cohomology of $s$ on the space of actions effectively amounts to calculate the cohomology of $s$ modulo $d$ on the space of local forms (polynomials of fields and their derivatives). Results in this direction were obtained in \cite{BBH,BarHenn}. If the cohomology of $s$ turns out to be trivial, the existence of $W_1$ is guaranteed and we can insert it back to the equation (\ref{W}) and calculate the higher order terms. This way we can reduce construction of $W$ to a strictly cohomological problem. Finding a solution $W$ of the {\qme} provides us with a map $S_1\mapsto W$. Note that this map satisfies the axioms \ref{Z0}-\ref{Zindep} (for the additivity we need to use the additivity of $\Lap_{S_1}(S_1)$) and therefore it is an element of the renormalization group. We can write the {\qme} in the form:
\[
e_{\sst{\TTR}}^{-iZ(S_1)/\hbar}\TR\left(\{e_{\sst{\TTR}}^{i Z(S_1)/\hbar},S_0\}_{\star}\right)\sim 0\,,
\]
where $Z(S_1)=W$. From the main theorem of renormalization theory and proposition \ref{renorm:qme} it follows that there exists a time ordering operator ${\TTR}'$ such that:
\[
e_{\sst{{\TTR}'}}^{-iS_1/\hbar}\cdot_{\sst{{\TTR}'}}\left(\{e_{\sst{{\TTR}'}}^{i S_1/\hbar},S_0\}_{\star}\right)\sim 0\,,
\]
This way we showed that the violation of the {\qme} can be also absorbed into the redefinition of the time-ordered product. This approach agrees with the one taken in \cite{H}.
\subsection{Relation to the regularized {\qme}} 
The construction of the renormalized quantum BV operator and the {\qme} we propose is completely independent of any regularization scheme, but it is still interesting to see how our approach relates to ones involving an explicit regularization. I particular  in this section we want to make contact with the works of K. Costello \cite{Costello,CostBV}.  Following \cite{BDF} we define the regularized time-ordered product corresponding to the scale $\Lambda$ as $\Tcal_\Lambda\doteq \exp(i\hbar \HL)$, where
\[
\HL=\frac{1}{2}\int dxdy(h_\Lambda-H)(x,y)\frac{\delta^2}{\delta\ph(x)\delta\ph(y)}\,,
\]
 and $h_\Lambda\xrightarrow{\Lambda\rightarrow\infty}H_F$ in the sense of H\"ormander. It is evident that $h_\Lambda-H\xrightarrow{\Lambda\rightarrow\infty}i\Delta_D$ and it provides a regularization of the Dirac propagator. The regularized S-matrix is now defined as $\Scal_\Lambda\doteq \exp_{\Tcal_\Lambda}$ and the regularized time-ordered Koszul operator is given by $\deL\doteq\TTL\circ\delta_{{\TTL}^{\minus}S}\circ{\TTL}^{\minus}$.  The regularized quantum BV operator\index{BV!operator!quantum regularized} is defined by replacing $\TT$ with a regularized time-ordered product in (\ref{QBV0}), i.e
\begin{multline}\label{RegBV0}
\hat{s}_\La X=e_{\sst{\TTL}}^{-i F/\hbar}\TL\left(\left(\frac{\delta}{\delta\ph^\ddagger(x)}( e_{\sst{\TTL}}^{i F/\hbar}\TL X)\right)\star\frac{\delta S}{\delta\ph(x)}\right)=\\
=e_{\sst{\TTL}}^{-i F/\hbar}\TL\left(\{e_{\sst{\TTL}}^{i F/\hbar}\TL X,S\}_{\star}\right)\,.
\end{multline}
The regularized quantum master equation\index{master equation!quantum regularized} can be understood as the condition that the regularized S-matrix is invariant under the quantum Koszul operator, i.e.:
\[
\{e_{\sst{\TTL}}^{iF/\hbar},S\}_{\star}=0
\]
 Let $F=F_0+\int dz F_1(z)\frac{\delta}{\delta \ph(z)}$, where $F_0$ doesn't depend on antifields. We can write the regularized {\qme} explicitly using the fact that:
 \begin{eqnarray*}
 \deL(\Scal_\Lambda(F))&=&\TTL\left(\delta_{{\TTL}^{\minus}S}\left(e^{i {\TTL}^{\minus}F/\hbar}\right)\right)=m\circ e^{i\hbar \HL'}\left(\int dx\frac{\delta S}{\delta\ph(x)}\otimes\left(F_1(x)\TL e^{i F/\hbar}_{\TL}\right)\right)=\\
 &=&\int dx\frac{\delta S}{\delta\ph(x)}\left({\TTL}F_1(x)\TL e^{F}_{\TL}\right)+\TTL\left(i\hbar\Lap_\Lambda F-\frac{1}{2}\{F,F\}_\Lambda \right)\TL e^{F}_{\TL}\,,
 \end{eqnarray*}
 where by $\Lap_\Lambda$ we denoted the differential operator:
 \[
\Lap_\Lambda \doteq\int \!\!dxdydz \frac{\delta^2 S}{\delta\ph(z)\ph(x)}(h_\Lambda-H)(x,y)\frac{\delta^2}{\delta\ph^\ddagger(z)\delta \ph(y)}\,,
\]
and $\{.,.\}_\Lambda$ is the scale $\Lambda$ antibracket defined as:
\[
\{A,B\}_\Lambda\doteq \Lap_\Lambda(AB)-\Lap_\Lambda(A)B-(-1)^{|A|}A\Lap_\Lambda(B)\,.
\]
We can conclude that the scale $\Lambda$ QME is the condition that:
\[
\deL F+\frac{1}{2}\{F,F\}_\Lambda-i\hbar\Lap_\Lambda F=0
\]
This is exactly the form of the regularized {\qme} provided in \cite{Costello}. 
%
%
\subsection{Yang-Mills theory}\label{YMquant}
To end this chapter we take a look at an explicit example, where the BV construction is nontrivial. The most straightforward one is provided by Yang-Mills theories with the generalized Lagrangian given by (\ref{LagrYM}). To construct the Hadamard solution we need again a hyperbolic system of equations, so we have to fix the gauge. To this end we can use the classical structure of the BV complex $\BV_\mc(M)$ described in section \ref{KTcom}, extended by the non-minimal sector as stated in \ref{gaugefixing}. We perform a suitable automorphism $\alpha_\Psi$ of the algebra $\BV_\mc(M)$ with the gauge-fixing fermion (\ref{Lorenz}) and we arrive finally at the extended Lagrangian of the form:
\begin{multline}\label{lagrfixed}
L^{\ex}_M(f)=-\frac{1}{2}\int_M f\tr(F\wedge *F)+\gamma^g\Psi_M(f)+\int_M f\tr\Big(*DC\wedge\frac{\delta}{\delta A}\Big)+\\
+\frac{1}{2}\int_M\dvol\, f [C,C]\frac{\delta}{\delta C}-i\int_M\dvol\, f B\frac{\delta}{\delta \bar{C}}\,,
\end{multline}
where $\gamma^g$ is the BRST operator. The full BV operator is the sum $s=\delta^g+\gamma^g$. Using formula (\ref{gfixh}) we see that the gauge invariant observables are encoded in the cohomology of $s$. 
If we denote by $\theta^g$ the natural transformation that implements $\gamma^g$ locally, we can rewrite formula (\ref{lagrfixed}) on the level of natural transformations as:
\be\label{gfixed}
S^{\ex}=S+\{\theta^g,\Psi\}+\theta^g\,.
\ee
This is the starting point for the quantization. Note that the gauge is not yet fixed, since we don't put antifields equal to 0. It is usefull to expand equation (\ref{gfixed}) in the free and interacting part:
\be\label{gfix2}
S^{\ex}=S_0+S_1+\theta_0+\theta_1\,,
\ee
where $S_0$ and $S_1$ have the total antifield number equal to 0 and $\theta_0$ $\theta_1$ are of $\#\ta=1$. As suggested by the notation $S_0$ is the free action quadratic in fields and $\theta_0$ implements the free BRST differential $\gamma_0$, which is linear in fields and antifields. Also the Koszul differential can be expanded into the free and interacting part $\delta^g=\delta_0+\delta_1$.
 With this structure we can now apply the results presented in previous sections. From the construction we know that $S_0$ provides a hyperbolic system of equations and we can construct the corresponding Hadamard solution $H$ \cite{H}. The $\star$-product and time ordered product $\T$ are constructed basing on the free action $S_0$. Note that algebra $\BV(M)$ is graded, so the axioms for the renormalized time-ordered products have to be replaced by their graded counterparts. Details of the construction can be found in \cite{H}, including the definitions of Wick powers and time-ordered products of ghosts, antighosts and Nakanishi-Lautrup fields. 
We understand the time ordering of antifields, corresponding to all these variables, similarly to the scalar case, i.e. by means of (\ref{rprodv}). The time-ordered free Koszul operator is defined again as $\delta^{\TTH}_{S_0}$. The time-ordered free BRST operator is given by:
\[
\gamma_0^{\sst{\TTH}}F\doteq \{F,{\theta_0}_M(f)\}_{\TTH}\,.
\]
where $F\in\TTH(\BV(M))$ with $\#\ta(F)=0$  $f$ is a test function equal to 1 on the support of $F$. Again we drop at this point the subscript $H$ for the clearer notation. It is also convenient to go one level of abstraction higher and work with the natural transformations, along the lines of the section \ref{adiablim}. It was proved in \cite{H} that the {\mwi} holds in the form:
\begin{multline}\label{YMQBV}
e_{\sst{\TTR}}^{-i (S_1+\theta_1)_M(f)/\hbar}\TR\left(\{e_{\sst{\TTR}}^{i (S_1+\theta_1)_M(f)/\hbar}\TR  X,(S_0+\theta_0)_M(f)\}_{\star}\right)=\\
=\{X,S^{\ex}_M(f)\}_{\TTR}-\Lap_{(S_1+\theta_1)_M(f)}(X)\,,
\end{multline}
where $f\equiv1$ on the support of $X$. In \cite{H} it is also argued that using consistency conditions on the anomaly term one can use the renormalization freedom to redefine time-ordered products and obtain $\Lap_{S_1+\theta_1}(S_1+\theta_1)=0$. With this time-ordering prescription from equation (\ref{YMQBV}) follows that:
\[
\hat{s}\left(e_{\sst{\TTR}}^{i (S_1+\theta_1)/\hbar}\right)=e_{\sst{\TTR}}^{i (S_1+\theta_1)/\hbar}\TR\{S^{\ex},S^{\ex}\}_{\TTR}\,,
\]
in the sense of natural transformations, where $\hat{s}\doteq\{.,S_0+\theta_0\}_{\star}$. Since the {\cme} holds for the extended action, i.e. $\{S^{\ex},S^{\ex}\}_{\TTR}\sim 0$, we obtain:
\[
\hat{s}\left(e_{\sst{\TTR}}^{i(S_1+\theta_1)/\hbar}\right)\sim 0
\]
To see that this is the on-shell gauge invariance of the S-matrix in the adiabatic limit, note that the BRST operator in the quantized algebra is given by $\{.,\theta_0\}_{\star}$ and we can rewrite the above equation as:
\[
\left\{e_{\sst{\TTR}}^{i(S_1+\theta_1)/\hbar},\theta_0\right\}_\star\sim 0\quad\textrm{on-shell}\,.
\]
\chapter*{Conclusions}
\addcontentsline{toc}{chapter}{\protect\numberline{}Conclusions}
\vspace{-5ex}
\begin{flushright}
 \begin{minipage}{10cm}
\selectlanguage{russian}(...) всегда останется нечто, что ни за что не захочет выйти из-под вашего черепа (...) с тем вы и умрете, не передав никому, может быть, самого-то главного из вашей идеи.
\begin{flushright}
Ф. М. Достоевский, «Идиот»
\end{flushright}
\selectlanguage{english}
\vspace{1ex}
\textit{(...) There is always a something, a remnant, which will never come out from your brain (...) and you will die, perhaps, without having imparted what may be the very essence of your idea to a single living soul.}
\begin{flushright}
F. M. Dostoyevsky, \textit{The Idiot}
\end{flushright}
 \end{minipage}
\end{flushright}
\vspace{5ex}
\noindent\rule[2pt]{\textwidth}{1pt}
\vspace{1ex}\\
It is always challenging to transform ideas into words and words into formulas. The difficulty in mathematical physics lies in finding the right concepts to describe the physical phenomena and it is often hard to judge which formulation would be more convenient in practice. We are of the opinion that introducing new mathematical structures is always worth the effort, if it allows for the conceptual understanding of physical theories. Therefore we advocated in this thesis the application of infinite dimensional calculus to classical and quantum field theory. We showed that by using these concepts many structures become more natural. For example the consequent treatment of antifields as infinite dimensional vector fields sheds light on the functional analytic aspects of the classical BV complex, but also allows for a systematic treatment of the BV quantization. The main result we obtained is the precise formulation of the renormalized quantum master equation and definition of the quantum BV operator, that doesn't involve an explicit regularization scheme. As an intermediate step, we also proved that the renormalized time-ordered product can be indeed defined as an algebraic product on a certain subspace of the quantum algebra. We showed that the BV quantization can be applied successfully without reference to the path integral approach and can be incorporated as an important tool in algebraic quantum field theory. All the constructions we presented were done in a completely covariant way and could be therefore applied to quantize theories on general globally hyperbolic spacetimes. This is important in the locally covariant approach to classical and quantum field theory. 

The principle of local covariance was the guiding principle throughout the whole thesis and it justifies many of the constructions we proposed. Already at the classical level the formulation in terms of category theory motivates the choice of the subspace of the space of infinite dimensional vector fields, that is identified with the antifields' space. The same is true also for the definition of the (modified) Chevalley-Eilenberg complex \ref{ChEil}. A very clear advantage of the locally covariant framework is an interpretation of the classical master equation as a condition formulated on the level of natural transformation. Another important consequence of   incorporating the local covariance principle is the possibility to treat theories with diffeomorphism invariance, including general relativity. In section \ref{grav} we showed that the natural notion of physical quantities in gravity is provided by the concept of locally covariant fields. This involves again a use of category theory and moves the discussion to a little bit more abstract level.  There is however a major gain coming from this construction. Treating physical fields as natural transformations allows us to recover structures known from classical GR and provides us also with a generalized notion of a symmetry transformation. On this more abstract level also the structure of the BV complex has to be modified, but this modification follows straightforwardly from the framework itself. We proved that the modified BV operator has a nontrivial cohomology which contains all the quantities, that one would intuitively consider as physical, for example scalars constructed covariantly from the metric. This way we can describe the space of physical observables of general relativity as a cohomology of the BV complex extended to the level of natural transformation. 

The motivation for this is of course the perspective for quantum gravity. The BV resolution is an algebraic structure, that can be quantized using the well known methods of deformation quantization. We developed a framework needed for this task in chapter \ref{pAQFT} and showed that it works for the better understood example of the Yang-Mills theory. We expect that the same can be done also for gravity and the resulting theory can be treated as an effective theory  in physical  situations where the quantum gravity effects are small. This thesis is the first step to fulfill this program and we hope that the conceptual understanding of the classical structure, that we gained, will allow us to successfully proceed in our quest to understand better the phenomena happening at the intersection of the quantum world and general relativity.

Although the main motivation for the present work was the possibility to apply it in quantum gravity, the results that we obtained are of interest on their own. The framework we present unifies in a consistent mathematical language three important physical principles: \textit{locality}, \textit{covariance} and \textit{gauge invariance}. We argue, that it provides means to investigate the structure of quantum field theories in the presence of symmetries, using modern mathematical tools. This is important not only from a purely conceptual point of view, since often a fresh look at the underlying structure allows one to simplify also some practical calculations. 
Besides, from the mathematical side, it is interesting, that the infinite dimensional differential geometry, a relatively new research field of mathematics, has also a nice application in physics. The constant exchange of ideas and concepts was always leading to progress both in mathematics and in physics. We hope, that our work will eventually become at least a small stone in this phantastic construction of human knowledge and creativity. With this thought we want to end this thesis, since no book, no paper and no lecture can give justice to the infinite complexity of nature. Even in a very long work, after a finite number of words, there always has to came the moment to put the final dot.
\chapter*{Acknowledgements}
\addcontentsline{toc}{chapter}{\protect\numberline{}Acknowledgements}
First of all I would like to thank my supervisor Professor Klaus Fredenhagen. I am very grateful for all the time and attention he gave to this project. His enlightening comments and suggestions were always helping me out when I had problems and deeply inspiring discussions with him allowed me to learn a lot from his knowledge and experience. Moreover I would like to thank him for his support and optimism, which kept me going even when I experienced difficulties in my research.

I would also like to thank Professor J\"org Teschner for agreeing to be the second referee of my thesis despite of his limited time.

I am very grateful to Pedro Lauridsen Ribeiro for valuable comments and enlightening discussions we had. Especially I would like to thank him for sharing with me and the other members of the group his knowledge on the Nash-Moser-H\"ormander theorem.

My special thanks go to Christoph Wockel for introducing me into the subtleties of infinite dimensional differential geometry. Without his help I wouldn't be able to get to grips with this formalism. At this point I would also like to thank Professor Helge Gl\"ockner for valuable comments on the infinite dimensional calculus and for his hospitality in Paderborn.

I would also like to thank all my other collaborators and friends. I am especially  grateful to Claudio Dappiaggi for his support and interesting scientific discussions. Moreover I would like to thank Wojciech Dybalski, Thomas Hack, Kai Keller, Falk Lindner, Nicola Pinamonti and Jochen Zahn for many valuable comments. In every scientific achievement there is always an enormous contribution from other peoples' ideas, remarks, suggestions or simply encouragement. Therefore I address my deepest thanks to all the members and friends of the Hamburg group with whom I shared my time during this project, especially to:  Dorothea Bahns, Romeo Brunetti, Andreas Degner, Elisabeth Duarte-Monteiro, Paniz Imani, Benjamin Lang, Mathias Makedonski, Christian Pfeifer,  Jan Schlemmer,  Sebastian Schubert, Daniel Siemssen, Kolja Them, Ole Vollertsen, Jan Weise and Mattias Wohlfarth. 

Last but not least I would like to thank my friends from Hamburg, in particular: Marco Argento, Shima Bayesteh, Deanne Litman, Roxana Tarkeshian and Alexander Zeplien.  I want to mention also my fellow students from Cracow, who, like me, decided to do a Ph.D. abroad: Iwona Mochol, Tomek Ku{\l}akowski and Tomek Rembiasz. Thanks for a great time and encouragement!

\selectlanguage{polish} Podziękowania należą się również mojej rodzinie i przyjaciołom z Polski, którzy udzielali mi przez cały czas wsparcia.  W szczególności chciałabym podziękować moim rodzicom, cioci Beacie i wspaniałym przyjaciółkom Kasi i Ejce. Bez Was nie dałabym rady!
\selectlanguage{english}
\cleardoublepage
\pagestyle{empty}
$\qquad$\\
\vspace{2cm}
$\qquad$\\
\begin{flushright}
\begin{minipage}{10cm}
{\large
\textit{What we call the beginning is often the end\\
And to make an end is to make a beginning.\\
The end is where we start from.\\}
$\ $\\
T.S. Eliot, \textit{Four Quartets}}
 \end{minipage}
 \end{flushright}
\printindex

\end{document}